\documentclass[journal,10pt]{IEEEtran}

 \def\BibTeX{{\rm B\kern-.05em{\sc i\kern-.025em b}\kern-.08em
     T\kern-.1667em\lower.7ex\hbox{E}\kern-.125emX}}

  \usepackage{caption} 
 \usepackage[font=small,labelfont=bf]{caption}
\usepackage{amsmath}
\usepackage{algorithm, amsbsy,amsmath,amssymb,epsfig,bbm,mathrsfs, bbm} 
\usepackage{amsthm,bm}
\usepackage{verbatim}
\usepackage[noadjust]{cite}
\usepackage[subfigure]{graphfig} 
\usepackage{amsmath,mathtools}

\usepackage{algorithm}
\usepackage{algpseudocode}
\usepackage{amsmath}
\usepackage{graphics}

\usepackage{epsfig}
\graphicspath{{./figures/}}
\theoremstyle{definition}
\newtheorem{definition}{Definition} 

\newtheorem{theorem}{Theorem}

\newtheorem{corollary}{Corollary}

 \usepackage{epstopdf}
 \usepackage{geometry}
\usepackage{xcolor} 
 \allowdisplaybreaks 
\usepackage{xcolor}

\usepackage{color}
\definecolor{MyDarkBlue}{rgb}{0,0.08,0.45}
\definecolor{yellow}{rgb}{0.99,0.99,0.70}
\definecolor{myback}{RGB}{204,232,207}  
\definecolor{white}{rgb}{1.0,1.0,1.0}                             
\definecolor{black}{rgb}{0.00,0.00,0.00}

\begin{document}
\title{Stochastic Geometry Analysis of Spatial-Temporal Performance in Wireless Networks: A Tutorial } 
\author{Xiao Lu, {\em Member, IEEE}, Mohammad Salehi, Martin Haenggi, {\em Fellow, IEEE}, 
\\
Ekram Hossain, {\em Fellow, IEEE}, and Hai Jiang, {\em Senior Member, IEEE} \thanks{X. Lu and H. Jiang are with the Department of Electrical and Computer Engineering, University of Alberta, Canada.
M. Salehi and E. Hossain are with the Department of Electrical and Computer Engineering, University of Manitoba, Canada. M. Haenggi is with the Department of Electrical Engineering, University of Notre Dame, Notre Dame, USA and with the Department of Information Technology and Electrical Engineering, Swiss Federal Institute of Technology, Zürich (ETHZ), Switzerland. }  }
\maketitle

\begin{abstract}

The performance of wireless networks is fundamentally limited by the aggregate interference, which  depends on the spatial distributions of the interferers, channel conditions, and user traffic patterns (or queueing dynamics). These factors usually exhibit spatial and temporal correlations and thus make the performance of large-scale networks environment-dependent (i.e., dependent on network topology, locations of the blockages, etc.). The correlation can be exploited in protocol designs (e.g., spectrum-, load-, location-, energy-aware resource allocations) to provide efficient wireless services. For this, accurate system-level performance characterization and evaluation with spatial-temporal correlation are required. In this context, stochastic geometry models and random graph techniques have been used to develop analytical frameworks to capture the spatial-temporal interference correlation in large-scale wireless networks. The objective of this article is to provide a tutorial on the stochastic geometry analysis of large-scale wireless networks that captures the spatial-temporal interference correlation (and hence the signal-to-interference ratio (SIR) correlation). We first discuss the importance of spatial-temporal performance analysis, different parameters affecting the spatial-temporal correlation in the SIR, and the different performance metrics for spatial-temporal analysis.  Then we describe the methodologies to characterize spatial-temporal SIR correlations for different network configurations (independent, attractive, repulsive configurations), shadowing scenarios, user locations, queueing behavior, relaying, retransmission, and mobility. We conclude by outlining future research directions in the context of spatial-temporal analysis of emerging wireless communications scenarios.

\end{abstract}


\section{Introduction}

Wireless communications systems are evolving toward a heterogeneous architecture (e.g., multi-tier and cell-free) with the dense deployment of different types of access points (e.g., smallcells and hotspots) to enable pervasive wireless Internet access \cite{E.2014Hossain}. The evolving wireless networks are expected to provide seamless connectivity to ubiquitous and/or high-mobility devices and users with millisecond delay and gigabits per second data rate \cite{M.2015Peng}. The ever-increasing demand for low-latency high-reliability services from pervasive terminals will lead to an explosive increase in mobile traffic. To accommodate the massive traffic volume, high network densification and aggressive spatial frequency reuse  will be required,  which will result in high levels of interference in the network.

\subsection{Background of Stochastic Geometry and Objective}
Signal propagation over a radio link is impaired by large-scale path loss and shadowing, small-scale fading,  as well as co-channel interference from concurrent transmissions.
Since all of these effects are heavily location-dependent, the network spatial configurations become a dominant factor that determines the system-level performance.
Hence, developing tractable approaches for modeling large-scale wireless systems and analyzing their statistical performance taking into account the randomness (due to the above-mentioned factors) have become compelling.  

In this context,  
stochastic geometry \cite{M.2013Haenggic} (also referred to as geometric probability), a probabilistic analytical approach to study (random) point configurations, 
has become a necessary theoretical tool for the analysis and characterization of large-scale wireless systems, including 
heterogeneous cellular networks  \cite{S.2012Dhillon,L.2018Wang,Wei2016H}, 
dynamic spectrum access systems \cite{Lee2012C,A.2015H}, wireless ad hoc networks \cite{K.2009Ganti,A.2007Hasan,M.2008Hunter,F.2009Baccelli2}, drone networks  
  \cite{Ravi2016Vishnu,M.2020Azari,Lahmeri2019Mohamed-Amine}, vehicular networks~\cite{Jeya2020Pradha,V.2020Chetlur,Konstantinos2020Koufos}, and low earth orbit satellite networks \cite{N.2020Okati,Talgat2020Anna}. 
 Spatial-temporal aspects of mobile communications, including the spatial distribution of network nodes, wireless channels and traffic patterns have to be considered for system development, resource allocation, performance evaluation and optimization. The purpose of this paper is to provide a tutorial on how to quantitatively analyze the effects of  spatial and temporal fluctuations of interference (resulting from the factors above) on the system-level network performance.

\begin{figure} 
\centering
 \subfigure [ Example 1: Temporal interference correlation in a static network. ]
  {
 \centering   
 \includegraphics[width=0.48 \textwidth]{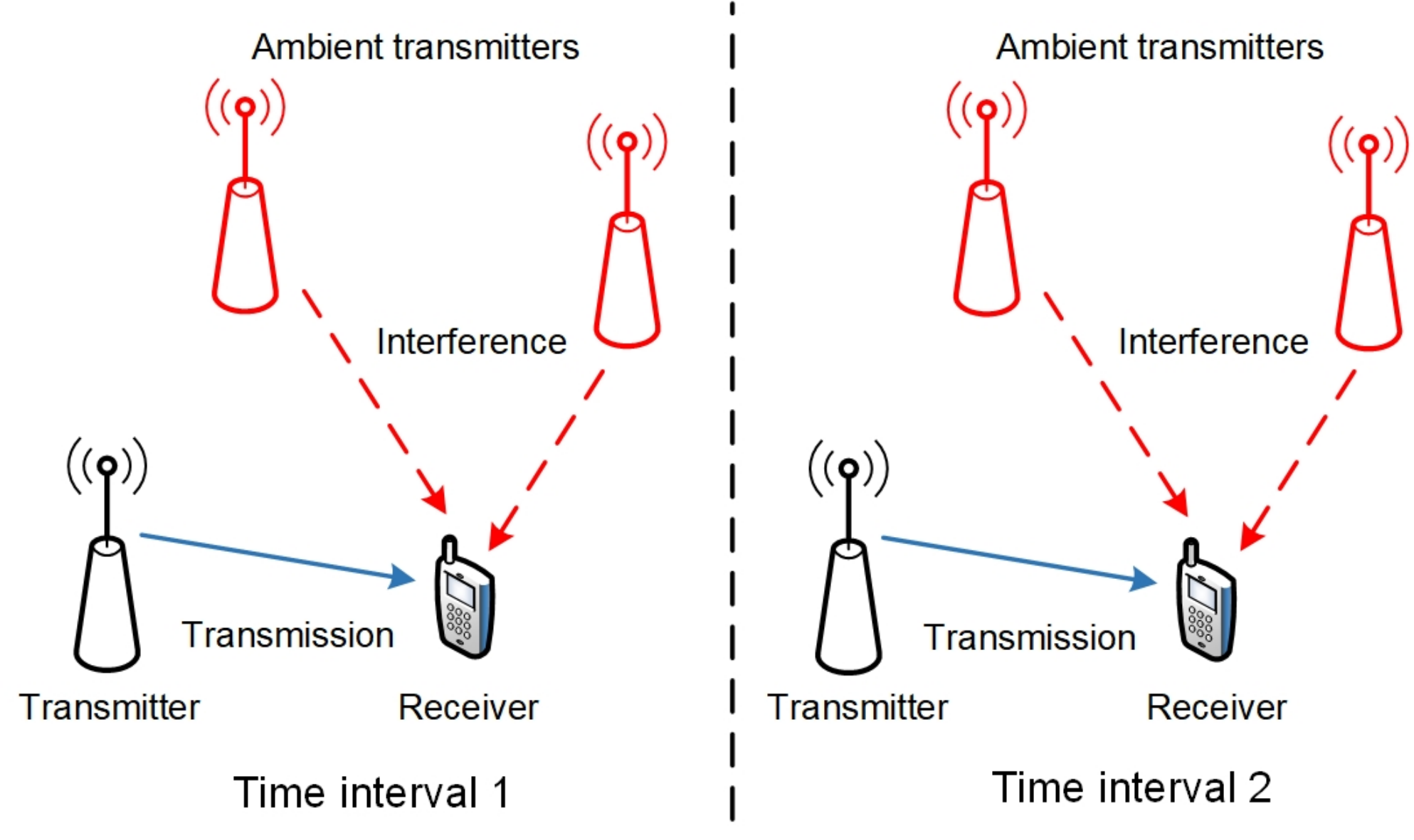}}\label{fig:example1}
 \centering  
 \subfigure  [ Example 2: Spatial interference correlation in a full-duplex communication system.
 ] {
 \centering
\includegraphics[width=0.3 \textwidth]{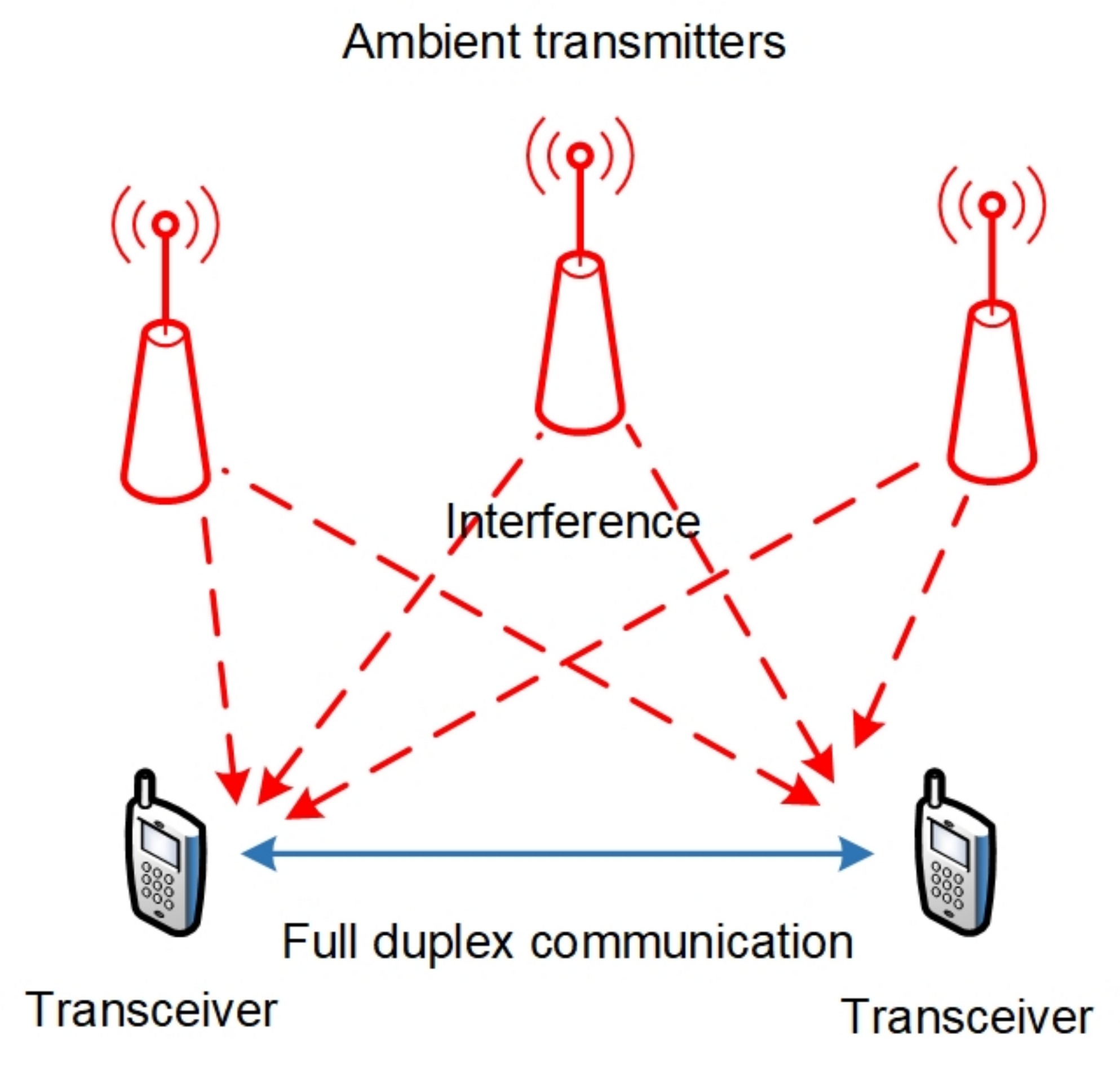}}
 \centering  
 \subfigure  [ Example 3: Spatial-temporal interference correlation in a mobile network. 
 ] {
 \centering
\includegraphics[width=0.48 \textwidth]{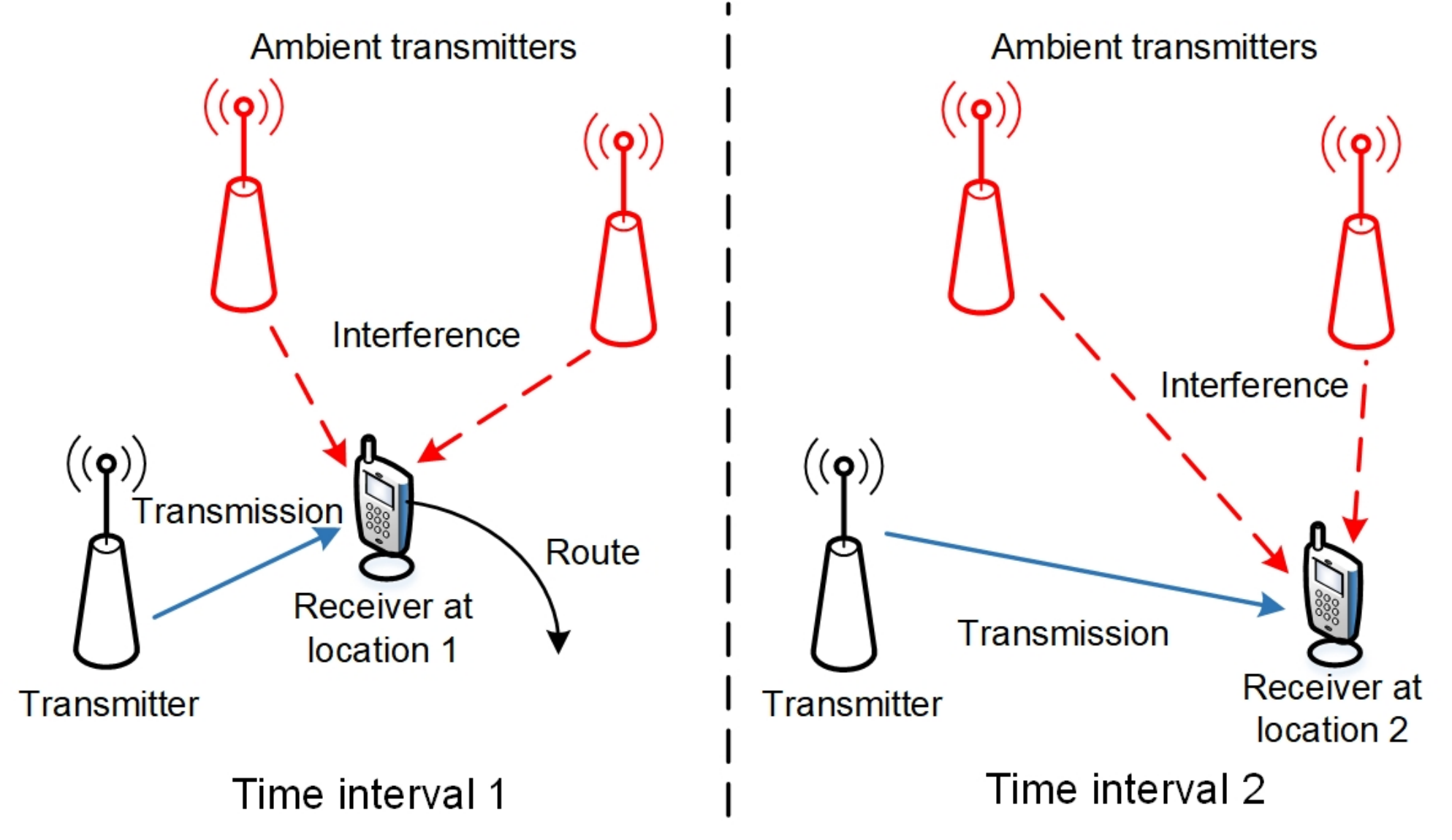}}
\caption{Examples of interference correlation.} 
\centering
\label{fig:example}
\end{figure}

\subsection{Importance of  Characterization of
Signal-to-Interference-plus-Noise Ratio Correlation}

Wireless communications systems need to preserve the quality of radio links in time-varying environments. The  signal-to-interference-plus-noise ratio (SINR)  statistics is commonly used  as the main measure of the quality of links~\cite{A.2005Goldsmith}, and most of the performance metrics for system-level evaluation (to be introduced in Section \ref{sec:PM}) are based on the SINR.
The variation of the wireless environment (and hence the SINR) is mainly attributed to two causes. On the one hand, the changes of the relative positions of communication devices and surrounding obstructions affect the multipath propagation and thus the received power of both desirable and interfering signals. Also, variations in traffic patterns cause fluctuations in the interference and hence the SINR.

Owing to the spatial-temporal fluctuations of network distributions, wireless channels and traffic patterns, the interferences and hence SINRs (at different locations and time instants) are correlated. 
Although rapid channel fluctuations due to small-scale fading can result in reduced interference correlation~\cite{R.2009Ganti}, the 
correlation still exists due to large-scale 
path-loss and shadowing~\cite{U2011Schilcher}. For example, in a static network, as shown in Fig.~\ref{fig:example}(a), if the ambient transmitters have data packets to send during two time intervals, the interferences at the receiver 
are temporally correlated. Besides, in a full-duplex communication system as shown in Fig.~\ref{fig:example}(b), the interferences at the transceiver pair are spatially correlated at any time instance. In a mobile network, as shown in Fig.~\ref{fig:example}(c), the interference at the receiver is spatially and temporally correlated. 
The causes of the temporal and/or spatial interference correlation in the above examples all arise from the fact that the interference comes from the same group of transmitters.  
The interference correlation results in the spatial-temporal correlation of transmission outage/success \cite{Z.2014Gong,K.2018Koufos,M.Aug.2017Gharbieh}, throughput \cite{H.Dec.2017Yang,J.Sep.2018Li},  
mean local delay~\cite{M.2013Haenggib,Y.Jun.2017Zhong}, etc., thus needs to be treated carefully in the designs of mobile systems. 

Since SINR correlation affects the performance at different transmission attempts (e.g., when using an error recovery method) and/or locations (e.g., in a relay-based system), an accurate characterization of 
it is essential to the understanding of wireless network performance. 
Information about SINR correlation can be exploited 
to optimize the performance and design of the system accordingly.

\subsection{Related Work}

Several survey and tutorial papers have focused on the stochastic geometry analysis of wireless communication networks.  
In particular, reference \cite{S.2009Zuyev} provides a survey of point process models and stochastic geometry tools that have been used to analyze static wired, wireless, ad hoc and cellular networks prior to 2009.
Reference \cite{JGAndrews2010} overviews the impact of spatial modeling on the SINR-based performance metrics, i.e.,  connectivity, coverage area, and capacity of different types of systems, including ad hoc, cellular and cognitive networks. 
The survey in \cite{ElSawy2013H} comprehensively reviews the  
works on stochastic geometry analysis of multi-tier and cognitive cellular systems prior to 2013.  
   
In addition to the above survey papers, tutorial papers on the mathematical tools used for stochastic geometry analysis of
 large-scale systems have also been written. Reference \cite{M.Haenggi2009} is the first tutorial on point process theory, random geometry graphs, and percolation theory for interference characterizations in ad hoc networks. 
Targeting cellular networks,  \cite{J.G.Andrews2016} provides a tutorial on stochastic geometry analysis of both  downlink and uplink networks based on Poisson point process (PPP) modeling.  The focus is on characterizing interference in different scenarios by exploiting the properties of the PPP under Rayleigh fading assumption. 
As PPP modeling fails to capture the spatial correlation among the random points, the authors in \cite{N.Miyoshi2016} emphasize the use of repulsive point processes to model  cellular networks, 
where the base station locations are usually planned with a moderate degree of irregularity due to different development issues. To this end, the authors present a tutorial on the SINR distribution analysis of downlink cellular networks based on the $\beta$-GPP (Ginibre point process), which is a fairly tractable model for random points with spatial repulsion. The tutorial in \cite{ElSawy2016Hesham} focuses on a unifying analysis of bit/symbol error probability, coverage outage probability, and ergodic capacity in cellular networks. More recently, reference~\cite{Y.2021Hmamouche} systematically tutors the analytical techniques to characterize interference, success probability and capacity in Poisson networks. 
\textbf{However, none of the existing survey and tutorial papers focus on the stochastic geometry techniques to 
characterize the  spatial and temporal correlations in their considered systems}. Moreover, 
this is the first article to include the refined-grained analysis tool of the signal-to-interference ratio (SIR) meta distribution.

\begin{figure} 
 \centering
 \includegraphics[width=0.5 \textwidth]{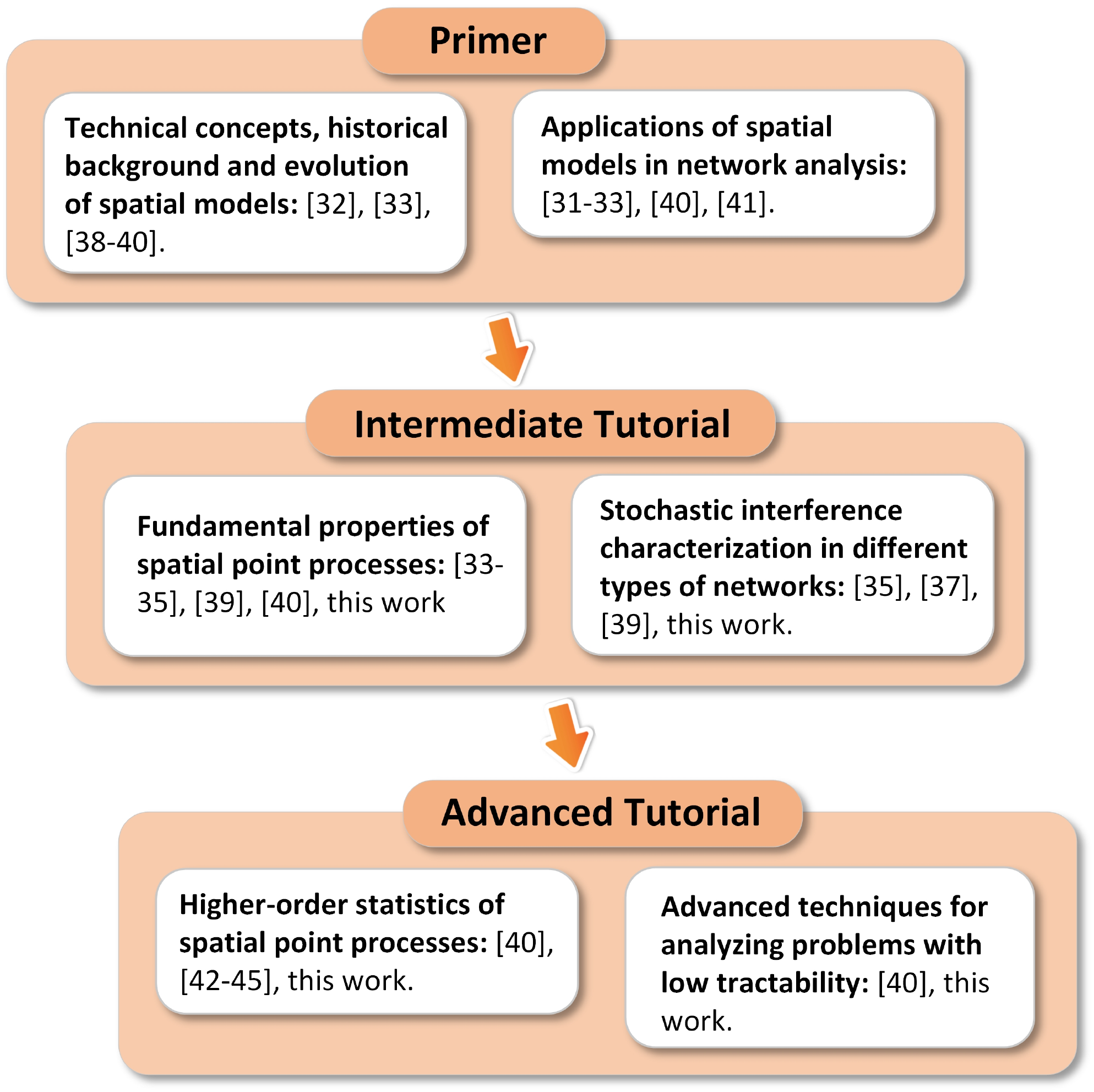} 
 \caption{The suggested learning path for stochastic geometry analysis in wireless networks.   } \label{fig:path} 
 \end{figure}

For a better understanding of the contents of this tutorial, 
we recommend taking some prerequisite knowledge of stochastic geometry regarding network analysis which has been comprehensively reviewed in the literature. The suggested learning path and the recommended readings corresponding to the prerequisite knowledge are shown in Fig. \ref{fig:path}, which includes the following. 


\begin{itemize} 

\item  Primers~\cite{JGAndrews2010, M.Haenggi2009,P.2010Cardieri,lattice,Y.2021Hmamouche,S.2009Zuyev,ElSawy2013H, J.Mar2013G} that cover the technical concepts of spatial models 
and their applications in modeling wireless networks;


\item Intermediate tutorials~\cite{ElSawy2013H,M.Haenggi2009,J.G.Andrews2016,lattice,Y.2021Hmamouche,ElSawy2016Hesham}  that tutor the fundamental properties of spatial point processes 
(e.g., counting measure,  superposition, thinning, and transforms) and analytical methods to characterize the interference distribution in different types of wireless networks 
(e.g., cellular and ad hoc networks);


\item Advanced tutorials~\cite{Y.2021Hmamouche,U.2016Schilcher,Yu2019Xinlei,M.Part1Haenggi,M.Part2Haenggi} that lays out the higher-order statistics of point processes for fine-grained analysis 
(e.g., product/joint meta distribution) and up-to-date analytical techniques (e.g., approximation and bounding methods) to analyze complicated scenarios (e.g., with spatial-temporal correlation) where the exact performance characterization is of low tractability or unavailable. 

 
 
\end{itemize}

 
\subsection{Contributions and Organization}
 
This tutorial aims to concisely present the analytical approaches for performance evaluation of large-scale wireless networks taking into account various correlation causes and effects, such as shadowing, traffic queueing, and spatial distribution of network nodes.  We focus on the characterization of the interference distribution and signal-to-interference ratio (SIR)-based performance  metrics\footnote{Since in large-scale cellular networks, the impact of aggregate interference typically dominates that of noise~\cite{S.2012Dhillon,G.2014Nigam}, 
this tutorial focuses on analyzing the interference-limited cases with the noise ignored. However, without loss of generality, the same analytical approaches can also be applied to characterize SINR-based performance metrics.}, e.g., {\em success probability}, {\em joint success probability}, and {\em moments of conditional success probability (CSP) given the point process}. 
For spatial point process models, we additionally derive the {\em asymptotic SIR gain}, which is the horizontal gap between a target SIR distribution and a reference SIR distribution. This metric directly reflects the variation of SIR due to the changes in the network model with respect to (w.r.t.) the reference model. Furthermore, this metric can be utilized to simplify the analysis of non-Poisson networks based on the PPP \cite{M.2014Haenggi} (to be introduced in Section~\ref{sec:spatial_PA}).

This tutorial considers the following correlation effects in wireless networks.

\begin{itemize}

\item The spatial correlation (i.e., attraction and repulsion) among the locations of the transmitters;    

\item The spatially-correlated shadowing experienced by the links that traverse common obstacles (e.g., buildings);
 
\item The spatially and temporally correlated queue status among the transmitters due to the cross-interference imposed on each other over space and time;

\item The spatial correlation experienced by users located in the cell-center and cell-boundary regions; 

\item The temporal interference correlation between multiple transmission attempts due to the correlation in the locations of the interferers over time; 

\item The spatial-temporal interference correlation among relaying nodes in a multihop network due to the correlation in the locations of the interferers over space and time;

\item The spatial-temporal interference correlation in mobile systems (i.e., where the users and/or base stations (BSs) are mobile) due to the correlation among interferers' locations over space and time. 

\end{itemize}

\begin{table*} 
\captionsetup{font=footnotesize}
\centering
\caption{\footnotesize Main differences between other survey and tutorial papers}  \label{Tab:difference} 
\begin{tabular}{|p{1.2cm}|p{1.cm}|p{0.9cm}|p{0.9cm}|p{0.9cm}|p{0.8cm}|l|p{0.9cm}|p{3.5cm}|p{1.3cm}|}
\hline
Reference (year) &  Type of Review & \multicolumn{4}{ c| }{Examined systems}  &  \multicolumn{2}{ c| }{Spatial models}   &  Target performance metrics  &   Spatial-temporal correlation analysis \\
\cline{3-8}
& & Ad hoc & Cellular & Multihop & Mobile & Poisson &  Non-Poisson  & & \\
\hline
\hline
\cite{S.2009Zuyev} (2009) & Survey & \checkmark 
&  \checkmark &    &   & \checkmark &
 & Success probability, paging, handover   &  No \\
\hline
\cite{M.Haenggi2009} (2009) & Tutorial &  
\checkmark &  &  & & \checkmark & &  Interference characterization, outage probability, capacity, and area spectral efficiency &  No  \\
\hline
\cite{JGAndrews2010} (2010) & Survey & 
\checkmark & \checkmark & & &  \checkmark & & Success probability, coverage area, capacity &  No \\
\hline
\cite{ElSawy2013H} (2013)  & Survey & 
\checkmark & \checkmark & & & & & Success probability, capacity  &  No \\
\hline
\cite{J.G.Andrews2016} (2016) &  Tutorial  &  
& \checkmark  &   & & \checkmark  & & Success probability &  No \\
\hline
\cite{N.Miyoshi2016} (2016) & Tutorial & 
& \checkmark  &   & &   & \checkmark  & Success probability &  No \\
\hline
\cite{ElSawy2016Hesham} (2017) &  Tutorial &  
& \checkmark & \checkmark  & & \checkmark &  & Interference characterization,  error probability, error rate, outage probability,  capacity, and handover &  No \\
\hline  
\cite{Y.2021Hmamouche} (2021) & Tutorial &   & \checkmark &   &   & \checkmark & \checkmark & Interference characterization, success probability, capacity  & No \\ 
\hline
This work & Tutorial  & 
\checkmark  & \checkmark  & \checkmark  & \checkmark  & \checkmark  & \checkmark & Success probability, joint success probabilities, conditional success probabilities, moments of conditional success probability given the point process, SIR meta distribution, SIR gain, interference correlation coefficient & Yes \\
\hline 
\end{tabular} 
\end{table*}

This tutorial presents the methodologies of analyzing the interference (and SIR) correlation and the effects of spatial-temporal SIR correlations. For a thorough exposition of the analytical techniques, we demonstrate the step-by-step derivations as well as numerical results\footnote{Note that the analytical methodologies introduced in this tutorial come from the referenced literature. The analytical results are either well-established or direct extensions of the ones derived in the references.}. The main differences between our tutorial
and the state-of-the-art discussed in the previous subsection are summarized in Table~\ref{Tab:difference}\footnote{Herein, success probability refers to the complementary cumulative distribution function (CDF) of the SINR/SIR.}. Herein, Poisson spatial models refer to the PPP 
and binomial point process, while non-Poisson spatial models refer to any point process whose points are not independently distributed. 

The organization of this tutorial and the relations among 
different sections are shown in Fig. \ref{fig:outlines}. We restrict the point processes to the Euclidean spaces $\mathbb{R}$ and $\mathbb{R}^2$. However, the same methodologies can be applied to analyze point processes in higher dimensions without loss of generality. Section~I introduces the role of stochastic geometry analysis and highlights the importance of spatial-temporal correlation characterization. Section II explains the correlation effects of different network factors that may affect the SIR-based network performance. Sections III-VI present exact methodologies to analyze network performance with correlations in node distribution, 
link distance distribution, shadowing, and queueing, respectively. Moreover, Sections VII and VIII present a performance characterization with multiple transmission attempts (i.e., retransmission) and multihop relaying, respectively, for correlated and independent interference. Section IX characterizes the spatial-temporal performance with mobility. Future directions and research challenges are then discussed in Section X followed by the conclusion in Section XI. Additionally, for convenience, we list the abbreviations used in Table \ref{tab:abb}.

\begin{table} 
\begin{center}
\caption{List of abbreviations} 
    \begin{tabular}{| l | p{5cm}  | }
    \hline
    Abbreviation & Description \\ \hline 
    \hline
PPP &  Poisson point process     \\ \hline 
PLP &  Poisson line process     \\ \hline 
BPP & Binomial point process  \\ \hline
MCP & Mat\'{e}rn cluster process \\ \hline
GPP &  Ginibre point process     \\ \hline 
RDP & Relative distance process \\ \hline
SINR &  Signal-to-interference-plus-noise ratio     \\ \hline    
SIR &  Signal-to-interference ratio     \\ \hline 
CSP  & Conditional success probability  \\ \hline
JSP & Joint success probability  \\ \hline
MISR & Mean interference-to-signal ratio \\ \hline
LSU  & Location-specific user \\ \hline 
ASAPPP & Approximate SIR analysis based on the PPP  \\ \hline 
MIMO & Multiple-input and multiple-output \\ \hline 
PDF &  Probability density function \\ \hline
CDF  & Cumulative distribution function \\ \hline
PCF & Pair correlation function \\ \hline
PGFL & Probability generating functional \\ \hline  
MAC &  Medium access control \\ \hline
CSMA & Carrier-sense multiple access \\ \hline
BS &  Base station \\ \hline
DF  & Decode-and-forward \\ \hline
HARQ & Hybrid automatic repeat request   \\ \hline 
QSI & Quasi-static interference  \\ \hline 
FVI  & Fast-varying interference \\ \hline    
IoT   &  Internet of Things \\ \hline    
IoNT & Internet of Nanothings  \\ \hline   
IoST  &     Internet of Space Things \\ \hline
ML & Machine learning \\ \hline
    \end{tabular}
    \label{tab:abb}
\end{center}
\end{table}

{\em Notations}: The notations defined in Table \ref{notation} are used throughout this tutorial. 
 
 \begin{figure*} [htp] 
 \hspace{10mm} 
 \includegraphics[width=0.9 \textwidth]{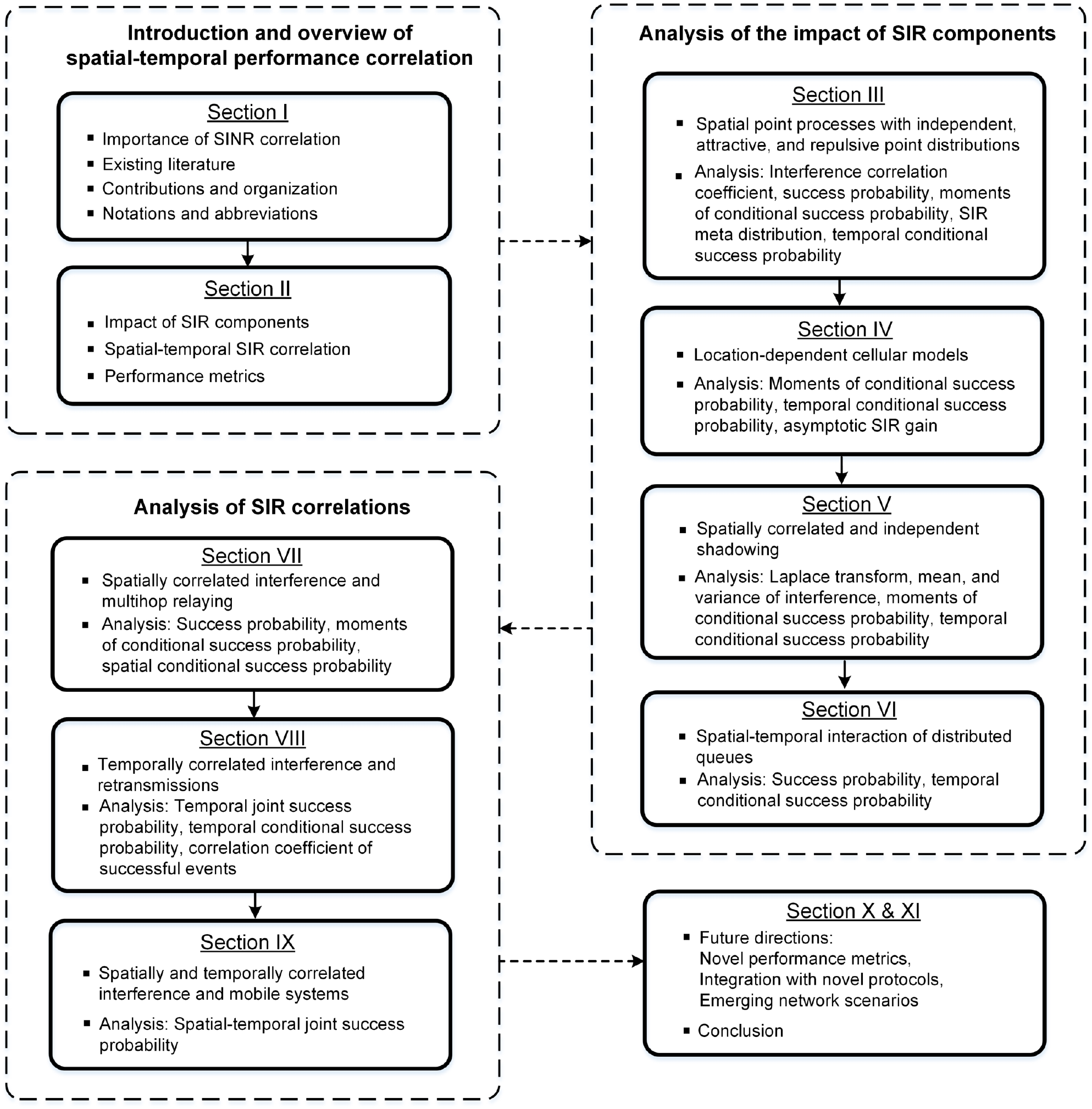} 
 \caption{Organization of this paper. } \label{fig:outlines} 
 \vspace{-3mm} 
 \end{figure*}

\begin{table*} [htp]
\centering
\caption{ Notations}  \label{notation} 
\begin{tabular}{|l|p{12cm}|}
\hline
Symbol & Definition\\ \hline
\hline
$\jmath = \sqrt{-1}$  &  The imaginary
unit \\
$\equiv$  &  Equivalence relation \\
$\mathbb{N}$,   $\mathbb{R}$, $\mathbb{R}^{+}$ &  Natural numbers,  real numbers, and positive real numbers, respectively  \\
$\mathbb{R}^{d}$ & $d$-dimensional Euclidean space \\
$o$ & The origin of $\mathbb{R}^{d}$ \\
$\mathbb{B}(x,R)$ & Disk of radius $R$ centered at $x$ \\
$\Phi$ &  Point process representing the nodes in the network \\
$\Phi^{!}$ & The set of interferers \\
$\Phi(A)$  &  
The number of elements in $\Phi \cap A$
\\
$\triangleq $ &   The definition operator \\
$\overset{(\rm d)}{=}$ & Equivalence in distribution \\
$\mathbbm{1}_{\{\cdot\}}$  & Indicator function which equals 1 and 0 if
the statement $\{\cdot \}$ is true and false, respectively \\
$\mathbb{E}[\cdot]$ & Expectation operator \\
$\mathbb{E}^{!}_{x}[\cdot]$ 
 &  Expectation operator 
 w.r.t. the reduced Palm measure at $x$ \\
$\mathbb{P}[\cdot]$ & Probability measure  \\
 $\mathbb{P}^{!}_{x}[\cdot]$ &  Reduced Palm probability measure at $x$ \\
$\mathbb{V}[\cdot]$ & Variance   \\ 
$|\cdot|$ & Modulus operator \\
$\|\cdot\|$ & Euclidean norm \\  
$\Gamma(\cdot)$ & Gamma function \\
$\gamma(s,x)$ & Lower incomplete gamma function, i.e., $\gamma(s,x)=\int^{x}_{0} u^{s-1} e^{-u} \mathrm{d}u$ \\
$\mathcal{E}(a)$ & Exponential distribution with rate parameter $a$   \\
$\mathcal{P}(a)$ & Poisson distribution with rate parameter $a$   \\
$\mathcal{G}(a,b)$ & Gamma distribution with shape parameter $a$ and scale parameter $b$  \\
$\mathrm{Det}$ &  Determinant operator \\
${_2}F_{1}(.,.;.; z)$  & The
Gauss hypergeometric function
\\
$f_{X}(\cdot)$,  $F_{X}(\cdot)$, $\mathcal{L}_{X}(\cdot)$ & The probability density function (PDF), the cumulative distribution function, and Laplace transform of random variable $X$.  \\
$\mathcal{W}(\cdot)$ &  The Lambert-$\mathcal{W}$ function, i.e., the inverse function of $f(x) = x e^{x}$ \\ \hline
\end{tabular} 
\end{table*}  
\vspace{2mm}

\begin{table*}[htp]
\caption{SIR components}  \vspace{-2mm}
\begin{center}  \label{SINR}   \footnotesize
\begin{tabular}{ | p{2cm} | p{5cm}| p{8cm} | } 
\hline
Component & Representation  & Characterization  \\ 
\hline
\hline
 $\Phi$   &  Spatial distribution of interferers &  Spatial point process models  (in Section III) \\ 
\hline
 $r_{1}=\|x_{1}\|$ & Contact distance distribution   &  Location-dependent analysis of cellular models (in Section IV)  \\ 
\hline
$S_{x}$  & Blockage &  Shadowing models (in Section V) \\ 
\hline
 $\iota_{x}$   &  Buffer status   &   Queueing models (in Section VI)  \\ 
\hline
\end{tabular} 
\end{center}
\end{table*}

\section{Overview of Spatial-Temporal Performance Correlation}

 This section first discusses the impact of SIR components and then elaborates on the SIR correlation effects in wireless networks and defines the performance metrics for spatial-temporal performance analysis.

\subsection{Parameters Impacting SIR Correlation}
 
Let $o$ denote the origin of $\mathbb{R}^{d}$. 
The SIR at a target receiver located at $o$ expressed as  
\begin{align} 
\textup{SIR} =   \frac{ P_{y} h_{y} S_{y} \ell(\| y  \|) }{ \sum_{x \in  \Phi^{!} } \iota_{x} P_{x} h_{x} S_{x} \ell(\| x \|)  },
\end{align} 
where $\Phi^{!}$ denotes the set of all the interfering transmitters, $y$ the location of the transmitter associated with the target receiver, $P_{x}$ the transmit power of transmitter at $x$, $h_{x}$ ($S_{x}$) the fading (shadowing) coefficient 
between the transmitter located at $x$ and receiver located at $o$,
$\ell$ the path loss function, and $\iota_{x}$ the state indicator of the transmitter located at $x$ which equals 1 and 0 when the transmitter is on and off, respectively. The physical implications of the SIR components are shown in Table~\ref{SINR}.   

The SIR varies due to the fading and shadowing effects and system factors such as timing-varying traffic loads and positions of transmitters and receivers. 
Their impacts are discussed below. 
  
\subsubsection{Spatial Distribution}

The spatial distribution of the network nodes can be categorized into three types: independent, repulsive, and attractive.

 \begin{figure*}
 \centering
 \includegraphics[width=0.75\textwidth]{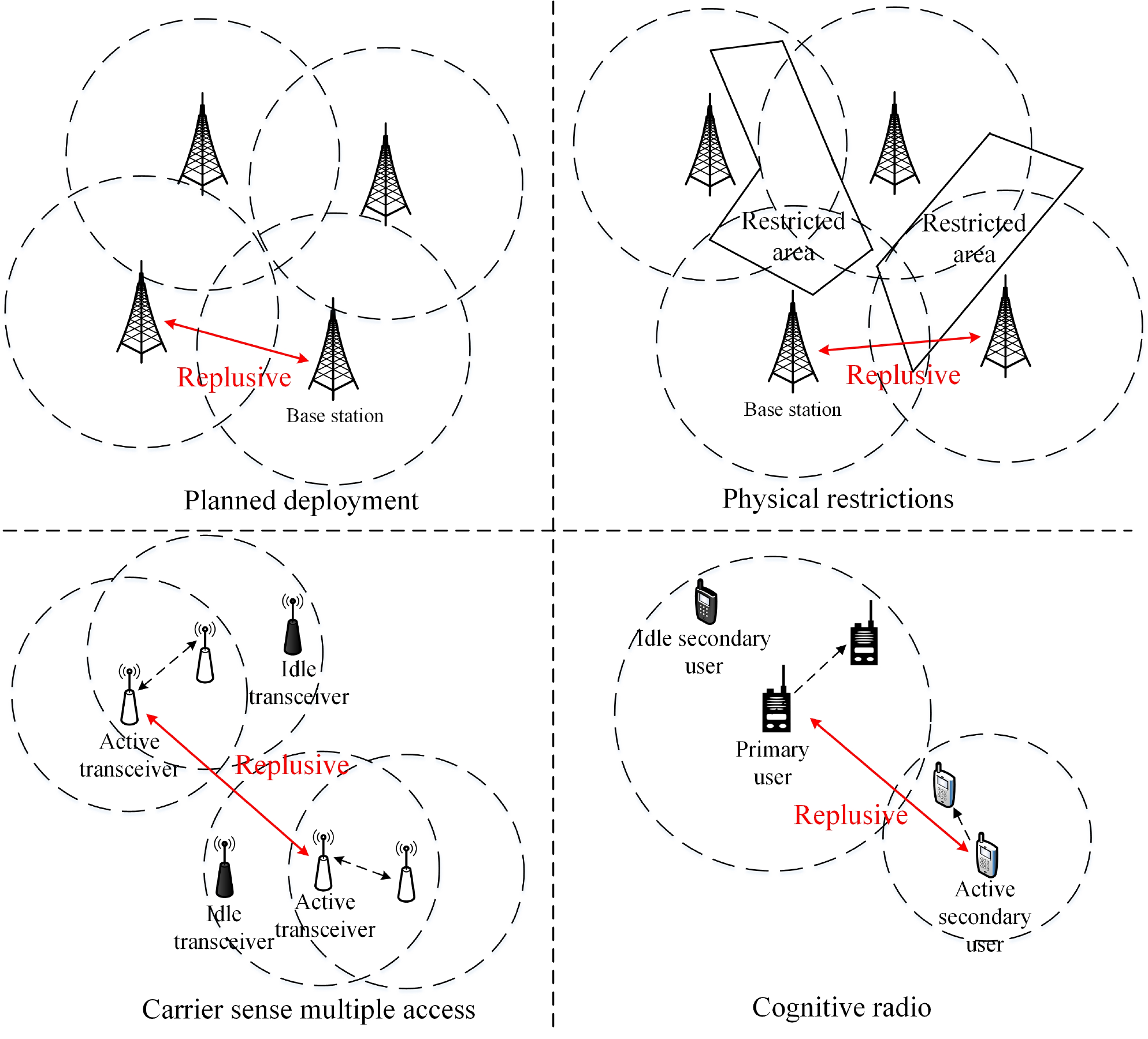} 
 \caption{Examples of wireless systems with spatially repulsive node locations.   } \label{fig:replusion} 
 \end{figure*}

\begin{itemize}

\item With independent distribution, the locations of the transmitters are independent of each other. 
Spatial point process models to characterize independent distributions include PPP   
and binomial point process (BPP) \cite{Afshang2017Mehrnaz} models.

\item With repulsive distribution, the wireless transmitters transmitting simultaneously are not too close to each other. 
The  repulsion  may arise from planned deployment, physical restrictions (e.g., geographic exclusion and terrain occlusion) and channel access control (e.g., carrier sense multiple access (CSMA) \cite{M2011Haenggi,H.2013ElSawy,G.2014Alfano,U.2020Schilcher2} and licensed-user activity \cite{Lee2012C,Z.2016Yazdanshenasan})  as shown in Fig.~\ref{fig:replusion}.
Lattice processes \cite{lattice}, Ginibre point processes \cite{N.Jan.2015Deng,N.Miyoshi2016,X.2016Lu}, and Mat\'{e}rn hardcore point processes \cite{M2011Haenggi,A.2016Al-Hourani} are some well-known examples of point process models to characterize this repulsive behavior.

\item  Attractive distributions can be observed when wireless transmitters are only clustered in certain regions, i.e., not identically distributed over the entire plane, as illustrated in Fig.~\ref{fig:cluster}. This can be caused by the base station (BS)-centric user gathering~\cite{D.2016Mankar}, (e.g., around open-access WiFi spots) or due to user-centric BS deployment \cite{C.2017Saha}. The Mat\'{e}rn cluster process \cite{C2017Saha,Azimi-Abarghouyi2018}, Thomas cluster process~\cite{D.2016Mankar}, 
Gauss-Poisson process~\cite{A2016Guo,N2018Deng},  
Cox process~\cite{Y2013Jeong,Chetlur2018,Jeya2020Pradha}, and the Poisson hole process \cite{Lee2012C,Z.2016Yazdanshenasan}
are some representative point processes to model attractive spatial distributions. 

\end{itemize}

Additionally, a point process can be a mix of the above three types. Whether a point process is of any type may depend on the distance between the locations considered. For example, the type I user point process introduced in \cite{M.2017Haenggi} (for modeling user distribution in cellular networks) is repulsive at short distances, attractive at intermediate distances and eventually approaches independence at larger distances.

 \begin{figure}
 \centering
 \includegraphics[width=0.40\textwidth]{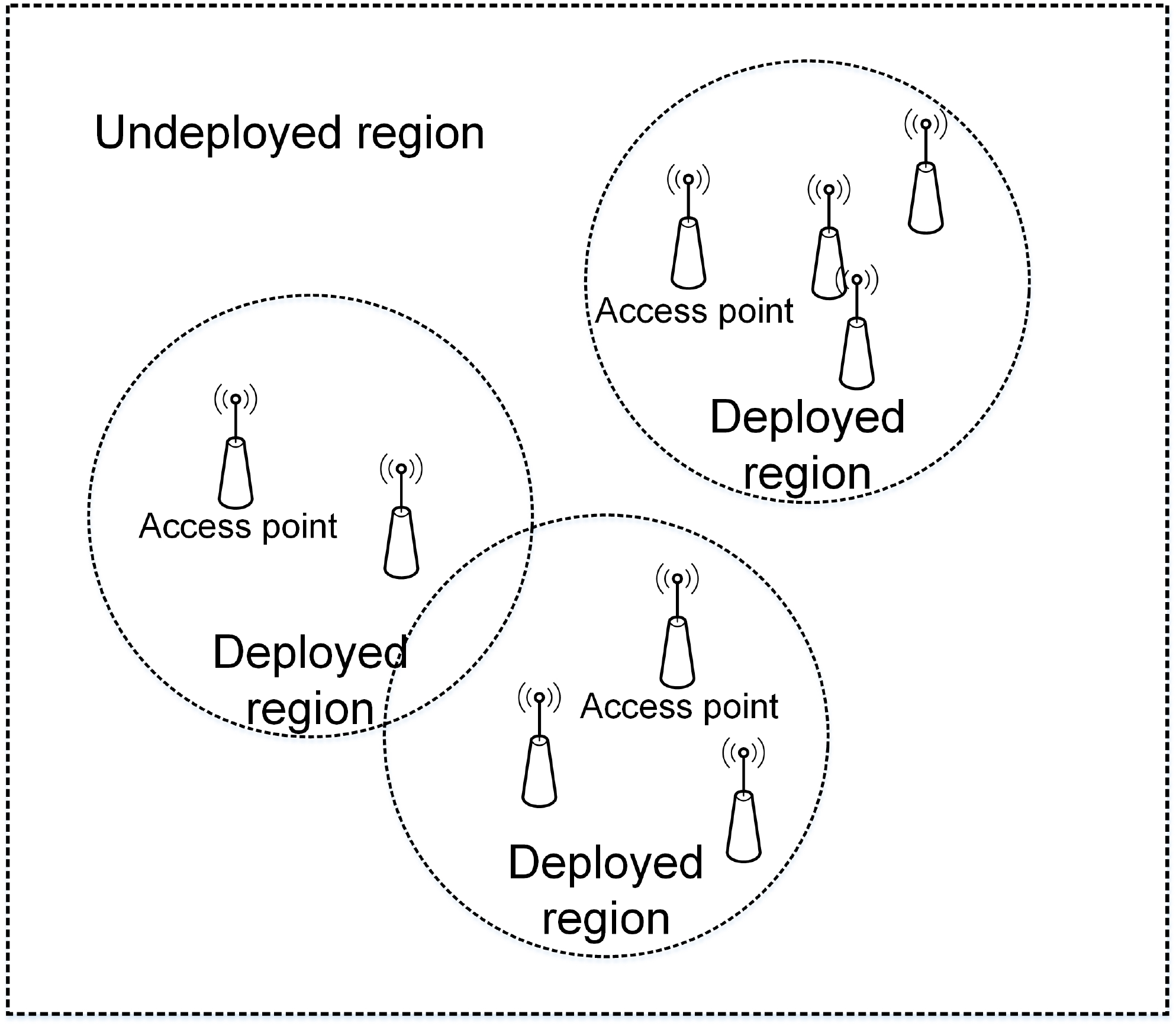} 
 \caption{Illustration of a wireless system with spatially attractive deployment.   } \label{fig:cluster} 
 \end{figure}

For network performance analysis, a significant body of literature adopts the independent distribution assumption for simplicity and tractability (see \cite{ElSawy2013H} and references therein). However, this assumption is too idealized to hold in practice as most of the cellular networks are deployed under system-level planning~\cite{G.2011Andrews}. Due to factors such as geographical restriction, access control and resource allocation, active network nodes can exhibit a spatial pattern 
with a degree of correlation (i.e., repulsion or attraction) which cannot be ignored in performance characterization.

\subsubsection{Locations of Receivers}
In cellular networks, the transmission performance of a mobile user is highly dependent on its location w.r.t. the serving BS and interfering BSs 
\cite{D.2010Gesbert}. 
Specifically, for a cell-edge user, the desired signal is  weaker and the interference signal is stronger compared to those for a cell-center user (Fig. \ref{fig:location_dependent}). 
Different techniques have been developed to strengthen the desirable signals and/or to mitigate the interference of cell-edge users, e.g., through coordinated beamforming \cite{H2017Yang} 
and intercell interference coordination \cite{J.2010Zhang,ICIC}. With these techniques, resource allocation  strongly depends on the spatial variation of mobile users. Therefore, location-dependent modeling of user performance is fundamental to the understanding of spatial-temporal performance of these systems. 
 
 \begin{figure}
 \centering
 \includegraphics[width=0.5\textwidth]{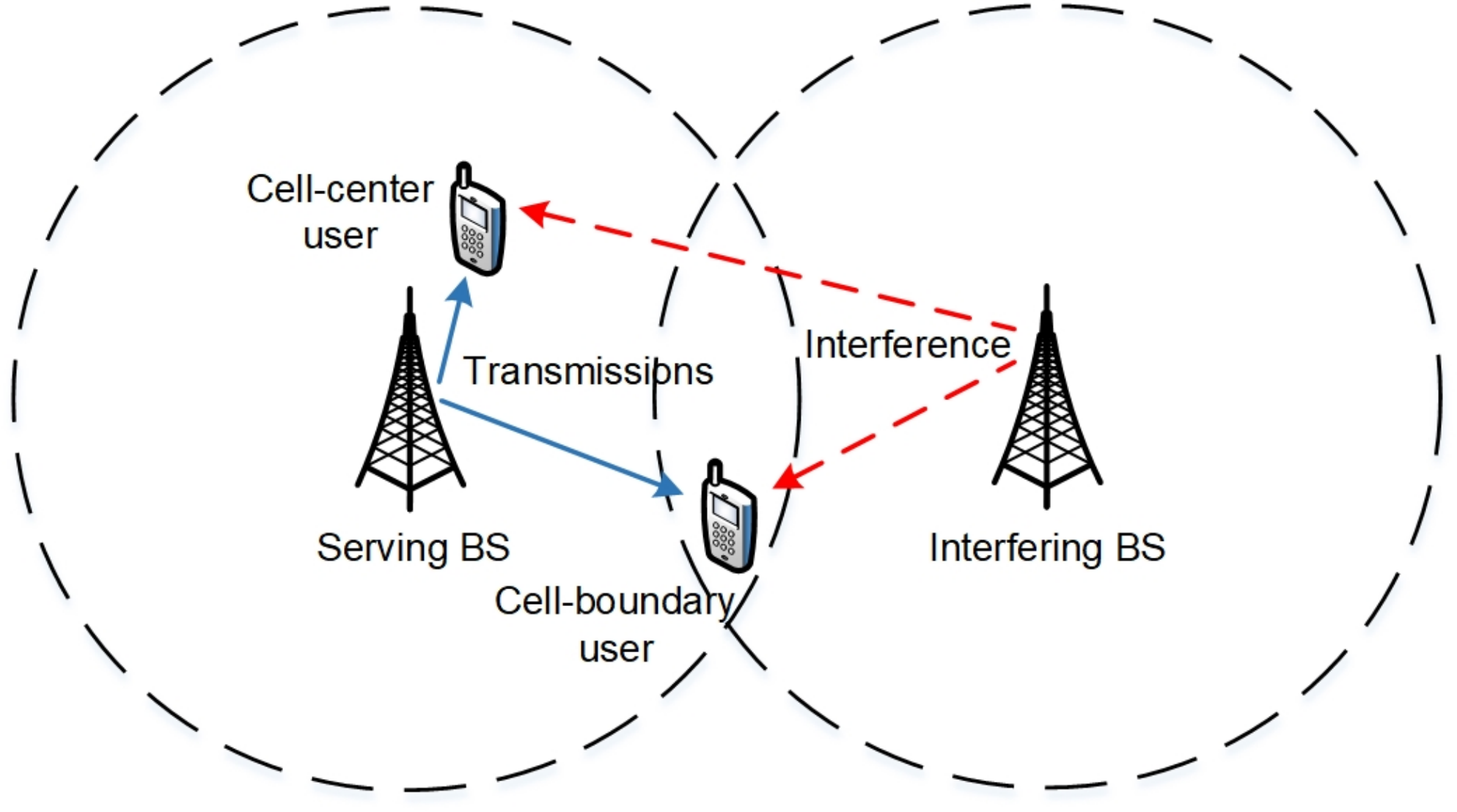} 
 \caption{Illustration of location-specific users.   } \label{fig:location_dependent} 
 \end{figure}

\subsubsection{Small-Scale Fading} 
Owing to small-scale fading, over only a fraction of a wavelength, the signal power variation may reach up to 40 dB \cite{T.2002Rappaport}. The correlation of the signal strength can be both temporal and spatial and is frequency dependent~\cite{A.2005Goldsmith}.
The temporal correlation occurs due to movement of environmental scatters.  
The spatial correlation is caused by correlated multipath components due to the common propagation environment.  
However, small-scale fading is only correlated over a short time duration and a small distance (e.g., several wavelengths) in environments with moving scattering objects that change the multipath propagation. Hence it is reasonable to adopt independent and identically distributed (i.i.d.) small-scale fading models for receiving antennas with wavelength separation. 
This assumption is widely used in the existing literature (see \cite{ElSawy2013H,ElSawy2016Hesham} and references therein).

\subsubsection{Shadowing}

 \begin{figure}
 \centering
 \includegraphics[width=0.47\textwidth]{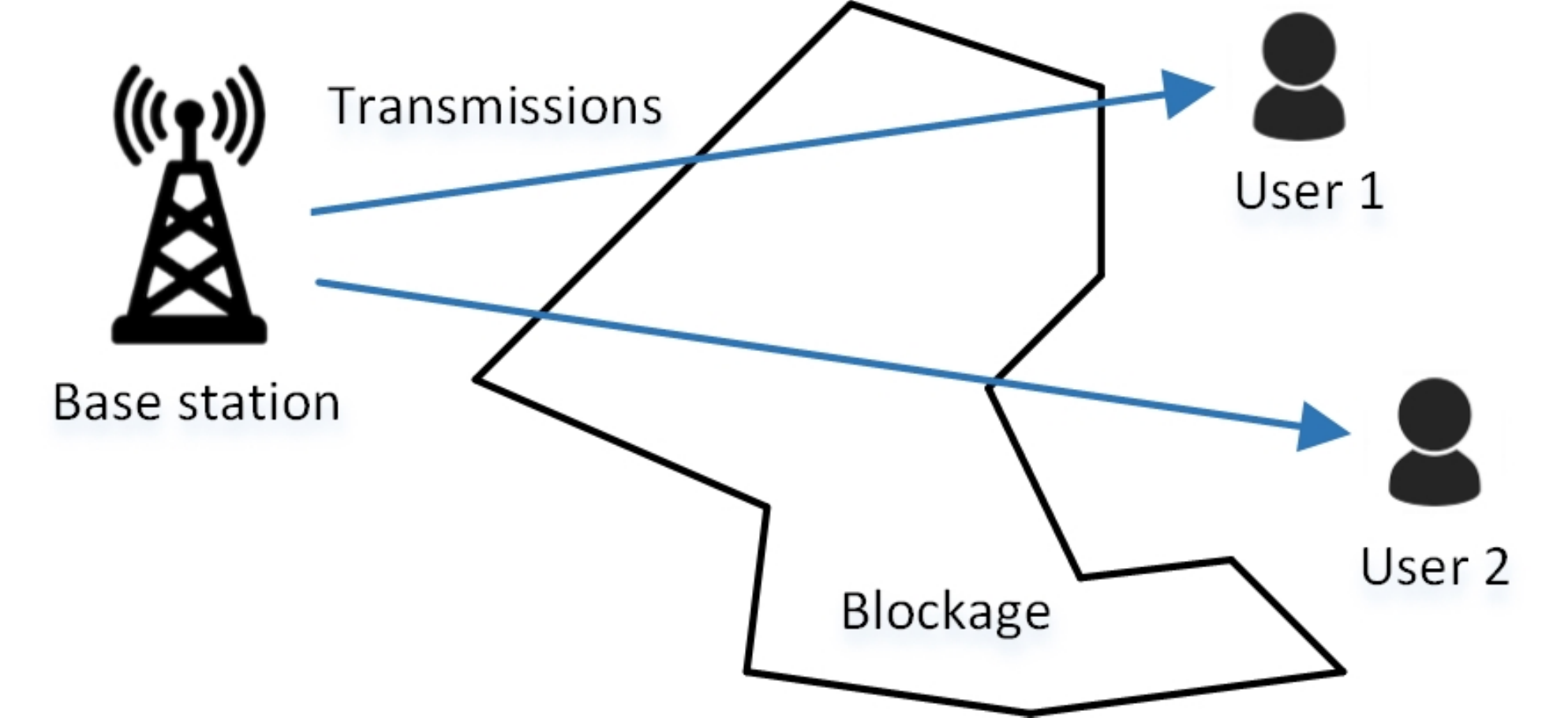} 
 \caption{Illustration of correlated shadowing.   } \label{fig:shadow} 
 \end{figure}

Compared to small-scale fading, shadowing is correlated at a much larger time and space scale \cite{M.1991Gudmundson}.
For example, in an urban environment, shadowing caused by obstacles in communication paths can be geographically correlated on a scale from 50 to 200 meters \cite{A.2005Goldsmith}. Moreover, temporal shadow fades with variations below 1 dB (i.e., highly correlated) over an IEEE 802.11ad channel can be commonly measured in urban streets \cite{L.2018Ahumada}.  
Fig. \ref{fig:shadow} demonstrates  an example of correlated shadowing.  As the two links between the BS and users experience a similar propagation environment, their shadowing attenuations tend to be correlated. As evidenced by experimental studies in \cite{V.1978Graziano} and \cite{Van1987Rees}, path attenuations in a region are correlated due to shadowing effects.  
Therefore, correlated shadowing should be carefully treated when assessing system performance that is heavily impacted by environmental blockages.  

\subsubsection{Transmission Buffer Status}

The buffer status (i.e., queueing status) of each transmitter
depends on the packet arrival and service processes, and it impacts the activation of the transmitters, and therefore, the mutual interference.  
In view of this, the buffer statuses of the transmitters are interdependent, leading to {\em interacting queues}.
The queues are spatially coupled since the mutual interference directly affects the transmission success probability (i.e., service processes of the queues). 
Also, the queues are temporally coupled since the current buffer status is affected by the previous departure process.  
Due to the random nature of channel fading and aggregate interference, the service processes can be very dynamic.

 \begin{figure}
 \centering
 \includegraphics[width=0.44\textwidth]{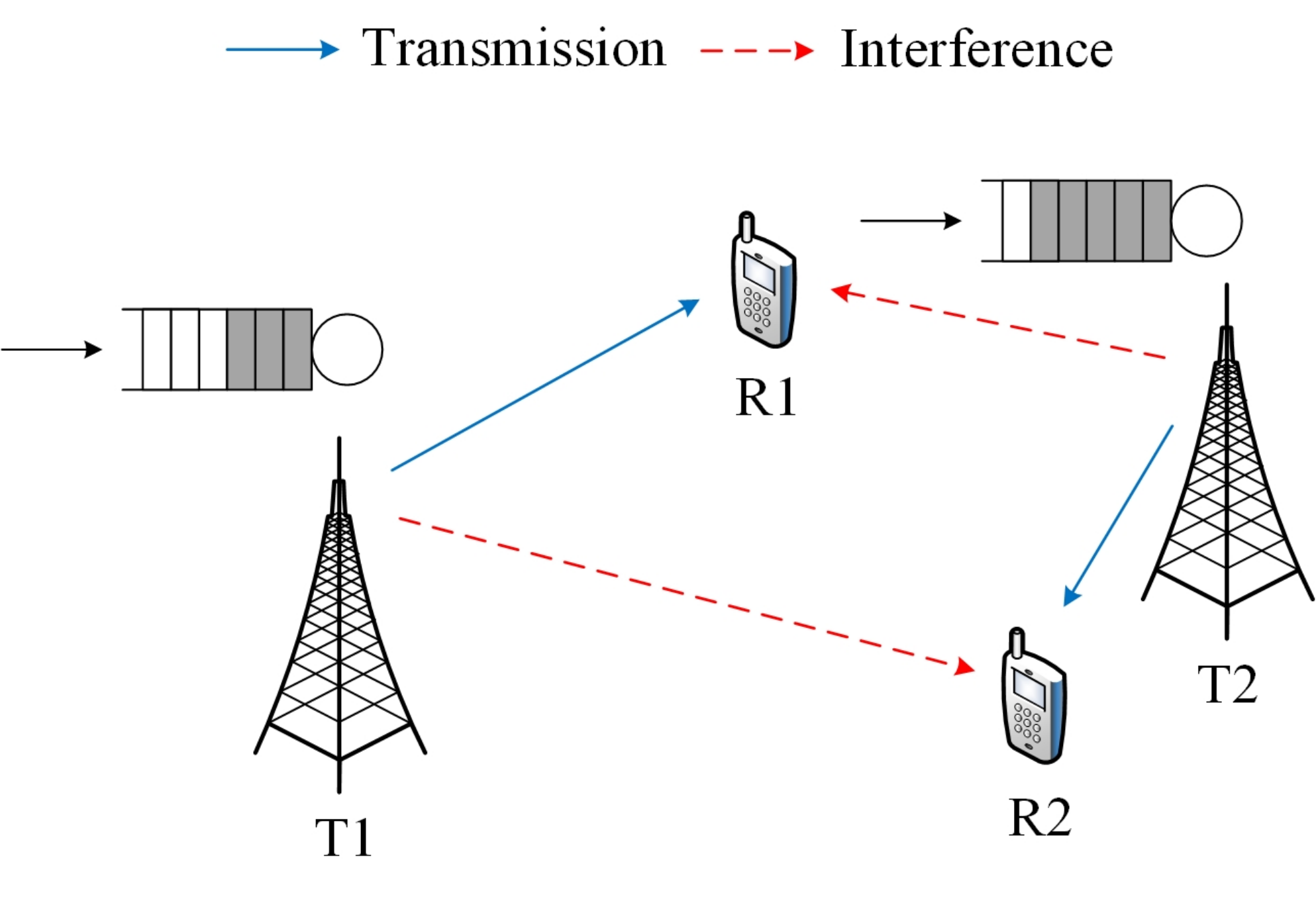} 
 \caption{Illustration of the interacting queues.   } \label{fig:queue} 
 \end{figure}

 Fig.~\ref{fig:queue} depicts an example of the correlation between the queues at two transmitters. Transmitter T1, which serves receiver R1, has a longer transmission link and a shorter interference link compared to those of transmitter T2, which serves receiver R2. 
 On the one hand,   
 the channel disparity between the two communication links results in different packet service rates. Compared to T2, T1 
 suffers from more severe path loss and interference. Correspondingly, given a similar traffic load 
 T1 tends to vacate its queue more slowly and thus remain active more frequently than T2.   
 On the other hand, the queue lengths of T1 and T2 determine the activation of T1 and T2.  
 In particular, if both transmitters are busy, 
 their transmissions will cause mutual interference, which slows down the departure process. 
 If one of the transmitters has an empty queue, the other  enjoys a speedy departure process.  
 Due to the interacting queues, transmissions in a large-scale system  inherently experience spatial and temporal correlation. 

\begin{figure*} 
 \centering
 \includegraphics[width=0.7 \textwidth]{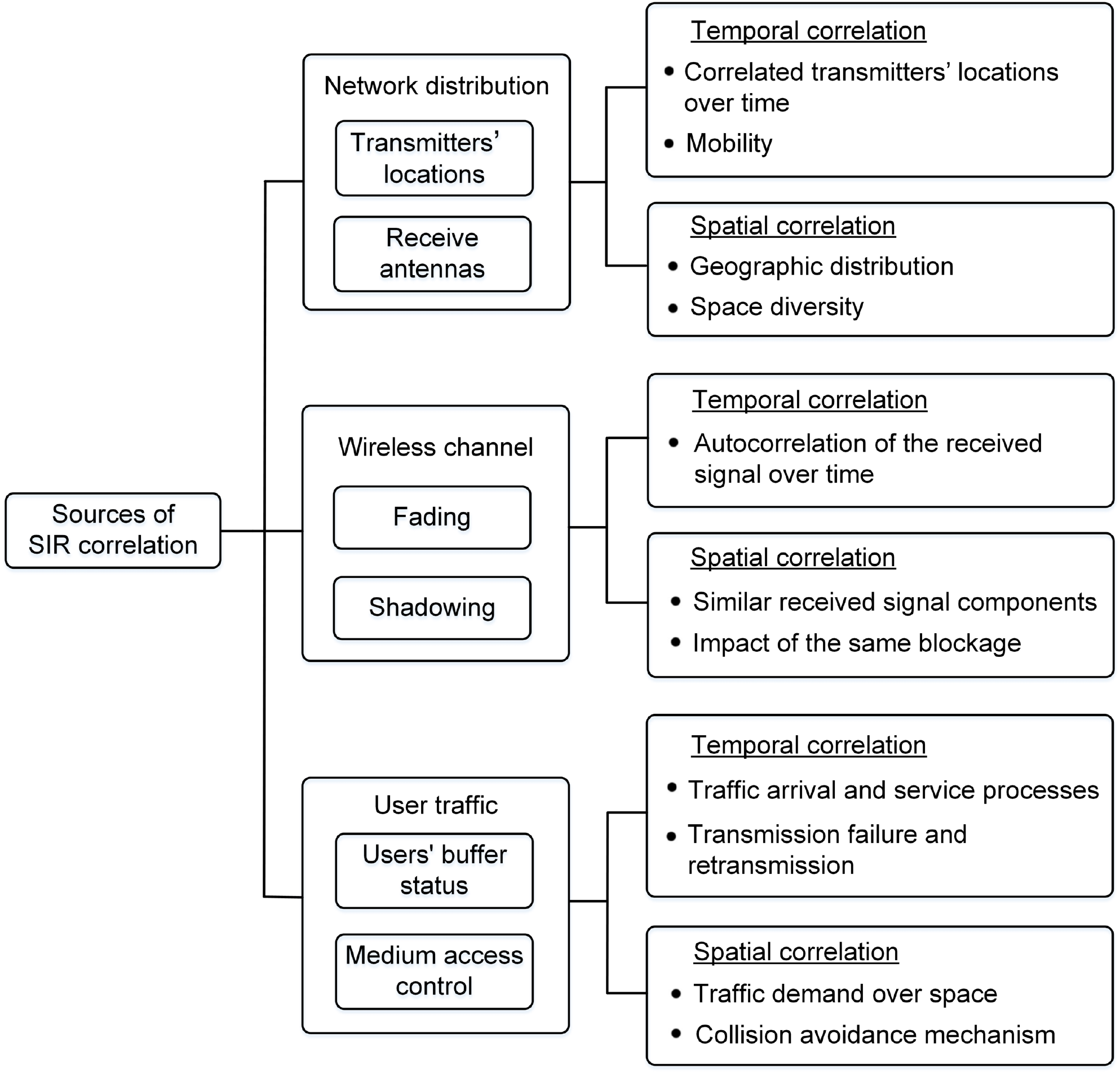} 
 \caption{Sources of SINR correlation.   } \label{fig:correlations} 
 \end{figure*}

Finally, Fig. \ref{fig:correlations} summarizes the sources of SINR correlation and their spatial-temporal impact discussed in this subsection.

\subsection{Spatial-Temporal SIR Correlation and Network Scenarios}

Even over independent fading channels, aggregated interference is spatially correlated when it comes from the same group of interferers. Similarly, interference becomes temporally correlated when common static interferers exist across time. Spatially and temporally-correlated interference causes spatial-temporal SIR correlation.
When different transmission attempts are affected by the same group of static interferers (e.g., the same point process), it gives rise to a {\em quasi-static interference} (QSI) scenario.  By contrast, a {\em fast-varying interference} (FVI) \cite{Z.2009Win} scenario results when different transmission attempts are affected by different sets of  interferers\footnote{The QSI and FVI assumptions can be used for modeling static and highly-mobile network scenarios, where the interferers across different transmissions remain the same and become completely different, respectively.  
} (e.g., independent point processes). 
As two extreme cases, QSI and FVI, respectively, render the highest and lowest interference correlation caused by the interferer locations.

 
Spatial-temporal SIR correlation affects the performance of many practical systems, where the transmission performance depends on the SIR measured over different space and time spots. Typical examples of these systems include the following:

\begin{itemize}
\item {\em Multihop relaying}: The end-to-end performance of a multihop relaying system depends on the  SIRs at the receivers at different hops. The SIRs at different hops can be correlated due to spatial-temporal interference correlation (e.g., common interferers exist from one transmission to another).

\item {\em Retransmission}: 
The reliability of a retransmission scheme or a multi-packet transmission scheme depends on the SIR at the receiver during multiple transmissions. The SIRs from different transmissions can are correlated if they are subject to the interference from the same group of transmitters. 
   
\item {\em Mobile networks}:  
For a mobile user, transmissions could occur at different spatial locations and times.
The correlation between the SIRs at the receiver for different transmissions are dependent on the mobility.  

\end{itemize}

The SIR correlation affects the performance metrics such as temporal joint success probability (JSP)~\cite{Krishnan2017}, end-to-end success probability~\cite{A.Oct.2015Crismani},  the local delay \cite{M.2013Haenggib}, and joint SIR meta distribution~\cite{Yu2019Xinlei}. 

\subsection{Metrics for Spatial-Temporal Performance Analysis} \label{sec:PM} 

\vspace{0.1cm}
\begin{definition}(Success Probability):
   The (average) success probability of a user 
   can be defined as the probability that the user's received SIR is greater than a threshold denoted by $\theta$~\cite{G.2011Andrews}, i.e., 
   \begin{align} \label{def:SP}
   P_{\rm s} = \bar{F}_{\mathrm{SIR}} (\theta) \triangleq  \mathbb{P} [ \mathrm{SIR} > \theta ], \quad \theta \in \mathbb{R}^{+}.
   \end{align}
   where $\bar{F}_{\mathrm{SIR}}$ represents the complementary CDF of SIR.  
\end{definition}

 \vspace{0.1cm} 

\begin{definition}(Moments of Conditional Success Probability given the Point Process): 
Let $P_{\rm s \mid \Phi} (\theta) = \mathbb{P} \big[  \mathrm{SIR}  > \theta \mid \Phi 
\big]$ represent the CSP given the point process $\Phi$  (abbreviated as CSP$_\Phi$), 
which is averaged over the fading of all the links and the random channel access (if applicable). 
The $b$-th moment of $P_{\rm s|\Phi}$ is defined as \cite{M.Apr.2016Haenggi}
\begin{align}
\mathcal{M}_{P_{\rm s}}(b) \triangleq  \mathbb{E}  \Big[ \big(  P_{\rm s \mid \Phi} (\theta)
\big)^b \Big] , \quad  b \in \mathbb{C}.
\end{align} 
\end{definition}

Note that the first moment of the CSP$_\Phi$ is the average success probability defined in (\ref{def:SP}), i.e., $\mathcal{M}_{P_{\rm s}}(1) \equiv P_{\rm s}$, and 
the variance is $\mathbb{V}(P_{\rm s}) = \mathcal{M}_{P_{\rm s}} (2) -  \mathcal{M}^2_{P_{\rm s}}(1)$. Besides, the mean local delay, defined as the average number of transmission attempts to accomplish a success~\cite{Baccelli2010}, is given by $\mathcal{M}_{P_{\rm s}}(-1)$. 
 It is worth mentioning that for static random networks and $b \in \mathbb{N}$, 
the $b$-th moment of the CSP$_{\Phi}$ 
is equivalent to the JSP of $b$ transmissions of the same link~\cite{M.2013Haenggia} and JSP that the $b$ antennas of a multiple-antenna receiver all succeed in reception (i.e., SIR exceeds $\theta$) \cite{M.2012Haenggi}.  

 \vspace{0.1cm} 
\begin{definition}(SIR Meta Distribution):
The SIR meta distribution is the complementary CDF of the CSP$_\Phi$, defined as~\cite{M.Part1Haenggi,M.Apr.2016Haenggi}
\begin{align}
\bar{F}_{P_{\rm s}} (\theta,s) \triangleq \mathbb{P}  \big[  P_{\rm s \mid \Phi} (\theta)
  > s \big] , \,\,\,  s \in [0, 1] , 
\end{align}
where $s$ represents the target success probability. 
 \end{definition} 
    
In ergodic point processes, the SIR meta distribution 
indicates the fraction of links that can achieve successful transmission with probabilities 
greater than $x$ in any realization of $\Phi$.
Compared to the average success probability defined in~(\ref{def:SP}), the SIR meta distribution also characterizes the variability of link success probabilities.  
For example, the inverse function of the SIR meta distribution $\bar{F}^{-1}_{P_{\rm s}} (p)$, $p \in [0,1]$, yields the success probability $x$ that a fraction $1-p$ of the links achieve while the rest do not.

\noindent
{\em Example:} With $\theta=1$ and $s=0.9$, $\bar{F}_{P_{\rm s}} (1,0.9)$ represents the fraction of links/transmissions that achieve a receive SIR greater than 1 with probability at least 0.9. 

\vspace{0.1cm}
 \begin{definition}(SIR Gain):  
The SIR gain is the horizontal gap between the complementary CDF of two SIR distributions.  
Evaluated at the target success probability  $P_{\rm t}$, the SIR gain is defined as \cite[Eq.  (1)]{M.Dec.2014Haenggi}
\begin{align}  \label{def:SIR_gain}
G( P_{\rm t})  \triangleq  \frac{ \bar{F}^{-1}_{ \mathrm{SIR}_{\rm tm} } ( P_{\rm t}) }{  \bar{F}^{-1}_{ \mathrm{SIR}_{\rm rm} } ( P_{\rm t})  }, \quad  P_{\rm t} \in (0,1), 
\end{align}
where $\bar{F}^{-1}_{\mathrm{SIR}}$ represents the inverse function of the complementary CDF of the SIR, and $\mathrm{SIR}_{\rm tm}$ and $\mathrm{SIR}_{\rm rm}$ denote the SIRs of a target model and a reference model, respectively. 
\end{definition}

The SIR gain can be used to quantify the impact of a target model on the SIR distribution w.r.t. that of a reference model. An SIR gain greater or less than 1 indicates that the target model can achieve the same success probability with a larger or smaller SIR threshold, respectively, than the reference model. 

In wireless networks, the SIR gain is usually not sensitive to the target success probability to be evaluated \cite{M.Dec.2014Haenggi,K.Mar.2016Ganti}. Therefore, the SIR gain can be approximated by the asymptotic SIR gain evaluated in the  high-reliability regime, i.e., $P_{\rm t} \to 1$ or $\theta \to 0$, defined as
\begin{align}
G_{0} \triangleq 
\lim_{ P_{\rm t} \to 1} G( P_{\rm t}),
\end{align}
whenever the limit exists.

Note that a necessary and sufficient condition for the asymptotic SIR gain to exist is that the slopes of the two CDFs of the SIR are asymptotically the same as $\theta \to 0$ \cite{M.Dec.2014Haenggi}.

\vspace{0.1cm}
\begin{definition} (Joint Success Probability): 
Let $\mathbf{x} =\{ x_{k}\}_{k \in \{1,\ldots,K\}} \in (\mathbb{R}^{d})^K$ and $\mathbf{t} =\{ t_{k}\}_{k \in \{1,\ldots,K\}} \in \mathbb{N}^K$ denote deterministic vectors of locations and time instances, respectively, and
$\textup{SIR}_{k}$ denote the SIR measured at location $x_{k}$ and time $t_{k}$.
The JSP is defined as the probability that $\textup{SIR}_{i}$ is greater than the corresponding SIR threshold $\theta_{k}$ for all $k \in \{1,\ldots,K\}$.
Given the locations $\mathbf{x}$, times $\mathbf{t}$, and the target SIR thresholds  $\mbox{\boldmath $\theta$ }=\{\theta_{k}\}_{k \in \{1,2,\ldots,K\}}   \in (\mathbb{R}^{+})^K$, the JSP is defined as
\begin{align}  \label{def:JCP}
& \mathcal{J}_{K} ( \mbox{\boldmath $\theta$},\mathbf{x},\mathbf{t}) \nonumber \\ &\triangleq \mathbb{P} [  \mathrm{SIR}_{1}> \theta_{1}, \mathrm{SIR}_{2}> \theta_{2}, \ldots,  \mathrm{SIR}_{K}> \theta_{K} ],  
\end{align} 
where $\mathrm{SIR}_{k}$ denotes the SIR of the $k$-th transmission. 
\end{definition}
The SIR corresponding to the different transmissions can exhibit spatial and/or temporal correlation, the effect of which will be reflected in the JSP. 
The JSP can be for temporal, spatial, or spatial-temporal transmission events. 
In particular, we have
\begin{itemize} 

\item  {\em Temporal JSP} of multiple transmission attempts occurring at the same location but different time instances (e.g., multiple transmissions~\cite{M.2013Haenggia} and retransmissions~\cite{G.2015Nigam} by a transmitter);

\item {\em Spatial JSP} of multiple transmission attempts occurring at the same time but different locations (e.g., joint uplink and downlink transmissions \cite{Singh2015Sarabjot,A.2017Kishk});

\item {\em Spatial-temporal JSP} of multiple transmission events occurring at different space and time intervals (e.g., for the same mobile user \cite{Krishnan2017}).  

\end{itemize}

For a multihop communication scenario, the probability that the transmissions at each hop are successful for a given packet is referred to as the {\em end-to-end success probability}.

\vspace{0.1cm} 
\begin{definition} (Conditional Success Probability): 
The CSP is the probability of achieving a successful transmission given that $K-1$ ($K\geq 2$) such events have occurred. Given the locations $\mathbf{x}$, times $\mathbf{t}$, and the target SIR thresholds $\mbox{\boldmath$\theta$}$, the CSP is given by
\begin{align}  \label{def:CSP} 
\mathcal{C}_{K}(\mbox{\boldmath$\theta$},\mathbf{x},\mathbf{t}) = \frac{\mathcal{J}_{K} (\mbox{\boldmath$\theta$},\mathbf{x},\mathbf{t})}{ \mathcal{J}_{K-1} (\mbox{\boldmath$\theta$},\mathbf{x},\mathbf{t}) },
\end{align}
where $\mathcal{J}_{K} (\mbox{\boldmath$\theta$},\mathbf{x},\mathbf{t})$ is defined in (\ref{def:JCP}). 
\end{definition}

The CSP reveals the dependence between two successful transmission events. If the two events are positively correlated (spatially or temporally), one has $\mathcal{C}_{2}(\mbox{\boldmath$\theta$},\mathbf{x},\mathbf{t}) >  \mathcal{J}_{1} (\mbox{\boldmath$\theta$},\mathbf{x},\mathbf{t})$~\cite{M.2013Haenggia}. 
Moreover, if the events are independent, one has $\mathcal{C}_{2}(\mbox{\boldmath$\theta$},\mathbf{x},\mathbf{t}) =  \mathcal{J}_{1} (\mbox{\boldmath$\theta$},\mathbf{x},\mathbf{t}) = \sqrt{\mathcal{J}_{2} (\mbox{\boldmath$\theta$},\mathbf{x},\mathbf{t})}$. Similar to the definition of the JSP, the CSP of spatial events, temporal events, and spatial-temporal events are referred to as {\em spatial CSP}, {\em temporal CSP}, and {\em spatial-temporal CSP}, respectively.

\vspace{0.1cm} 

\begin{definition} (Product SIR Meta Distribution): 
Let $\mathcal{J}_{K|\Phi} (\mbox{\boldmath$\theta$},\mathbf{x}, \mathbf{t}) \triangleq \mathbb{P} [  \mathrm{SIR}_{1}> \theta_{1},  \ldots,  \mathrm{SIR}_{K}> \theta_{K} \mid \Phi ]$ denote the JSP of $K$ transmissions given the point process. 
The product SIR meta distribution is defined as the complementary CDF of the JSP at locations $\mathbf{x}$ and times $\mathbf{t}$ given the point process $\Phi$ \cite[Eq. (11)]{Yu2019Xinlei}, i.e.,
\begin{align}
\bar{F}_{\mathcal{J}_{K}} \!(\mbox{\boldmath$\theta$}, \mathbf{x},  \mathbf{t}, s  ) \triangleq  \mathbb{P} \Big[ \mathcal{J}_{K|\Phi} (\mbox{\boldmath$\theta$},\mathbf{x}, \mathbf{t}) > s \Big], \, s \!\in\! [0,1].
\end{align}  
\end{definition}

\noindent
{\em Example:} Let $K=2$,  ${\bm \theta}=\{\theta_{1},\theta_{2}\}$, $\mathbf{x}=\{x_{1},x_{2}\}$, $\mathbf{t}=\{t_{1}, t_{2} \}$. Considering a stationary and ergodic point process $\Psi=\{u_{j}\}_{j \in \mathbb{N}}$ representing the users, the product SIR meta distribution measures the fraction of users
for which the joint probability that the SIR at their location $u$ at time $t_{1}$ exceeds $\theta_{1}$ and the SIR at location $u + (x_{2}-x_{1})$ at time $t_{2}$ exceeds $\theta_{2}$ is larger than $s$. 

\vspace{0.1cm} 
\begin{definition} (Joint SIR Meta Distribution): 
 Let $P^{(k)}_{\rm s|\Phi}(\theta_{k})= \mathbb{P} \big[ \textup{SIR}_{k} > \theta_{k} \mid \Phi \big] $ represent the CSP$_{\Phi}$ at $x_{k}$.  
The joint SIR meta distribution is defined as the joint distribution of CSP$_{\Phi}$ at locations $\mathbf{x}$ and times $\mathbf{t}$ \cite[Eq. (10)]{Yu2019Xinlei}, i.e., 
\begin{align}
\bar{F}^{(K)}_{P_{\rm s } } (\mbox{\boldmath$\theta$}, \mathbf{x}, \mathbf{t}, \mathbf{s} ) &  \triangleq \mathbb{P} 
\bigg[  \bigcap^{K}_{k=1}  \big \{ P^{(k)}_{s|\Phi} (\theta_{k}) > s_{k}   \big \} \bigg] ,  
\end{align}
where $\mathbf{s}=\{s_{1},s_{2},\ldots,s_{K}\} \in [0,1]^{K}$ is the vector of target success probabilities corresponding to the locations $\mathbf{x}$. 
 \end{definition}

\noindent
 {\em Example:} Let $K=2$,  $\bm{\theta}=\{\theta_{1},\theta_{2}\}$, $\mathbf{x}=\{x_{1},x_{2}\}$, $\mathbf{t}=\{t_{1}, t_{2} \}$, $\mathbf{s}=\{s_{1},s_{2}\}$. Considering a stationary and ergodic point process $\Psi=\{u_{j}\}_{j \in \mathbb{N}}$  representing the users, the joint SIR meta distribution measures the fraction of users that meet the following conditions: 1) the probability that the SIR at the user location $u$ at time $t_{1}$ exceeds $\theta_{1}$ is larger than $s_{1}$ and 2) the probability that the SIR at location $u+ (x_{2}-x_{1})$ at time $t_{2}$ exceeds $\theta_{2}$ is larger than $s_{2}$. 

\vspace{0.1cm} 
\begin{definition} (Interference Correlation Coefficient): 
Let $I^{(t)}_{x}$ denote the aggregated interference received at location $x$ at time $t$. The correlation degree of interference at two locations and times can be quantified by the {\em Pearson correlation coefficient} defined as \cite[Eq. (2)]{U2011Schilcher} 
\begin{align} \label{def:ICC}
& \zeta_{t_{1},t_{2}} (u_{1},u_{2})  \triangleq \nonumber\\
& \frac{\mathbb{E} \big[ I^{(t_{1})}_{u_{1}}  I^{(t_{2})}_{u_{2}} \big] - \mathbb{E} \big[ I^{(t_{1})}_{u_{1}} \big] \mathbb{E} \big[I^{(t_{2})}_{u_{2}}\big] }{\sqrt{ \mathbb{E} \big[ \big( I^{(t_{1})}_{u_{1}} \big)^{\!2}    \big]\! -\! \mathbb{E} \big[   I^{(t_{1})}_{u_{1}}  \big]^{\!2}   } \sqrt{ \mathbb{E} \big[ \big( I^{(t_{2})}_{u_{2}} \big)^{\!2}    \big]\! - \!\mathbb{E} \big[   I^{(t_{2})}_{u_{2}}    \big]^{\!2}  } } ,
\end{align} 
where the numerator computes the covariance of $I^{(t_{1})}_{u_{1}}$ and $I^{(t_{2})}_{u_{2}}$ and the denominator is the product of the standard deviations of $I^{(t_{1})}_{u_{1}}$ and $I^{(t_{2})}_{u_{2}}$.
\end{definition} 
 
The interference correlation coefficient is a statistical measure that quantifies the extent to which the interferences at $u_{1}$ and $u_{2}$ are associated. 
The value of the interference correlation coefficient
ranges between 0 and 1. The coefficients of 0 and 1, respectively, indicate no linear  relationship and full correlation between the interferences at $u_{1}$ and $u_{2}$.

Note that when the interferers are motion-invariant,  
 one has $I^{(t_{1})}_{o} \overset{(d)}{=} I^{(t_{1})}_{u}$.   
In this case, $\zeta (\|u\|)$ can be simplified to  \cite[Eq. (12)]{R.2009Ganti} 
\begin{align}  \label{def:ICC2}
\zeta_{t_{1},t_{2}} (\|u_{1}-u_{2}\|)  \triangleq  \frac{ \mathbb{E} \Big[ I^{(t_{1})}_{u_{1}}  I^{(t_{2})}_{u_{2}} \Big] - \mathbb{E} \Big[   I^{(t_{1})}_{u_{1}}   \Big]^{\!2} }{ \mathbb{E} \Big[ \big( I^{(t_{1})}_{u_{1}} \big)^2    \Big]\! - \! \mathbb{E} \Big[  I^{(t_{1})}_{u_{1}}  \Big]^2    } .
\end{align}

\vspace{0.1cm}
\begin{definition} (Interference Coherence Time):
The interference coherence time is defined as the minimum time lag such that the interference correlation coefficient is below a threshold $\zeta_{\rm th}$ \cite[Eq. (35)]{U.2020Schilcher}, i.e.,
\begin{align}
\tau_{\rm ct} \!\triangleq\! \min \big \{  \tau =  t_{2} \!-\! t_{1} \in \mathbb{N}  \mid \zeta_{t_{1},t_{2}}  \leq   \zeta_{\rm th}  \big \}, \, \zeta_{\rm th} \in \mathbb{R}^{+} .
\end{align}
 \end{definition}

Fig.~\ref{fig:metrics} summarizes the above performance metrics for spatial-temporal performance analysis.

\begin{figure}
 \centering
   \includegraphics[width=0.35  \textwidth]{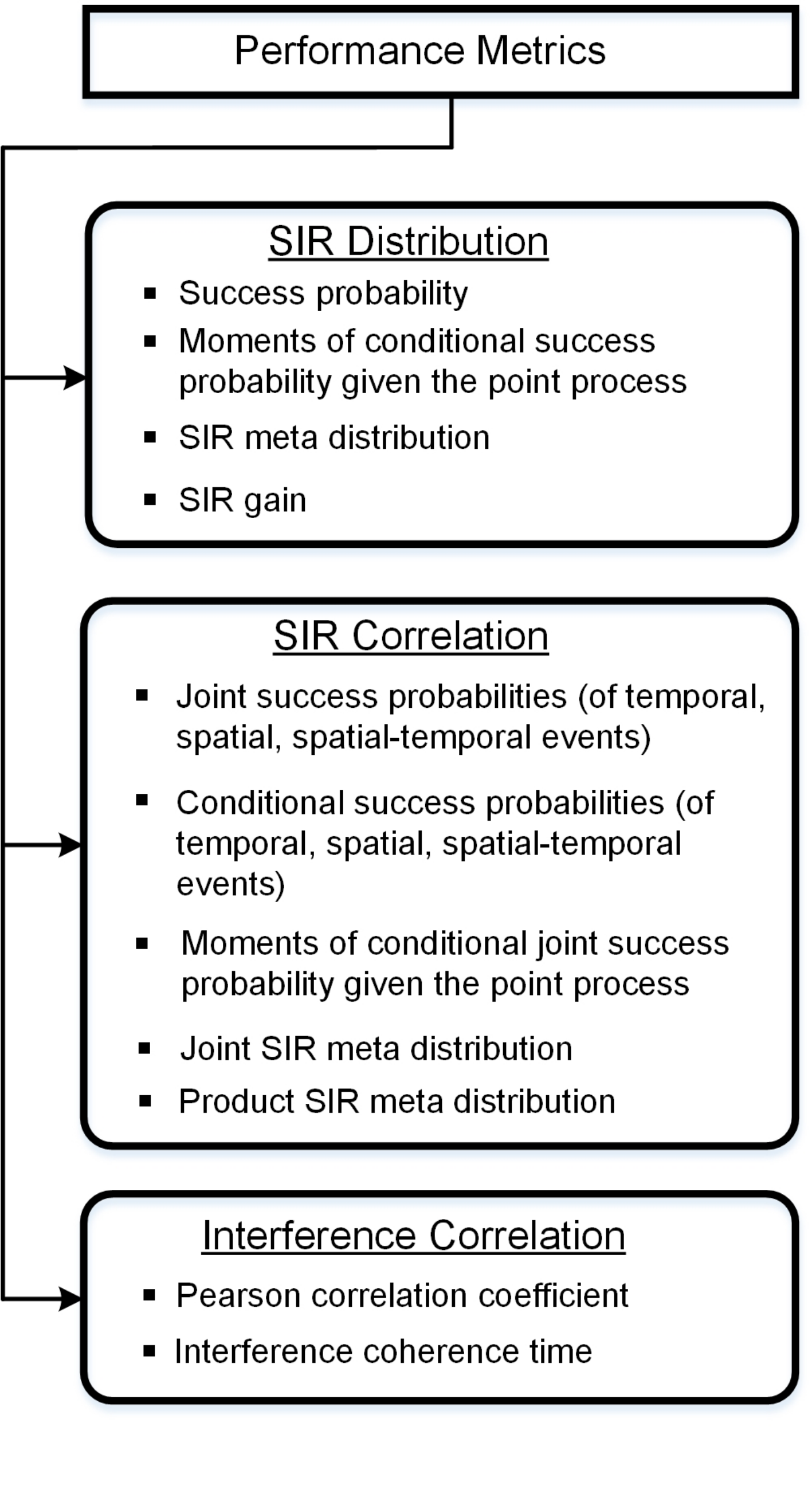} 
 \caption{Metrics for spatial-temporal performance analysis. } \label{fig:metrics} 
 \end{figure}

\section{Spatial Point Process Models} \label{sec:spatial}

This section formally defines some common point processes, also referred to as random point fields, and illustrates how the spatial distribution of the random points affects the network performance. To evaluate the impacts of independent, attractive and repulsive spatial configurations, we choose three representative point processes, namely, the Poisson point process (PPP), the Mat\'{e}rn cluster process (MCP), and the $\beta$-Ginibre point process (GPP), due to their tractability.  

\begin{itemize}
\item  In a PPP, each point is located independently from the others. 

\item  In an MCP, the locations of the random points have a propensity to be clustering.

\item  In a $\beta$-GPP, the random points are scattered with repulsion.

\end{itemize}

This section considers both ad hoc and downlink cellular networks modeled based on the above point processes and analyzes the interference correlation coefficient, success probability, SIR meta distribution and SIR gain. 
  
\subsection{System Models} \label{sec:sectionIII_SM}

\subsubsection{System Configurations}

We focus on analyzing the performance of a target in both ad hoc and downlink cellular networks. 
The target receiver is considered to be at the origin $o \in \mathbb{R}^2$ and attempts to decode the transmitted signals from the associated transmitter subject to the aggregate interference. 

\begin{itemize}
 
\item {\em Ad hoc networks}: The target receiver at $o$ is served by the transmitter located at $x_{\rm t} \in \mathbb{R}^2$, where $\|x_{\rm t}\|=r_{\rm t}$. The target link is impaired by the interference from a random field of interferers 
modeled by a stationary point process $\Phi$ of intensity $\lambda$.  $\Phi$ does not contain the serving transmitter at $x_{\rm t}$.

\item {\em Downlink cellular networks}: The transmitters (i.e., BSs) form a stationary point process $\Phi$ with intensity $\lambda$. The users are located according to a stationary point process independent of the BS process.  
Each user is served by the nearest BS in $\Phi$. 
 
\end{itemize}

In the networks considered, the transmitters stay active with unit transmit power, i.e.,  there is no MAC scheme.
All the transmitters and receivers are equipped with one antenna. 
The system employs universal frequency reuse.  
The channels of the links 
experience i.i.d. 
block Rayleigh fading and power-law path loss, 
i.e., 
$d^{-\alpha}$, where $d$ and $\alpha > 2$ are the link distance and the path-loss exponent, respectively. 
 
For the convenience of notation, the points in $\Phi$ are assumed to be ordered from nearest to farthest to the origin, i.e., $\|x_{j}\| < \|x_{j+1}\|$.  
 Subsequently,  
 the set of interferers in the ad hoc networks and downlink cellular networks is $\Phi^{!}=\Phi=\{x_{j}\}_{j \in \mathbb{N}}$ and $\Phi^{!}=\Phi \backslash \{x_{1}\}$, respectively.  The distance from the $j$-th nearest point in $\Phi$ to the origin is denoted as $r_{j}=\|x_{j}\|$.
 
The SIR at the target receiver in an ad hoc and downlink cellular network is given by, respectively, as
\begin{align}  \label{eqn:eta_adhoc}
\eta = \frac{ h_{\rm t}  \|x_{\rm t}\|^{- \alpha}  }{  \sum_{ j \in \mathbb{N} } h_{j}  \|x_{j}\|^{-\alpha} 
} 
\end{align}
and
\begin{align}  \label{eqn:eta_downlink}
\eta = \frac{ h_{1}  \|x_{1}\|^{- \alpha} }{  \sum^{\infty}_{ j =2 } h_{j} \|x_{j}\|^{-\alpha  }},
\end{align}  
where $h_{\rm t}$ and $h_{j}$ denote the power fading coefficients from the transmitters at $x_{\rm t}$,  $x_{j}$ 
to the target receiver, which are exponential random variables with unit mean, i.e., $h_{\rm t}, h_{j} \sim \mathcal{E}(1)$. It is noted that (\ref{eqn:eta_adhoc}) represents the SIR of a specific receiver instead of the typical receiver in ad hoc networks, while (\ref{eqn:eta_downlink}) is the SIR of a user at the origin that, upon averaging over the base station and user point processes, becomes the typical user in the downlink cellular network, for any stationary point process of users that is independent of the base station point process.

\subsubsection {Spatial Configurations}
 
We consider the homogeneous PPP, MCP, and $\beta$-GPP, formally defined as follows.

\begin{definition} (Homogeneous Poisson point process):  
A  point process $\Phi = \{ x_{j} \}_{j \in \mathbb{N}} \subset \mathbb{R}^{d} $ is 
a homogenous PPP if it satisfies two conditions: i) for any compact set $\mathcal{B} \subset \mathbb{R}^{d}$, the number of points inside it follows a Poisson distribution with average value being $\lambda|\mathcal{B}|$; and ii) disjoint sets are independent in terms of the number of points inside them.
\end{definition}

\vspace{0.2cm}
{\em Properties:}
\vspace{0.2cm}

1) {\bf Density}: For a homogeneous PPP, the first moment density  (or first-order density) and second moment density (or second-order product density) are given, respectively, as~\cite{N.2013Chiu}
\begin{align} \label{eqn:rho1_PPP}
\rho^{(1)}(x) = \lambda, 
\end{align} 
and
\begin{align} \label{eqn:rho2_PPP}
\rho^{(2)}(x,y) = \big(\rho^{(1)}(x) \big)^2=\lambda^2.
\end{align}
The $n$-th moment density is the density pertaining to the $n$-th order factorial moment measure, which, for $n>1$,  indicate the spatial correlation.  
For example, $\rho^{(1)}(x) \mathrm{d}x$ measures the probability of having a point at $x$ in some infinitesimal region $\mathrm{d}x$.

\vspace{0.2cm}
2) {\bf Contact Distance Distribution}: Contact distance refers to the distance between a reference location and the nearest point in $\Phi$~\cite{M.2017Afshang}. 
 The probability density function (PDF) and CDF of the contact distance in a homogeneous PPP are given, respectively, as~\cite{M.Haenggi2005}
\begin{align} \label{eqn:PDF_CD_PPP}
f_{r_{1}}(r)= 2 \pi \lambda \exp (-\lambda \pi r^2) 
\end{align}
and
\begin{align}
F_{r_{1}}(r)= \exp (-\lambda \pi r^2).
\end{align}

3) {\bf Joint Distance Distribution}:
The joint PDF of the distances to $n$ nearest points is~\cite[Eq. (30)]{D.2012Moltchanov}
\begin{align}
  f_{r_{1}, r_{2}, \ldots, r_{n}} (x_{1}, x_{2}, \ldots\!, x_{n})  \! =\! e^{- \lambda \pi x^2_{n}} (2 \lambda \pi )^{n} x_{1} x_{2} \ldots x_{n}.
\end{align}

4) {\bf Distance Ratio Distribution}:
Let $\varrho_{j}=\frac{r_{1}}{r_{j}}, j \in \mathbb{N}$, denote the distance ratio of the nearest point to the $j$-th nearest point. The CDF and PDF of $\varrho_{j}$
are given, respectively, by~\cite[Lemma 3]{X.Dec.2014Zhang}
\begin{align} \label{eqn:CDF_ratio}
 F_{\varrho_{j}} (\varrho) = 1 - (1-\varrho^2)^{j-1}, \hspace{10mm} \varrho \in [0,1]    
\end{align}
and
\begin{align} \label{eqn:PDF_ratio}
 f_{\varrho_{j}} (\varrho) =  2 (j-1) \varrho (1 - \varrho^2)^{j-2} , \hspace{2.5mm} \varrho \in [0,1]. 
\end{align}
 
5) {\bf Probability Generating Functional (Product Functional)}: 
For $\upsilon(x) \in [0,1]$ and $\int_{\mathbb{R}^2} (1-\upsilon(x)) \mathrm{d} x < \infty $,  the probability generating functional (PGFL) for the PPP is given by \cite[Eq. (4.8)]{M.2013Haenggic}
\begin{align}  
\mathbb{E} \bigg[ \prod_{j \in \mathbb{N} }  \upsilon(x_{j}) \bigg] = \exp \bigg (  - \lambda \int_{\mathbb{R}^d} \big ( 1 -  \upsilon(x )    \big ) \mathrm{d} x  \bigg ). \label{eqn:PPP_PGFL}  
\end{align}   

According to Slivnyak's theorem \cite[Theorem  8.10]{M.2013Haenggic}, the reduced Palm distribution of the PPP is the same as its ordinary distribution. Therefore, the conditional PGFL for the PPP can also be expressed as~(\ref{eqn:PPP_PGFL}).

Let  $\Phi^{\rm R} = \big \{ x \in \Phi \backslash \{ x_{1} \} :   \|x_{1}\| /\|x \|    \big\}    \subset (0,1)$  denote the {\em relative distance process} (RDP) of a PPP $\Phi$.  The PGFL of  $\Phi^{\rm R}$ is given by \cite[Lemma 1]{K.Mar.2016Ganti}
\begin{align} \label{eqn:PGFL_RDP}
\mathbb{E} \bigg[ \prod_{  \varrho \in \Phi^{\rm R} }  \upsilon ( \varrho )    \bigg] = \frac{1}{ 1 + 2 \int^{1}_{0} \big(1 - \upsilon(r) \big) r^{-3} \mathrm{d} r } , 
\end{align}
for functions $\upsilon(r) \in [0,1]$ such that $\int_{\mathbb{R}^+} (1-\upsilon(r))r^{-3} \mathrm{d}r$ is finite.

6) {\bf Sum Functional}: For any measurable function $g \geq 0$ on $\mathbb{R}^{d}$, the sum functional for a stationary point process is given by Campbell's Theorem~\cite[Theorem 4.1]{M.2013Haenggic} as
\begin{align}
\mathbb{E} \bigg[ \sum_{   j \in \mathbb{N} }  g ( x_{j} )    \bigg] = \lambda \int_{\mathbb{R}^{d}} g(x) \mathrm{d} x .
\end{align} 

7) {\bf Sum-Product Functional}: 
Let $\bm{p} \triangleq (p_{1},p_{2},\ldots,p_{q}) \in \mathbb{N}^{q}$, $\|\bm{p}\|_{1} \triangleq \sum^{q}_{k=1} p_{k}$ and $\{ M_{j } \}_{j  \in \mathbb{N} }$ denote a set of i.i.d. random marks associated with the points in $\Phi=\{x_{j}\}_{j\in\mathbb{N}}$.  
If $g_{k}(x) \in [0,\infty), 1 \leq k \leq q$,  and $\upsilon(x) \in [0,1]$ are measurable functions for all $x \in \mathbb{R}^d$, the sum-product functional for the PPP with $\|\bm{p}\|_{1}>0$ is given by \cite[Theorem 1]{U.2016Schilcher}
\begin{align} 
& \mathbb{E} \Bigg[ \prod^{q}_{k=1} \bigg( \sum_{j  \in \mathbb{N}} g(x_{j},M_{j})  \bigg)^{p_{i}} \prod_{j  \in \mathbb{N}} \upsilon(x_{j} ,M_{j})   \Bigg]  \nonumber \\
& = \exp \bigg( - \lambda \int_{\mathrm{R}^d}  \Big( 1 - \mathbb{E}_{M_{j}} \Big[ \upsilon(x,M_{j})  \Big] \Big)  \mathrm{d} x  \bigg) \nonumber \\
& \hspace{5mm}  \times \sum^{\|\bm{p}\|_{1}}_{l=1} \sum_{N \in \mathcal{N}^{p}_{l} } \frac{D_{N}}{l!} \prod^{l}_{k=1} \lambda \int_{\mathbb{R}^d} \mathbb{E}_{\bm{M}} \bigg[ \upsilon(x,M_{j}) \nonumber \\
& \hspace{5mm}  \times  \prod^{q}_{i=1} f^{n_{ki}}_{i} (x,M_{j})  \bigg] \mathrm{d} x,
\end{align} 
where $D_{N} \triangleq \prod^{q}_{k=1} \frac{p_{k}!}{\prod^{l}_{i=1} m_{ik} ! } $ and $\mathcal{N}^{\bm{p}}_{l} \subset \mathbb{N}^{q \times l}$ is the class of all $q \times l$ matrices with the columns $\|n_{\cdot i}\|_{1} > 0$, $\forall i \in \{1,2,\ldots, l\}$ and the rows $\| n_{ k \cdot } \|_{1} = p_{k}$,  $\forall k \in \{1,2\dots,q\}$.

In the special case $q=1$ and $p=1$, the sum-product functional for the PPP is given by the Campbell-Mecke Theorem~\cite[Theorem 8.2]{M.2013Haenggic} as
\begin{align} \label{eqn:Campbell-Mecke}
& \mathbb{E} \Bigg[ \sum_{j  \in \mathbb{N}} g(x_{j},M_{j}) \prod_{j  \in \mathbb{N}} \upsilon(x_{j},M_{j}) \Bigg]  \nonumber  \\
& = \exp \bigg(  - \lambda \int_{\mathbb{R}^d} \Big( 1 - \mathbb{E}_{M_{j}} \Big[ \upsilon(x,M_{j})  \Big]  \Big)  \mathrm{d}x \bigg)  \nonumber \\
& \hspace{10mm} \times \lambda \int_{\mathbb{R}^d} \mathbb{E}_{M_{j}} \Big[ g(x,M_{j}) \upsilon(x,M_{j}) \Big] \mathrm{d} x.
 \end{align}

\begin{definition} (Mat\'{e}rn cluster process): 
The MCP $\Phi^{\rm M}$ is a doubly Poisson cluster process 
constructed from a parent PPP $\Phi_{\rm p}=\{x_{j}\}_{ j \in \mathbb{N}}$ with intensity $\lambda_{\rm p}$ with each point of $\Phi_{\rm p}$ substituted by a
daughter cluster consisting of a PPP with an average number of points $\bar{c}$ within a disk of radius $R_{\rm d}$ centered at that point. 
\end{definition}

\vspace{0.2cm}
{\em Properties:}
\vspace{0.2cm}

1) {\bf Density}: For an MCP, the first moment and second moment densities are given, respectively, by  \cite{M.2013Haenggic} 
\begin{align}  
& \rho^{(1)}(x)=\lambda = \lambda_{\rm p} \bar{c},  \label{eqn:rho1_MCP} \\
& \rho^{(2)}(x,y) =  \lambda^2_{\rm p} \bar{c}^2 +  \lambda_{\rm p}  \frac{\bar{c}^2 }{\pi^2 R^4_{\rm d} } A_{R_{\rm d}}(\|x-y\|),  \label{SOD_MCP1}
\end{align}
$A_{R_{\rm d}}$ represents the area of the intersection of two disks with radius $R_{\rm d}$ at a distance $r > 0$ given as \cite[Eq. (6.4)]{M.2013Haenggic} 
\begin{align} \label{eqn:spatial_AR}
A_{R_{\rm d}}(r) =  \begin{dcases}   2 R^2_{\rm d}  \arccos(r/2R_{\rm d}) - r \sqrt{R^2_{\rm d}-r^2/4 } ,  \quad  
\\
\hspace{38mm}	0 \leq r \leq 2 R_{\rm d} ,     \\
0 \hspace{36mm} \textup{otherwise.}
\end{dcases}  
\end{align}

2) {\bf Daughter Point Distribution}:  Each daughter point is located uniformly within a disk of radius $R$ around the origin with the PDF given by  \cite[Eq. (5)]{K.2009Ganti}
\begin{align} \label{PDF_MCP_daughter}
 f^{\rm M}(y)  = \begin{cases}
\frac{1}{\pi R^2_{\rm d} },  \quad \text{if} \,\, \|y\| \leq R_{\rm d}  \\
0,   \quad \quad \text{otherwise}.
\end{cases}   
\end{align} 

3) {\bf Contact Distance Distribution}:
The PDF and CDF of the contact distance in the MCP are given as (\ref{eqn:pdf_CD_MCP}) and  (\ref{eqn:cdf_CD_MCP}), respectively~\cite{M.2017Afshang}. 
\begin{figure*}
\begin{align}  \label{eqn:pdf_CD_MCP}  
  f^{\rm M}_{r_{1}} (r|x)  = \begin{dcases}
 \frac{2 r }{R^2_{\rm d}} ,   \hspace{39mm} \textup{if} \,\, 0 \leq r \leq R_{\rm d} -x  \,\, \textup{for} \,\, \|x\| < R_{\rm d}  \\
 \frac{2 r }{\pi R^2_{\rm d}} \cos^{-1} \Big ( \frac{r^2+x^2-R^2_{\rm d}}{2 r x}  \Big) ,   \hspace{2mm} \textup{if} \,\,  R_{\rm d} - \|x\| \leq r \leq R_{\rm d} - \|x\|   \,\, \textup{for} \,\,  \|x\| < R_{\rm d} , \,\,  \textup{and} \,\,   \textup{if} \,\, \|x\|> R_{\rm d} ,   
\end{dcases}   
\end{align}
\hrulefill
\end{figure*}
\begin{figure*}
\vspace{-3mm}
\begin{align}
 F^{\rm M}_{r_{1}} (r) & = 1 - \exp\Bigg(  -  \lambda_{\rm p} \bigg(  \int_{\mathbb{B}(0,R_{\rm d})}   \bigg[ 1 - \exp \bigg( - \bar{c} \bigg( \int^{\min(r,R_{\rm d}-x)}_{0} \frac{2y}{R^{2}_{\rm d} } \mathrm{d}y \nonumber \\
& \hspace{10mm} + \int^{\min(r,R_{\rm d}+x)}_{\min(r,R_{\rm d}-x)}   \frac{2 y }{\pi R^2_{\rm d}} \cos^{-1} \Big ( \frac{y^2+x^2-R^2_{\rm d}}{2 y x}  \Big) \mathrm{d} y  \bigg)  \bigg) \bigg]  \mathrm{d} x  \nonumber \\
& \hspace{20mm}  + \int_{\mathbb{R}^2 \backslash \mathbb{B}(0,R_{\rm d})} \bigg[  1 - \exp \bigg( - \bar{c} \int^{\min(r, R_{\rm d}+x) }_{\min(r, R_{\rm d}-x)} \frac{2 y }{\pi R^2_{\rm d}} \cos^{-1} \Big ( \frac{y^2+x^2-R^2_{\rm d}}{2 y x} \mathrm{d} y \Big)   \bigg) \bigg]   \mathrm{d} x  \bigg)  \Bigg),   \label{eqn:cdf_CD_MCP}
 \end{align}  
 \hrulefill
 \end{figure*}

4) {\bf Probability Generating Functional}:  
 The PGFL for an MCP is given as follows~\cite{K.2009Ganti}:
 \begin{align} \label{eqn:PGFL_MCP}
   & \mathbb{E} \bigg[ \prod_{j \in \mathbb{N} } \upsilon  (x_{j}) \bigg] \nonumber \\ &= \exp \Bigg( \! - \lambda \int_{\mathbb{R}^2} \bigg[ 1   - M \bigg(  \int_{\mathbb{R}^2} \upsilon(x+y) f^{\rm M}(y) \mathrm{d} y   \bigg) \bigg] \mathrm{d} x  \Bigg)    \mathrm{d} y,
 \end{align}
 where $M(t)=e^{-\bar{c}(1-t) }$ is the moment generating function of the representative cluster in an MCP. 
 
Let $\mathbb{E}^{!}_{o}[\cdot]$ denote the expectation operation based on the reduced Palm measure~\cite{N.2013Chiu} which takes the expectation for a point process conditioned at a point of the process at $o$ without including the point. The conditional PGFL of an MCP is~\cite[Lemma 1]{K.2009Ganti}  
\begin{align} \label{eqn:reduced_PGFL}
& \mathbb{E}^{!}_{o} \bigg[ \prod_{j \in \mathbb{N} }  \upsilon(x_{j}) \bigg] \nonumber \\
 & \! =  
\exp \! \Bigg( \! - \lambda_{\rm p} \int_{\mathbb{R}^2} \!\! \bigg[  1 - M \bigg( \int_{\mathbb{R}^2} \! \upsilon(x+y) f(y) \mathrm{d}y \bigg) \! \bigg] \mathrm{d}x \! \Bigg) \nonumber \\
& \hspace{22mm} \times \int_{\mathbb{R}^2}  G^{\rm M}_{\rm d} \Big( \upsilon(x-y) \Big) f^{\rm M}(y) \mathrm{d} y,
\end{align}
where $G^{\rm M}_{\rm d}$ is the PGFL for the representative cluster given by
\begin{align}
G^{\rm M}_{\rm d} (\upsilon) = M \bigg( \int_{\mathbb{R}^2}  \upsilon(x) f^{\rm M}(x) \mathrm{d} x \bigg).
\end{align}

\begin{definition} ($\beta$-Ginibre point process):
A $\beta$-GPP $\Phi^{\rm G}_{\beta} =  \{ x_{j}  \}_{j \in \mathbb{N}}$ is a determinantal point process with the kernel\footnote{The kernel represents the interaction force among the points of the process.} given by \cite{N.Jan.2015Deng}
\begin{align}
\mathcal{K}_{\beta,\lambda} (x,y) = \lambda    e^{-\frac{\pi \lambda |x-y|^2}{2 \beta} }   , \quad  x,y \in \mathbb{C} ,
\end{align} 
w.r.t. the Lebesgue measure  
 on $\mathbb{C}$~\cite{N.2014Miyoshi}. 
\end{definition} 

\vspace{0.2cm}
{\em Properties:}
\vspace{0.2cm}

1) {\bf Density}: For a $\beta$-GPP, the first moment density and the second moment density of the $\beta$-GPP are given, respectively, by~\cite{N.Jan.2015Deng}
\begin{align}  
& \rho^{(1)}(x)= \det[ K_{\beta,\lambda}(x,x)]= \lambda, \label{eqn:rho1_GPP}  \\
 &   \hspace{-3mm} \quad \rho^{(2)}(x,y)= \det \left[ \begin{array}{cc}
    \mathcal{K}_{\beta,\lambda} (x,\bar{x}) & \mathcal{K}_{\beta,\lambda} (x,\bar{y} )\\
    \mathcal{K}_{\beta,\lambda} (\bar{x},y) & \mathcal{K}_{\beta,\lambda} (y,\bar{y} ) 
   \end{array} 
   \right ]\nonumber\\
   & =  \lambda^2  \bigg(1 -  \exp \bigg( \! \! - \! \frac{\pi  \lambda  |x -  y|^2  }{ \beta } \bigg)   \bigg) .  \label{SOD_GPP1}
   \end{align}

2)  {\bf Link Distance Property}:   
  Let $\lambda$ represents the intensity of $\Phi^{\rm G}_{\beta}$.
  Let $\{Q_{j}\}_{ j\in \mathbb{N} }$ be a set of  independent gamma random variables with PDF 
  \begin{align} \label{eqn:PDF_gamma}
  f_{Q_{j}} (q) = \frac{q^{j-1} e^{-\frac{\pi \lambda}{\beta} q}}{ \Big ( \frac{\beta}{ \pi \lambda}  \Big)^{j} \Gamma (j) },
  \end{align} 
 i.e., $Q_{j} \sim \mathcal{G}(j,\beta/ \pi \lambda)$.  Then the set $\{ \|x_{j}\|^2 \}_{j \in \mathbb{N}}$ 
 is equivalent in distribution with the set constructed by retaining each element from $\{Q_{j}\}_{ j\in \mathbb{N} }$ independently with probability $\beta$ \cite[Theorem 4.7.1]{B.2009Hough}.

\begin{figure*} [htp]
\centering
 \subfigure [ Realization of an MCP ($\lambda_{\rm p}=0.1$, $\bar{c}=5$ and $R_{\rm d}=1$). ]
  {
 \centering   \hspace{-3mm}
 \includegraphics[width=0.53\textwidth]{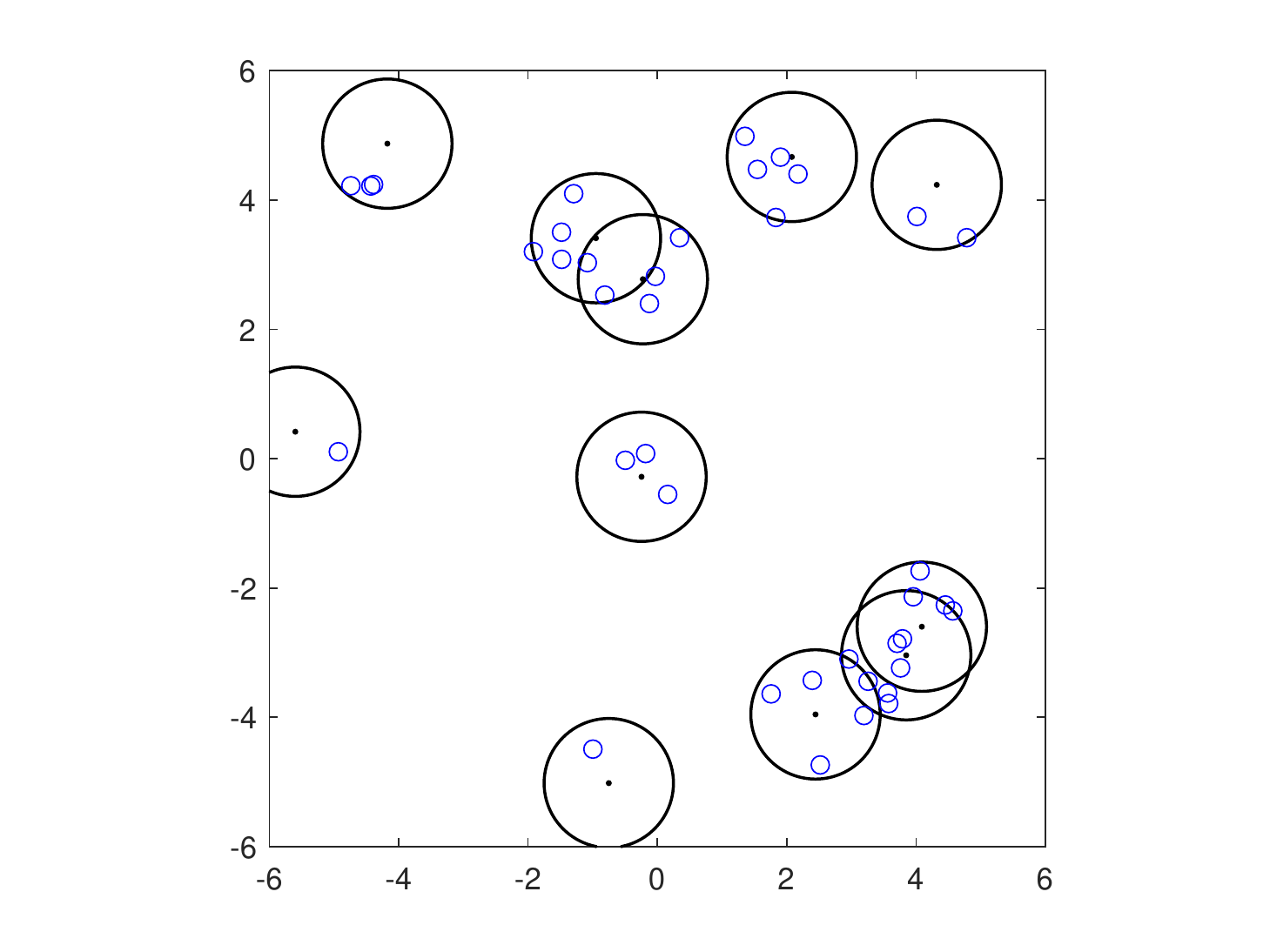}}
 \centering  
 \subfigure  [ Realization of a homogeneous PPP ($\lambda=0.5$).
 ] {
 \centering \hspace{-2mm}
\includegraphics[width=0.41 \textwidth]{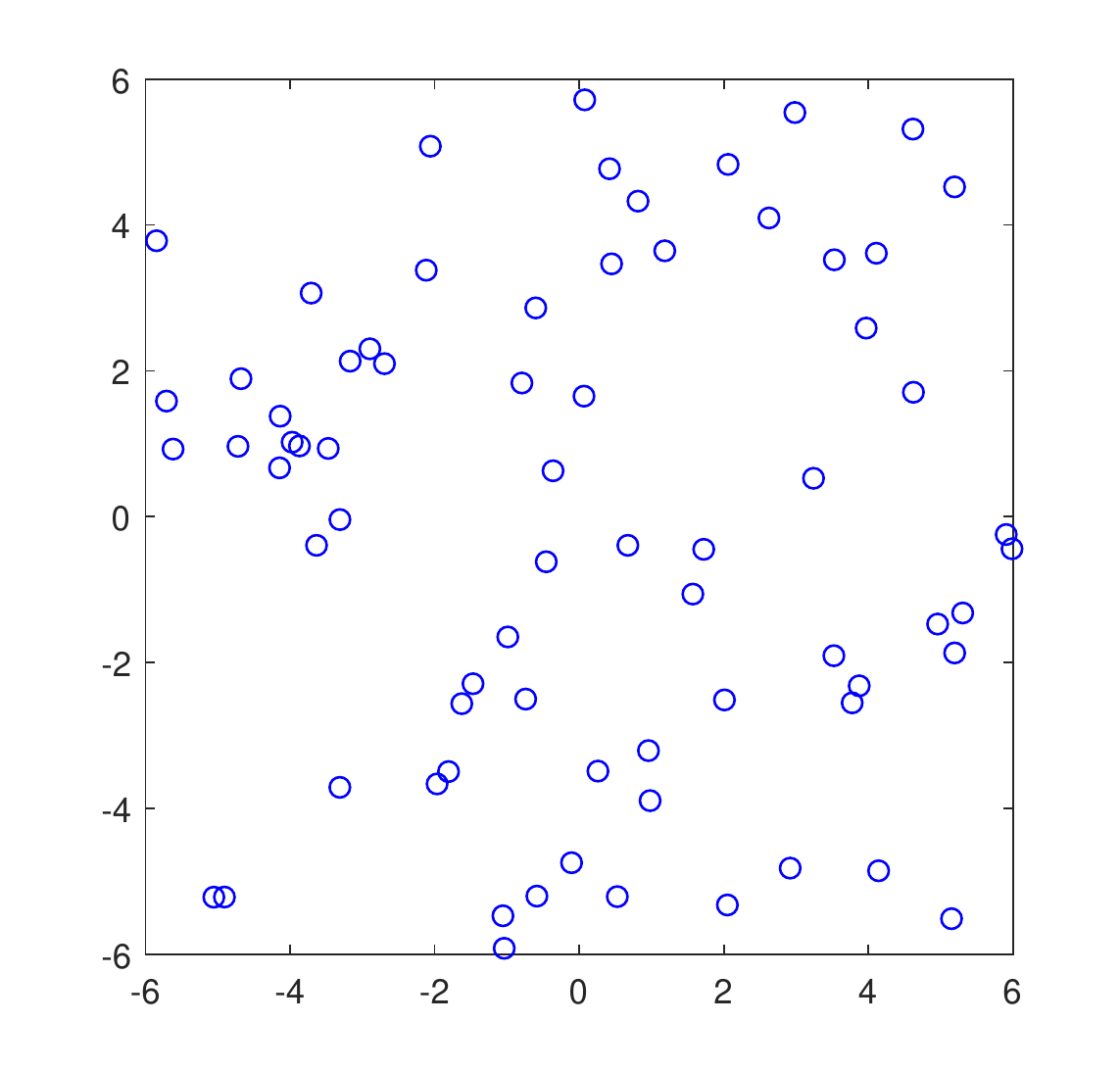}}
\centering
 \subfigure  [ Realization of a $\beta$-GPP ($\lambda=0.5$, $\beta=0.5$). 
 ] {
 \centering \hspace{4mm}
\includegraphics[width=0.41 \textwidth]{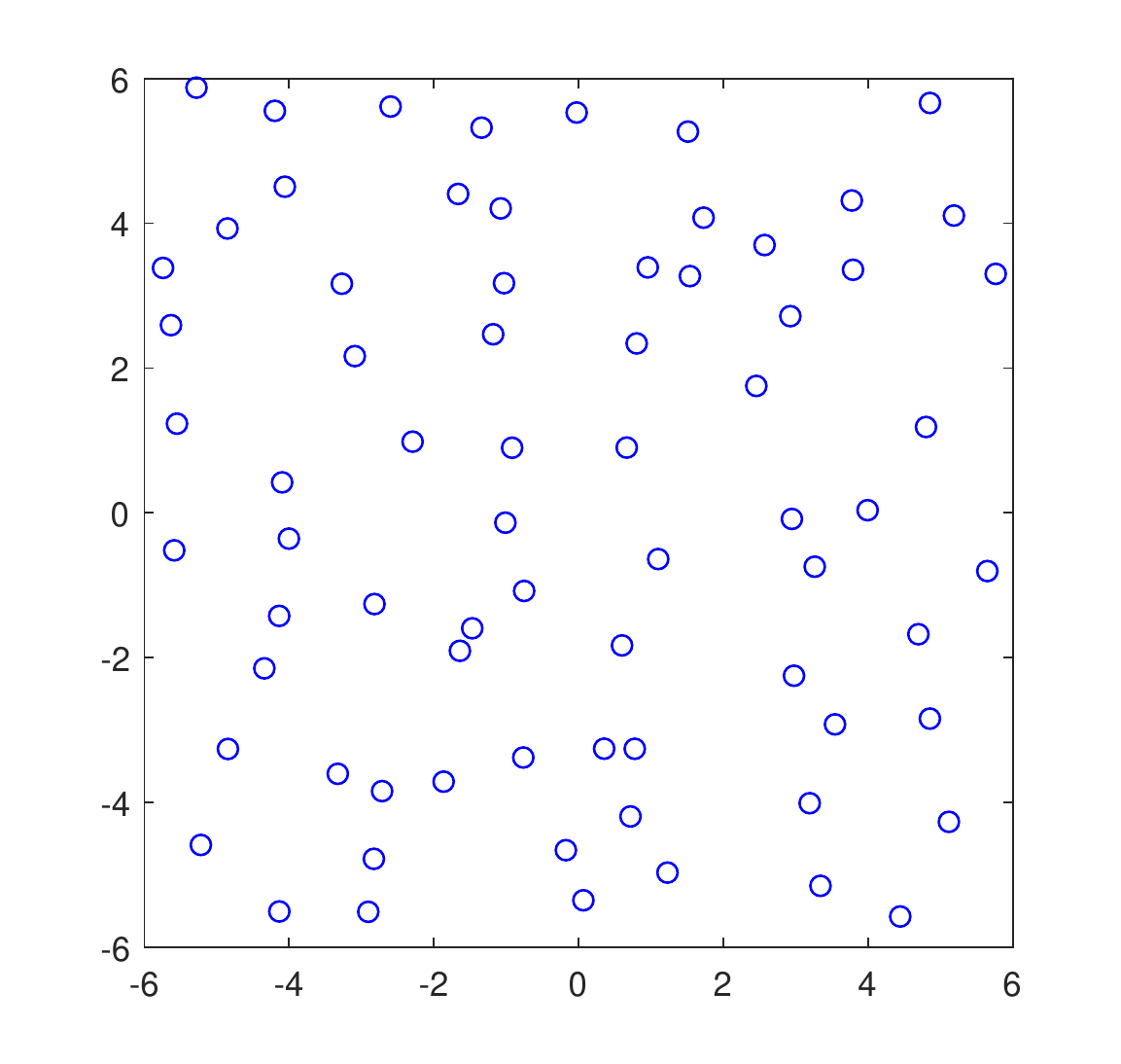}}
\centering
 \subfigure  [ Pair correlation functions ($\lambda=1$, $R_{\rm d}=1$). 
 ] {
 \label{fig:Spatial_PCF}
 \centering  \hspace{7mm}
\includegraphics[width=0.41 \textwidth]{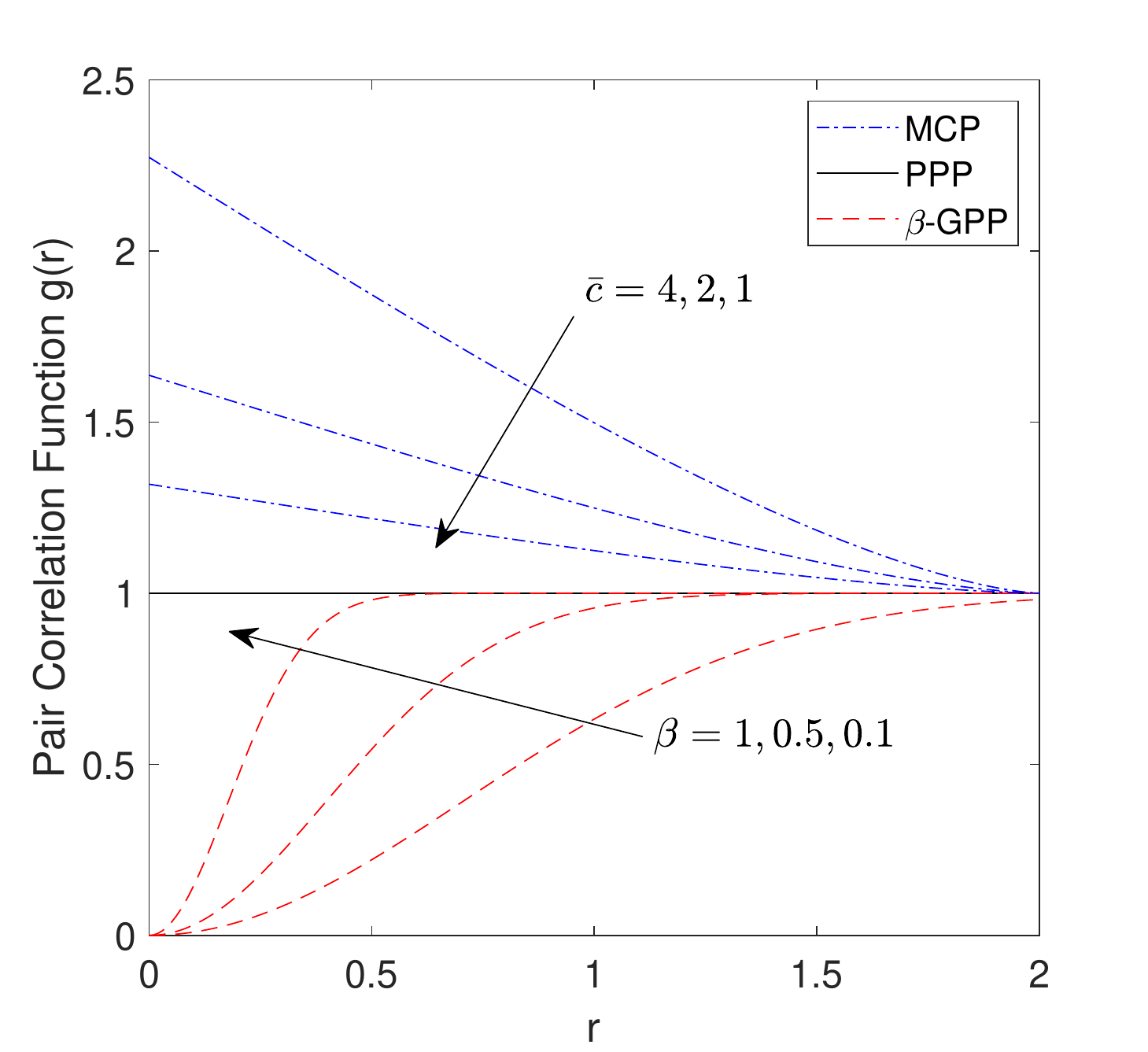}}
\caption{Illustration of spatial point processes. } 
\centering
\label{fig:Illustration_SpatialPP} 
\end{figure*}

Fig. \ref{fig:Illustration_SpatialPP}(a), (b), (c), respectively, illustrate realizations of MCP, PPP, and $\beta$-GPP. As shown by the realizations of MCP and $\beta$-GPP, the point sets are attractive and repulsive, respectively.  
 
The correlation between the spatial points can be measured by the pair correlation function (PCF). 
For a point process $\Phi \subset \mathbb{R}^{d}$, the PCF is defined as $g(x,y) \triangleq  \frac{\rho^{(2)}(x,y)}{\rho^{(1)}(x)\rho^{(2)}(y)}$. If $\Phi$ is motion-invariant, 
the first moment density is the constant intensity and the second moment density $\rho^{(2)}(x,y)$ only depends on the difference  $r=\|x-y\|$. 
Hence, the PCF can be expressed as $g(r)= \frac{\rho^{(2)}(r)}{\lambda^2}$.

Let $\Phi(A)$ denote the number of points in $\Phi \cap A$.  
The PCF quantifies the degree of correlation between the random variables $\Phi(A)$ and $\Phi(B)$ in a non-centered way.  
If $g(x,y)\equiv 1$
then for disjoint $A$ and $B$ the  covariance of $\Phi(A)$ and $\Phi(B)$ is zero, which means the two variables are uncorrelated. For spatial point processes, the PCF describes how a point is surrounded by others.
The PCF equals one if the points are uncorrelated (e.g., as in the PPP), and is greater (smaller) than 1 if the points are attractive (repulsive).

The PCFs of the MCP, PPP and $\beta$-GPP are given, respectively, as \cite[Page 153]{M.2013Haenggic}, \cite{N.Jan.2015Deng} 
\begin{align}
g(r)= \begin{dcases}1+ \frac{ \bar{c} A_{R_{\rm d}}(r) }{\lambda  \pi^2 R^4_{\rm d}   } , \hspace{24mm} \textup{MCP} \\
1,  \hspace{42mm} \textup{PPP}  \\
1-\frac{\exp( -r^2/\beta)}{\pi \lambda},  \hspace{18mm} \textup{$\beta$-GPP} ,
\end{dcases} 
 \end{align}
 where $A_{R_{\rm d}}(r)$ is given in (\ref{eqn:spatial_AR}). 
 
 \vspace{0.2cm}
{\em Properties:}
\vspace{0.2cm}
 
1)  For $r\geq 2R_{\rm d}$, for an MCP, the PCF is the same as that of the PPP, because two points with a distance greater than $2R_{\rm d}$ must belong to different clusters and thus are independent. 
 
2) Given the intensity $\lambda$, the PCF of an MCP approaches that of the PPP as $\bar{c} \to 0$, since the points are less likely to belong to the same cluster.

3) For a PPP, the PCF is not affected by the intensity since the points are independently distributed.  
 
4) For a $\beta$-GPP, the PCF approaches that of the PPP as $\beta \to 0$ or $\lambda \to \infty$.
 
Fig.~\ref{fig:Illustration_SpatialPP}(d) shows the PCFs of the spatial point processes and demonstrates the properties discussed above.

\subsubsection{Validity} 
 
The validity of the point process approximation of real-world cellular networks has been demonstrated in the literature. 
In particular, an early study in~\cite{X.2000Brown} found that the success probability of an actual cellular network (with a degree regularity) is lower-bounded by that of a Poisson network (with complete randomness). 
 More importantly, reference~\cite{M.Dec.2014Haenggi} 
 revealed that adding a horizontal SIR threshold shift from 0 dB to 3.4 dB to the success probability in a Poisson network yields a tight approximation of that in any network with a spatial layout from complete randomness to perfect regularity (i.e., triangular lattice).
 Hence, the analysis of an actual cellular network can be performed based on that of a Poisson network with a horizontally shifted SIR threshold calibrated to the regularity of the topology.
 
 The validity of the $\beta$-GPP to model actual cellular networks has been rigorously explored. The studies in \cite{N.Jan.2015Deng,S.2015Gomez}
 show that the repulsion parameter $\beta$ that adjusts the degree of repulsion among points can be numerically fitted with the data set from actual network deployments. The simulations demonstrate 
 that the success probabilities from $\beta$-GPP analysis closely approximates the ones from simulations with the actual deployment topologies. Hence, it is valid to employ either a horizontally shifted PPP  or the fitted $\beta$-GPP to analyze actual cellular networks.

\subsection{Performance Analysis} \label{sec:spatial_PA}
 
\subsubsection{Spatial-Temporal Interference Correlation in Ad Hoc Networks} \label{sec:spatial_PA_IC} 
We first investigate the correlation of interference observed at two locations $o$ and $u$ in two time slots $t_{1}$ and $t_{2}$, respectively, with the interferers distributed as the three types of point processes considered. As PPP, MCP and $\beta$-GPP are motion-invariant, the spatial-temporal interference correlation can be measured by the Pearson correlation coefficient defined in (\ref{def:ICC2}).
 
As can be seen from (\ref{def:ICC2}), the expectation and the second moments of $I^{(t_{1})}_{o}$ and the mean product of $I^{(t_{1})}_{o}$ and $I^{(t_{2})}_{u}$ are needed to quantify the interference correlation.
However, these two quantities do not
exist because of the singularity of the path-loss function $\ell(x)$.  To cope with this issue, we follow the approach in~\cite{K.2009Ganti} by defining $\ell_{\epsilon}(x)=\frac{1}{\epsilon+\|x\|^{\alpha}}, \alpha>2, \epsilon>0$, such that $\ell(x)=\lim_{\epsilon\to 0} \ell_{\epsilon}(x)$.

The expectation of $I^{(t_{1})}_{o}$ can be derived as
 \begin{align}  
 \mathbb{E} \big[I^{(t_{1})}_{o} \big] & \overset{(a)}{=} \mathbb{E} \big[I_{o} \big] \nonumber \\
  & = \mathbb{E} \bigg [      \sum_{j \in \mathbb{N} } \!    h_{j} \ell_{\epsilon}(x_{j} )      \bigg] \nonumber \\ 
  &   = \mathbb{E} \big[  h   \big]   \int_{\mathbb{R}^2}   \ell_{\epsilon}(x ) \rho^{(1)}(x) \mathrm{d} x   \nonumber \\
 & \overset{(b)}{=} 2 \pi \lambda  \int^{\infty}_{0} \frac{x}{ \epsilon + x^{\alpha} }
 \mathrm{d}x  \nonumber \\
 & \overset{(c)}{=} \delta \pi^2 \lambda \epsilon^{\delta-1} \csc(\delta \pi) , \label{MeanInterference}
 \end{align}
 where $(a)$ follows since MCP, PPP, and $\beta$-GPP are all motion-invariant and the superscript $(t_{1})$ is dropped for conciseness, $(b)$ applies the conversion from Cartesian to polar coordinates, and $(c)$ substitutes $\frac{2}{\alpha}$ with $\delta$.

If the densities of the interferers following a PPP, MCP, and $\beta$-GPP as given in (\ref{eqn:rho1_PPP}),  (\ref{eqn:rho1_MCP}), and (\ref{eqn:rho1_GPP}), respectively,  are the same, 
we can observe from (\ref{MeanInterference}) that the three point processes cause the same mean interference at an arbitrary location. This indicates that spatial attraction and repulsion do not affect the first moment of the interference. 
 
We then continue to derive the second moment of  $I^{(t_{1})}_{o}$ as    
\begin{align}  
    & \mathbb{E} \Big[ \Big( I^{(t_{1})}_{o} \Big)^2 \Big] 
    \nonumber \\
    & =  \mathbb{E}  \Bigg [   \bigg( \sum_{ j \in 
    \mathbb{N} }    h_{j}  \ell_{\epsilon} ( x_{j} ) \bigg)^{\! 2}  \Bigg]   \nonumber \\  
    & = \! \mathbb{E} \Bigg[  \sum_{ j \in 
        \mathbb{N}  }           h^2_{j} \ell^2_{\epsilon}(x_{j})   + \sum^{ j \neq i }_{  j,i  \in \mathbb{N} }    h_{j} h_{i} \ell_{\epsilon}(x_{j}) \ell_{\epsilon}(x_{i})  \Bigg] \nonumber \\
    & = \mathbb{E} [h^2]   \int_{\mathbb{R}^2} \ell^2_{\epsilon}(x) \rho^{(1)}(x) \mathrm{d} x \nonumber\\
    & \hspace{10mm} +\mathbb{E} [h]^2  \int_{\mathbb{R}^2}  \int_{\mathbb{R}^2}  \ell_{\epsilon}(x) \ell_{\epsilon}(y) \rho^{(2)}(x,y)  \mathrm{d} x  \mathrm{d} y  \nonumber \\
    & = 2 \delta \lambda \pi^2 (1-  \delta  ) \epsilon^{\delta-2} \csc(\delta \pi)\nonumber \\
    & \hspace{15mm} +  \int_{\mathbb{R}^2}  \int_{\mathbb{R}^2}  \ell_{\epsilon}(x) \ell_{\epsilon}(y) \rho^{(2)}(x,y)  \mathrm{d} x  \mathrm{d} y.   \label{SM_Interference}
    \end{align}  By plugging in
 the second moment density $\rho^{(2)}(x,y)$ for the MCP, PPP, and $\beta$-GPP given, respectively, in  (\ref{SOD_MCP1}), (\ref{eqn:rho2_PPP}), (\ref{SOD_GPP1}), we obtain the second moment of the interference in (\ref{eqn:SM_I}), 
 \begin{figure*}
\begin{align}
  \mathbb{E} \bigg[ \Big( I^{(t_{1})}_{o} \Big)^2 \bigg]   \overset{(a)}{=} \! \begin{dcases} \bar{I}^{2}_{\textup{PPP}}  +   \frac{ \lambda \bar{c}}{\pi^2 R_{\rm d}^4} \int_{\mathbb{R}^2}\! \! \bigg(   \!\! \int_{\mathbb{R}^2} \! \ell_{\epsilon} ( x ) 
     \ell_{\epsilon} (y) A_{R_{\rm d}}(\|x-y\|)   \mathrm{d} x  \bigg)  \mathrm{d} y , \hspace{13mm} \textup{MCP} \\
   \bar{I}^{2}_{\textup{PPP}},  \hspace{85 mm}  \textup{PPP}  \\
\bar{I}^{2}_{\textup{PPP}}  -   \lambda^2 \! \int_{\mathbb{R}^2} \!\!\! \bigg( \!  \int_{\mathbb{R}^2} \!\!  \ell_{\epsilon}(x) \ell_{\epsilon}(y)   \bigg(   \exp \bigg( \!- \frac{\pi \lambda  \|x-y\|^2}{\beta} \bigg) \!  \bigg)  \mathrm{d} x  \bigg)   \mathrm{d} y,   \quad \!\! \beta \textup{-GPP} .  
    \end{dcases}  \label{eqn:SM_I}
    \end{align}
        \hrulefill
 \end{figure*} 
 where  $\bar{I}^{2}_{\textup{PPP}}$ is
\begin{align}
& \bar{I}^{2}_{\textup{PPP}} = 2 \delta \lambda \pi^2 (1-  \delta  ) \epsilon^{\delta-2} \csc(\delta \pi) \nonumber \\ 
& \hspace{30mm}  +  \delta^2 \pi^4  \lambda^2 \epsilon^{2\delta-2}   \csc(\delta \pi)^2 . 
\end{align}

Table~\ref{Tab:SM_Interference} shows the variance of the interference with different fields of interferers. It can be observed that the MCP and the $\beta$-GPP cause larger and smaller interference variance than the PPP, respectively, as illustrated in Fig.~\ref{fig:Spatial-Interference-Variance}. 
Moreover, it can be found that, for an MCP, given the interference density $\lambda=\lambda_{\rm p} \bar{c}$, the variance increases when the points are more densely clustered (i.e., with smaller cluster density $\lambda_{\rm p}$ and a larger average number of points $\bar{c}$ within each cluster). 
 
 \begin{figure}[htp]   
 \centering   
 \includegraphics[width=0.48 \textwidth]{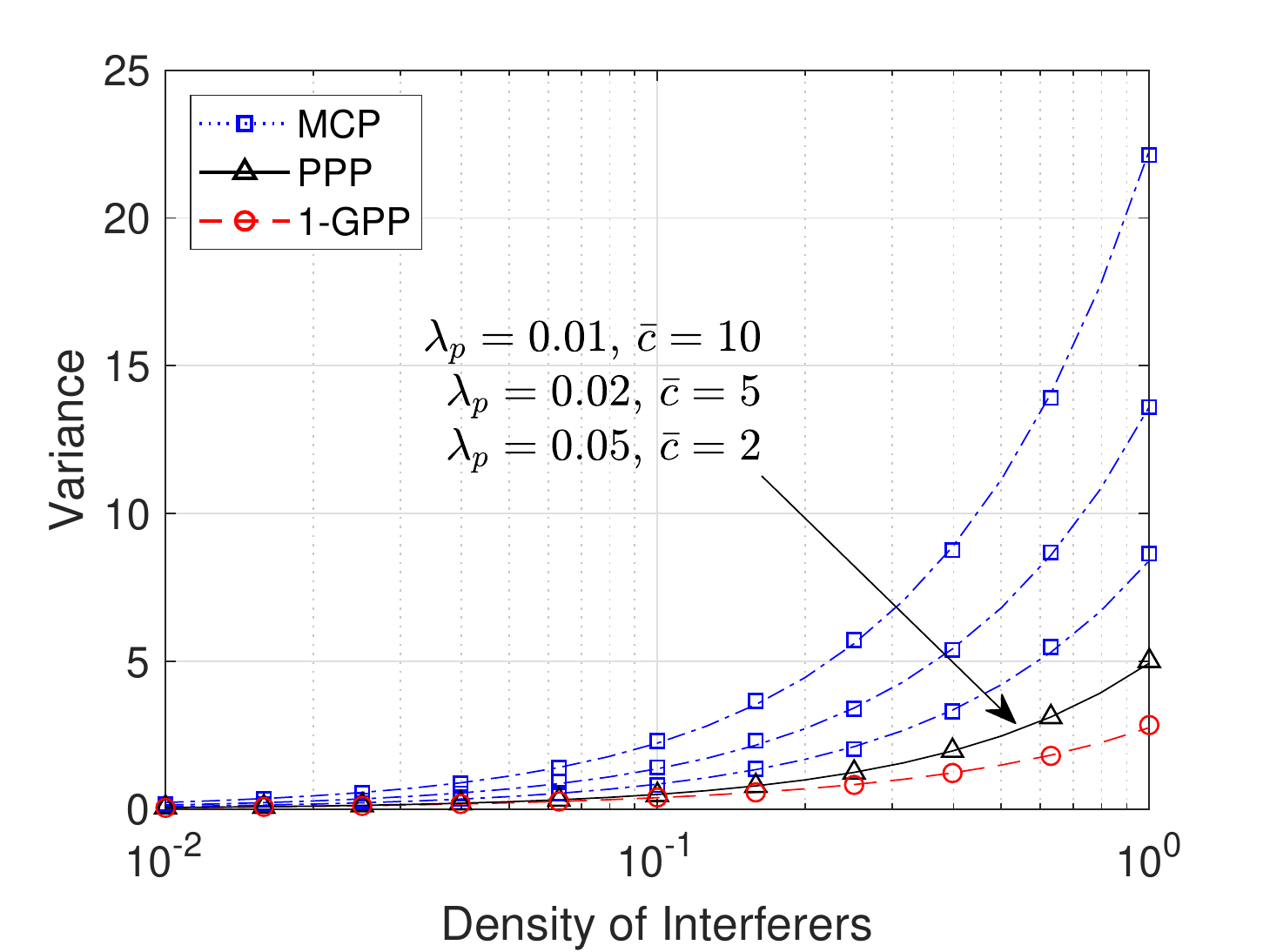}
 \caption{Interference variance in ad hoc networks with different fields of interferers  ($\ell(x)=\frac{1}{1+x^{\alpha}}$ and $R_{\rm d}=1$).}  
   \label{fig:Spatial-Interference-Variance} 
\end{figure}

\begin{table*}
\centering
\caption{\footnotesize Variance of Interference}  \label{Tab:SM_Interference} 
\begin{tabular}{|p{1cm}|p{10cm}|}
\hline
Point Process &  $\mathbb{E} [I^2_{o}] - \mathbb{E} [I_{o}]^2$  \\ \hline
\hline
MCP & $\delta  \pi^2 \lambda (1- \delta ) \epsilon^{\delta-2} \csc(\delta \pi) +   \frac{\lambda \bar{c} }{\pi^2 R_{\rm d}^4} \int_{\mathbb{R}^2} \big(  \int_{\mathbb{R}^2} 
\ell_{\epsilon} (x) \ell_{\epsilon} (y) A_{R_{\rm d}}(\|x-y\|) \mathrm{d} x  \big)   \mathrm{d} y $  \\
\hline  
PPP & $\delta  \pi^2 \lambda (1-\delta ) \epsilon^{\delta-2} \csc(\delta \pi)  $  \\
\hline  
$\beta$-GPP &  $\delta \pi^2 \lambda (1-\delta ) \epsilon^{\delta-2} \csc(\delta \pi)   -    \lambda^2 \int_{\mathbb{R}^2} \Big(  \int_{\mathbb{R}^2}  \ell_{\epsilon}(x) \ell_{\epsilon}(y)      \exp \big( - \pi \lambda \|x-y\|^2 / \beta \big)    \mathrm{d} x  \Big)   \mathrm{d} y $  \\
\hline   
\end{tabular} 
\end{table*}

Similarly, we have the mean product of $I^{(t_{1})}_{o}$ and $I^{(t_{2})}_{u}$, $t_{1} \neq t_{2}$, as 
\begin{align}  
 & \mathbb{E} \Big[ I^{(t_{1})}_{o} I^{(t_{2})}_{u} \Big]   \nonumber\\
 & =  \mathbb{E}  \Bigg [     \sum_{ j \in \mathbb{N}  }    h^{(t_{1})}_{j}  \ell_{\epsilon} (x_{j} - o)   \sum_{ i \in \mathbb{N} }    h^{(t_{2})}_{i}  \ell_{\epsilon} (x_{i} - u )  \Bigg]   \nonumber \\  
 & = \! \mathbb{E} \Bigg[  \sum_{ j \in \mathbb{N}  }             h^{(t_{1})}_{j}  h^{(t_{2})}_{j}  \ell_{\epsilon}(x_{j} ) \ell_{\epsilon}(x_{j} - u)   \nonumber\\
 &\hspace{10mm}+ \! \sum^{ j \neq i}_{j, i \in \mathbb{N}  }    h^{(t_{1})}_{j} h^{(t_{2})}_{i} \ell_{\epsilon}(x_{j} ) \ell_{\epsilon}(x_{i} - u)   \Bigg] \nonumber \\
 & = \mathbb{E} [h]^2  \! \int_{\mathbb{R}^2} \ell_{\epsilon}(x ) \ell_{\epsilon}(x - u) \rho^{(1)}(x) \mathrm{d} x  \nonumber\\
 &\hspace{10mm}  + \mathbb{E} [h]^2 \! \int_{\mathbb{R}^2}  \int_{\mathbb{R}^2} \! \ell_{\epsilon}(x ) \ell_{\epsilon}(y ) \rho^{(2)}(x,y)  \mathrm{d} x  \mathrm{d} y,  \label{eqn:MP_I} 
\end{align}
which is an integral function of the first and second moment densities. Subsequently, (\ref{eqn:MP_I}) can be obtained by following the derivation of the second moment of interference as in (\ref{eqn:MP_Interference}). 
 \begin{figure*}  
 \begin{align}
& \mathbb{E} \Big[ I^{(t_{1})}_{o} I^{(t_{2})}_{u} \Big] = \! 
  \begin{dcases}  \lambda \int_{\mathbb{R}^2} \ell_{\epsilon}(x) \ell_{\epsilon}(x-u)   \mathrm{d} x  +  \delta^2 \pi^4  \lambda^2 \epsilon^{2\delta-2}   \csc(\delta \pi)^2 + \frac{ \lambda \bar{c} }{\pi^2 R_{\rm d}^4} \\
  	\hspace{45mm} \times \int_{\mathbb{R}^2} \!\! \bigg( \!  \int_{\mathbb{R}^2} \ell_{\epsilon} (  x)   \ell_{\epsilon} (y) A_{R_{\rm d}}(\|x-y\|)  \mathrm{d} x  \bigg)  \mathrm{d} y,  \hspace{4mm}  \textup{MCP} \\
  \lambda \int_{\mathbb{R}^2} \ell_{\epsilon}(x) \ell_{\epsilon}(x-u)   \mathrm{d} x  + \delta^2 \pi^4  \lambda^2 \epsilon^{2\delta-2}   \csc(\delta \pi)^2,  \hspace{42mm}     \textup{PPP}  \\
    \lambda \int_{\mathbb{R}^2} \ell_{\epsilon}(x) \ell_{\epsilon}(x-u)   \mathrm{d} x  + \delta^2 \pi^4  \lambda^2 \epsilon^{2\delta-2}   \csc(\delta \pi)^2 \\ 
      \hspace{33mm}   -  \lambda^2 \!  \int_{\mathbb{R}^2} \!\!\! \bigg( \!  \int_{\mathbb{R}^2}\!\!  \ell_{\epsilon}(x) \ell_{\epsilon}(y)   \exp \bigg( \!- \frac{\pi \lambda \|x-y\|^2 }{\beta}\bigg)    \mathrm{d} x  \bigg)  \mathrm{d} y,   \quad \!\! \beta\textup{-GPP}   .
          \end{dcases}   \label{eqn:MP_Interference}
          \end{align}
             \hrulefill
          \end{figure*}

Finally, by inserting the expectation, second moment, and mean interference product given in (\ref{MeanInterference}), (\ref{SM_Interference}), and (\ref{eqn:MP_I}), respectively, 
into (\ref{def:ICC2}), we have the spatial-temporal interference correlation coefficient in the following theorem.

\begin{theorem}
The spatial-temporal correlation coefficient with interferers distributed as MCP, PPP, and $\beta$-GPP and path-loss function $\ell_{\epsilon}(x)=\frac{1}{\epsilon+x^{\alpha}}$ is given by (\ref{eqn:zeta}),  
\begin{figure*} 
\begin{align}
\zeta(\|u\|) = \begin{dcases}\frac{  \int_{\mathbb{R}^2} \ell_{\epsilon}(x) \ell_{\epsilon}(x-u)   \mathrm{d} x + \varpi^{\textup{MCP}} (R_{\rm d},\bar{c}) }{2 \delta  \pi^2 ( 1 -  \delta  ) \epsilon^{\delta-2} \csc(\delta \pi) + \varpi^{\textup{MCP}} (R_{\rm d},\bar{c}) },  \hspace{18mm} \,  \textup{MCP} \\
\frac{   \int_{\mathbb{R}^2} \ell_{\epsilon}(x) \ell_{\epsilon}(x-u)   \mathrm{d} x   }{2 \delta   \pi^2 (1-  \delta  ) \epsilon^{\delta-2} \csc(\delta \pi)  }, \hspace{42mm}  \textup{PPP}  \\
 \frac{   \int_{\mathbb{R}^2} \ell_{\epsilon}(x) \ell_{\epsilon}(x-u)   \mathrm{d} x - \varpi^{\beta\textup{-GPP}} (\lambda) }{2 \delta   \pi^2 ( 1 -  \delta  ) \epsilon^{\delta-2} \csc(\delta \pi) - \varpi^{\beta\textup{-GPP}} (\lambda) } , \hspace{22mm}  \beta\textup{-GPP} ,  
 \end{dcases} \label{eqn:zeta} 
\end{align}
 \hrulefill  
\end{figure*}
where  $\varpi^{\textup{MCP}} (R_{\rm d},\bar{c})$ and $\varpi^{\beta\textup{-GPP}} (\lambda)$ are given, respectively, by
 \begin{align}
&\varpi^{\textup{MCP}} ( R_{\rm d},\bar{c}) \nonumber\\
&= \frac{  \bar{c}}{\pi^2 R_{\rm d}^4} \int_{\mathbb{R}^2} \!\! \bigg( \! \int_{\mathbb{R}^2}  \!\! \ell_{\epsilon} (x) \ell_{\epsilon} (y)   A_{R_{\rm d}}(\|x-y\|)  \mathrm{d} x   \bigg)   \mathrm{d} y, \nonumber 
\end{align} 
and
\begin{align}
& \varpi^{\beta\textup{-GPP}} (\lambda) \nonumber \\  & = \! \lambda  \! \int_{\mathbb{R}^2} \!\!\! \bigg( \!\int_{\mathbb{R}^2} \!\!   \ell_{\epsilon}(x) \ell_{\epsilon}(y)    \exp \bigg( \! - \! \frac{\pi \lambda  \|x\!-\!y\|^2 }{\beta} \bigg)     \mathrm{d} x  \bigg) \mathrm{d} y.  \nonumber
\end{align} 
\end{theorem}

\noindent
{\bf Remark 1}: The interference correlation coefficient for the MCP is greater than for the PPP, 
since
 \begin{align}
 \zeta^{\textup{MCP}}(\|u\|)-\zeta^{ \textup{PPP}}(\|u\|)  
 = \frac{ (C_{2} - C_{1})  \varpi^{\textup{MCP}}(R_{\rm d},\bar{c}) }{ C_{2}(C_{2} + \varpi^{\textup{MCP}}(R_{\rm d},\bar{c}))   } \overset{(a)}{>} 0 ,  \nonumber
 \end{align}
where $(a)$ holds as  $\varpi^{\textup{MCP}}(R_{\rm d},\bar{c})>0$ and $C_{2}-C_{1} \geq 1/2$ \cite{K.2009Ganti}, and
$C_{1}=  \int_{\mathbb{R}^2} \ell_{\epsilon}(x) \ell_{\epsilon}(x-u)   \mathrm{d} x $ and $C_{2}= 2 \delta  \pi^2 (1-  \delta  ) \epsilon^{\delta-2} \csc(\delta \pi)$.

\noindent
{\bf Remark 2}: The interference correlation coefficient for the $\beta$-GPP is smaller than for the PPP, since
\begin{align} 
 \zeta^{\beta\textup{-GPP}}(\|u\|)-\zeta^{ \textup{PPP}}(\|u\|)  
 = \frac{ (C_{1} - C_{2})  \varpi^{\beta\textup{-GPP}}(\lambda) }{ C_{2}(C_{2} - \varpi^{\beta\textup{-GPP}}(\lambda))   } \overset{(a)}{<} 0, \nonumber 
\end{align}
where $(a)$ holds as   $C_{1}<C_{2} $ and $C_{2}>\varpi^{\beta\textup{-GPP}}$ (since $\exp \big( \!- \pi \lambda \|x-y\|^2 / \beta \big) < 1$). 

Fig. \ref{fig:Spatial_CC} shows the spatial-temporal correlation coefficient for different fields of interferers, which illustrates the properties discussed above.

\begin{figure}[htp]
\centering
\includegraphics[width=0.5\textwidth]{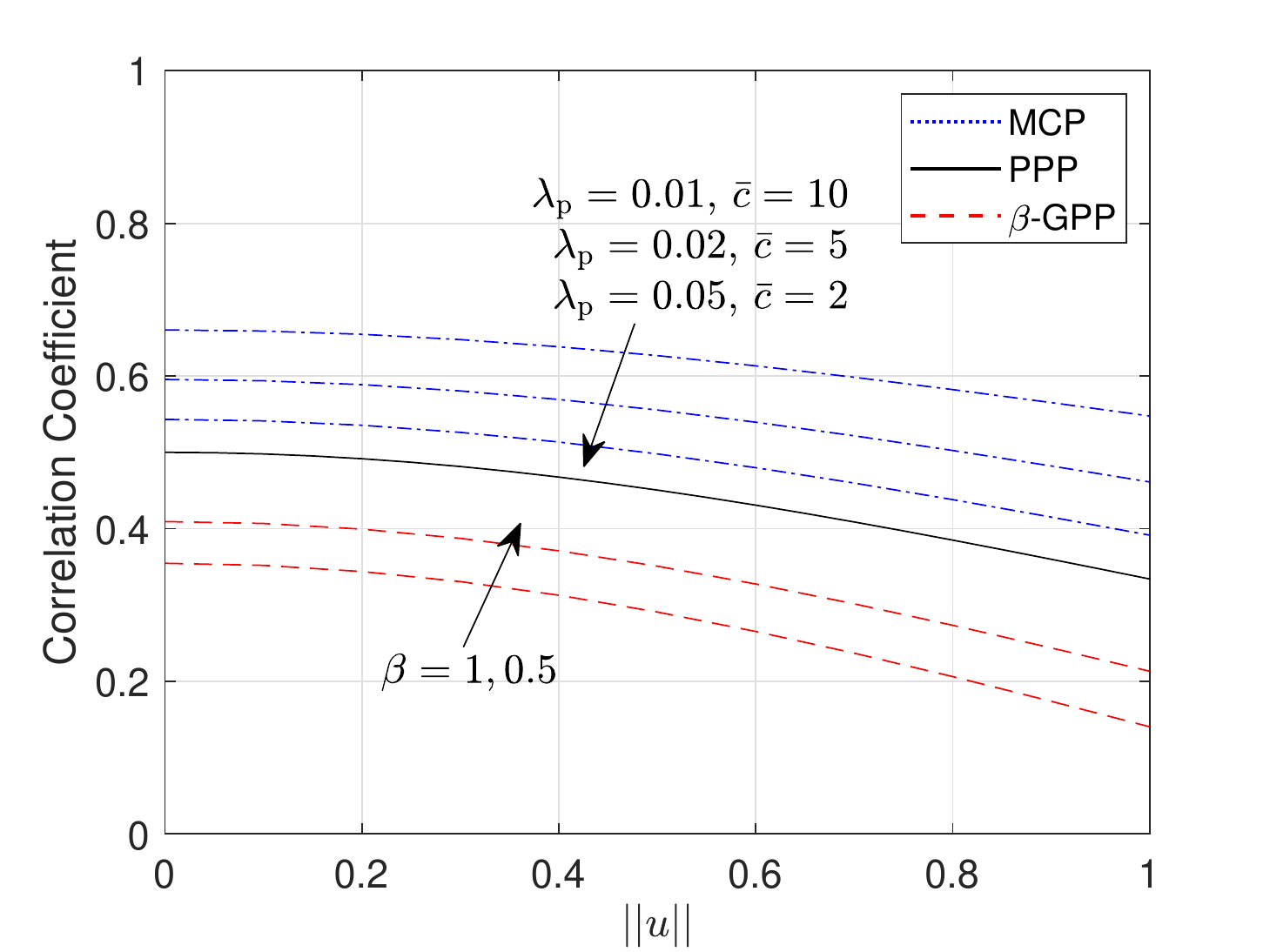} 
\caption{Correlation coefficient in ad hoc networks with different fields of interferers ($\lambda=0.1$, $\epsilon=1$ and $R_{\rm d}=1$).} \label{fig:Spatial_CC}
\end{figure}

\vspace{0.2cm}
\subsubsection{Moments of the CSP$_{\Phi}$ in Ad Hoc Networks}
 

Next, we derive the moments of the CSP$_\Phi$ in three different random fields of interferers by following the methodology in \cite{M.Apr.2016Haenggi}. 
For this,  
\begin{itemize}
\item We start by obtaining the CSP$_{\Phi}$ by averaging out the randomness of channel gains for a given point process $\Phi$ and a link distance.

\item     
We then derive the moments of the CSP$_\Phi$ by averaging over the 
spatial distributions of the interferers in the different random fields. 
Specifically,

\begin{itemize}  
\item  after averaging over the channel gains of the contact link and interfering links based on their distributions, the moments of the CSP$_\Phi$ can be represented by the expectation of the product of a function w.r.t. the locations of the interferers, i.e., $\mathbb{E} \big[ \prod_{x \in \Phi} f(x) \big] $. 

\item the expectation for Mat\'{e}rn cluster and Poisson fields of interferers can be derived based on the PGFLs of MCP and PPP given in (\ref{eqn:PGFL_MCP}) and (\ref{eqn:PPP_PGFL}), respectively, and that for the $\beta$-Ginibre field of interferers can be derived based on the distributions of the distances of the interfering links given in (\ref{eqn:PDF_gamma}).
 \end{itemize}
\end{itemize}

\begin{theorem} \label{thm:SP_MCP_field}
The moments of the CSP$_{\Phi}$ for a Mat\'{e}rn cluster field, Poisson field and $\beta$-Ginibre fields of interferers are given by (\ref{eqn:moments_field}), 
\begin{figure*} 
\begin{align} 
\mathcal{M}_{P_{\rm s}} (b) = \begin{dcases} \exp \! \Bigg(   - \lambda_{\rm p} \! \int_{\mathbb{R}^2} \!\! \bigg( \!1 \! - \exp \! \bigg( \! - \bar{c}    \! +  \!  \frac{ \bar{c} V_{b}(x,\theta)}{\pi R^2_{\rm d}}    \!   \bigg) \! \bigg)  \mathrm{d} x \!  \Bigg), \hspace{25mm}  \textup{MCP}  \\
\exp \Bigg( - \pi \lambda \theta^{\delta} r^2_{\rm t}  \frac{\Gamma(1-\delta) \Gamma(b+\delta)}{\Gamma(b)}       \Bigg ) ,  \hspace{40mm} \, \, \,  \textup{PPP} \\ 
 \int^{\infty}_{0}  \!\! e^{-\pi \lambda q / \beta} \bigg( \frac{ \beta   }{1 + \theta   r^{\alpha}_{\rm t}  q^{-\alpha/2} }  +  1 - \beta \bigg)^{b}  \prod_{ j\geq 1  } \frac{ (\pi \lambda / \beta )^{j} }{\Gamma(j) } q^{j-1}  \mathrm{d} q, \hspace{6mm}  \textup{$\beta$-GPP} 
\end{dcases} \label{eqn:moments_field}
\end{align}
 \hrulefill 
\end{figure*}
where 
\begin{align}
 V_{b}(x,\theta) & = \int_{ \mathbb{B}(0,R_{\rm d}) } \!\! \bigg ( \frac{1}{  1 \!+\! \theta r^{\alpha}_{\rm t}  \| x-y \|^{-\alpha}   } \bigg)^{\!b}  \mathrm{d} y .   
\end{align}
\end{theorem}

\begin{proof}
See \textbf{Appendix A}.
\end{proof}

Fig.~\ref{fig:Spatial_SP_compare} shows the average success probabilities (i.e., $\mathcal{M}_{P_{\rm s}}(1)$) in ad hoc networks with different fields of interferers.
In order to reveal the entire SIR distribution, we plot the success probabilities in the M\"{o}bius homeomorphic (MH) scale. The conversion from linear scale to MH scale is given by the function $x \, \text{MH} = \frac{x}{1-x}$~\cite{M.2020Haenggi}, which maps the   one-sided infinite support  $[0,\infty)$ to the unit interval $[0,1)$. 
 It can be observed that, in ad hoc networks, spatial repulsion and attraction among the interferers result in lower and higher success probabilities, respectively, compared to the independently located interferers. This can be intuitively understood from the fact that stronger spatial attraction (repulsion) increases the chance that the interferers are located further away (closer to) the target receiver (as shown in Fig.~\ref{fig:Illustration_SpatialPP}). 
  
\begin{figure}[htp]
\centering
\includegraphics[width=0.5\textwidth]{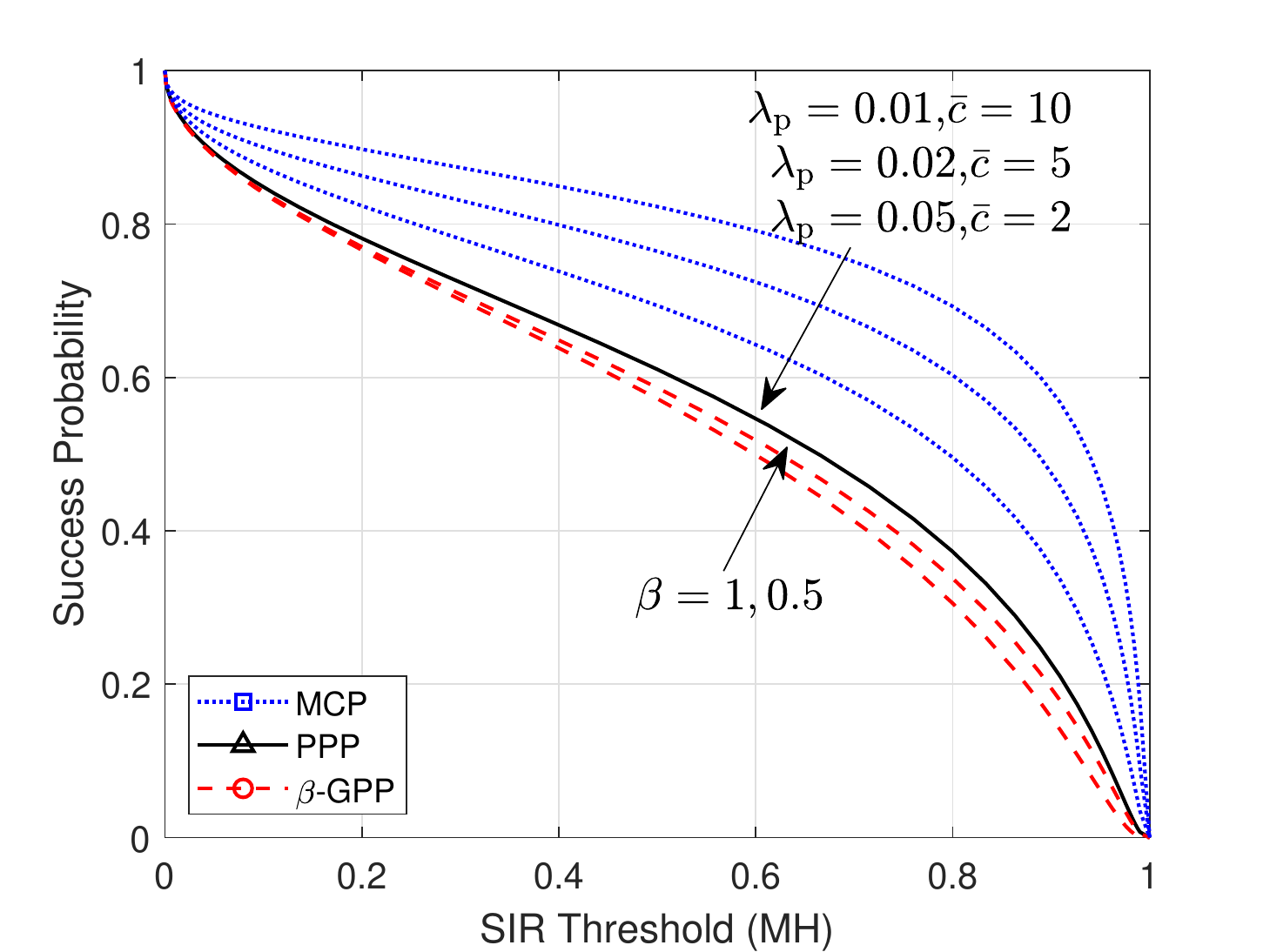} 
\caption{Success probability in different ad hoc networks ($\lambda=0.1$, $r_{\rm t}=1$, 
and $R_{\rm d}=1$. The curves and markers correspond to the 
analytical and simulation results, respectively.).} \label{fig:Spatial_SP_compare}
 \vspace{-3mm}
\end{figure}
 
 \begin{figure*} [htp]
  \centering
   \vspace{-3mm}  
   \subfigure [  MCP ($\lambda=0.1$) ]
    {
   \centering   
   \includegraphics[width=0.48 \textwidth]{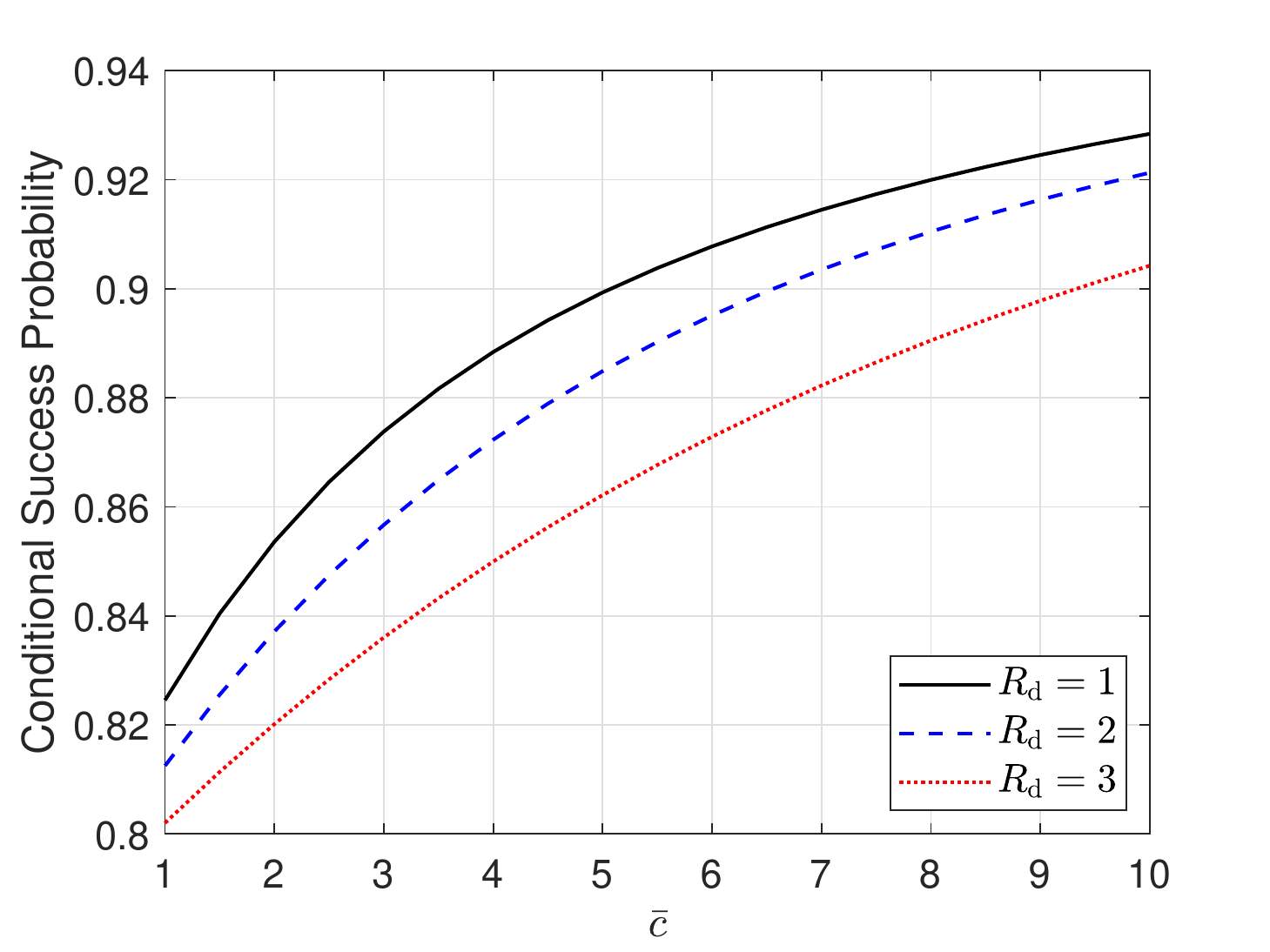}}
   \centering  
   \subfigure  [ $\beta$-GPP  ($\lambda=0.1$)
   ] {
   \centering 
  \includegraphics[width=0.48 \textwidth]{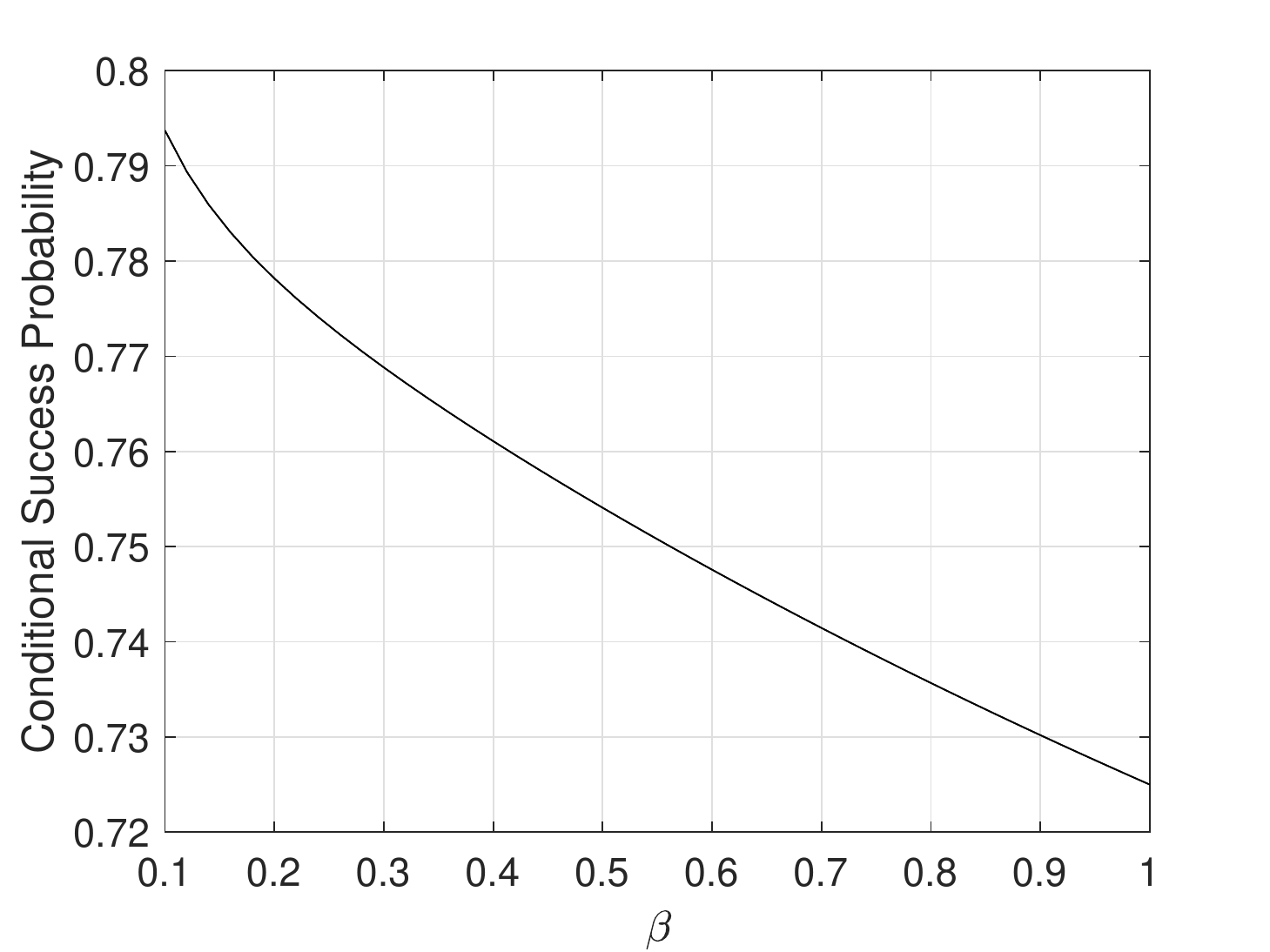}}
  \captionsetup{justification=centering}
  \caption{Temporal CSP in ad hoc networks with non-Poisson fields of interferers  ($\alpha=4$). } 
  \centering
  \label{fig:non_Poisson_field}
  \end{figure*}

Subsequently, we investigate the temporal dependence of successful transmissions by evaluating the temporal CSP, i.e., $\mathcal{M}_{P_{\rm s}}(2)/\mathcal{M}_{P_{\rm s}}(1)$, in  Fig.~\ref{fig:non_Poisson_field}. It can be found that the CSP increases when the interferers are more clustered (i.e., with smaller $R_{\rm d}$ or larger $\bar{c}$) and decreases when the interferers are more scattered (i.e., with larger $\beta$).  The reason is that spatial attraction and repulsion cause a lower and higher intensity of aggregated interference at the target receiver, and thus a transmission attempt is more likely to succeed given a previous successful transmission.
 
Furthermore, we investigate the SIR meta distribution for a target link in ad hoc networks. 
  According to the Gil-Pelaez theorem~\cite{Gil-Pelaez1951}, the exact SIR meta distribution can be represented as an integral function of the moments of the CSP$_\Phi$ as  
 \begin{align}
 \bar{F}(\theta,x) = \frac{1}{2} \! + \! \frac{1}{\pi} \int^{\infty}_{0} \frac{\Im \big(e^{-\jmath u \log x } M_{\jmath u} (\theta) \big) }{u } \mathrm{d} u,  
 \end{align}
 herein  $\Im(z)$ is the imaginary part of $z$ and $\jmath = \sqrt{-1}$ denotes the imaginary unit. 
  
 Fig.~\ref{fig:Spatial_meta_field} shows the SIR meta distribution in ad hoc networks with different fields of interferers. Compared with the success probability in Fig.~\ref{fig:Spatial_SP_compare}, the SIR meta distribution gives the entire distribution of the CSP$_{\Phi}$. For example, for a $0.5$ MH SIR threshold, although the average success probability for PPP is slightly greater than that of $\beta$-GPP, the percentage of the links achieving  $90\%$ reliability in PPP is higher than twice of that in $\beta$-GPP.
 
\subsubsection{Moments of the CSP$_\Phi$ in Downlink Cellular Networks}
  
With Rayleigh fading, the exact moments of the CSP given $\Phi^{\rm M}$ and $\Phi^{\rm G}_{\beta}$ can be obtained using the following steps~\cite{M.Apr.2016Haenggi}: 
\begin{itemize}
\item Deriving the CSP$_\Phi$ by averaging over the randomness of the channel gains of the contact link and interfering links based on the PDF of the exponential distribution due to Rayleigh fading;

\item Deriving the moments of the CSP$_\Phi$ by deconditioning on 
the spatial distributions of the interferers in the different random fields and contact distance.
Specifically,

\begin{itemize}

\item The spatial randomness of the interferers in Mat\'{e}rn cluster and Poisson downlink networks is averaged out based on the reduced PGFL of MCP and PPP, respectively, given in (\ref{eqn:reduced_PGFL}) and (\ref{eqn:PPP_PGFL}), 
and that in Ginibre downlink networks is averaged out based on the PDF of the distances of the interfering links given in (\ref{eqn:PDF_gamma});

\item Based on the nearest-BS associated rule, the contact distances in Mat\'{e}rn cluster, Poisson, and Ginibre downlink networks are averaged out based on their PDFs given in  (\ref{eqn:pdf_CD_MCP}), (\ref{eqn:PDF_CD_PPP}), and (\ref{eqn:PDF_gamma}), respectively.

\end{itemize}

\end{itemize}

  \begin{figure}[htp] 
  \centering
  \includegraphics[width=0.5\textwidth]{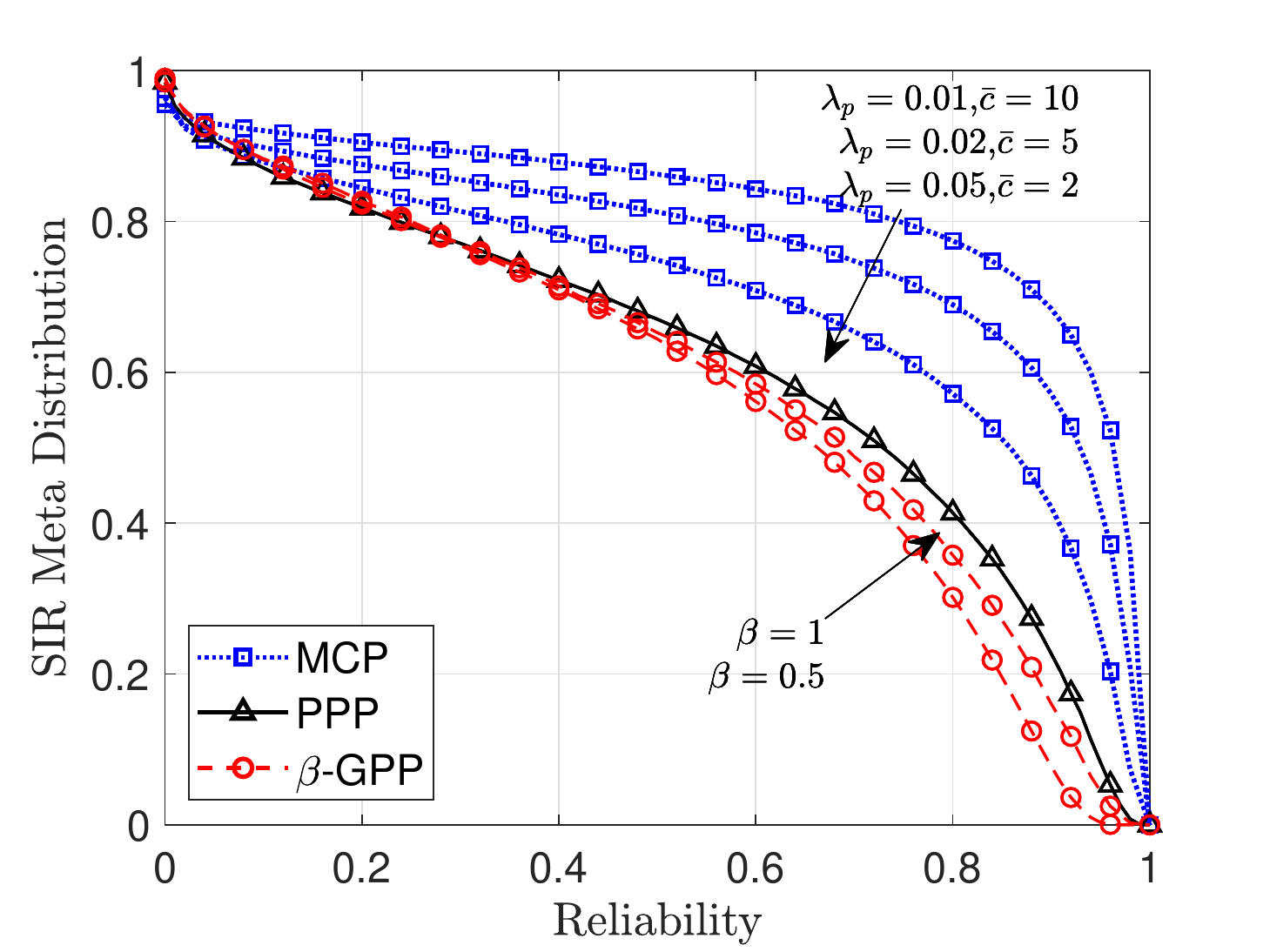} 
  \caption{SIR meta distribution in ad hoc networks  with different fields of interferers. ($\theta=0.5$ MH, $r_{\rm t}=1$, $\lambda=0.1$,  
  and $R_{\rm d}=1$. The curves and markers 
 correspond to the analytical and simulation results, respectively.)} \label{fig:Spatial_meta_field}
  \end{figure}
  
For non-Poisson networks, since deriving the exact downlink success probability 
is tedious (if not impossible) and the resulting expressions are cumbersome, we use an 
approximation method, referred to as {\em Approximate SIR analysis based on the PPP (ASAPPP)} method \cite{A.2015Guo,M.2014Haenggi,S.2019Kalamkar}, to simplify the evaluation of the SIR distribution.
ASAPPP provides an approximate SIR distribution in a non-Poisson network, obtained from the SIR distribution in a Poisson network along with a horizontal shift (in the dB scale).   
This method is based on the insight that when the disparity between the target system model and the Poisson model purely lies in the spatial configuration of the points, shifting the SIR threshold $\theta$ of the Poisson model by a coefficient $G_0$ results in a close approximation of the success probability and the meta distribution of the target system model.    
The subscript in $G_0$ indicates that the shift is calculated when $\theta$ approaches $0$. In other words, the ASAPPP method becomes exact as $\theta \to 0$. 
As shown in \cite{K.Mar.2016Ganti}, ASAPPP yields a very good approximation across the whole SIR distribution and 
is barely susceptible to the fading and path-loss models.

 The asymptotic gain $G_0$ can be obtained by taking the ratio of the mean interference-to-signal ratios (MISRs) of the considered point process to that of the PPP as \cite{M.Dec.2014Haenggi}
 \begin{align} 
 G_{0} = \frac{\textup{MISR}_{\textup{PPP}} }{ \textup{MISR}}   = \frac{2}{\alpha-2} \frac{1}{\textup{MISR}},
 \end{align}
 where MISR is defined as
  \begin{align}
 \textup{MISR} &  \triangleq \mathbb{E} \Bigg[  \frac{\sum^{\infty}_{ k = 2}    \|x_{k}\|^{-\alpha} }{ \|x_{1}\|^{-\alpha} } \Bigg],   
 \end{align}  
and $\textup{MISR}_{\textup{PPP}}$ 
is the MISR in Poisson downlink networks, which can be obtained by utilizing the distance ratio distribution as follows~\cite{M.Dec.2014Haenggi}. 
\begin{align} \label{eqn:MISR_2D}
 \textup{MISR}_{\textup{PPP}} & =   \sum^{\infty}_{j=2} \mathbb{E} \bigg[  \bigg ( \frac{r_{1}}{r_{j}} \bigg)^{\! \alpha}  \bigg] \nonumber \\
& =     \sum^{\infty}_{j=2} \mathbb{E} \big[     \varrho_{j} ^{ \alpha}  \big]   \nonumber \\
& \overset{(a)}{=}  \sum^{\infty}_{j=2} \int^{1}_{0} \varrho^{\alpha} 2 (j-1) \varrho (1 - \varrho^2)^{j-2} \mathrm{d} \varrho  \nonumber \\ 
&  =   \sum^{\infty}_{j=2} \frac{\Gamma(1+\alpha/2 ) \Gamma(j)}{ \Gamma(j+\alpha/2 ) }  \nonumber \\
&  =  \frac{2}{\alpha-2}, 
\end{align}  
 where $(a)$ follows the PDF of  $\varrho_{j}$ given in (\ref{eqn:PDF_ratio}).
 
The numerical value of $G_0$ can be obtained easily from Monte Carlo simulations for a given network geometry, i.e., $\lambda_{\rm p}$, $\bar{c}$, and $R_{\rm d}$ for the MCP and $\lambda$ and $\beta$ for the $\beta$-GPP. Note that, as found in \cite{Wei2016H},  the MISR-based gain of the $\beta$-GPP can be accurately approximated as 
\begin{align} \label{eqn:G_GPP}
G_{0} \approx 1 + \beta/2,
\end{align}
which is 
insensitive to the network density and path-loss exponent~\cite{K.Mar.2016Ganti}. Therefore, the success probability of the typical user in a $\beta$-GPP network is approximately identical to that in a Poisson network with the SIR threshold scaled from $\theta$ to $\theta/(1+\beta/2)$.

Fig. \ref{fig:Spatial_SP_ASAPPP} and Fig. \ref{fig:Spatial_meta_AMAPPP}, respectively, confirm the effectiveness of the ASAPPP method for approximating the SIR distribution and the SIR meta distribution in MCP and $\beta$-GPP downlink networks. 
As expected, for both types of networks, the ASAPPP method yields a more accurate approximation of SIR distribution as the SIR threshold decreases. 
For the MCP model, we can observe that the ASAPPP method is more accurate when $\Phi^{\rm M}$ is less clustered (e.g., with smaller $\bar{c}$ given $\lambda$).
When the network is more clustered (e.g., when $\bar{c}=10$), there exists an observable 
gap between the approximation and the simulation results for both success probability and SIR meta distribution. The reason is that a higher degree of clustering results in a slowly growing rate of the asymptotics~\cite{S.2019Kalamkar}.
Moreover, for the $\beta$-GPP, the insignificant disparity   between the simulation and the approximation results 
can be ascribed to the approximation of the MISR-based gain given in (\ref{eqn:G_GPP}).

   \begin{figure*}[htp]  
    \centering
      \begin{minipage}[c]{0.48 \textwidth}
       \includegraphics[width=0.98\textwidth]{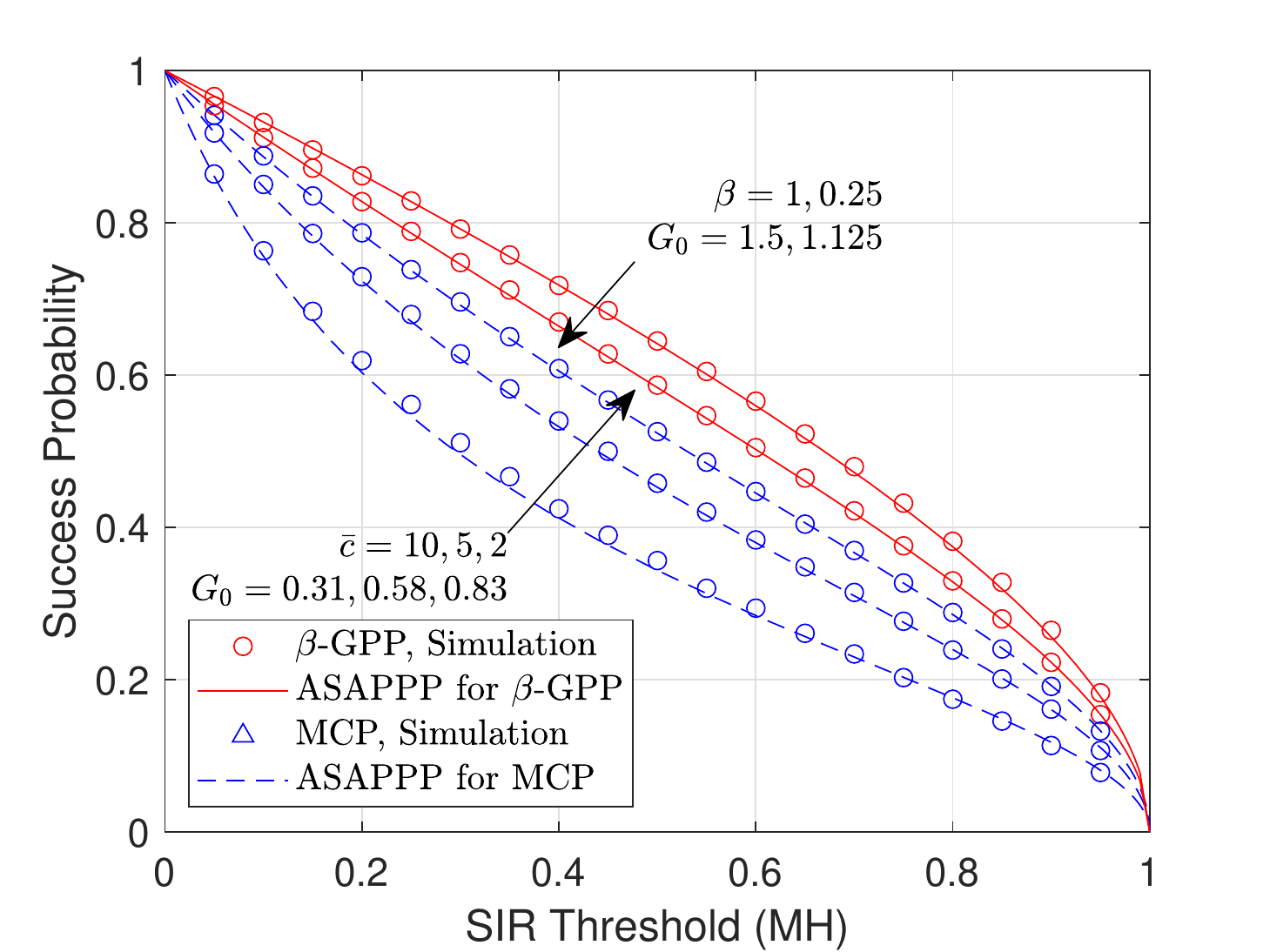} 
       \caption{The ASAPPP approximation of non-Poisson cellular networks ($\alpha=4$, $\lambda=0.1$, $R_{\rm d}=4$). 
       }  \vspace{4mm} \label{fig:Spatial_SP_ASAPPP}
      \end{minipage}  \hspace{3mm}
      \begin{minipage}[c]{0.48 \textwidth}\vspace{0mm}
        \includegraphics[width=0.98\textwidth]{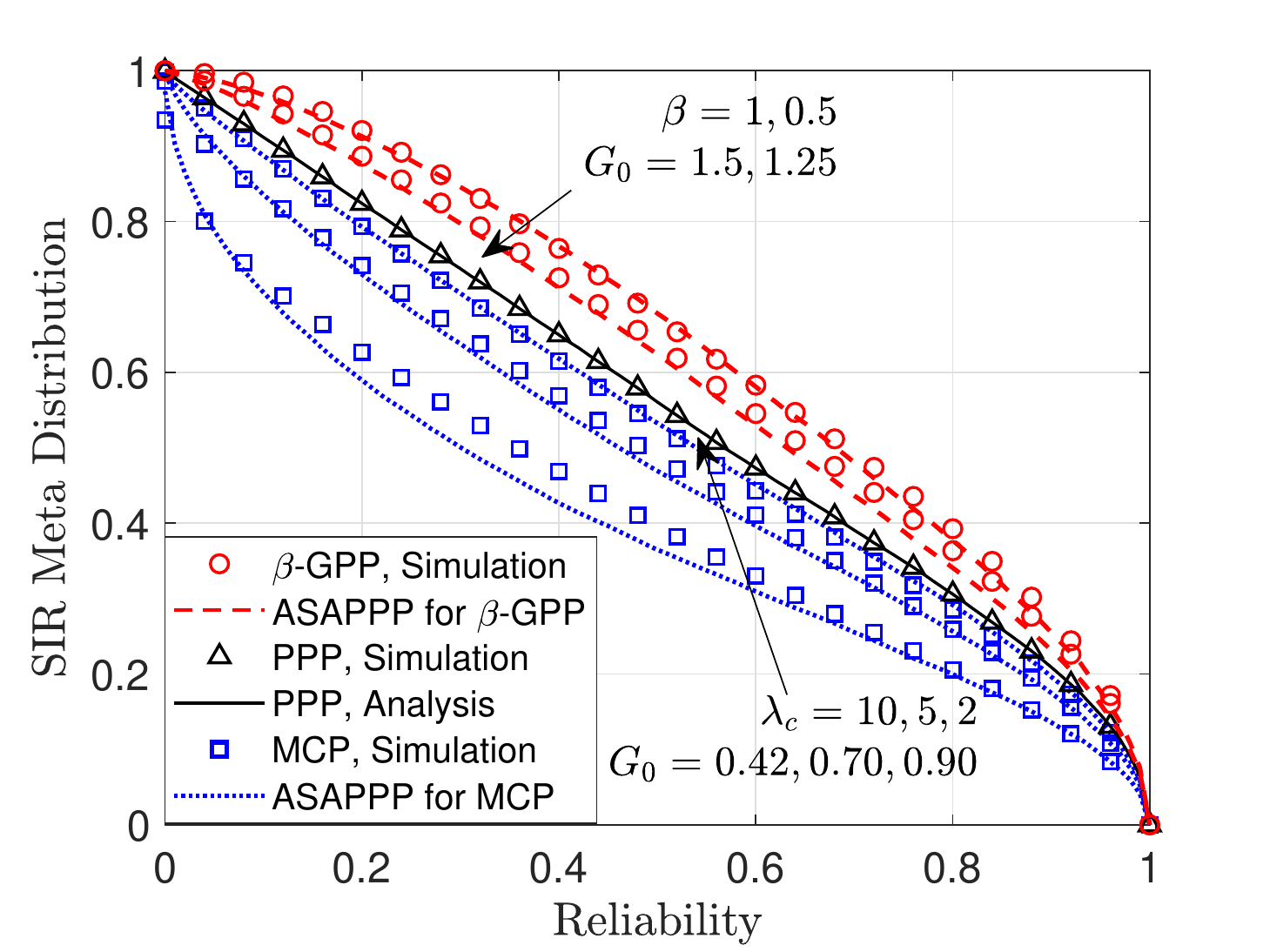}  
        \caption{The ASAPPP approximation of SIR meta distribution in non-Poisson cellular networks ($\theta=0.5$ MH, $\alpha=4$, $\lambda=0.1$, $R_{\rm d}=5$). } \label{fig:Spatial_meta_AMAPPP}
        \end{minipage}  
     \end{figure*}

 Next, in Fig. \ref{fig:spatial_DL}(a) and Fig. \ref{fig:spatial_DL}(b), we evaluate the temporal CSPs  for MCP and $\beta$-GPP downlink networks, respectively. We observe that the CSP decreases slightly when the network becomes more clustered (i.e., with larger $\bar{c}$ and/or smaller $R_{\rm d}$ with fixed $\lambda$) and increases slightly when the network becomes more repulsive (i.e., with larger $\beta$). 
 However, the curves are almost flat, i.e., the temporal CSPs in Mat\'{e}rn cluster and Ginibre downlink networks are generally non-sensitive to attractiveness and repulsiveness of the spatial points.

  \begin{figure*} [htp]
  \centering
   \subfigure [  MCP   ]
    {
   \centering   
   \includegraphics[width=0.48 \textwidth]{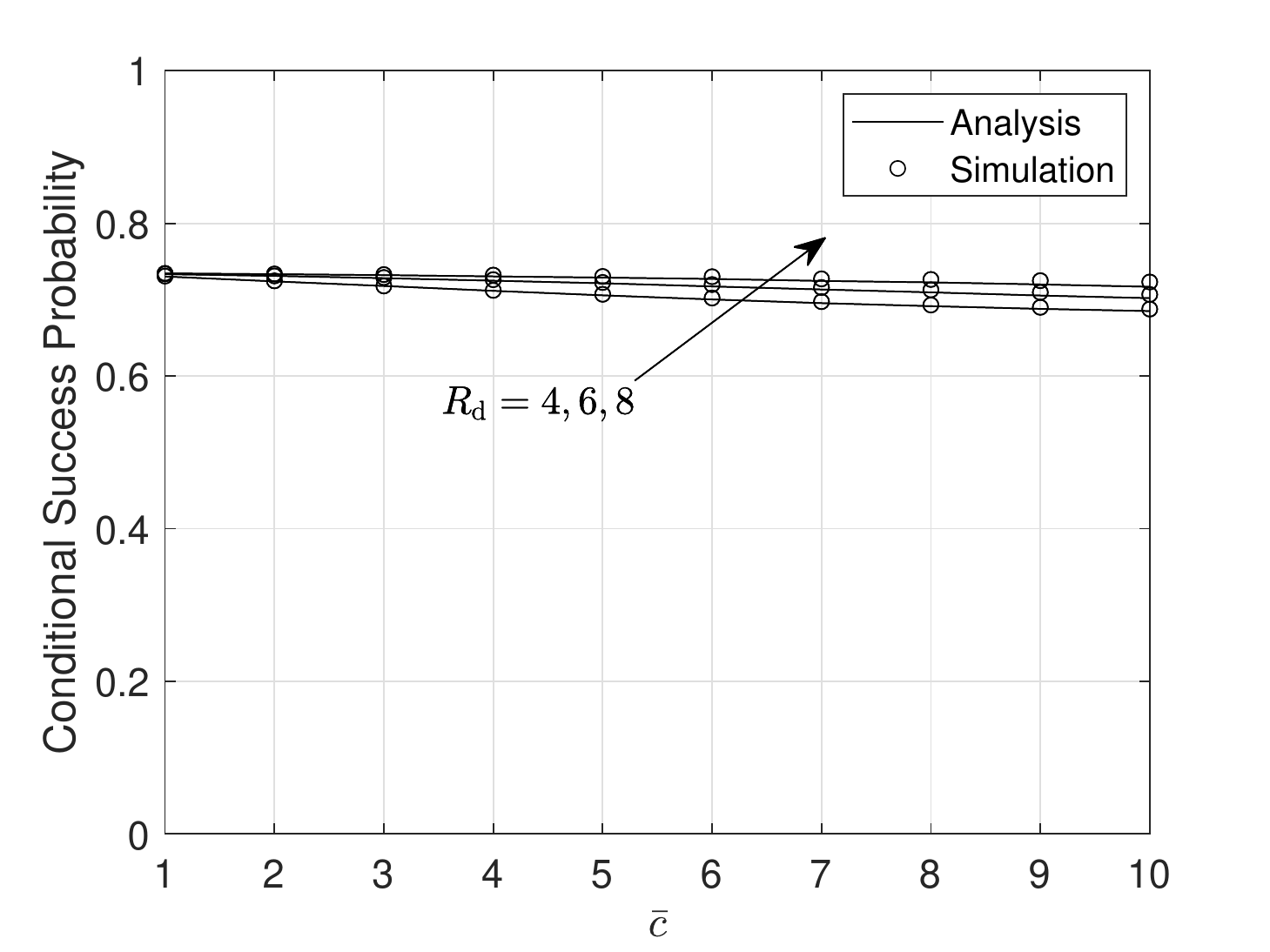}} \label{fig:spatial_DL1}
   \centering  
   \subfigure  [ $\beta$-GPP   
   ] {
   \centering
  \includegraphics[width=0.48 \textwidth]{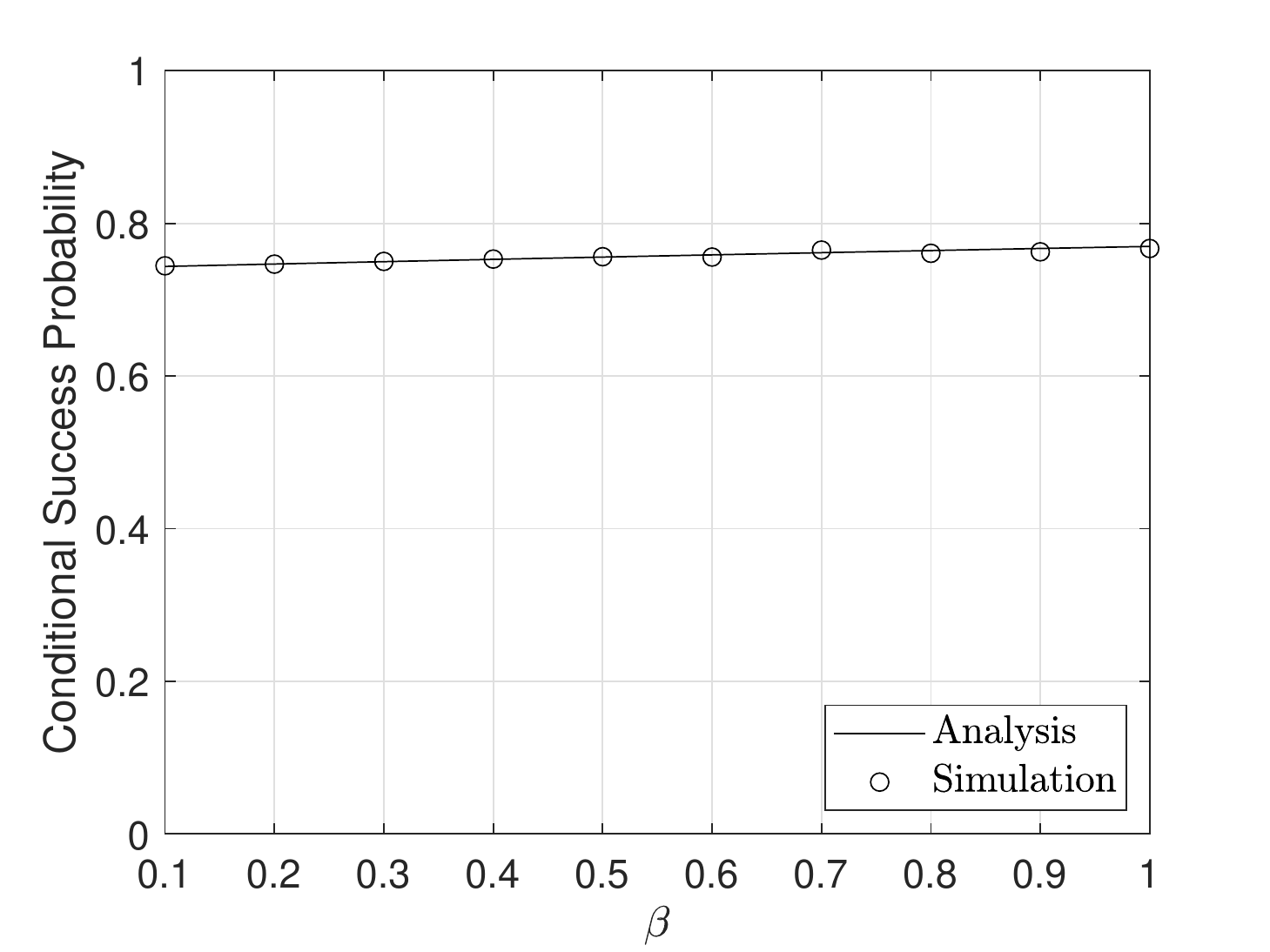}}
  \label{fig:spatial_DL2}
  \caption{Temporal CSP in non-Poisson cellular networks ($\alpha=4$, $\lambda=0.1$). } 
  \centering
  \label{fig:spatial_DL}
  \end{figure*}

\subsection{Summary and Discussion}

We have analyzed the impact of the spatial distribution of transmitters on the interference correlation coefficient and success probability.
Specifically, we have presented the derivations of spatial-temporal interference correlation coefficient and success probability and SIR meta distribution of a target link in MCP, PPP, and $\beta$-GPP fields of interferers. We have also introduced an approximation of the success probabilities and SIR meta distributions in non-Poisson networks based on the analysis of a Poisson network by scaling the SIR threshold. The main lessons learned are as follows.

\begin{itemize}
 
 \item With non-Poisson fields of interferers, the spatial-temporal interference is  more and less correlated when the interferers are distributed more attractively and repulsively, respectively.

\item Compared with a Poisson field of interferers, the Mat\'{e}rn cluster and $\beta$-Ginibre field of interferers increases and decreases the success probability of a target link, respectively. Moreover, the successful transmission events are more and less temporally-correlated when the locations of the interferers are more clustered and scattered, respectively.

\item Compared with a Poisson downlink network, the Mat\'{e}rn cluster and Ginibre downlink networks render a lower and a higher success probability, respectively. Additionally, the successful transmission events are less and more temporally-correlated when the interferers are more attractive and repulsive, respectively.

\end{itemize}

 \vspace{0.2cm}
{\em Open Technical Issues}:
The point process models presented in Section~\ref{sec:spatial} have been adopted to characterize the general spatial features (e.g., independence, repulsion and clustering) of wireless networks. In addition to these generic models, some point process models have been established for some specialized network scenarios. For example, reference~\cite{N.2020EPPDeng} proposes an {\em energized point process} to model the spatial distribution of wireless-powered devices. The model exhibits spatial correlation of the RF-powered nodes that can harvest sufficient energy from a Poisson field of RF sources.
Besides, reference~\cite{K.2014Stamatiou} introduces a {\em Poisson rain process} to model cellular networks with spatial-temporal traffic. Specifically, the model employs a space-time PPP to model traffic arrives which are then assigned to PPP-distributed BSs based on different allocation schemes.
Future efforts should be dedicated to developing point processes that incorporate spatial-temporal features (e.g., fading, shadowing, blockage, and renewable energy arrivals) for more specific system models.

\section{Location-Dependent Analysis of Cellular Models} \label{eqn:location}

The locations of the users play a crucial part in their performance. Conditioning the regions (e.g., cell center and boundary) of the user of interest allows a location-dependent analysis \cite{K.2019Feng,K.2019Feng2} for a location-specific user (LSU) in cellular networks. The objective of this section is to demonstrate the effect of the location of the target user on the SIR distribution. For this, we derive the moments of the CSP$_\Phi$ and the SIR gain for an LSU in different regions.  

\subsection{System Models}

For a stationary point process $\Phi \subset \mathbb{R}^2$ of BSs, let $x_{j}(u) \in \Phi$ represent the location of the $j$-th nearest BS to the LSU at $u$.  Under the nearest-BS association principle, 
$x_{1}(u)$ and $x_{2}(u)$ are the locations of the associated (serving) BS and the nearest interfering BS to the LSU, respectively.  
The region of a user's location can be defined by the relationship between the link distances to the associated BS and to the nearest interfering BS. 
 
For $\rho \in [0,1]$, the regions of {\em cell-center user} and {\em cell-boundary user} are defined, respectively, as
\begin{align}
&\mathcal{R}_{\rm c} \triangleq \{ u  \in \mathbb{R}^2: \|   x_{1}(u) - u \| \leq \rho \|  x_{2}(u) - u \|  \} ,  \nonumber  \\
&\mathcal{R}_{\rm b} \triangleq \{ u  \in \mathbb{R}^2:  \|   x_{1}(u) - u \| > \rho \|   x_{2}(u) - u \|  \}.  \nonumber
\end{align}

The area fraction of a region 
can be determined by the probability that 
an arbitrary location, say $o$,  falls into that region, which solely depends on $\rho$. 
For the analysis, we focus on the case where $\Phi$ is a PPP. 

The area fractions of $\mathcal{R}_{\rm c}$ and $\mathcal{R}_{\rm b}$ 
can be obtained based on the CDF of $\varrho_{2}$ (given in~(\ref{eqn:CDF_ratio})) as 
 \begin{align} 
&   \mathbb{P} [ o \in \mathcal{R}_{\rm c} ] =  \mathbb{P} [ \varrho_{2} \leq  \rho ] = F_{\varrho_{2}}  (\rho)
 =
  \rho^2  \nonumber
  \end{align}
  and
  \begin{align}
  \mathbb{P} [ o \in \mathcal{R}_{\rm b} ] = 1 - \mathbb{P} [ o \in \mathcal{R}_{\rm c} ] 
  = 1 - \rho^2 . \nonumber 
 \end{align}

Moreover, in the special cases when a cell-boundary user has two or three equidistant closest BSs, it is referred to as an {\em edge user} or a {\em vertex user}, respectively.
The regions of the edge users (one-dimensional) and vertex users (zero-dimensional) are defined, respectively, as
\begin{align}
& \mathcal{R}_{\rm e} = \{ u  \in \mathbb{R}^2: \|   x_{1}(u) - u \| = \|  x_{2}(u) - u \|   \},  \nonumber  \\
& \mathcal{R}_{\rm v} = \{ u  \in \mathbb{R}^2: \|   x_{1}(u) - u \| = \|  x_{2}(u) - u \| \nonumber \\
& \hspace{40mm} = \|  x_{3}(u) - u \|  \}.
\end{align} 
Note that in a two-dimensional homogeneous PPP there exists no location with more than three nearest BSs almost surely~\cite{N.2013Chiu}.

In the following, we explore and the performance at five different types of typical users: (i) The  standard typical user, whose performance represents the average of all users; it is referred to in this section as the typical general user; (ii) the typical cell-center user, whose performance represents the average of all cell-center users; (iii) the typical cell-boundary user, whose performance represents the average of all cell-boundary users; (iv) the typical cell-edge user whose performance represents the average of all locations on the cell edges; and (v) the typical vertex user, whose performance represents the average of all vertex points of the cells.

In Poisson cellular networks, the PDF of the contact distance of the typical vertex user is given by~\cite{L.2005Muche,Y.Aug.2013Jung} 
\begin{align} \label{eqn:PDF_VU_2D}
f^{\rm v}_{r_{1}}(r) = 2 (\lambda \pi)^2 r^3 \exp(-\lambda \pi r^2).
\end{align}

The success probability of the typical cell-center and cell-boundary users can be defined, respectively, as 
\begin{align}
& \bar{F}^{\rm c}_{\eta} (\theta) = \mathbb{P} [ \eta > \theta \mid o \in \mathcal{R}_{\rm c} ],  \nonumber \\
  \quad \textup{and} \quad & 
\bar{F}^{\rm b}_{\eta} (\theta) = \mathbb{P} [ \eta > \theta \mid o \in \mathcal{R}_{\rm b} ].
\end{align} 
 
Fig.~\ref{fig:2D_regions} shows a realization of network partitions of a Poisson network with unit intensity, where the locations of the BSs are represented by the black circles. 
White and cyan regions represent the cell-center regions $\mathcal{R}_{\rm c}$ and cell-boundary regions $\mathcal{R}_{\rm b}$, respectively. The blue triangle markers represent 
the region of the vertex users and blue lines 
represent the region of the edge users.

\begin{figure}[htp]
\centering
\includegraphics[width=0.5\textwidth]{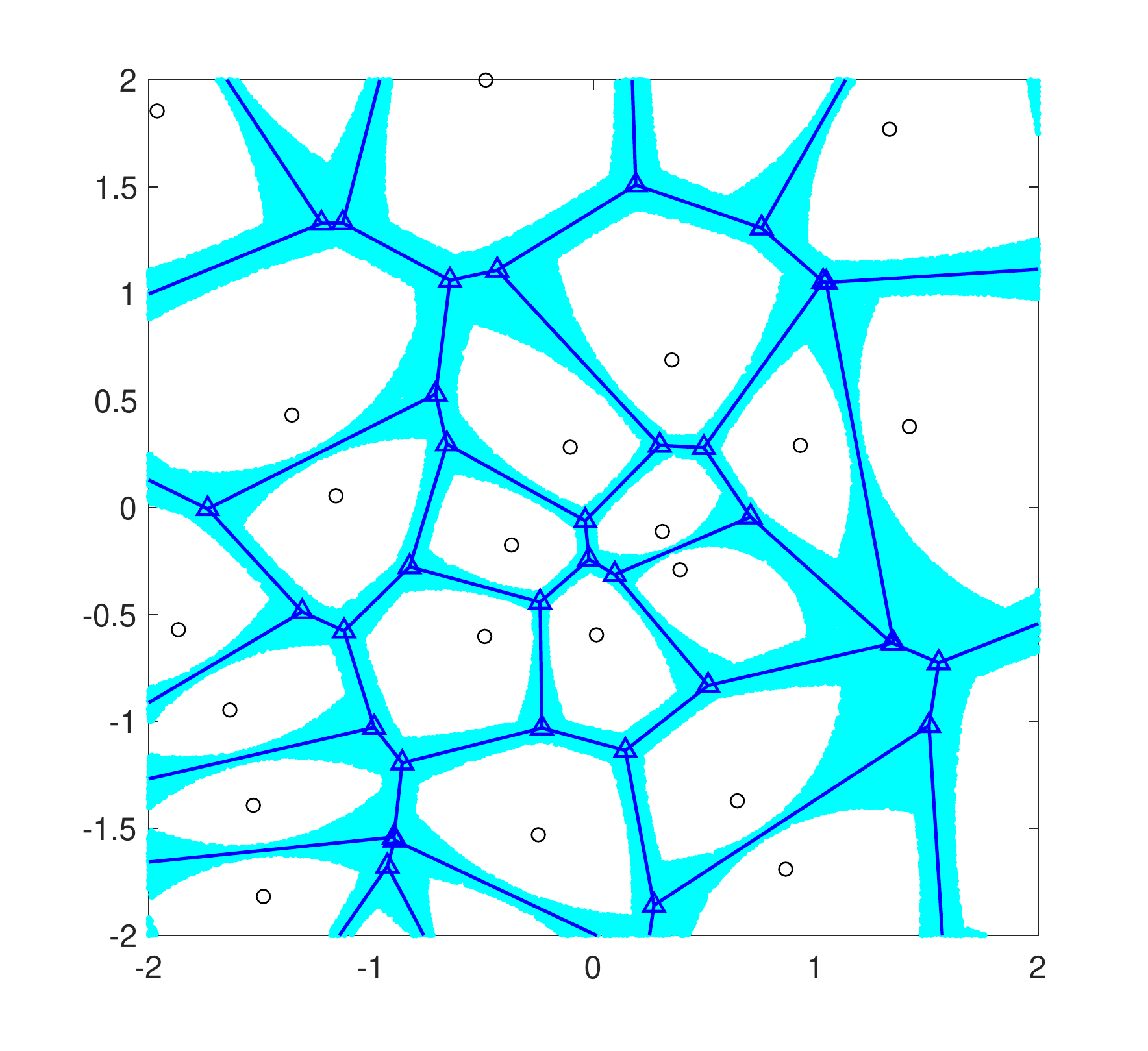} 
\caption{Illustration of network partitions in a two-dimensional Poisson network for $\rho=0.8$.} \label{fig:2D_regions}
\end{figure}

 \subsection{Performance Analysis}

 \subsubsection{Cell-Center and Cell-Boundary User}
 
 We start by showing how to obtain the $b$-th moments of the CSP$_\Phi$ for cell-center and cell-boundary users, denoted as $\mathcal{M}^{\rm c}_{P_{\rm s}}(b)$ and $\mathcal{M}^{\rm b}_{P_{\rm s}}(b)$, respectively. The methodology of obtaining the former is to derive the SIR meta distribution for the typical general user 
 by conditioning its location to be in the cell-center and cell-boundary regions \cite{K.2019Feng,K.2019Feng2}. Such conditioning restricts the distribution regions of interferers which will be reflected in the lower bounds of the integrals after applying the PGFL of PPP.

 Based on this approach, we obtain $\mathcal{M}^{\rm c}_{P_{\rm s}}(b)$ in the following theorem.

 \begin{theorem} \label{thm:2D_SP_CCU} 
 The moments of the CSP$_\Phi$ for the typical cell-center user, i.e., the typical general user conditioned on $\mathcal{R}_{\rm c}$, in a Poisson downlink network with Rayleigh fading are
 \begin{align} \label{eqn:2D_SP_CCU}
 \mathcal{M}^{\rm c}_{P_{\rm s}}(b)
 = \frac{1}{ {_2}F_1(b,-\delta;1-\delta;- \rho^{\alpha} \theta )},
 \end{align} 
 where $\delta=2/\alpha$.
 \end{theorem}
 
 \begin{proof}
 See \textbf{Appendix B}. 
 \end{proof}

\noindent
{\bf Remark 3}:
The SIR gain of the typical cell-center user compared to the typical general user is $\rho^{-\alpha}$, which is easily obtained by comparing (\ref{eqn:2D_SP_CCU}) with the moments of the CSP$_\Phi$ for the typical general user in a Poisson downlink network 
given as~\cite{M.Apr.2016Haenggi}:
 \begin{align} \label{eqn:M_SP_TU}
 \mathcal{M}_{P_{\rm s}}(b) = \frac{1}{ {_2} F_1(b,-\delta;1-\delta;- \theta)} .
 \end{align}

\noindent
{\bf Remark 4}: A fraction $x=\rho^2$ of users who have the lowest distance ratio of the associated BS to the interfering BS has an SIR gain of $-5 \alpha \log_{10} x$ dB.

Applying the law of total probability, the moments of the CSP$_\Phi$ for the typical general user can be expressed 
 as follows:
\begin{align} 
  \mathcal{M}_{P_{\rm s}}(b)  =&    \mathbb{E} \Big[ \mathbb{P} \big[ \eta > \theta   \big]^{b}  \Big] \nonumber \\ 
  =& \mathbb{E} \Big[ \mathbb{P} \big[ \eta > \theta   \big]^{b} \bigm|  o \in \mathcal{R}_{\rm c} \Big] \mathbb{P} \big[  o \in \mathcal{R}_{\rm c} \big] \nonumber \\
 &   + \mathbb{E} \Big[ \mathbb{P} \big[ \eta > \theta   \big]^{b} \bigm| o \in \mathcal{R}_{\rm b} \Big] \mathbb{P} \big[  o \in \mathcal{R}_{\rm b} \big] \nonumber \\
   = &  \mathcal{M}^{\rm c}_{P_{\rm s}}(b)  \mathbb{P} \big[  o \in \mathcal{R}_{\rm c} \big]   + \mathcal{M}^{\rm b}_{P_{\rm s}}(b)  \Big( 1-  \mathbb{P} \big[  o \in \mathcal{R}_{\rm c} \big] \Big). \nonumber
\end{align}

Then, the moments of the CSP$_\Phi$ for the typical cell-boundary user 
are given by
\begin{align} \label{eqn:SP_CBU_2D} 
\mathcal{M}^{\rm b}_{P_{\rm s}}(b) 
& = \frac{ \mathcal{M}_{P_{\rm s}}(b) - \mathcal{M}^{\rm c}_{b}(\theta)\mathbb{P} [ o \in \mathcal{R}_{\rm c} ] }{ 1 - \mathbb{P} [ o \in \mathcal{R}_{\rm c} ] } \nonumber \\
& = \frac{ \mathcal{M}_{P_{\rm s}}(b) -  \rho^2 \mathcal{M}^{\rm c}_{P_{\rm s}}(b) }{  1 - \rho^2}  . 
\end{align}

We then have the following corollary by using (\ref{eqn:M_SP_TU}) and (\ref{eqn:2D_SP_CCU}) in (\ref{eqn:SP_CBU_2D}).

\begin{corollary}  
The moments of the CSP$_\Phi$ of the typical cell-boundary user, i.e., $o \in \mathcal{R}_{\rm b}$, in a Poisson downlink  network with Rayleigh fading are
\begin{align} \label{eqn:SP_CBU_2D}
    \mathcal{M}^{\rm b}_{P_{\rm s}}(b) 
 =  & \frac{1}{ (1-\rho^2)\, {_2} F_1(b,-\delta;1-\delta;-  \theta)} \nonumber\\
&   -   \frac{ \rho^2  }{(1-\rho^2) \, {_2}F_1(b,-\delta;1-\delta ;- \rho^{\alpha}\theta)},
\end{align}  
where $\delta=2/\alpha$.
\end{corollary}

\subsubsection{Edge User} 
An edge user can be considered as an asymptotic case of a cell-boundary user when $\rho \to 1$, i.e., $\mathbb{P} [ o \in \mathcal{R}_{\rm b} ] \to 0 $. In this case, the cell-boundary user is equidistant from 
the serving BS and the nearest interferer, i.e., on the cell edge. 
Therefore, the performance of the typical cell-edge user can be obtained from the asymptotics of the typical cell-boundary user as $\rho \to 1$. 

By taking the limit $\rho \to 1$ of $\mathcal{M}^{\rm b}_{P_{\rm s}}(b)$ in (\ref{eqn:SP_CBU_2D}), we have the moments of the CSP$_\Phi$ for the typical edge user as follows: 
\begin{align}
 \mathcal{M}^{\rm e}_{P_{\rm s}}(b) & \overset{(a)}{=}   \lim_{\rho \to  1}  \mathcal{M}^{\rm b}_{P_{\rm s}}(b)   \nonumber \\
&   =  \frac{1}{{_2}F_{1} ( b, - \delta; 1 - \delta; - \theta ) } \nonumber \\
& \hspace{5mm} -  \frac{  b  \theta \,   {_2}F_{1} ( b + 1, 1-\delta; 2-\delta; - \theta ) }{ (1 - \delta)  \, {_2}F_{1} ( - b, -\delta;  1-\delta; - \theta )^2  }, 
\end{align}
where $(a)$ follows from L'H\^opital's rule. 

After some mathematical manipulations, we have  $\mathcal{M}^{\rm e}_{P_{\rm s}}(b)$ in  the following corollary. 
\begin{corollary} \label{M_SI_EU}
With Rayleigh fading, the moments of the CSP$_\Phi$ for the typical edge user in a Poisson downlink network are 
\begin{align} \label{eqn:M_EU}
\mathcal{M}^{\rm e}_{P_{\rm s}}(b) = \frac{1}{ (1+\theta)^{b}{_2}F_1(b,- \delta;1\!-\!\delta;-\theta)^2}, \quad b \in \mathbb{C}
 \end{align}
 where $\delta = 2/ \alpha$.
 \end{corollary}
 
 \noindent
{\bf Remark 5}: Compared with $\mathcal{M}_{P_{\rm s}}(b)$ in (\ref{eqn:2D_SP_CCU}), $\mathcal{M}^{\rm e}_{P_{\rm s}}(b)$ in (\ref{eqn:M_EU}) can be expressed as $ \mathcal{M}^{\rm e}_{P_{\rm s}}(b) = \frac{  \mathcal{M}_{P_{\rm s}}(b) ^2 }{(1+\theta)^{b}}$, $\forall b \in \mathbb{R}^{+}$. $ \mathcal{M}^{\rm e}_{P_{\rm s}}(b)$ is smaller than $ \mathcal{M}_{P_{\rm s}}(b)$ and their gap depends only on $\theta$ and $b$.

\begin{figure*} 
\centering
 \subfigure [ Typical general user and cell-center/boundary user ]
  {
 \centering   
 \includegraphics[width=0.48 \textwidth]{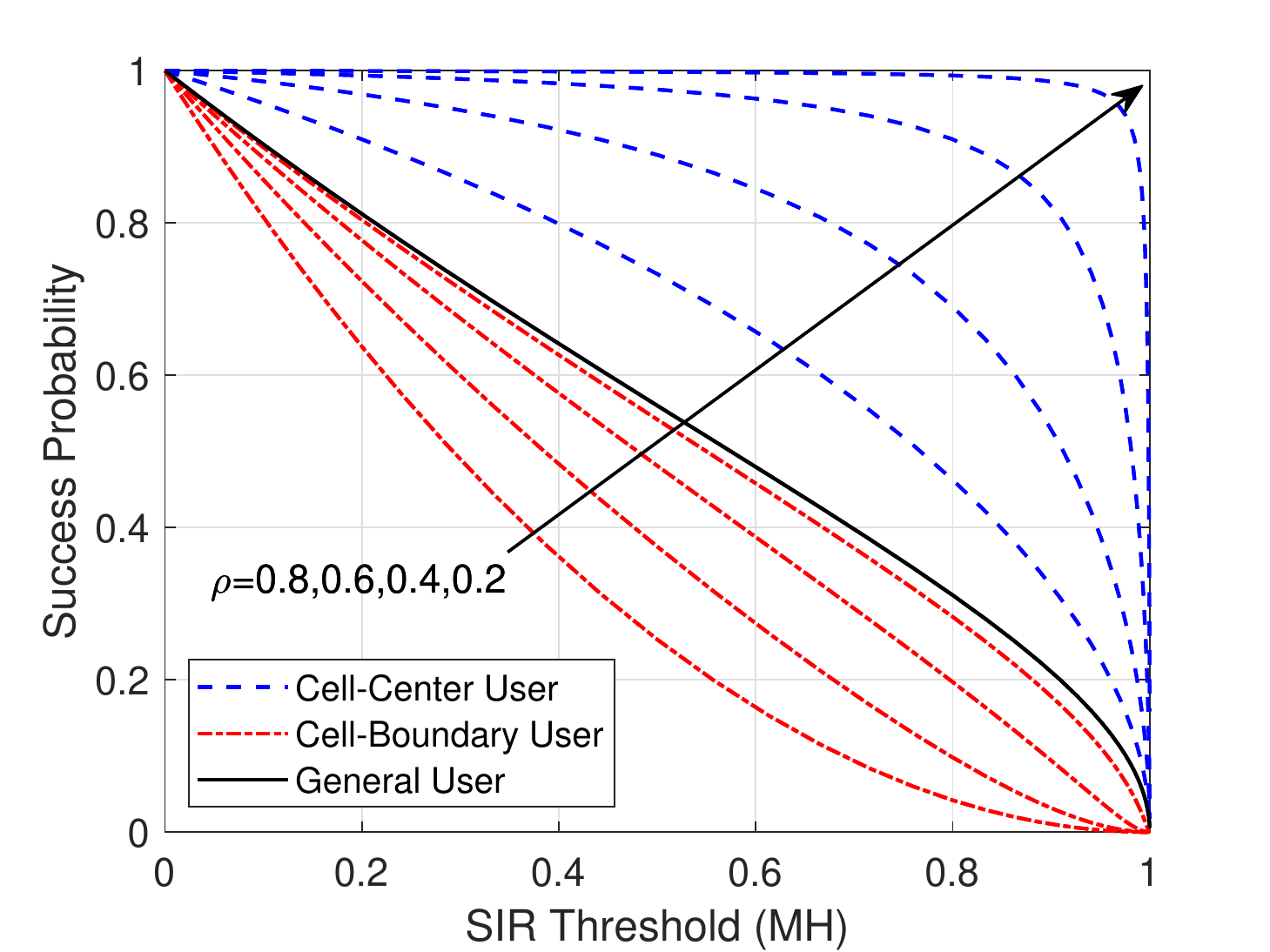}}
 \centering  
 \subfigure  [ Typical edge user and vertex user 
 ] {
 \centering
\includegraphics[width=0.48 \textwidth]{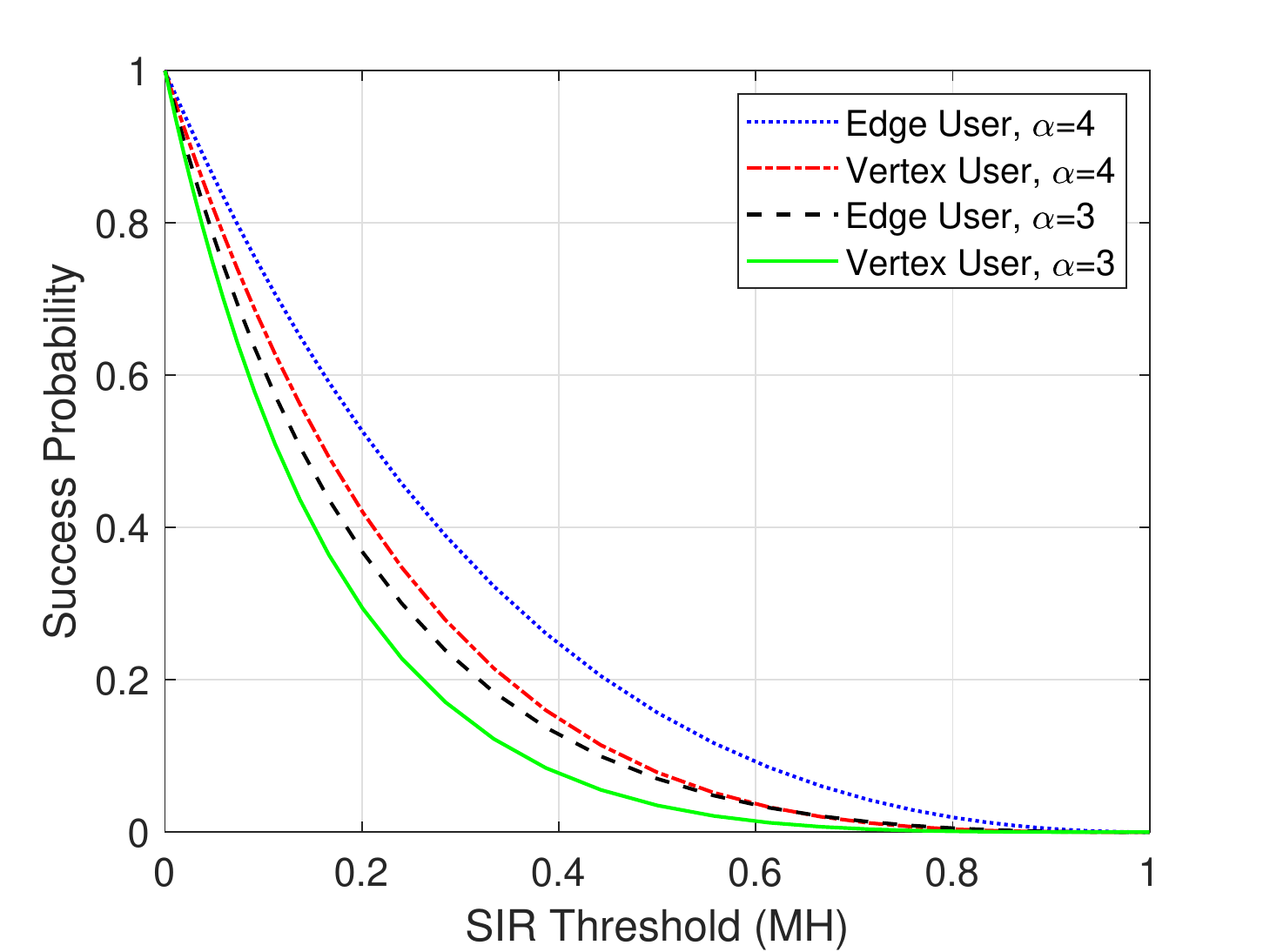}}
\caption{Success probability of LSUs in Poisson downlink networks ($\lambda=1$, $\alpha=4$). } 
\centering
\label{fig:CP_2D}
\end{figure*}

\subsubsection{Vertex User}

Finally, we compute the performance of the typical vertex user, which represents the worst-case performance.  
To obtain the moments of the CSP$_\Phi$, denoted as $\mathcal{M}^{\rm v}_{P_{\rm s}}(b)$, we follow the derivation steps similar to those for the typical general user with the condition that the user is equidistant from the serving BS and the two nearest interfering BSs \cite{K.2019Feng}. 
Under this condition, the contact distance is averaged out based on its PDF given in (\ref{eqn:PDF_VU_2D}). With this methodology, we have $\mathcal{M}^{\rm v}_{P_{\rm s}}(b)$ given in the following theorem.

\begin{theorem} \label{thm:vertex}
With Rayleigh fading, the moments of the CSP$_\Phi$ for the typical vertex user in a Poisson downlink network are
\begin{align} \label{eqn:M_2Dvertex}
\mathcal{M}_{P_{\rm s}}(b) = \frac{1}{(1+ \theta )^{2b} {_2}F_1(b,- \delta;1-\delta;- \theta)^{ 2}},
\end{align}
where $\delta=2/\alpha$.
\end{theorem}

\begin{proof}
See \textbf{Appendix C}.
\end{proof}

\noindent
{\bf Remark 6}: Compared with (\ref{M_CP_worstcase}), the success probability of the typical vertex user in (\ref{eqn:M_2Dvertex}) can be expressed as $ \mathcal{M}^{\rm v}_{P_{\rm s}}(b) = \frac{ \mathcal{M}^{\rm e}_{P_{\rm s}}(b) }{(1+\theta)^b}$. $ \mathcal{M}^{\rm v}_{P_{\rm s}}(b)$ is smaller than $ \mathcal{M}^{\rm e}_{P_{\rm s}}(b)$ due to the two equidistant closest interfering BSs.

Fig.~\ref{fig:CP_2D} depicts the average success probabilities (i.e., $\mathcal{M}_{P_{\rm s}}(1)$) of the LSUs. It is straightforward that the success probability monotonically decreases with $\rho$. The weighted mean of the success probabilities for the typical cell-center user and the typical cell-boundary user with the same value of $\rho$ are equal to that of the typical general user.

 \subsubsection{Temporal Effect of Location Dependence} 

Next, we investigate how location dependence affects the temporal correlation among different successful transmission events by evaluating the CSP $\frac{\mathcal{M}_{P_{\rm s}}(2)}{\mathcal{M}_{P_{\rm s}}(1)}$.  Fig. \ref{fig:CCP_locationdependent} shows how the CSP varies with $\rho$. It is found that the CSP for the typical cell-center user is greater than the typical cell-boundary user. Moreover, for both typical cell-center and cell-boundary users, their CSPs decrease with $\rho$. 
The reason is that the successful transmission events are more temporally correlated when the received signal dominates the interference. Such temporal correlation decreases when the interference becomes stronger. 

 \begin{figure} \centering \includegraphics[width=0.5\textwidth]{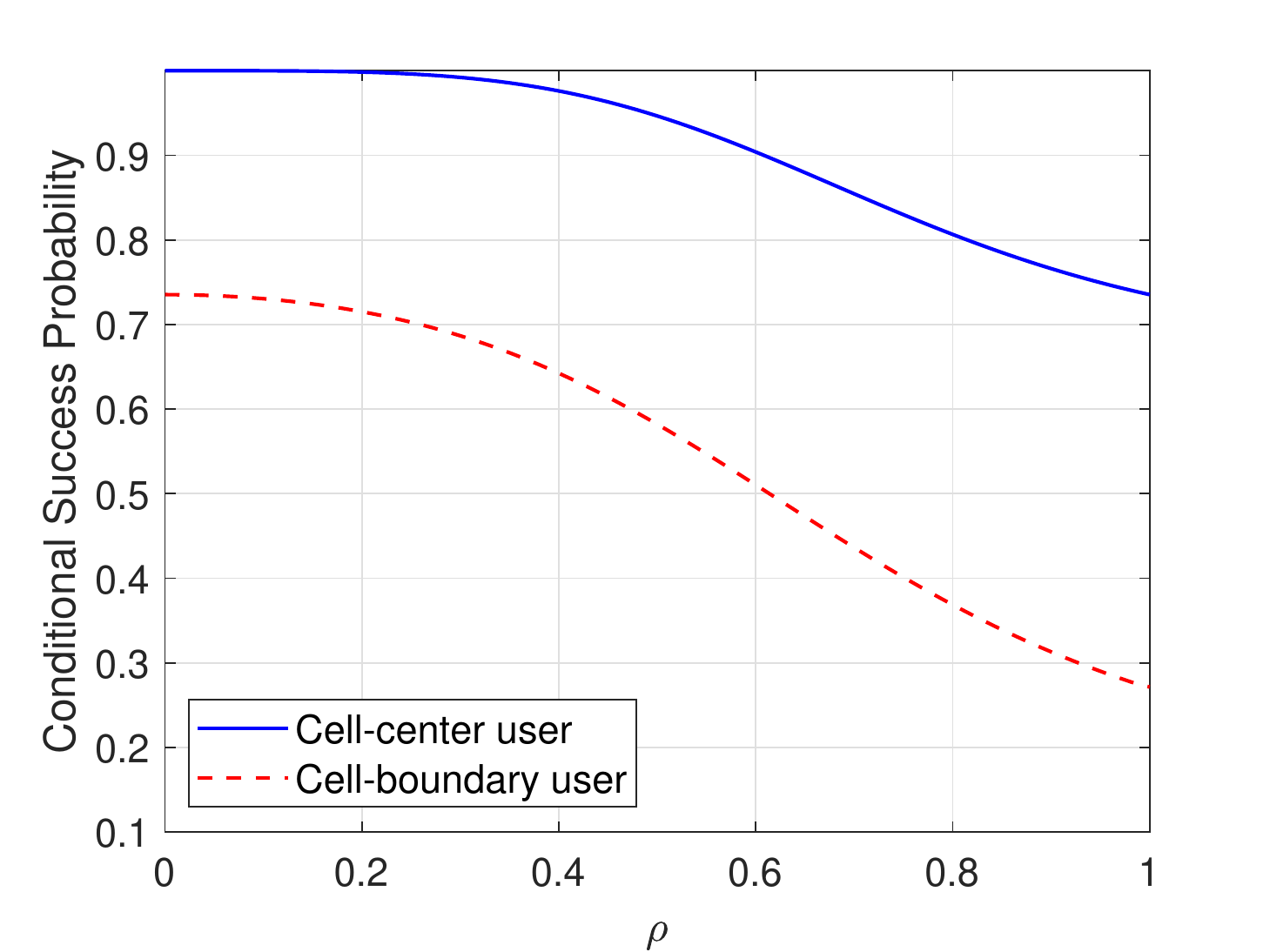} \caption{The temporal CSP of LSUs.} \label{fig:CCP_locationdependent} \end{figure}

\subsubsection{SIR Gain} \label{Sec:SIR_gain_1D}
Next, we investigate the asymptotic SIR gain to illustrate 
the SIR improvement over the typical general user due to location dependence.   
The MISR of the typical cell-center user can be derived as
\begin{align}
 \textup{MISR}_{\mathcal{R}_{\rm c}} & \! = \sum^{\infty}_{i=2} \mathbb{E} \bigg[ \bigg( \frac{r_{1}}{r_{i}}  \bigg)^{\!\alpha} \mid  o \in \mathcal{R}_{\rm c}   \bigg] \nonumber \\
 & \! \overset{(b)}{=}   \mathbb{E} \bigg[ \bigg(\frac{r_{1}}{r_{2}}  \bigg)^{\!\alpha}  \mid  o \in \mathcal{R}_{\rm c}  \bigg]  \sum^{\infty}_{i=2} \mathbb{E} \bigg[ \bigg( \frac{r_{2}}{r_{i}} \bigg)^{\!\alpha}   \bigg], \nonumber \\
 & \!  = \frac{\mathbb{E} \Big[ \Big(\frac{r_{1}}{r_{2}}  \Big)^{\!\alpha}  ,  o \in \mathcal{R}_{\rm c}  \Big]}{ \rho^2}  \sum^{\infty}_{i=2} \mathbb{E} \bigg[ \bigg( \frac{r_{2}}{r_{i}} \bigg)^{\!\alpha}   \bigg], \label{eqn:MISR_Rc}
\end{align} 
where $(b)$ holds as the condition that the typical general user is in the cell-center region, i.e., $o \in \mathcal{R}_{\rm c}$, only affects   the relative distance between $r_{1}$ and $r_{2}$.

The first expectation in (\ref{eqn:MISR_Rc}) can be derived based on the joint distribution of $r_{1}$ and $r_{2}$ as
\begin{align}
 \mathbb{E} \Big[ \Big(\frac{r_{1}}{r_{2}}  \Big)^{\!\alpha}  ,  o \in \mathcal{R}_{\rm c}  \Big]  & = \int^{\rho}_{0} f_{v_{2}} (v) v^{\alpha} \mathrm{d} v 
  = \frac{2 \rho^{2+\alpha}}{2+ \alpha}. \label{eqn:FT_MISR_R1}
\end{align}

The second expectation in (\ref{eqn:MISR_Rc}) can be obtained by exploiting the distribution of the distance ratios of a PPP as  
\begin{align}
 \sum^{\infty}_{j=2}   \mathbb{E} \bigg[ \bigg( \frac{r_{2}}{r_{i}} \bigg)^{\!\alpha}   \bigg] & = \frac{\sum^{\infty}_{j=2}   \mathbb{E} \Big[ \Big( \frac{r_{1}}{r_{i}} \Big)^{\!\alpha}   \Big] }{   \mathbb{E} \Big[ \Big( \frac{r_{1}}{r_{2}} \Big)^{\!\alpha}   \Big] } \nonumber \\
& \overset{(c)}{=} \frac{2+\alpha}{2} \sum^{\infty}_{j=2}   \int^{1}_{0}     v_{i}^{ \alpha}  
f_{i}(v) \mathrm{d} v     \nonumber \\
&  = \frac{2+\alpha}{2} \int^{1}_{0} v^{\alpha} 2 v^{-3}  \mathrm{d} v \nonumber \\
& =   \frac{\alpha+2}{\alpha-2},   \label{eqn:ST_MISR_Rc}
\end{align}
where $(c)$ follows the distribution of $\varrho_{2}$. 

 Plugging (\ref{eqn:FT_MISR_R1}) and (\ref{eqn:ST_MISR_Rc}) into (\ref{eqn:MISR_Rc}) yields
 \begin{align} \label{eqn:MISRP_RC}
 \textup{MISR}_{\mathcal{R}_{\rm c}} = \frac{2 \rho^{\alpha}}{\alpha-2}. 
 \end{align}

Subsequently, with $\textup{MISR}_{\textup{PPP}}$ and  $\textup{MISR}_{\mathcal{R}_{\rm c}}$ given in (\ref{eqn:MISR_2D}) and (\ref{eqn:MISRP_RC}), respectively, the SIR gain of the typical cell-center user can be obtained as 
\begin{align}
G_{\mathcal{R}_{\rm c} }= \frac{\textup{MISR}_{\textup{PPP}}}{\textup{MISR}_{\mathcal{R}_{\rm c}}} =  \rho^{-\alpha}.
\end{align}
Similarly, the SIR gains of different LSUs can be obtained as in Table VI. 

\noindent
{\bf Remark 7}: Since $\alpha > 2$ and $\rho \in [0,1]$, the SIR gain of the typical cell-center user is greater than or equal to one, and the equality holds when $\rho=1$. Moreover, the SIR gain of the typical cell-boundary user is smaller than or equal to one and the equality holds when $\rho=0$.
It can be found that the influence of $\rho$ on the SIR gain of the typical cell-center user is much more significant than that of the typical cell-boundary user. Moreover, the gap between the SIR gains of the typical edge user and vertex user increases with  $\alpha$.

\begin{table*}   
\begin{center} 
\caption{MISR and SIR gain of location-specific users} 
\begin{tabular}{ | p{3cm} | p{3.5cm}| p{2.5cm} | }   
\hline 
User Location & MISR & SIR Gain \\ 
\hline
Typical general user & $2/(\alpha-2)$ & 1 \\ 
\hline
Typical cell-center user & $2\rho^{\alpha}/(\alpha-2)$  & $\rho^{-\alpha}$ \\ 
\hline
Typical cell-boundary user & $2(1-\rho^{\alpha+2})/ (\alpha-2)(1-\rho^2)$  &  $(1\!-\!\rho^2)/ (1- \rho^{\alpha+2})$ \\ 
\hline
Typical edge user &  $(\alpha+2)/(\alpha-2)$  & $2/(\alpha+2)$  \\ 
\hline
Typical vertex user   & $2  \alpha  / ( \alpha -2)$ & $1/\alpha$ \\ 
\hline
\end{tabular}   
\end{center}    \label{Tab:SIR_gain_2D}
\end{table*}

\subsection{Summary and Discussion}

We have discussed the impact of the relative distance of a random user on its performance in Poisson downlink networks. Based on the relationship between the contact distance and the distance to the two nearest interferers, a user can be categorized as cell-center, cell-boundary, cell-edge, and vertex user. We present the derivations of  
the moments of the CSP$_\Phi$ for the typical cell-center, cell-boundary, cell-edge and vertex users in two-dimensional Poisson downlink networks. To show the direct impact of the location on the SIR, we also derive the SIR gain of the four types of users. The 
major lessons learned
are as follows.

\begin{itemize}

\item The gaps among the $b$-th moments of the CSP$_{\Phi}$ for the typical cell-center, cell-edge and vertex users depend only on $b$ and the SIR threshold $\theta$, i.e., $\mathcal{M}^{\rm c}_{P_{\rm s}}(b)= (1+\theta)^{b} \mathcal{M}^{\rm e}_{P_{\rm s}}(b) = (1+\theta)^{2 b} \mathcal{M}^{\rm v}_{P_{\rm s}}(b) $, $\forall b \in \mathbb{R}^{+}$.  

\item The successful transmission events are more temporally correlated for the typical cell-center user than the typical cell-boundary user.

\item The SIR gains for the typical cell-center and cell-boundary user depend on the path-loss exponent and location dependence coefficient $\rho$ while those for the typical cell-edge and vertex user only depend on the path-loss exponent.

\end{itemize} 

\vspace{0.2cm}
{\em Open Technical Issues}: 
In the literature, a location-dependent analysis has been carried out~\cite{K.2019Feng}, and BS cooperation schemes have been designed based on the location-dependent modeling~\cite{K.2019Feng2}. However, the location-dependent performance for Poisson uplink networks remains unexplored. Characterizing the location-dependent uplink user is more challenging as the uplink model is not fully tractable~\cite{M.2017Haenggi}. Another interesting direction is to analyze the performance of LSUs in non-Poisson cellular networks, e.g., when the transmitters exhibit repulsive or attractive distribution.

Additionally, cell-free communication systems  
where densely deployed BSs could simultaneously serve a number of users have emerged as a practical solution for future-generation communication systems.  
In such an infrastructure, the BS cooperation schemes need to take into account the user locations. For example, it is more meaningful to assign more resources for cell-boundary users to improve their success probabilities. Thus, location-dependent analysis is the key to the design of BS cooperation schemes in cell-free massive MIMO systems. 

 \section{  
 Spatially Correlated and Independent Shadowing} \label{sec:shadowing}
 
 The main objective of this section is to investigate the effect of correlated and independent shadowing in large-scale systems based on stochastic geometry analysis.  
 
 \subsection{System Model}

 To model the effect of shadowing,  we consider that the 
 whole plane is partitioned into a set of deterministic shadowing cells $\mathbf{S}= \{ \mathbb{S}_{k}  \subset \mathbb{R}^2 \}_{k \in \mathbb{N}}$ in which the transmitters  
 experience similar shadowing. 
 We consider an ad hoc network with a Poisson field of interferers  $\Phi$ introduced in Section~\ref{sec:sectionIII_SM} as an example system model for the analysis in this section.
  
 The channel gain from $x_{j} \in \Phi$ to $o$ is $h_{j} S_{j} \ell(\|x_j\|)$, where $h_{j}$ and $S_{j}$ are the fading coefficient and shadowing coefficient, respectively, from $x_{j}$ to $o$, and $\ell(\|x_j\|)= \|x_{j}\|^{-\alpha}$ is the path-loss function.  
With a Poisson field of interferers, the aggregated interference at the receiver of interest is
 \begin{align}   
  I_{o} & = \sum_{ j \in \mathbb{N} } h_{j} 
  S_{j}  
  \ell(\|x_j\|),   
  \end{align}
  and the SIR follows as
  \begin{align}
  \eta & = \frac{  h_{\rm t}  S_{\rm t}
  \ell(\|x_{\rm t}\|)  }{ I_{o} },  
 \end{align} 
 where $S_{\rm t}$ denotes the shadowing coefficient of the link under consideration.  
 
  Let $c(j)$ be the cell index that $x_j$ resides in, i.e., $c(j)=k$ if $x_{j} \in \mathbb{S}_{k}$.  
 We consider the extreme cases of (fully) correlated and independent intra-cell shadowing.
 In independent shadowing, the $S_{c(j)}$ are all independent and distributed with CDF $F_{c(j)}$. In correlated shadowing, $S_{c(j)}$ and $S_{c(i)}$ are independent only if $c(j)\neq c(i)$, and $S_{c(j)}=T_{k}$ for all $x_{j}  \in \mathbb{S}_{k}$. By choosing $F_{k}$, different shadowing properties can be assigned to the individual cells.

 \subsection{Performance Analysis}
  
 We first characterize the Laplace transform of the interference. Then,  based on it, we derive the mean and variance of the interference. Furthermore, we obtain the moments of the CSP$_\Phi$ under both correlated and independent shadowing.  
 
 \subsubsection{Laplace Transform of the Interference}

 To analyze the interference distribution, we start by characterizing the Laplace transform of the aggregated interference at the target receiver at the origin with correlated and independent shadowing, denoted by $I^{\rm Cor}_{o}$ and $I^{\rm Ind}_{o}$, respectively, by following the methodology in \cite{J.toappearLee}.
 Similar to the derivation steps of the moments of the CSP given a Mat\'{e}rn cluster field of interferers presented in \textbf{Appendix A}, we first express the Laplace transform as the expectation of the product of a function of the locations of the interferers given the distribution of the shadowing coefficients.
 Then, by the PGFLs of MCP  given in (\ref{eqn:PGFL_MCP}), we convert the expectation into an integral expression as a function of the shadowing coefficients.  Subsequently, the randomness of the shadowing can be averaged out by conditioning on that the daughter points in each cluster are associated with the same and independent shadowing coefficients with correlated shadowing and independent shadowing, respectively.

 Based on this methodology, the Laplace transforms of   $I^{\rm Cor}_{o}$ and $I^{\rm Ind}_{o}$ are given in the following theorem.

 \begin{theorem} \label{thm:shadowing}
 In a Rayleigh fading environment with a Poisson field of interferers, the conditional Laplace transform of the interference at the target receiver 
 is
 \begin{align} \label{eqn:L_cor_shadowing}
 & \mathcal{L}_{I^{\rm Cor}_{o}   } (s) = \nonumber \\ 
 & \prod_{k \in \mathbb{N}} \mathbb{E}_{T_{k}} \Bigg[  \exp  \bigg( \! - \! \lambda    \int_{\mathbb{S}_{k}} \! \! \bigg( 1 -     \frac{1}{1\!+\!s    
  \ell(\|x\|)   T_{k} }   \! \bigg)     \mathrm{d} x \bigg)\Bigg] ,
 \end{align}
 with correlated shadowing, and
 \begin{align}  \label{eqn:L_ind_shadowing}
  & \mathcal{L}_{I^{\rm Ind}_{o} } (s) = \nonumber \\ 
 & \prod_{k \in \mathbb{N}} \exp  \Bigg( \! - \! \lambda     \int_{\mathbb{S}_{k}} \! \! \bigg( 1 -   \mathbb{E}_{T_{k}} \bigg[   \frac{1}{1\!+\!s    
 \ell(\|x\|)  
  T_{k} }   \bigg] \! \bigg)     \mathrm{d} x \Bigg) ,
 \end{align}   
 with independent shadowing.
  \end{theorem}
   
 \begin{proof} 
 See \textbf{Appendix D}.
 \end{proof}

 According to Jensen’s inequality, i.e., the convex transformation of an expectation equals or exceeds the expectation over the convex transformation, we have the inequality in (\ref{eqn:inequaility_shadowing}).
 \begin{figure*}
 \begin{align}
 &  \mathbb{E}_{T_{k}} \bigg[ \exp \bigg( - \bar{c} \int_{\mathbb{S}_{k}} \Big( 1 -    \frac{1}{1+s    \ell(\|x\|) T_{k}  }  \Big) \mathrm{d} x     \bigg) \bigg]  \nonumber \\
 & \hspace{45mm} > \exp \bigg( - \bar{c} \int_{\mathbb{S}_{k}} \bigg( 1 -  \mathbb{E}_{T_{k}} \bigg[  \frac{1}{1 + s   \ell(\|x\|) T_{k}  }    \bigg] \mathrm{d} x \bigg)   \bigg)  . \label{eqn:inequaility_shadowing}
 \end{align} 
 \hrulefill
 \end{figure*}
 Furthermore, since $\exp(-x)$ is a completely monotone function, we readily obtain the following observation by comparing (\ref{eqn:L_cor_shadowing}) and (\ref{eqn:L_ind_shadowing}). 
   
 \vspace{0.1cm}
 \noindent
 {\bf Remark 8}:  For all $s>0$,  $\mathcal{L}_{I^{\rm Cor}_{o}   } (s) >  \mathcal{L}_{I^{\rm Ind}_{o}   } (s)$.

 \subsubsection{Mean and Variance of Interference} 
 
 Due to the singularity of the path-loss function $\ell(x)$,  we adopt $\ell_{\epsilon}(x)=\frac{1}{\epsilon+\|x\|^{\alpha}}, \alpha>2, \epsilon>0$, to evaluate the mean and variance of interference, similar to Sec.~\ref{sec:spatial_PA_IC}.
 As the mean of a random variable can be derived by taking the first-order derivative of 
 its Laplace transform w.r.t. $s$ and setting $s = 0$,  i.e., $\mathbb{E} [X] = - \frac{ \mathrm{d} \mathcal{L}_{X}(s) }{ \mathrm{d} s } |_{s=0} $, we can easily obtain the mean interference as
 \begin{align} 
 & \mathbb{E} \Big[ I^{\rm Ind}_{o}  \Big]   = \mathbb{E} \Big[ I^{\rm Cor}_{o}   \Big]   = \lambda \sum_{k \in \mathbb{N} }  \mathbb{E} [ T_{k} ] \int_{\mathbb{S}_{k}}  \frac{1}{\epsilon + \|x\|^{\alpha}}  \mathrm{d} x   , \nonumber 
 \end{align} 
 which 
 shows that the mean interference at the typical receiver with independent shadowing and correlated shadowing are identical, i.e., $I^{\rm Ind}_{o} = I^{\rm Cor}_{o}$.

 The same result can also be obtained from Campbell's Theorem~\cite[Thm. 4.1]{M.2013Haenggic}, which shows that it holds for all stationary point processes with an arbitrary fading model.    
 
 As the second moment of a random variable can be derived by evaluating the second derivative of its Laplace transform w.r.t. $s$ at $s = 0$, 
 the variance of $I^{\rm Ind}_{o} $ can be obtained based on the Laplace transform as \cite{J.toappearLee} 
 \begin{align}
 \mathbb{V} [ I^{\rm Ind}_{o} ] &   = \mathbb{E} \Big[ \big( I^{\rm Ind}_{o} \big)^2    \Big ] - \mathbb{E} \Big[   I^{\rm Ind}_{o}    \Big ]^2 \nonumber \\
 &  = \frac{ \mathrm{d}^2   \mathcal{L}_{I^{\rm Ind}_{o}}(s) }{\mathrm{d}^2  s^2 }  \big|_{s=0} - \bigg(  \frac{ \mathrm{d}   \mathcal{L}_{I^{\rm Ind}_{o}}(s) }{\mathrm{d}  s } \big|_{s=0} \bigg)^2 \nonumber \\
 &  = 2 \lambda \sum_{k \in \mathbb{N} }   \mathbb{E} \big[ T^2_{k} \big] \int_{\mathbb{S}_{k}}   \frac{1}{ \big(\epsilon + \|x\|^{\alpha} \big)^2}    \mathrm{d} x  \nonumber \\
 &  \hspace{2mm} +   \lambda^2 \sum_{k \in \mathbb{N} }    \mathbb{E} [T_{k}]^2  \bigg( \int_{\mathbb{S}_{k}}  \frac{1}{\epsilon + \|x\|^{\alpha}}   \mathrm{d} x \bigg)^2 . \label{eqn:V_I_cor_ind}
 \end{align}
  
 In the same manner, the variance of $I^{\rm Cor}_{o}$ can be derived as
 \begin{align} 
\mathbb{V} [ I^{\rm Cor}_{o} ]  &    = 2 \lambda \sum_{k \in \mathbb{N} }   \int_{\mathbb{S}_{k}}  \frac{\mathbb{E} \big[ T^2_{k} \big]}{ \big(\epsilon + \|x\|^{\alpha} \big)^2}     \mathrm{d} x  \nonumber \\
 &  \hspace{5mm} +   \lambda^2 \sum_{k \in \mathbb{N} }    \mathbb{E} [T_{k}]^2  \bigg( \int_{\mathbb{S}_{k}}  \frac{1}{\epsilon + \|x\|^{\alpha}}   \mathrm{d} x \bigg)^2  \nonumber \\
 &   \hspace{5mm} +  \lambda^2  \sum_{k \in \mathbb{N}}   \mathbb{V} [ T_{k} ]  \bigg(  \int_{\mathbb{S}_{k}}  \frac{1}{\epsilon + \|x\|^{\alpha}}  \mathrm{d} x \bigg)^2  \nonumber \\
 & \overset{(a)}{=} \! \mathbb{V} [ I^{\rm Ind}_{o} ] \! + \! \underbrace{  \lambda^2 
 \! \sum_{k \in \mathbb{N} }   \! \mathbb{V} [T_{k}]^2  \bigg( \int_{\mathbb{S}_{k}} \! \frac{1}{\epsilon \! + \! \|x\|^{\alpha}}  \mathrm{d} x \bigg)^2 }_{ >  0},  \label{eqn:V_I_cor}
 \end{align} 
 where $(a)$ follows from (\ref{eqn:V_I_cor_ind}). Comparing (\ref{eqn:V_I_cor_ind}) and (\ref{eqn:V_I_cor}) yields the following observation.

 \vspace{0.1cm}
 \noindent
 {\bf Remark 9}: The variance of the interference at the typical receiver with correlated shadowing is greater than that with independent shadowing, i.e., $\mathbb{V}[I^{\rm Cor}_{o}] >  \mathbb{V}[I^{\rm Ind}_{o}]$.

 \subsubsection{Moments of the CSP$_\Phi$}

 Given the point process $\Phi$, the CSP is  
 \begin{align} 
  \mathbb{P} [ \eta > \theta    \mid  \Phi ] &  =   \mathbb{P} \bigg[  \frac{ h_{\rm t} S_{\rm t}   r^{-\alpha}_{\rm t}  }{  I_{o}   } > \theta  \Bigm|   \Phi \bigg]   \nonumber \\ 
 & \overset{(a)}{=}  \mathbb{E}_{S_{\rm t}, (S_{j}) } \Bigg[ \prod_{ j \in \mathbb{N} } \bigg( 1 + \frac{\theta r^{\alpha}_{\rm t} r_{j}^{-\alpha} S_{j} }{ S_{\rm t} } \bigg)^{\!-1} \Bigg] 
 , \nonumber  
 \end{align} 
 where $(a)$ follows the derivations of (\ref{eqn:spatial_SP1}) in \textbf{Appendix A}.

 Then, the moments of the CSP$_\Phi$ can be represented as
 \begin{align}
    \mathcal{M}_{P_{\rm s}} (b) =    \mathbb{E}  \Bigg[  \prod_{ j \in \mathbb{N} }  
    \! \bigg( 1 \! + \! \frac{\theta r^{\alpha}_{\rm t} r_{j}^{-\alpha} S_{j} }{ S_{\rm t} } \bigg)^{\!-b}  \Bigg].  \nonumber
 \end{align}
  
 By following the same methodology to the proof of Theorem~\ref{thm:shadowing}, we have $\mathcal{M}_{P_{\rm s}} (b)$ derived in the following theorem.
 
 \begin{theorem}
 In a Rayleigh fading environment, the conditional moments of the CSP$_\Phi$ under correlated shadowing and independent shadowing are given, respectively, by (\ref{eqn:MCSP_shadowing_cor}) and (\ref{eqn:MCSP_shadowing_ind}), shown on the top of the next page.  
 \begin{figure*}
 \begin{align}
 & \mathcal{M}^{\rm Cor}_{P_{\rm s}} (b) 
   =   \prod_{k \in \mathbb{N}} \! \mathbb{E}_{T_{k},T_{c(\rm t)}} \! \Bigg[  \exp \! \bigg( \! - \! \lambda  \!  \int_{\mathbb{S}_{k}} \! \!   \bigg ( \! 1 -    \bigg ( \frac{1}{1\!+\! \theta r^{\alpha}_{\rm t} \|x\|^{-\alpha}  T_{k} / T_{c(\rm t)} } \bigg)^{\!b}     \bigg)     \mathrm{d} x \!\bigg) \!\Bigg], \label{eqn:MCSP_shadowing_cor} \\ 
 & \mathcal{M}^{\rm Ind}_{P_{\rm s}} (b) 
  =    \prod_{k \in \mathbb{N}} \mathbb{E}_{T_{c(\rm t)}} \Bigg[  \exp \! \Bigg( \! - \! \lambda   \!  \int_{\mathbb{S}_{k}} \! \! \bigg( \! 1 -   \mathbb{E}_{T_{k}} \! \bigg[   \bigg( \frac{1}{ 1\!+\! \theta r^{\alpha}_{\rm t}  \|x\|^{-\alpha}   T_{k} / T_{c(\rm t)} } \bigg)^{\!b}  \bigg] \! \bigg)     \mathrm{d} x \! \Bigg) \Bigg] . \label{eqn:MCSP_shadowing_ind}
  \end{align}  
  \hrulefill
  \end{figure*}
 \end{theorem}

 \noindent
 {\bf Remark 10}: It follows from \textbf{Remark 8} that with Rayleigh fading, the moments of the CSP$_\Phi$ under correlated shadowing are greater than that with independent shadowing.

 Next, we numerically evaluate the  
 network performance under correlated and independent shadowing assumptions. We assume  $T_{k}=\kappa^{N_{k}(x)}$
 ~\cite{J.toappearLee}, where $\kappa$ is the attenuation factor accounting for the signal loss while penetrating a blockage (e.g., a wall) and $N_{k}(x)$ represents the number of blockages between $x$ and the receiver of interest.  
 $N_{k}(x)$ is assumed to be a  Poisson random variable for each  link~\cite{T.2014Bai}, i.e., $N_{k}(x) \sim \mathcal{P}( \lambda_{\rm b}\|x\|)$, where $\lambda_{\rm b}$ denotes the blockage density. With correlated shadowing, the transmitters  in the same cell $\mathbb{S}_{k}$  
  are associated with the same shadowing coefficient  $T_{k}=\kappa^{N_{k}(d_{k,o})}$, $N_{k}(d_{k,o}) \sim \mathcal{P}( \lambda_{\rm b}d_{k,o})$,  where $d_{k,o}$ represents the distance between the center of $\mathbb{S}_{k}$ and the origin.  Differently, with independent shadowing, the transmitters  
   at $x_{j} \in \mathbb{S}_{k}$ are associated with  
  i.i.d. shadowing coefficients  $S_{j}=\kappa^{N_{k}(d_{k,o})}$, $N_{k}(d_{k,o}) \sim \mathcal{P}( \lambda_{\rm b}d_{k,o})$. Besides, we assume that the link between the target receiver and the serving transmitter is subject to no shadowing, i.e., $S_{\rm t}=1$.
 For the spatial configuration of the shadowing cells,
 we consider the  
 example in Fig.~\ref{fig:grid}, where $\mathbf{S}=\{\mathbb{S}_{k}\}_{k \in \mathbb{N}}$ is a set of squares with length $L$. 
 It is worth noting that the analytical framework can be extended to other types of point processes by changing the distribution of points as long as $|\mathbb{S}_{k}|<\infty$, i.e., the point process has a finite number of points in each shadowing cell almost surely.
 
 Figure~\ref{fig:CP_shadowing} illustrates the success probability, i.e., $\mathcal{M}^{\rm Cor}_{P_{\rm s}} (1)$ and $\mathcal{M}^{\rm Ind}_{P_{\rm s}} (1)$, as a function of the SIR threshold. It can be observed that the success probability with correlated shadowing is greater than that with independent shadowing, which agrees with {\bf Remark 10}. Moreover, the gap between the success probabilities with correlated and independent shadowing decreases as the cell size (i.e., $L^2$) becomes smaller, as in this case the signals from the interferers  experience nearly independent shadowing even with the correlated shadowing model.  

  \begin{figure}
  \centering
  \includegraphics[width=0.35\textwidth]{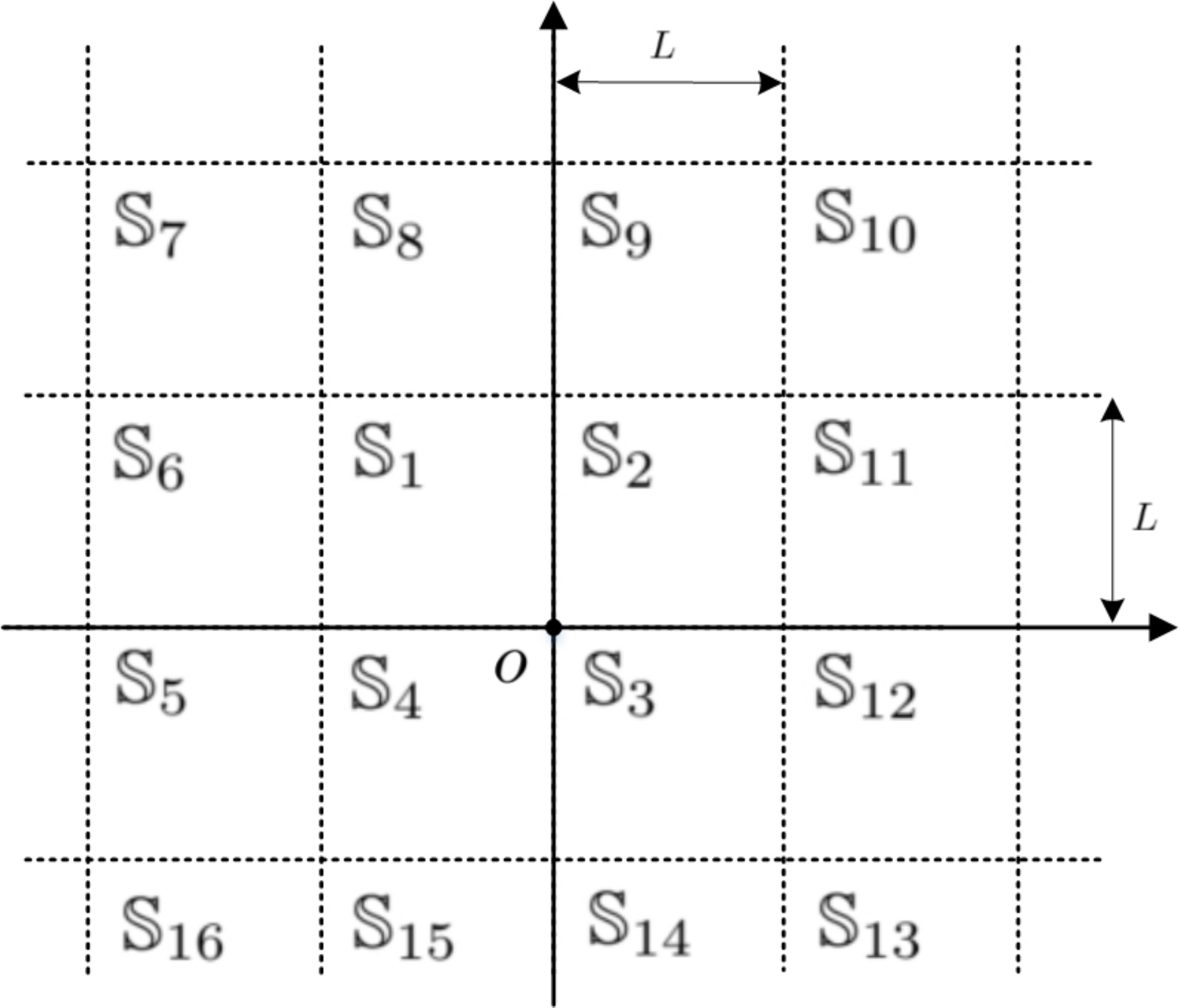} 
  \caption{Configuration of shadowing cells for simulations.  } \label{fig:grid} 
  \end{figure}

  \begin{figure*}[htp]  
     \centering
       \begin{minipage}[c]{0.48 \textwidth}
        \includegraphics[width=0.98\textwidth]{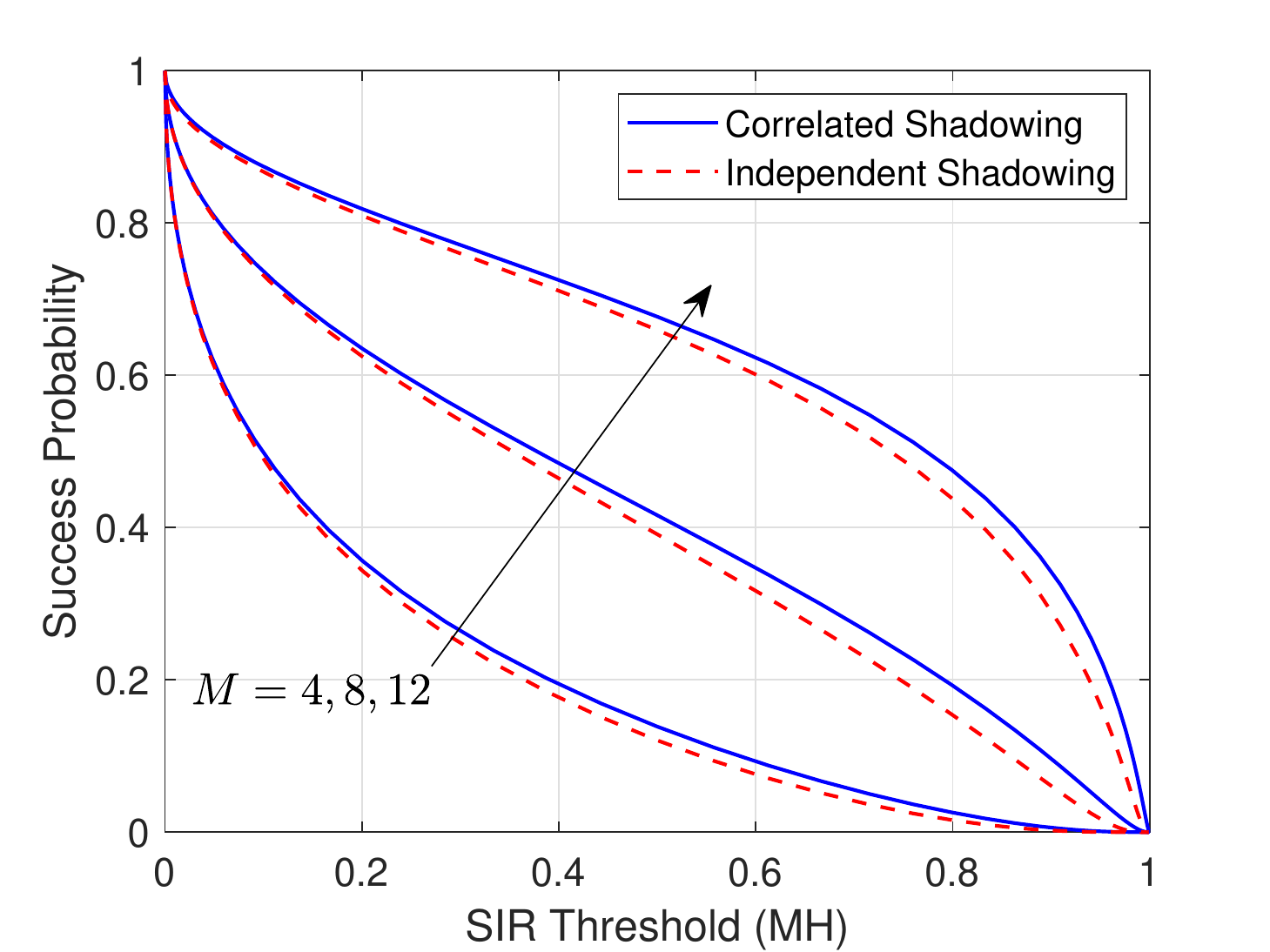} 
        \caption{Success probability versus SIR threshold under correlated and independent shadowing ($\lambda=1$,  $\kappa=0.5$, $\alpha=4$, $r_{\rm t}=1$).   
        } \label{fig:CP_shadowing} 
       \end{minipage}  \hspace{3mm}
       \begin{minipage}[c]{0.48 \textwidth}\vspace{0mm}
         \includegraphics[width=0.98\textwidth]{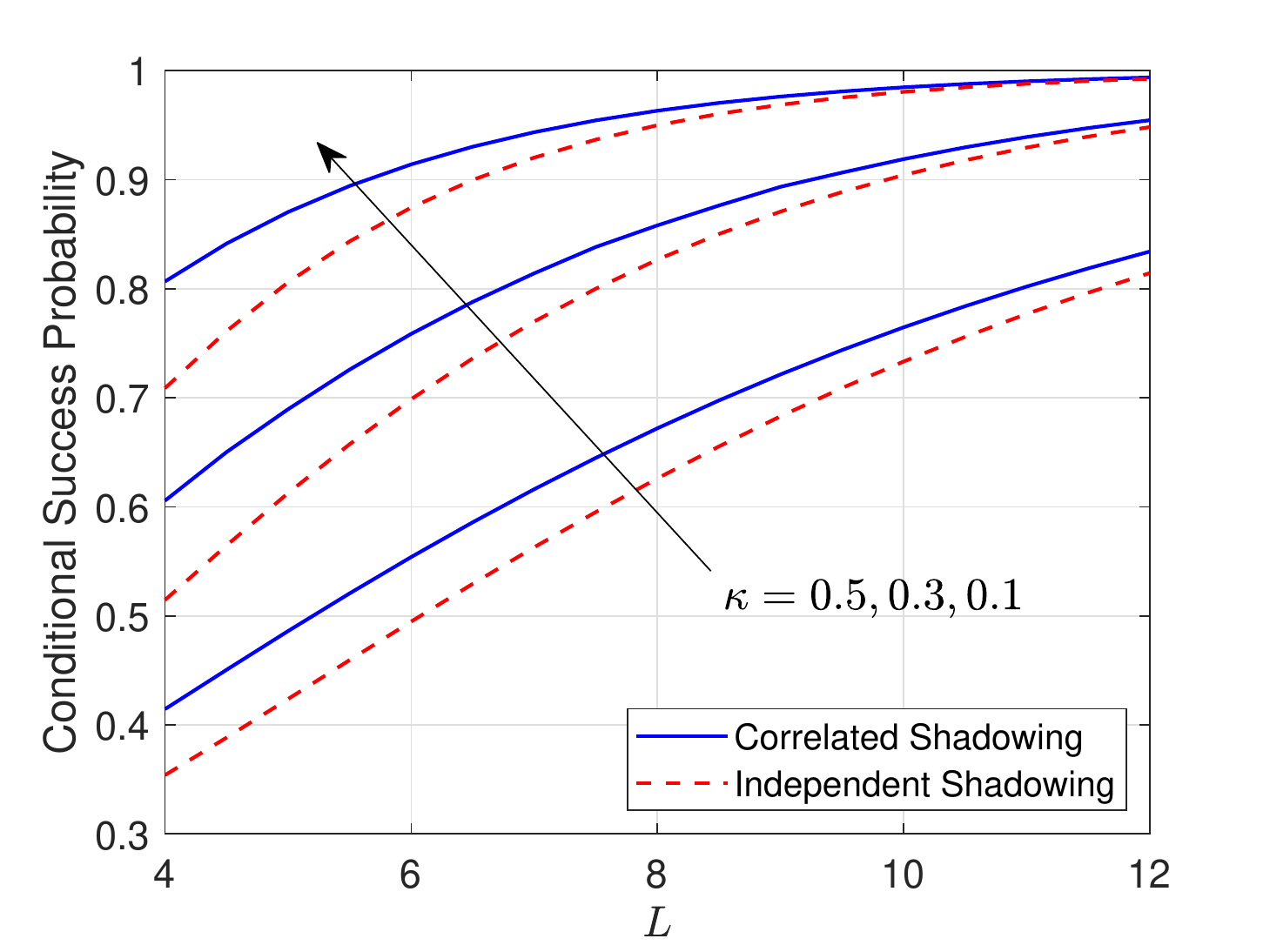}  
         \caption{CSP versus $L$ under correlated and independent shadowing  ($\lambda=1$,   $\alpha=4$, $r_{\rm t}=1$, $\theta=1$).   } \label{fig:JCP_shadowing} 
         \end{minipage}  
      \end{figure*}  
  
  Furthermore, to evaluate the temporal effect of shadowing, we evaluate the temporal CSP with both correlated and independent fading, i.e., $\frac{\mathcal{M}^{\rm Cor}_{P_{\rm s}}(2)}{\mathcal{M}^{\rm Cor}_{P_{\rm s}}(1)}$ and $\frac{\mathcal{M}^{\rm Ind}_{P_{\rm s}}(2)}{\mathcal{M}^{\rm Ind}_{P_{\rm s}}(1)}$, as shown in Figure~\ref{fig:JCP_shadowing}. 
 We can observe that correlated shadowing results in more temporal dependence between two successful transmission events. 
 The gap of CSPs between correlated and independent shadowing decreases as the cell size $L^2$ becomes larger.  
 Moreover, the dependence between two successful temporal transmission events increases as the cell size becomes smaller, especially when the SIR threshold is small.  The reasons for the above observations can be ascribed to the adopted shadowing model, which essentially leads to exponential path loss.

 \subsection{Summary and Discussion}
  
 In this section, we have established an analytical framework to model the interference distribution and success probabilities in networks with spatially-correlated shadowing and spatially-independent shadowing. In particular, the analytical framework characterizes deterministic shadowing cells where the transmitters inside are associated with the same shadowing effects. 
 The key  
 lessons learned are as follows.
 
 \begin{itemize}
 
 \item Spatially-correlated shadowing generates the same interference as but higher variance than spatially-independent shadowing.
 
 \item The moments of the CSP$_\Phi$ under spatially-correlated shadowing is greater than that with spatially-independent shadowing. The performance increase reduces when the cell size shrinks.

 \item The successful transmission events are more temporally correlated with spatially-correlated shadowing than with spatially-independent shadowing.

 \end{itemize}
   
 \vspace{0.2cm}
 \noindent
 {\em Open Technical Issues}: 
 Most of the existing literature, e.g., \cite{F.May2015Baccelli,X.Dec.2015Zhang,J.Nov.2016Lee,J.toappearLee,T.2018Kimura}, only investigates network performance 
 under the assumptions that the shadowing coefficients of the interferers in the same shadowing cells are either fully spatially correlated or independent. However, in practice, the interferers 
 in the same shadowing cell may only experience partially correlated shadowing effect due to their location difference. More 
 accurate shadowing models need to be developed by taking into account the properties (e.g., shape, density and mobility) of the obstacles that may exhibit location-dependency (e.g., urban and rural areas). Moreover, temporal shadowing variation due to mobility of obstacles (e.g., vehicles), access points (e.g., drone hot spots) or users is an important factor to be investigated in practical systems. 
 
 Different from the deterministic shadowing cells introduced in this section, reference~\cite{K.Feng2021} introduces a correlated shadowing model dependent on the Poisson Voronoi cell. This model is suitable for studying coverage-oriented  cellular networks where BSs deployment is configured to guarantee the cell-boundary users to acquire sufficient signal strength.  An intriguing future direction is to extend the study of cell-dependent correlated shadowing to capacity-oriented cellular networks.

\section{Spatial-Temporal Interactions Between Queues} \label{sec:queueing}
The majority of the existing literature heavily relies on the assumption that each transmitter always has packets in the buffer to send out, which does not characterize random traffic flows. 
While the temporal randomness 
of the traffic flows complicates the  analysis, 
it is nevertheless essential to 
understanding system-level performance. This is due to the fact that the traffic patterns in the evolving wireless networks are getting increasingly more dynamic and heterogeneous. 
The main difficulty of random traffic characterization 
originates from the correlation among the buffer statuses of different transmitters, 
often referred to as {\em interacting queues} \cite{H.2019Yang,H2021Yang}. Since the queues interact spatially and temporally, an exact analysis of the mutual interference is quite challenging. 

\subsection{System Model}

We consider both Poisson downlink networks (as introduced in Section~\ref{sec:sectionIII_SM}) and
{\em Poisson bipolar networks}~\cite[Def. 5.8]{M.2013Haenggic}.
In a Poisson downlink network, the transmitters (i.e., BSs) and receivers (i.e., users) are distributed following independent homogeneous PPPs, denoted as  $\Phi_{\rm B}=\{x_{j}\}_{j \in \mathbb{N}}$  and $\Phi_{\rm u}$ with intensity $\lambda_{\rm B}$ and $\lambda_{\rm u}$, respectively. 
  The points in $\Phi_{\rm B}$ are assumed to be ordered from nearest to farthest to the origin, i.e., $\|x_{j}\| < \|x_{j+1}\|$.   
It is assumed that each user is associated with its nearest BS for downlink transmission. In a Poisson bipolar network, the transmitters are distributed as a homogeneous PPP $\Phi$ with intensity $\lambda$. Each transmitter is paired with one receiver in a uniformly random direction with a link distance $r_{\rm t}$.  
 Without loss of generality, we study the performance of the typical receiver, conditioned to be at the origin, in both models.   

We consider a discrete-time transmission and queueing model. 
Specifically, the data transmissions are 
divided into equal-duration time slots. We consider fixed-length data packets and assume it takes exactly one time slot to send out one packet.  
If a transmitter is scheduled for transmission, it can only send out the accumulated packet(s) that arrived prior to the transmission. In light of queueing, we assume that the incoming packets at each queue 
are stored in a buffer with infinite size and sent out on a first-in-first-out basis. A transmitted packet is removed from the head of the queue only if it is successfully decoded at the target receiver.

For Poisson downlink networks, each BS maintains an individual queue for the arrived packets of each associated user~\cite{Y.Jun.2017Zhong}. Hence, each BS has a number of queues equal to the number of users in its Voronoi cell. The temporal arrival of traffic at each 
queue follows an i.i.d. Bernoulli process with arrival rate $\xi_{\rm u}$ representing the probability of a new arrival per time slot.  
The users associated with the same BS are served based on {\em random scheduling}~\cite{Y.Jun.2017Zhong}, i.e., each BS randomly selects one user within its Voronoi cell with equal probability to serve in each time slot. If the selected user has a non-empty queue at the BS, the BS is scheduled for transmission. Otherwise, the BS is muted.

For Poisson bipolar networks, the packet arrival at each transmitter $x_{j} \in \Phi$ follows an i.i.d. Bernoulli process with arrival rate $\xi$. At each time slot, 
the transmitters with non-empty buffer are all scheduled for transmission \cite{Y.2018Zhong}.

Given that the typical receiver is receiving data, its SIR in Poisson downlink and Poisson bipolar networks are given, respectively, by
\begin{align}
& \eta = \frac{ h_{1} \|x_{1}\|^{-\alpha} }{ \sum^{\infty}_{j=2} \iota_{j}  h_{j} \| x_{j} \|^{-\alpha} }, \,\, x_{j} \in \Phi_{\rm B},
\end{align}
and
\begin{align}
& \eta = \frac{ h_{\rm t} \|x_{\rm t}\|^{-\alpha} }{ \sum_{j \in \mathbb{N}  } \iota_{j}  h_{j} \| x_{j} \|^{-\alpha} }, \, \, x_{j} \in \Phi, 
\end{align}
where 
$\iota_{j}$ denote the state indicator of the transmitter located at $j$ which equals 1 and 0 when the transmitter is on and off, respectively.

\subsection{Performance Analysis}

In large random networks, characterizing the exact queue interaction among the transmitters is quite challenging. Fortunately, in a large-scale network,  the correlation among the interacting queues tends to be ``weak" and ``global"~\cite{Howard2018Yang}.   
Therefore, the impact of interacting queues tends to be negligible. In the following, we show how to approximate the success probability in Poisson downlink networks with ``interacting queues" by exploiting the mean-field property (e.g., as in \cite{Y.2018Zhong,Howard2018Yang}). 

\begin{itemize}

\item  We first compute the success probability of the typical user's cell, i.e., the cell containing the origin,  
based on the assumption that each BS at $x_{j} \in \Phi \backslash \{x_{1}\} $ is active independently with probability  $p_{\rm A}=\mathbb{E}[\iota_{j}]$.

\item We then derive the active probability of the serving BS of the typical user as a function of packet arrival rate and success probability (i.e., service rate) based on queueing theory. By inserting the success probability obtained in the previous step, we can establish a fixed-point equation of $p_{\rm A}$.

\item We finally obtain the success probability by plugging in $p_{\rm A}$ obtained by solving the fixed-point equation. 

\end{itemize} 

Based on the above methodology, for a Poisson downlink network, we have the  following result. 

\begin{theorem} \label{thm:DSP}
With infinite buffer size and random scheduling, the success probability of a Poisson downlink network under Rayleigh fading can be approximated by $P_{\rm s}$, obtained by solving the fixed-point equation 
\begin{align}
P_{\rm s}  =  \bigg( 1 +    \frac{  \xi_{\rm u}   ( {_2}F_{1} (1, - \delta; 1 - \delta ; -\theta) - 1  ) }{\sum^{\infty}_{n=1} p_{N_{\rm u}} (n)  P_{\rm s} /n }    \bigg)^{-1},
\end{align} 
where   
\begin{align} p_{ N_{\rm u} }(n) = \frac{\nu^{\nu} \Gamma(n+\nu) (\lambda_{\rm u}/\lambda_{\rm B})^n}{n!\Gamma(\nu)(\lambda_{\rm u}/\lambda_{\rm B} + \nu)^{n+\nu}},
\end{align}
with $\delta=2/\alpha$ and $\nu=3.5$.
\end{theorem} 

\begin{proof}
See \textbf{Appendix E}.
\end{proof}

Fig. \ref{fig:CP_queue_downlink} shows the success probability under different packet arrival rates  $\xi_{\rm u}$ in a Poisson downlink network. We observe that ignoring the temporal and spatial correlation among queue iterations closely approximates the success probability achieved in the presence of the interactions. This is due to the ``mean field" effect  in a large-scale network. 
We also note that the traffic arrival rate $\xi$ plays a pivotal part in the success probability.
For example, to achieve a target success probability of $80\%$, the disparity of the supported SIR thresholds with $\xi_{\rm u}=0.01$ and $\xi_{\rm u}=0.05$  
can be over 10 dB.   
A larger $\xi_{\rm u}$ decreases the success probability due to the increased density of interferers.

\begin{figure}
 \centering
 \includegraphics[width=0.5\textwidth]{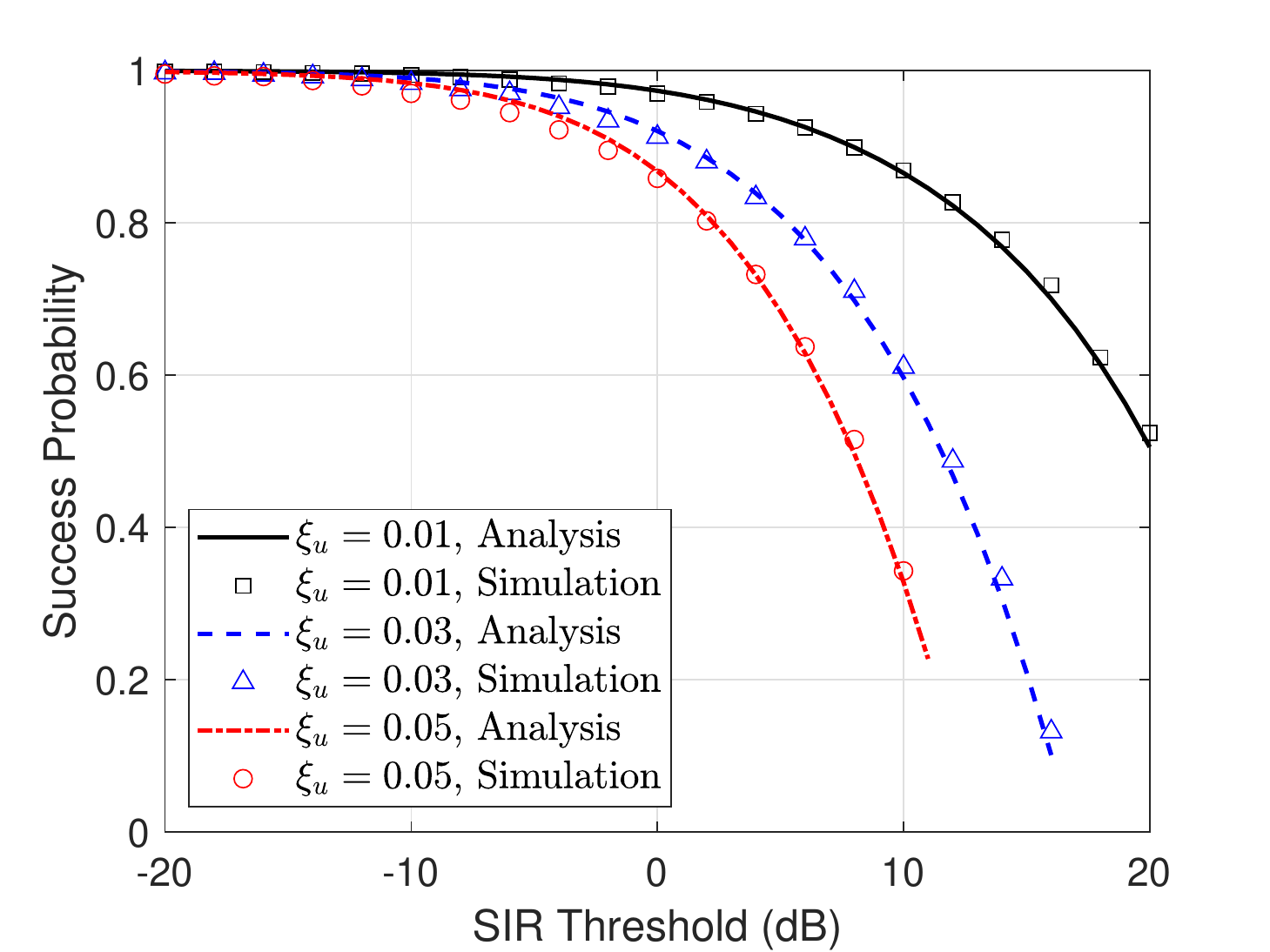} 
 \caption{Success probability versus SIR threshold in Poisson downlink networks ($\alpha=4$,  $\lambda_{\rm u}/\lambda_{\rm B}=5$). 
   } \label{fig:CP_queue_downlink} 
 \end{figure}

Next, we discuss the success probability in Poisson bipolar networks. By following the same methodology as used for
Poisson downlink networks, we obtain the result  
presented in the following theorem.

\begin{theorem} \label{thm:buffer_bipolar}
If all links in a Poisson bipolar network experience Rayleigh fading and have infinite buffer size, the success probability
of a typical link is given by (\ref{eqn:CP_buffer}), 
\begin{figure*}
\begin{align}  
P_{\rm s} \approx \max \bigg \{ \! \exp \bigg( \mathcal{W} \big( - \xi \lambda \pi r^2_{\rm t}  \theta^{\delta}  \Gamma(1\!+\!\delta ) \Gamma(1\!-\!\delta ) \big) \! \bigg) , \exp \Big( \!  -   \lambda \pi r^2_{\rm t}  \theta^{\delta}    \Gamma(1\!+\!\delta ) \Gamma(1\!-\!\delta ) \Big) \! \bigg \},
 \label{eqn:CP_buffer}
\end{align}
\hrulefill
\end{figure*}
where $\delta= \frac{2}{\alpha}$.
\end{theorem}

\begin{proof}
See \textbf{Appendix F}.
\end{proof}

Next, we explore the asymptotics of the success probability by utilizing the expansion of the Lambert-$\mathcal{W}$ function as $\mathcal{W}(z)\sim z$ when $z \to 0$, 
which follows from $x e^x \sim x$ when $x \to 0$.

 \begin{figure}
 \centering
 \includegraphics[width=0.5\textwidth]{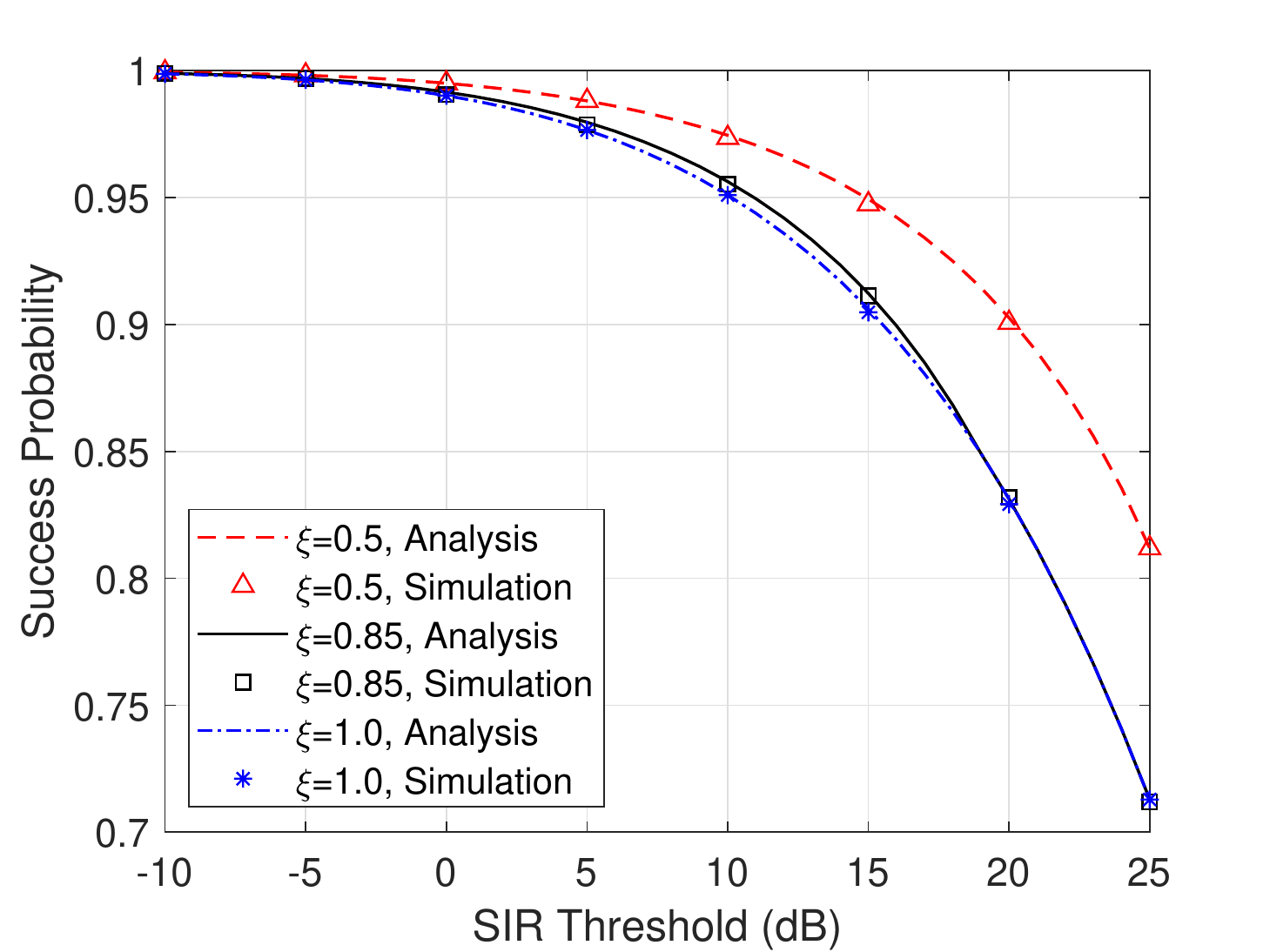} 
 \caption{Success probability as a function of the SIR threshold in Poisson bipolar networks ($\alpha=4$, $r_{\rm t}=2$, $\lambda=0.001$).  } \label{fig:CP_queue} 
 \end{figure}

Fig.~\ref{fig:CP_queue} depicts the success probability 
versus the SIR threshold with different settings of the packet arrival rate.  
The analytical results closely match the simulation results, which validates the effectiveness of the adopted approximation. 
Additionally, 
the success probability with $\xi=1.0$ overlaps that with $\xi=0.85$ when the SIR threshold is large (e.g., when $\theta>20$ dB). The reason is that in both cases the service rate of each queue is below the packet arrival rate, and thus the buffer is always non-empty.
As a result, all the transmitters remain active for transmission, which renders the same success probability under different packet arrival rates. 
This observation also shows that the success probability is lower-bounded by the case with a full load.

  \begin{figure}
 \centering
 \includegraphics[width=0.5\textwidth]{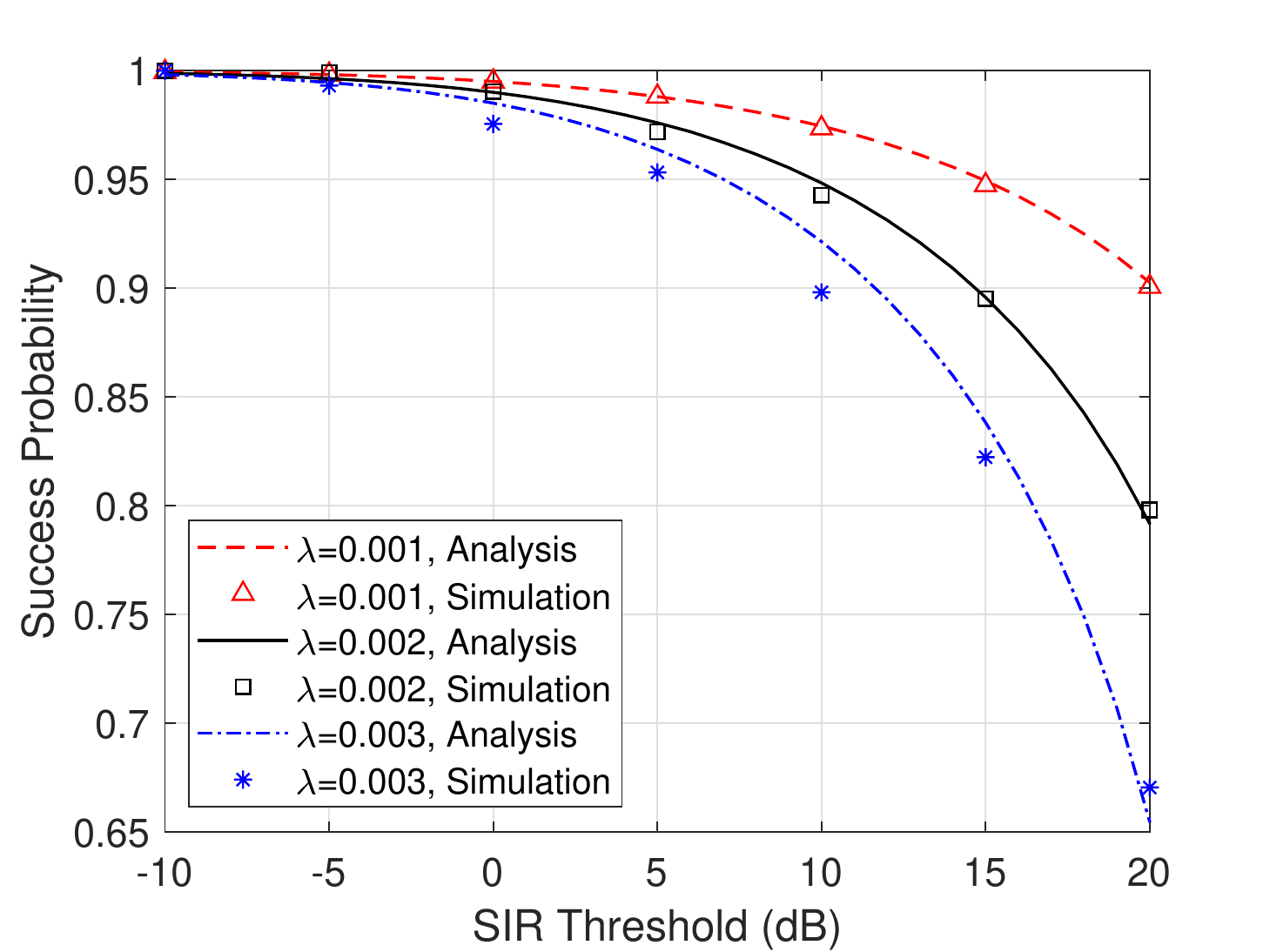} 
 \caption{Success probability in Poisson bipolar networks under different intensities.   } \label{fig:SP_queue_density} 
 \end{figure}

Fig.~\ref{fig:SP_queue_density} further illustrates the success probability in a Poisson bipolar network under different densities. It can be observed that the approximation tends to lose its accuracy with the increase of the network density. The reason is that, in a Poisson bipolar network, the interferers can be arbitrarily close to the target receiver, and thus the queues of the serving transmitters and the interferers are strongly coupled. The spatial-temporal correlation of the buffer status cannot be ignored in this regime.

\subsection{Summary and Discussion}

This section has developed models to analyze wireless systems with unsaturated buffers. Specifically, 
by integrating results from queueing theory, we have presented the derivations of success probabilities for Poisson bipolar and Poisson downlink systems given the packet arrival rate. The key 
lessons learned are as follows.

\begin{itemize}

\item The spatial correlation among the queue statuses of different transmitters and the temporal correlation among the  queue status of the same transmitter can be ignored when different queues are weakly coupled, e.g., in Poisson downlink networks, and cannot be ignored when the queues are strongly coupled, e.g., in Poisson bipolar networks. 


\item The temporal correlation among the successful transmission events decreases with increasing packet arrival rate. 

\end{itemize}

\vspace{0.2cm}
\noindent 
{\em Open Technical Issues}: In Section \ref{sec:queueing}, we have shown that in large Poisson downlink networks, the interaction among the queues becomes weak, and the impact of the temporal and spatial correlation tends to be negligible. However, this effect only appears due to the cellular infrastructure, where the interfering BSs are further away from the typical user than the serving BS. In infrastructure-less networks that do not impose any restrictions on the interferers' locations, the interaction between the queues tends to be strong, thus, their spatial-temporal correlation cannot be ignored. 
Thus, further research efforts are needed to characterize the dynamics of strongly coupled queues in large-scale networks. 
  
\section{Spatially-Correlated Interference and Relaying} \label{sec:relaying}


Multihop relaying is an effective technique to extend the communication range with limited transmit power at each hop,  enhancing the reliability and throughput of point-to-point communication~\cite{R.July2013Tanbourgi}. In multihop relaying, the transmission performance at different hops can be impacted by some common interferers. Therefore, the spatial characteristics of the interferers can 
have a substantial impact on
the end-to-end transmission performance \cite{A.Oct.2015Crismani}. The goal of this section is to quantify the spatially and temporally-correlated interference and demonstrate its effect on the performance of multihop relaying.

\subsection{System Model}

We consider a multihop relay network consisting of a source node and a receiver node, which is $M$ hops away from the transmitter node, in a random field of interferers. The locations of the source node $\mathrm{S}$ and the $m$-th hop receiver (i.e., $(m+1)$-th hop transmitter) are deterministic and  denoted by $x_{\mathrm{S}}$ and $z_{m}$, respectively. The transceiver node in each hop works in a half-duplex fashion.  
Let $d_{m}$ represent the Euclidean distance of the $m$-th hop link, i.e., $d_{1}=\|z_{1}-x_{\mathrm{S}}\|$ and $d_{m}=\|z_{m}-z_{m-1}\|$, for $ 2 \leq m \leq M$. The route from the source node to the $M$-th hop receiver is fixed. The decode-and-forward (DF) relaying protocol is adopted such that each hop first decodes the received signal and forwards the re-encoded version to the next hop. We consider a Poisson field of interferers where the locations of the interferers $\Phi \subset \mathbb{R}^2$ are a PPP with intensity $\lambda$, as illustrated in Fig.~\ref{fig:Cox_system}(a). 
Note that the relaying system of interest is considered to be independent of the interferers.

We assume that each transmitter in the considered system uses unit transmit power. Let $\Phi_{m}=\{x_{j,m}\}_{j \in \mathbb{N}}$ denote the point process of interferers during the $m$-th hop transmission. 
The SIR at the $m$-th hop can be expressed as 
\begin{align}
& \eta_{m} = \frac{ h_{m}   d_{m}  ^{- \alpha} }{  \sum_{ j \in \mathbb{N} } h_{j,m} \| x_{j,m} - z_{m} \|^{-\alpha  } }, 
\end{align}
where $h_{m}$ represents the power gain of small-scale fading for the $m$-th transmission hop and $h_{j,m}$ represents the small-scale fading gain between the interferer $j$ and the $m$-th hop receiver, which are both i.i.d. exponential random variables with unit mean.
   
For the analysis of the success probability of multihop relaying, we consider both QSI and FVI with which the transmission of different hops are subject to the interference from the same point process, i.e., $\Phi_{1}=\Phi_{2}=\cdots=\Phi_{M}=\Phi$, and independent point processes, respectively.    
 With DF relaying protocol~\cite{O.2003Hasna}, the end-to-end success probability of an $M$-hop relaying system with QSI and FVI are given, respectively, by
 \begin{align} \label{def:SP_relaying}
  \mathcal{P}^{\textup{QSI}}_{M} & =  \mathbb{E}  \Bigg[  \mathbb{P} \bigg[ \bigcap^{M}_{m=1} \big \{ \eta_{m} >  \theta \big \} \bigm| \Phi \bigg] \Bigg]  \nonumber \\
  & =  \mathbb{E}  \Bigg[ \prod^{M}_{m=1}  \mathbb{P} \big[    \eta_{m} >  \theta    \mid \Phi \big] \Bigg] 
  \end{align}
  and
   \begin{align}
 \mathcal{P}^{\textup{FVI}}_{M}  & = \mathbb{P} \Bigg[ \bigcap^{M}_{m=1} \big \{ \eta_{m} >  \theta \big \}  \Bigg]  \nonumber \\
 & = \prod^{M}_{m=1} \mathbb{P} \big[   \eta_{m} >  \theta   \big]  . 
 \end{align}

  \begin{figure}
 \centering
 \includegraphics[width=0.4\textwidth]{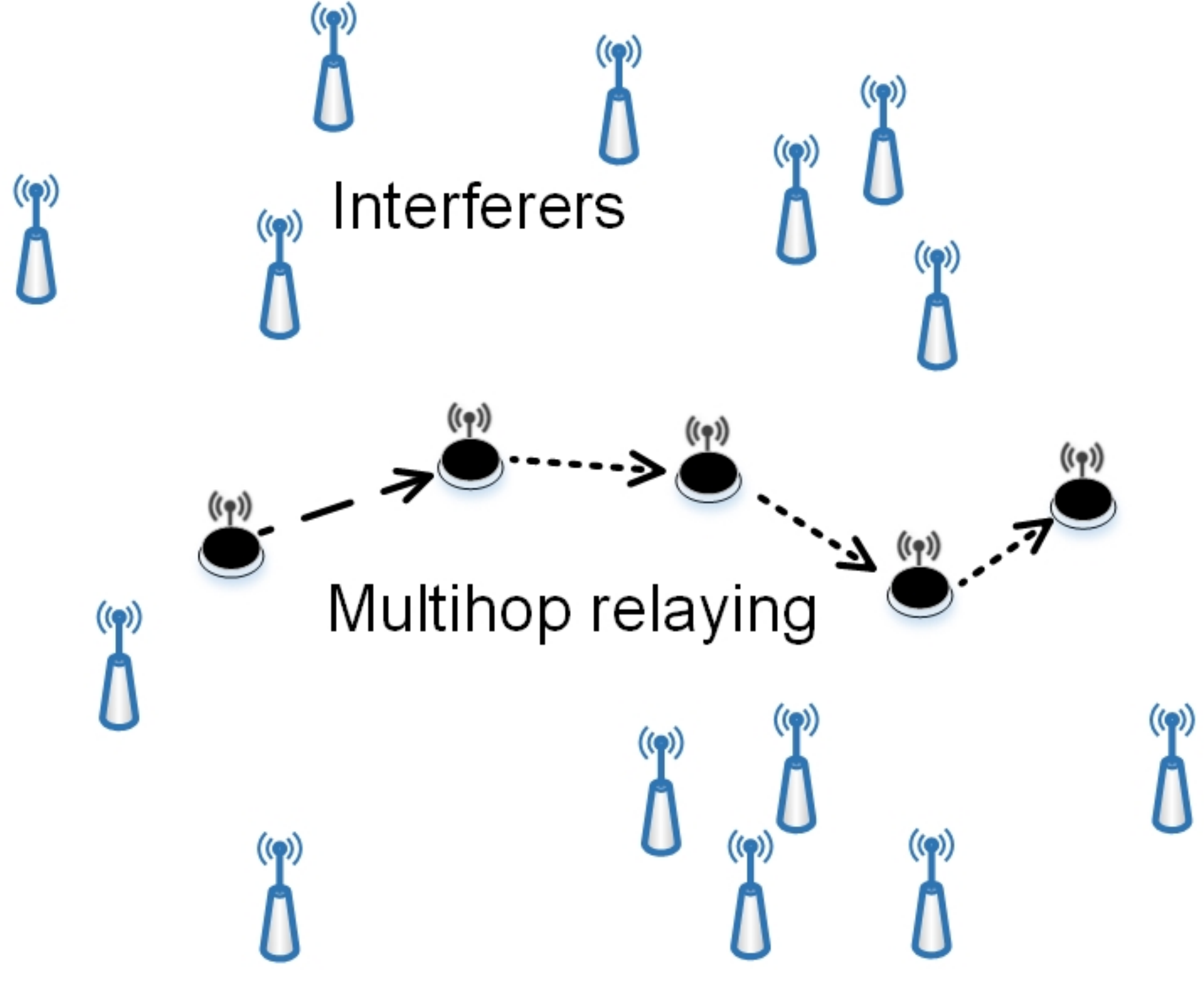} 
 \caption{Multihop relaying in the presence of random interferers \cite{B2}.  } \label{fig:Cox_system} 
 \end{figure}

 \subsection{Moments of the End-to-End CSP$_\Phi$} 

This subsection characterizes the moments of the end-to-end CSP$_\Phi$ of the multihop relaying. 
 
 \begin{itemize}
 
 \item We compute the end-to-end JSP that the transmissions of the $M$ hops all succeed given the point process, i.e., $\mathbb{P} \big[  \bigcap^{M}_{m=1} \{ \eta_{m} > \theta \} \mid \Phi_{m} \big]$ (e.g., as in \cite{J.2013Lee}).
 
 \item  We derive the moments of the CSP$_\Phi$ based on the PGFL of the PPP. 
  
 \end{itemize}

Following the above methodology, we obtain the moments of the end-to-end CSP$_\Phi$ under both QSI and FVI in the following theorem.  

\begin{theorem} \label{thm:M_ESP}
The moments of the end-to-end CSP for an $M$-hop relaying system in a Poisson field 
of interferers are given by 
(\ref{eqn:relay_PPP}).  
\begin{figure*} 
\begin{align} \label{eqn:relay_PPP}
&\mathcal{M}^{\textup{Poi}}_{P_{\rm s}} (b,M)  \!=\! \begin{dcases} \exp \Bigg( \!\! -  \lambda \! \int_{\mathbb{R}^2}  \bigg(  1  -  	\prod^{M}_{m=1} \bigg( \frac{1}{  1  + \theta d^{\alpha}_{m}  \|x-z_{m}\|^{-\alpha}   } \! \bigg)^{\! b} \bigg)   \mathrm{d}x \Bigg), \quad \textup{QSI} \\
 \exp \Bigg( \! -  \lambda \prod^{M}_{m=1}   \int_{\mathbb{R}^2}    \! \! \bigg(  1  -  \bigg( \frac{1}{   1 \! +\! \theta d^{\alpha}_{m} \| x  - z_{m} \|^{-\alpha}  } \! \bigg)^{\!b} \bigg)    \mathrm{d}x \Bigg), \quad \,\, \textup{FVI} 
\end{dcases}   
\end{align} 
\hrulefill
\end{figure*}

\end{theorem}
 
 \begin{proof}
See \textbf{Appendix G}. 
 \end{proof}

Fig.~\ref{fig:CP_SP_relay} illustrates the success probability for an $M$-hop linear-route multihop relaying system where all relays are placed on the source-destination line and all $M$ links have distance $l$. 
With both fields of interferers, we can observe that correlated QSI provides a higher success probability than independent FVI. 
This can be understood from the perspective of the CSP of $m$-th hop given that the transmissions of the previous hops all succeed. 
With FVI, we can see from  (\ref{eqn:relay_PPP}) 
that   $\mathcal{M}^{\textup{Poi}}_{P_{\rm s}}(b,M)/\mathcal{M}^{\textup{Poi}}_{P_{\rm s}}(b,M-1)=\mathcal{M}^{\textup{Poi}}_{P_{\rm s}}(b,1)$.  
This is consistent with the fact that with FVI, the successful transmission event of the $m$-th hop is independent of those of the previous hops. 
  
 \begin{figure}
  \centering
  \includegraphics[width=0.5\textwidth]{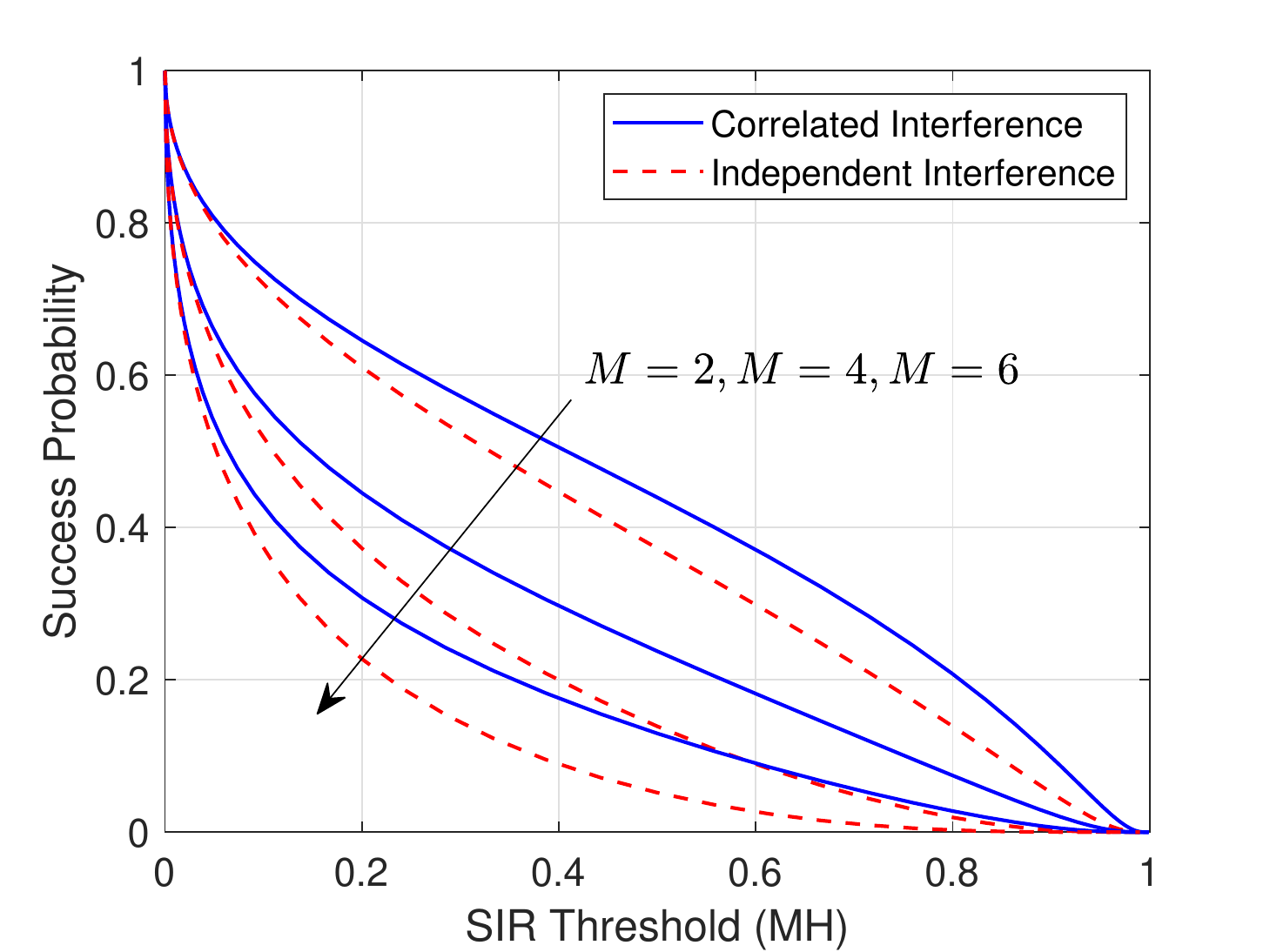} 
  \caption{Success probability of a multihop relaying system in a Poisson field of interferers ($\alpha=4$, $\lambda=0.1$, $d_{1}=d_{2}=\cdots=d_{M}=1$). } \label{fig:CP_SP_relay} 
  \end{figure}

  
Fig.~\ref{fig:CCP} demonstrates the spatial CSP of the \mbox{$m$-th} hop given that the transmissions in the previous \mbox{$m-1$} hops all succeed in the scenario with QSI, i.e., $\big(\mathcal{M}^{\text{Poi}}_{P_{\rm s}}(1,M)/\mathcal{M}^{\text{Poi}}_{P_{\rm s}}(1,M-1)\big)$ for $M \geq 2$. It can be seen that, with QSI, the CSP  increases considerably given the successful transmission of the first hop, especially when the SIR threshold $\theta$ is high. Thus, the end-to-end success probability with QSI exceeds that with FVI. It is worth noting that with temporally correlated  QSI, the conditional outage probability, given the outage events of the previous hops, also increases. However, this does not worsen the end-to-end success probability, since the outage event of any single hop leads to an end-to-end outage.

Moreover, Fig.~\ref{fig:CCP_relay} depicts the temporal CSP of an $M$-hop relaying system given a previous end-to-end successful transmission, i.e., 
$\mathcal{M}^{\text{Poi}}_{P_{\rm s}}(2,M)/\mathcal{M}^{\text{Poi}}_{P_{\rm s}}(1,M)$.
We can observe that the CSP increases with the number of hops. The reason is that a successful transmission with a larger number of hops indicates a good channel condition, and thus the next end-to-end transmission is more likely to succeed.

 \begin{figure*}[htp]  
      \centering
        \begin{minipage}[c]{0.48 \textwidth}
         \includegraphics[width=0.98\textwidth]{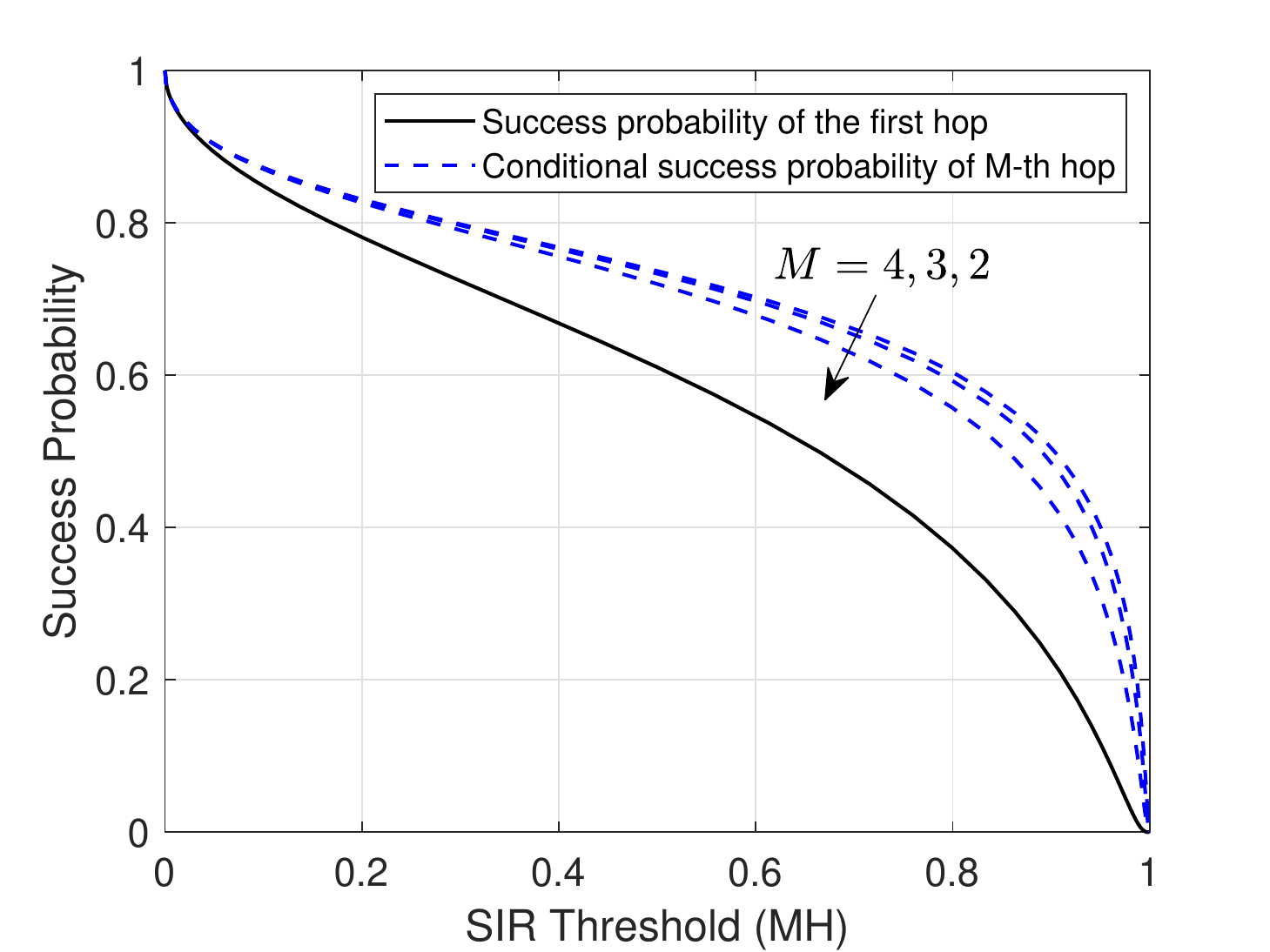} 
         \caption{ The spatial CSP of the $M$-th hop relaying in a Poisson field of interferers with QSI.
         } \label{fig:CCP} 
        \end{minipage} \hspace{3mm}
        \begin{minipage}[c]{0.48 \textwidth}\vspace{0mm}
          \includegraphics[width=0.98\textwidth]{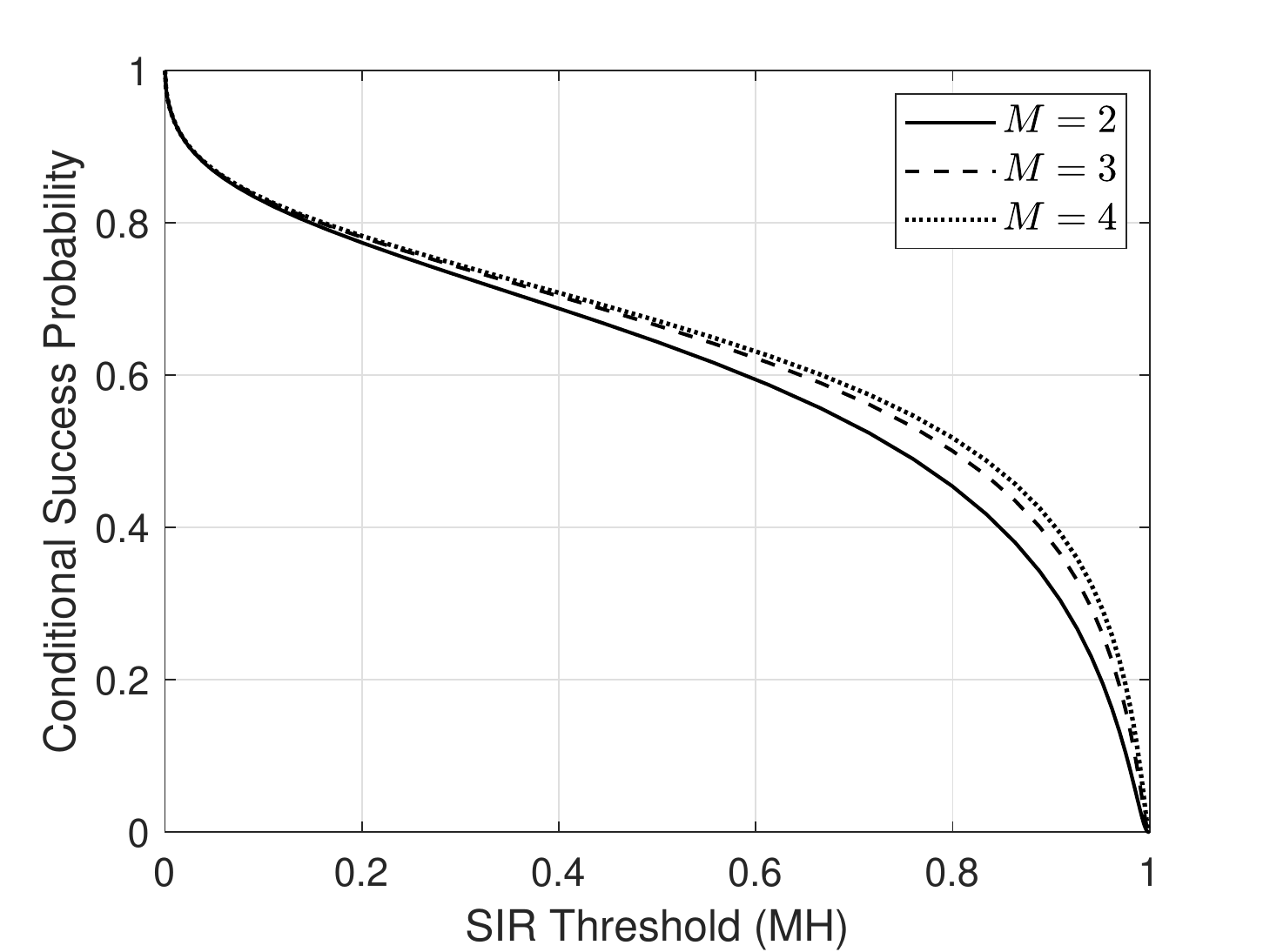}  
          \caption{The temporal CSP of an $M$-hop relaying system in a Poisson field of interferers with QSI.  } \label{fig:CCP_relay} 
          \end{minipage}  
       \end{figure*}

\subsection{Summary and Discussion}

In this section, we have discussed the impact of spatial-temporal interference on the performance of multihop relaying. In particular, we have presented the derivations of end-to-end success probability of multihop relaying in Poisson fields of interferers under both FVI and QSI. The main 
lessons learned are as follows.

\begin{itemize}

\item Correlated spatial-temporal interference (i.e., QSI) imposed by the same interferers at different hops results in higher success probability than independent spatial-temporal interference (i.e., FVI).

\item With correlated spatial-temporal interference, as $m$ increases, the successful transmission of the $m$-th hop is more dependent on the successful transmissions of the previous hops.

\item  With independent spatial-temporal interference, the CSP of the $m$-th hop given the successful transmission of the previous hops is the same as the success probability of the $m$-th hop without the condition.

\item With correlated spatial-temporal interference, the  end-to-end successful transmission events are more correlated when the number of hops is larger.

\end{itemize} 

 \vspace{0.2cm}
 \noindent
{\em Open Technical Issues}: This section has presented the characterization of end-to-end relaying performance under QSI and FVI, which considers a purely static network and an independent network, respectively. The two types of interference represent two extreme cases with fully correlated and independent interferer locations.  
In general, the ambient interferers (e.g., mobile users) may have a certain degree of movement during different transmission attempts and result in neither a static nor independent network environment. In such an environment, the resulting temporal interference is only partially correlated. An accurate characterization of the spatial-temporal correlation at different locations would be required for performance evaluation of multihop relaying systems. 

\section{Temporally-Correlated Interference and Retransmission}  \label{sec:retransmission}

The goal of this section is to derive the success probabilities of a target transmission link with retransmissions under temporally correlated and independent interference. 

\subsection{System Model}
 
This section considers the Poisson ad hoc network model (as introduced in Section~\ref{sec:sectionIII_SM}). Let \mbox{$\Phi^{(k)}=\{x^{(k)}_{j}\}_{j \in \mathbb{N}}$} denote the node locations in time slot $k$. 
The aggregated interference and receive SIR at the target receiver located at $o$ in time slot $k$ can be expressed, respectively, as
\begin{align}  
& I^{(k)}_{o} =  \sum_{ j \in \mathbb{N} } h^{(k)}_{j}  \|x^{(k)}_{j}\| ^{-\alpha  } \nonumber \\
& \eta^{(k)} = \frac{ h^{(k)}_{\rm t}   \|x^{(k)}_{\rm t}\|^{- \alpha} }{  I^{(k)}_{o}  }, 
\end{align}  
where $h^{(k)}_{\rm t}$ and $h^{(k)}_{j}$ denotes the  small-scale fading gains between the 
transmitters at $x^{(k)}_{\rm t}$, $x^{(k)}_{j}$ and the target receiver in time slot $k$ which are i.i.d. exponential random variables with unit mean.

We consider both QSI and FVI with which different transmission attempts are influenced by the same set and different sets of interferers, respectively. As a result, the interferences, and thus the successful transmission events, in the considered system are temporally correlated with and independent of QSI and FVI, respectively.

\subsection{Performance Analysis}

Let $A_{k} \triangleq  \{  \eta^{(k)} > \theta \}$ denote the successful transmission event in time slot $k$. 
The JSP of $K$ transmissions in the cases with QSI and FVI are defined, respectively, as
 \begin{align}  
 \mathcal{J}^{\textup{QSI}}_{K} & \triangleq \mathbb{E}  \Bigg[ \mathbb{P} \bigg[ \bigcap^{K}_{k=1}   A_{k} \bigm|   \Phi^{(k)}  \bigg]    \Bigg] \nonumber \\ & =  \mathbb{E}  \Bigg[ \prod^{K}_{k=1}  \mathbb{P} \big [   A_{k} \mid  \Phi^{(k)}  \big ]    \Bigg] \label{eqn:JSP_QSI}  
 \end{align}
 and
 \begin{align}
 \mathcal{J}^{\textup{FVI}}_{K}  & \triangleq   \mathbb{P} \bigg[ \bigcap^{K}_{k=1} A_{k}       \bigg]    = \prod_{k=1}^{K}    \mathbb{P} \big[     A_{k}    \big]   .   
\label{eqn:JSP_FVI} 
\end{align}

\subsubsection{Joint Success Probability}

In this subsection, we show how to derive the JSP of $K$ transmissions based on the PGFL of the PPP.  
  The final results are presented in the following theorem \cite{M.2013Haenggia}.

\begin{theorem} \label{thm:Retx_JSP} ({\bf Temporal JSP}) 
With a Poisson field of interferers, the probability that a link over distance $r_{\rm t}$ has $K$ successful transmissions in a row with QSI and FVI are given, respectively, as
\begin{align}
& \mathcal{J}^{\textup{QSI}}_{K} = \exp \big(-c \lambda  \theta^{\delta} r^{2}_{\rm t} D_{K}(\delta) \big), \label{eqn:J_QSI} \\
& \mathcal{J}^{\textup{FVI}}_{K}  =  \exp \big(-c \lambda   \theta^{\delta} r^{2}_{\rm t} K  \big), \label{eqn:J_FVI}
\end{align}
where $c= \pi   \Gamma(1+\delta) \Gamma(1-\delta) $ and $D_{K}(\delta)= \frac{\Gamma(K+\delta)}{\Gamma(K)\Gamma(1+\delta)}$. 
\end{theorem}

\begin{proof} 
See \textbf{Appendix H}.
\end{proof}


\noindent
{\bf Remark 11}: $\mathcal{J}^{\textup{QSI}}_{1}= \mathcal{J}^{\textup{FVI}}_{1}$ as $D_{1}(\delta)=1$, which indicates that the success probability of any single time slot is not influenced by the type of interference experienced.

\noindent
{\bf Remark 12}: Since $\delta < 1$ (i.e., $\alpha>2$), it is readily checked that $D_{K}(\delta) < K$ for $K \in \mathbb{N}$, and thus $\mathcal{J}^{\textup{QSI}}_{K} > \mathcal{J}^{\textup{FVI}}_{K}$. This reveals that temporally-correlated QSI results in higher JSP than the temporally-independent FVI.  Fig.~\ref{Retx_JSP} illustrates the JSP for $K=2,3,4$ with both QSI and FVI.

\noindent
{\bf Remark 13}: When $\delta \to 0$ (i.e., $\alpha \to \infty$), it can be found that  $\mathcal{J}^{\textup{QSI}}_{1} = \mathcal{J}^{\textup{QSI}}_{2} = \cdots = \mathcal{J}^{\textup{QSI}}_{K} = e^{-\lambda \pi r^2_{\rm t}}$, equivalently $\mathbb{P}[A_{k+1}|A_{k } ]=1$ for $k=1,\ldots,K$. This indicates that despite the effect of small-scale fading, the successful transmission events are fully correlated. 
  
\noindent
{\bf Remark 14}: When $\delta \to 1$ (i.e., $\alpha \to 2$), it can be found that a) $  \mathcal{J}^{\textup{QSI}}_{K}  \to e^{- c   \lambda  \theta^{\delta} r^2_{\rm t} K }$, equivalently $\mathbb{P}[A_{k+1}|A_{k } ]=e^{- c   \lambda \theta^{\delta} r^2_{\rm t} }$ for $k=1,\ldots,K$. This means the successful transmission events are independent; b) $c \to \infty$, which indicates that $\mathcal{J}^{\textup{QSI}}_{K} \to 0$.

 \subsubsection{Conditional Success Probability}
 As a consequence of Theorem~\ref{thm:Retx_JSP}, the CSPs of succeeding at $K+1$-th transmission given the previous $K$ successful transmissions can be directly obtained following Bayes rule as                       
\begin{align} 
\mathcal{C}^{\textup{QSI}}_{K+1,K} & = \mathbb{P} \big[ A_{K+1} \mid A_{1},\ldots,A_{K} \big] \nonumber \\
&  = \frac{ \mathcal{J}^{\textup{QSI}}_{K+1}   }{ \mathcal{J}^{\textup{QSI}}_{K } } \nonumber\\
& =  \exp \Big( -  c  \lambda \theta^{\delta} r^2_{\rm t} \big( D_{K+1}(\delta) -  D_{K }(\delta)  \big) \! \Big),       \nonumber                                     
\end{align}
 and
 \begin{align}
   \mathcal{C}^{\textup{FVI}}_{K+1,K} =  \frac{ \mathcal{J}^{\textup{FVI}}_{K+1}   }{ \mathcal{J}^{\textup{FVI}}_{K } } = \exp \big( - c \lambda \theta^{\delta} r^2_{\rm t}   \big).  \nonumber
  \end{align}


 \noindent
 {\bf Remark 15}: It can be readily checked that   $\partial \mathcal{C}^{ \textup{QSI} }_{K+1,K} / \partial  K >0 $ which indicates that   $\mathcal{C}^{\textup{QSI}}_{K+1,K}$    monotonically increases with the number of transmission attempts $K$.  
                                                 
                                                                                             
Fig. ~\ref{Retx_CSP} shows the temporal CSPs with QSI when $K=1,2,3,4$. The numerical results illustrate the above-discussed properties. Given that the previous transmissions succeed, the CSP considerably increases, especially with a high SIR threshold $\theta$.

\subsubsection{Correlation Coefficient}

Let $Q_{k}=\mathbbm1_{A_{k}}$ be the indicator that event $A_{k}$ happens. The correlation coefficient between $Q_{k}$ and $Q_{j}$ with $k \neq j$ is
\begin{align}
\zeta^{\textup{QSI}}_{Q_{k}, Q_{j} } & = \frac{ \mathbb{P} [ A_{k} \cap A_{j} ] - \mathbb{P}^2[A_{k}]    }{  \mathbb{P} [A_{k} \cap A_{k}]   - \mathbb{P}^2[A_{k}]    }  \nonumber \\
& \overset{(a)}{=} \frac{  \mathcal{J}^{\textup{QSI}}_{2}  - \big( \mathcal{J}^{\textup{QSI}}_{1} \big)^2  }{   \mathcal{J}^{\textup{QSI}}_{1}(1- \mathcal{J}^{\textup{QSI}}_{1})  }   \nonumber \\
& = \frac{  \exp  ( -  c   \lambda  \theta^{\delta} r^{2}_{\rm t} ( \delta +1 )     ) - \exp ( - 2   c  \lambda  r^2_{\rm t} \theta^{\delta}   ) }{ \exp ( -     c  \lambda  r^2_{\rm t} \theta^{\delta}   ) (1 - \exp ( -     c  \lambda  r^2_{\rm t} \theta^{\delta}   ))}
\nonumber \\
& = \frac{ \exp  (   c \lambda  \theta^{\delta} r^{2}_{\rm t} (1- \delta   )     ) - 1  }{ \exp (   c   \lambda  \theta^{\delta}   r^2_{\rm t} )  - 1  }, \label{eqn:zeta_QSI}
\end{align}
where $(a)$ holds as $\{Q_{k}\}_{k\in\{1,\ldots, K \}}$ are identically distributed. Moreover, $\zeta^{\textup{FVI}}_{Q_{k}, Q_{j} }=0$ as $J^{\textup{FVI}}_{2}= \big(J^{\textup{FVI}}_{1}\big)^2$. 
Since $\zeta^{\textup{QSI}}_{Q_{k}, Q_{j} }>0$ and  $\zeta^{\textup{FVI}}_{Q_{k}, Q_{j} }=0$, it is evident that QSI and FVI result in temporally-correlated and independent successful transmission events, respectively. 


\noindent
{\bf Remark 16}: When $\delta \to 0$ and $\delta \to 1$,  $\zeta^{\textup{QSI}}_{Q_{k}, Q_{j} } \to 1$ and $\zeta^{\textup{QSI}}_{Q_{k}, Q_{j} } \to 0$, indicating that the $Q_{k}$ and $Q_{j}$ become fully correlated and fully uncorrelated, respectively. 

\noindent
{\bf Remark 17}:   For given $\delta$ and $r$, $\zeta^{\textup{QSI}}_{Q_{k}, Q_{j} }$ is a decreasing function of $\lambda$ and $\theta$. This can be checked that since $0<\delta < 1$ (i.e., $\alpha > 2$), the denominator scales up at a higher rate than the numerator of (\ref{eqn:zeta_QSI}) with the increase of $\lambda$ or $\theta$.
Fig.~\ref{fig:Retx_CC} shows the correlation coefficient of successful transmission events with QSI. 

  \begin{figure*}[htp]  
  \centering
    \begin{minipage}[c]{0.48 \textwidth}
     \includegraphics[width=0.98\textwidth]{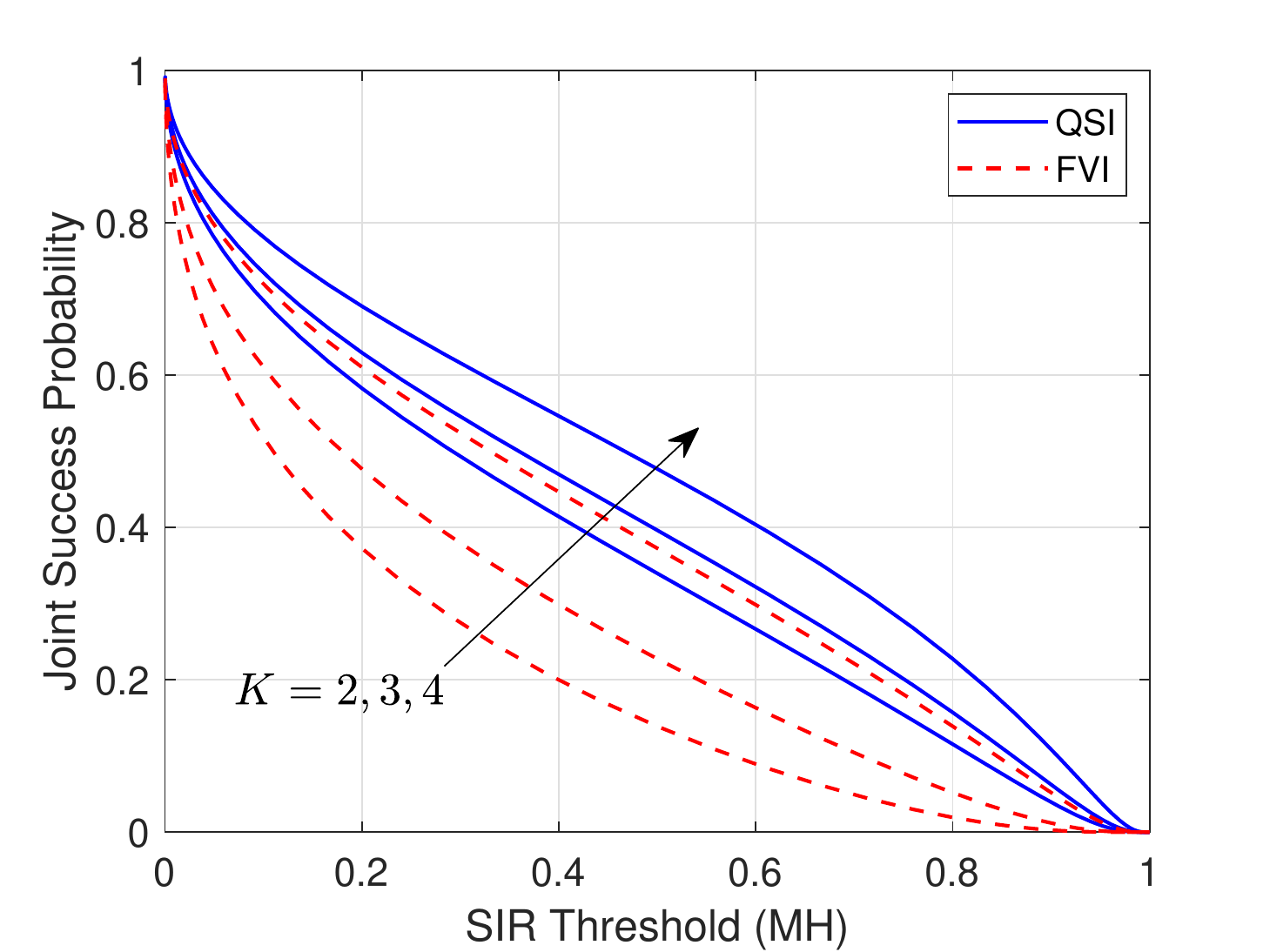} 
     \caption{JSP with a different number of transmission attempts in Poisson ad hoc networks ($\alpha=4$, $\lambda=0.1$
     ). 
     } \label{Retx_JSP}
    \end{minipage} 
    \hspace{3mm}
    \begin{minipage}[c]{0.48 \textwidth}\vspace{0mm}
      \includegraphics[width=0.98\textwidth]{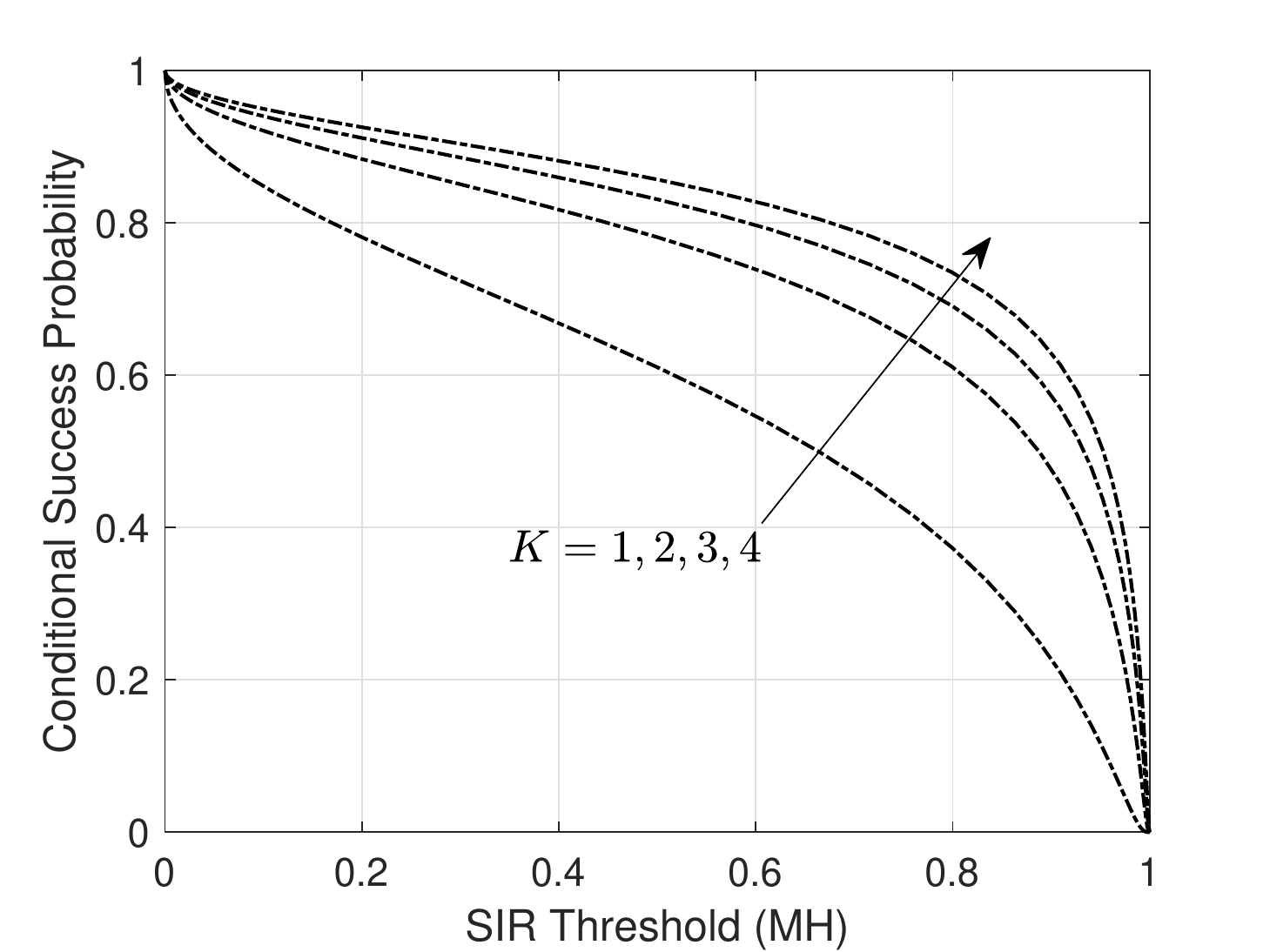}  
      \caption{CSP with a different number of transmission attempts in Poisson ad hoc networks ($\alpha=4$, $\lambda=0.1$
      ). } \label{Retx_CSP}
      \end{minipage}  
   \end{figure*}  

  \begin{figure*}[htp]  
  \centering
    \begin{minipage}[c]{0.48 \textwidth}
     \includegraphics[width=0.98\textwidth]{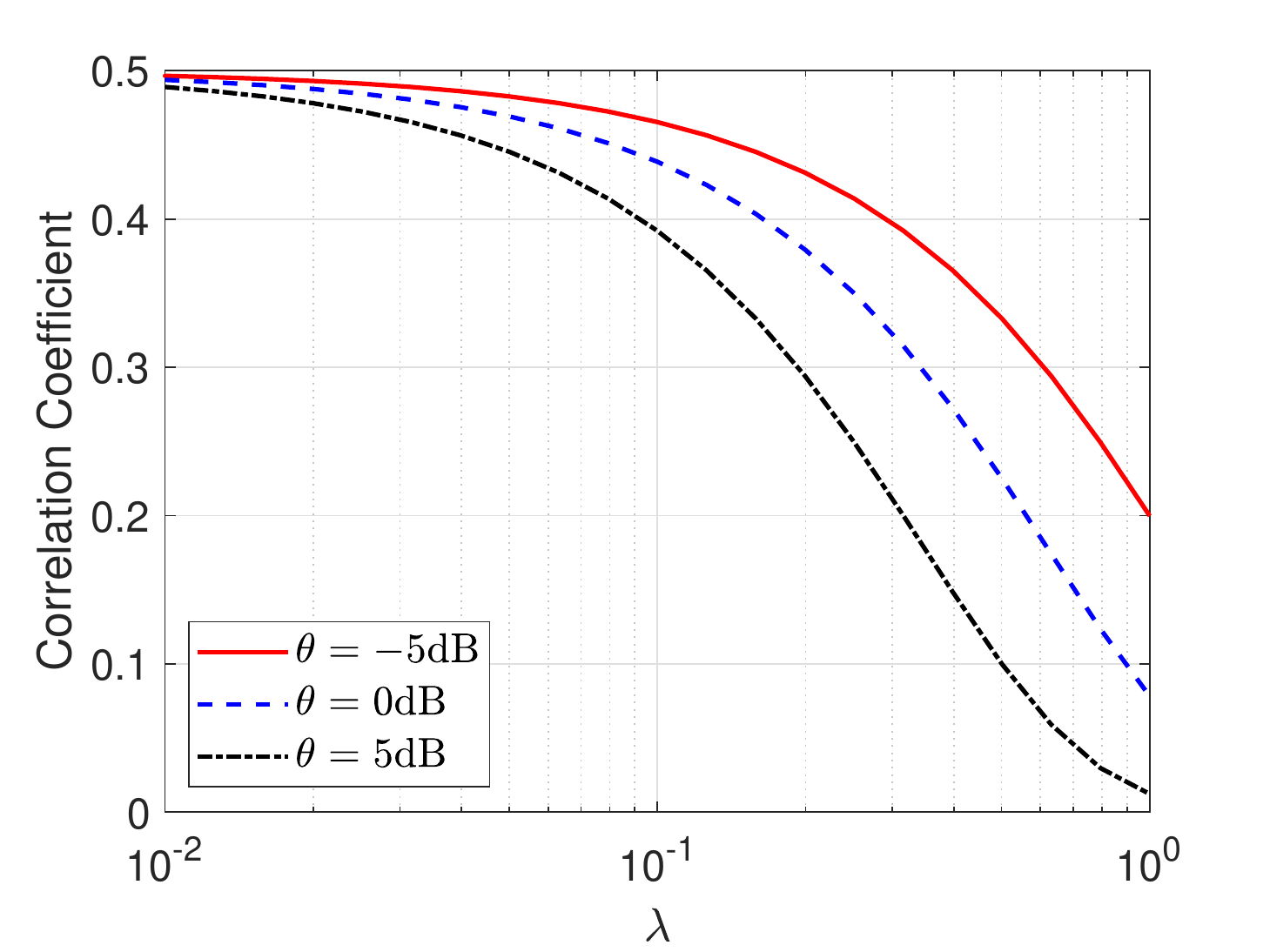}   
     \caption{Correlation coefficient in Poisson ad hoc networks ($\alpha=4$).
     } \label{fig:Retx_CC}  \vspace{2mm}
    \end{minipage}  \hspace{3mm} 
    \begin{minipage}[c]{0.48 \textwidth}\vspace{0mm}
      \includegraphics[width=0.95\textwidth]{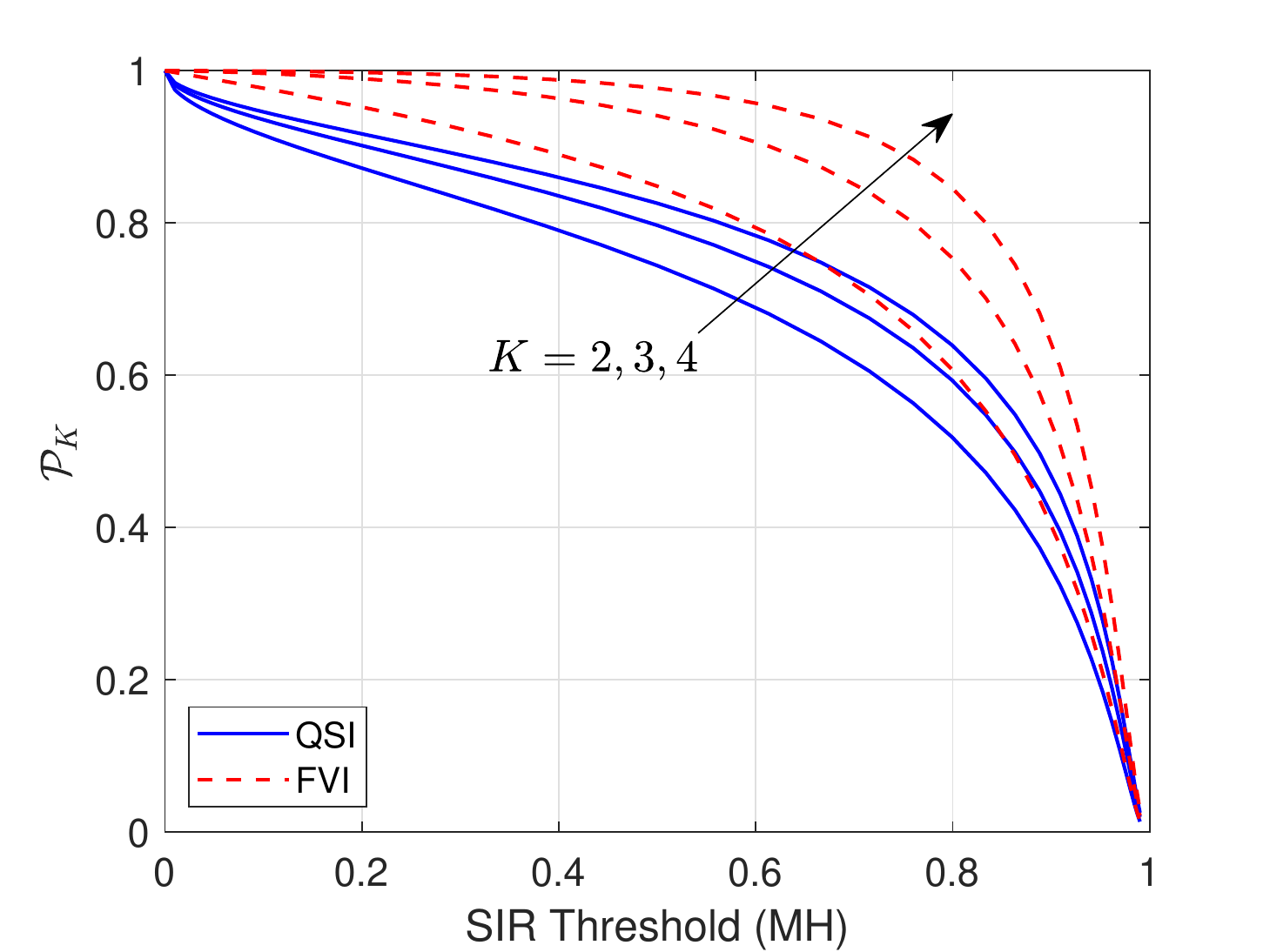} 
      \caption{Success probability with retransmission in Poisson ad hoc networks ($\alpha=4$, $\lambda=0.1$). } \label{fig:Retx_SP}
      \end{minipage}  
   \end{figure*}

 \subsubsection{Success Probability with Retransmissions}
                                                            
In the case of transmission failure, multiple transmissions can be carried out to deliver a message. 
Let us consider a retransmission protocol where the receiver requests the associated transmitter to send the message again upon a transmission failure until reaching a maximum number of transmission attempts denoted as $K$. At the receiver side, the received signal at each time slot is decoded independently.  
The success probability with retransmissions can be expressed as
 \begin{align}
  \mathcal{P}_{K} \triangleq  \mathbb{P} \bigg[ \bigcup^{K}_{k=1} A_{k} \bigg].   \end{align}
 Note that from the above definition, we have $\mathcal{P}_{1}=\mathcal{J}_{1}$, where $\mathcal{J}_{k}$, $k \in \mathbb{N}$, is defined in (\ref{def:JCP}).
                                                                                             
 Let $\mathcal{A} \triangleq \{ A_{1}, \ldots, A_{K}  \}$ denote the set of the successful transmission events and $\mathrm{P}(\mathcal{A})$ represent the power set of $\mathcal{A}$. 
 By applying the inclusion-exclusion principle, we obtain the success probability with retransmissions as
                                                                                                   \begin{align}  \label{eqn:PK}
                                                                                                  \mathcal{P}_{K} & = \sum_{A \in \mathrm{P} (\mathcal{A})   } (-1)^{|A|+1} \mathbb{P}  \big[ A \big]  \nonumber \\ 
                                                                                                   & \overset{(a)}{=} \sum^{K}_{k=1} (-1)^{k+1}   \binom{K}{k} \mathcal{J}_{k},
                                                                                                  \end{align}
where $(a)$ follows as $\{\eta_{k}\}_{ k \in \{1,2,\ldots,K\}}$ are identically distributed, and $\mathcal{J}_{k}$ has been obtained in Theorem~\ref{thm:Retx_JSP}.
                                                                                                     
 Fig. \ref{fig:Retx_SP} shows the success probability with retransmissions with both QSI and FVI. It can be seen that $\mathcal{P}^{\textup{FVI}}_{K}$ exceeds $\mathcal{P}^{\textup{QSI}}_{K}$. The reason is that the transmission failures are temporally correlated and thus are likely to occur in succession with QSI. By contrast, with FVI, transmissions have a better chance to succeed as previous transmission failure does not infer a lower success probability in the current time slot.  This can be understood by checking the dependency of transmission failure. Let $\bar{S}\triangleq \{ \eta^{(k)} < \theta \}$. By following the inclusion-exclusion principle, the joint outage probability of two transmission events is given by 
      \begin{align}
     \bar{\mathcal{J}}_{2} & = \mathbb{P} \big[ \bar{A}_{1} \cap \bar{A}_{2} \big] \nonumber \\
     &= 1 - \mathbb{P} \big[ A_{1} \cup A_{2} \big]  
     \nonumber \\
     &= 1 -  (2 \mathcal{J}_{1} - \mathcal{J}_{2}) \nonumber\\
     & = 1 - 2\exp(-c \lambda \theta^{\delta} r^2_{\rm t} )  + \exp(-c \lambda \theta^{\delta} r^2_{\rm t} (1+\delta) ).  \nonumber
      \end{align}
      Subsequently, 
      \begin{align}
      \frac{\bar{\mathcal{J}}^{\textup{QSI}}_{2}}{\big(\bar{\mathcal{J}}^{\textup{QSI}}_{1}\big)^2} = 1 + \frac{ \exp( c \lambda \theta^{\delta} r^2_{\rm t} (1-\delta) ) -1 }{ (\exp( c \lambda \theta^{\delta} r^2_{\rm t} ) -1 )^2 } >1, \nonumber
      \end{align}
      and $\frac{\bar{\mathcal{J}}^{\textup{FVI}}_{2}}{\big(\bar{\mathcal{J}}^{\textup{FVI}}_{1}\big)^2} = 1$.
     
     Moreover, it is evident that the conditional outage probability increases given the occurrence of previous transmission failure by checking 
     \begin{align}
     \mathbb{P} \big[ \bar{A}_{1} \mid \bar{A}_{2} \big] &  = \frac{\mathbb{P} \big[ \bar{A}_{1} \cap \bar{A}_{2} \big]}{ \mathbb{P} \big[ \bar{A}_{1} \big]} \nonumber \\
     & = 1 - \exp( - c \lambda \theta^{\delta} r^2_{\rm t}) \frac{1 - \exp( - \delta c \lambda \theta^{\delta} r^2_{\rm t}) }{1 - \exp( -   c \lambda \theta^{\delta} r^2_{\rm t})} \nonumber \\
     & \geq 1 - \exp( - c \lambda \theta^{\delta} r^2_{\rm t}) = \mathbb{P} \big[ \bar{A}_{1} \big]. \nonumber
     \end{align}

\subsection{HARQ Retransmission Schemes} \label{sec:HARQ}

For transmission and decoding, similar to  \cite{G.2015Nigam}, we consider two categories of HARQ, i.e., Type-I HARQ and Type-II HARQ schemes.

\begin{itemize} 

\item  If the SIR at a target receiver for the initial transmission does not exceed the threshold $\theta$, a one-time retransmission request is sent to the serving transmitter. After receiving the retransmitted signals, the target receiver with Type-I HARQ abandons the received signal from the initial transmission and decodes only from the received signal from the retransmission. Given a maximum number of transmissions $K$ and a SIR threshold $\theta$, the success probability 
is given by
\begin{align} 
\mathcal{P}_{\mathrm{I}} = \mathbb{P} \bigg[   \bigcup^{ K }_{k =1} \Big \{  \eta^{(k)} > \theta  \Big \}  \bigg],  \label{eqn:Type_I}
\end{align} 
where  $\eta^{(k)} $ is the SIR  of $k$-th transmission of the same content with Type-I HARQ retransmission scheme.

\item  With Type-II HARQ with chase combing (CC) codes and maximal ratio combining (MRC) of signals from both the initial transmission and retransmission,  the success probability  is given by 
\begin{align}
\mathcal{P}_{\mathrm{II}}  = \mathbb{P} \bigg[   \bigcup^{K}_{k =1 } \Big \{\Upsilon_{k} > \theta \Big \}     \bigg], \label{eqn:Type_II}
\end{align}
where  $\Upsilon_{k}  = \sum^{k}_{i=1} \eta^{(k)} $ is the effective SIR after the $k$-th transmission of the same content with the Type-II HARQ retransmission scheme. 
 
\end{itemize}

In what follows, the success probabilities of Type-I HARQ and Type-II HARQ-CC schemes are derived for the case with $K=2$, i.e., each packet is retransmitted once if the initial transmission is not successful. The cases with $K \geq 3$ can be obtained by following the same methodology straightforwardly. 

\subsubsection{Success Probability of Type-I HARQ} When $K=2$, the success probability with Type-I HARQ given in (\ref{eqn:Type_I}) can be expressed as \cite{X.2019Lu} 
 \begin{align}
  \mathcal{P}_{\mathrm{I}} &  =   
  \mathbb{P} \big[ \eta^{(1)}  > \theta  \big] + \mathbb{P} \big[ \eta^{(1)} > \theta, \eta^{(2)} \leq \theta \big]     \! \nonumber \\ 
&  \! \overset{\text{(a)}}{=}     \sum^{2}_{ k=1 }     
 \mathbb{P} \big[ \eta^{(k)}   >   \theta   \big]     -      \mathbb{P} \big[ \eta^{(1)}  >  \theta , \eta^{(2)} > \theta   \big]    \nonumber\\
 &   \overset{\text{(b)}}{=} \! \begin{dcases}
      \sum^{2}_{ k=1 }     
      \mathbb{P} \big[ \eta^{(k)}  \! > \!  \theta  \big] \!    -      \mathbb{P} \big[ \eta^{(1)} \! > \! \theta , \eta^{(2)} \! > \! \theta \mid \Phi  \big] , \, \, \textup{QSI}    \\
       \sum^{2}_{ k=1 }      
       \mathbb{P} \big[ \eta^{(k)}  \! > \!  \theta    \big]  \!   - \!  \prod^{2}_{k=1}    \mathbb{P} \big[ \eta^{(k)} \! > \!  \theta    \big],  \hspace{14mm}  \textup{FVI} 
 \end{dcases}   
\end{align}
where $(a)$ applies the inclusion-exclusion principle and $(b)$ follows as the point processes across different time slots are the same and i.i.d. for QSI and FVI, respectively. 
Subsequently, based on the definition of $\mathcal{J}^{\textup{QSI}}_{K}$ and $\mathcal{J}^{\textup{FVI}}_{K}$ in (\ref{eqn:JSP_QSI}) and (\ref{eqn:JSP_FVI}), respectively. 
we can obtain the following corollary.

\begin{corollary}
Given that each unsuccessfully transmitted packet is retransmitted once,  the success probability under Type-I HARQ retransmission scheme with QSI and FVI are given, respectively, as
 \begin{align}
& \mathcal{P}_{\mathrm{I}}^{\textup{QSI}}    = 2 \mathcal{J}^{\textup{QSI}}_{1} - \mathcal{J}^{\textup{QSI}}_{2},  \\
& \mathcal{P}_{\mathrm{I}}^{\textup{FVI}}   = 2 \mathcal{J}^{\textup{FVI}}_{1} - \Big(\mathcal{J}^{\textup{FVI}}_{1}\Big)^{\!2 },
\end{align}
where $\mathcal{J}^{\textup{QSI}}_{K}$ and $\mathcal{J}^{\textup{FVI}}_{K}$ for $K \in \mathbb{N}$ are given in~(\ref{eqn:J_QSI}) and~(\ref{eqn:J_FVI}), respectively.
\end{corollary}

\subsubsection{Success Probability of Type-II HARQ-CC}

Next, we discuss how to obtain the success probability of Type-II HARQ-CC.
When $K=2$, the success probability with Type-II HARQ given in (\ref{eqn:Type_I}) can be expressed as  
$\mathcal{P}_{\mathrm{II}}    =   
  \mathbb{P} \big[ \eta^{(1)}  > \theta \big] + \mathbb{P} \big[ \eta^{(1)} + \eta^{(2)} > \theta, \eta^{(1)}\leq \theta  \big]$      
where the first term on the right-hand side of the equality has been obtained in Section~\ref{sec:retransmission}.C.1. To obtain the second term, we need to compute the distribution of the SIR of any single time slot. In particular, we  first calculate the CDF of the SIR and then obtain the corresponding PDF by taking the derivative of the CDF.
With the PDF of the SIR, we then derive the joint probability that events  $\eta^{(1)} + \eta^{(2)} > \theta$ and $\eta^{(1)}  < \theta$ both occur. The final results are presented in the following corollary. 

\begin{corollary} \label{cor:typeII}
Given that each unsuccessfully transmitted packet is retransmitted once,  the success probability under Type-II HARQ-CC retransmission scheme with QSI and FVI are given, respectively, as
\begin{align}
& \!\! \mathcal{P}^{\textup{QSI}}_{\mathrm{II}} \! =  \exp(-c \lambda  \theta^{\delta} r^2_{\rm t} ) \! + \! 2 \pi \lambda    \!  \int^{\theta}_{0} \! \!\! \int^{\infty}_{0} \!\! \frac{ r^{\alpha}_{\rm t} r^{-\alpha}  J(r,u) }{ ( 1+ u r^{\alpha}_{\rm t}   r^{-\alpha} )^2 }      
 r \mathrm{d} r \nonumber  \\
 & \hspace{6mm} \times  \exp \bigg( - 2 \pi \lambda \int^{\infty}_{0} \!   \bigg(  1 -  \frac{ J(r,u) }{ 1 +  u r^{\alpha}_{\rm t}  r^{-\alpha}  }  \bigg) r  \mathrm{d} r  \bigg) \mathrm{d} u , 
 \end{align}
 and
 \begin{align}
& \!\! \mathcal{P}^{\textup{FVI}}_{\mathrm{II}} \! =  \! \exp(-c \lambda  \theta^{\delta} r^2_{\rm t} ) \! + \! 2 \pi \lambda    \!  \int^{\theta}_{0} \! \! \! \int^{\infty}_{0} \!\! \frac{ r^{\alpha}_{\rm t} r^{-\alpha}   }{ ( 1+ u r^{\alpha}_{\rm t}   r^{-\alpha} )^2 }  
 r \mathrm{d} r \nonumber  \\ 
 & \hspace{6mm} \times  \exp \bigg( \! - 2 \pi \lambda \int^{\infty}_{0} \! \bigg(  1 -  \frac{ 1 }{ 1 +  u r^{\alpha}_{\rm t}  r^{-\alpha}  }  \bigg) r  \mathrm{d} r  \bigg)   \nonumber \\
 & \hspace{6mm} \times \exp \bigg( \! - 2 \pi \lambda \int^{\infty}_{0} \! \Big(  1 -  J(r,u) \Big) r  \mathrm{d} r  \bigg) \mathrm{d} u ,
\end{align}
where $c \! = \! \pi   \Gamma(1+\delta) \Gamma(1-\delta)$  and $ J(r,u) =
   \frac{1}{  1 +  ( \theta - u ) r^{\alpha}_{\rm t}  r^{-\alpha} }$.
\end{corollary}
 
\begin{proof}
See \textbf{Appendix I}.
\end{proof}

 \begin{figure}[htp]
 \centering
 \includegraphics[width=0.5\textwidth]{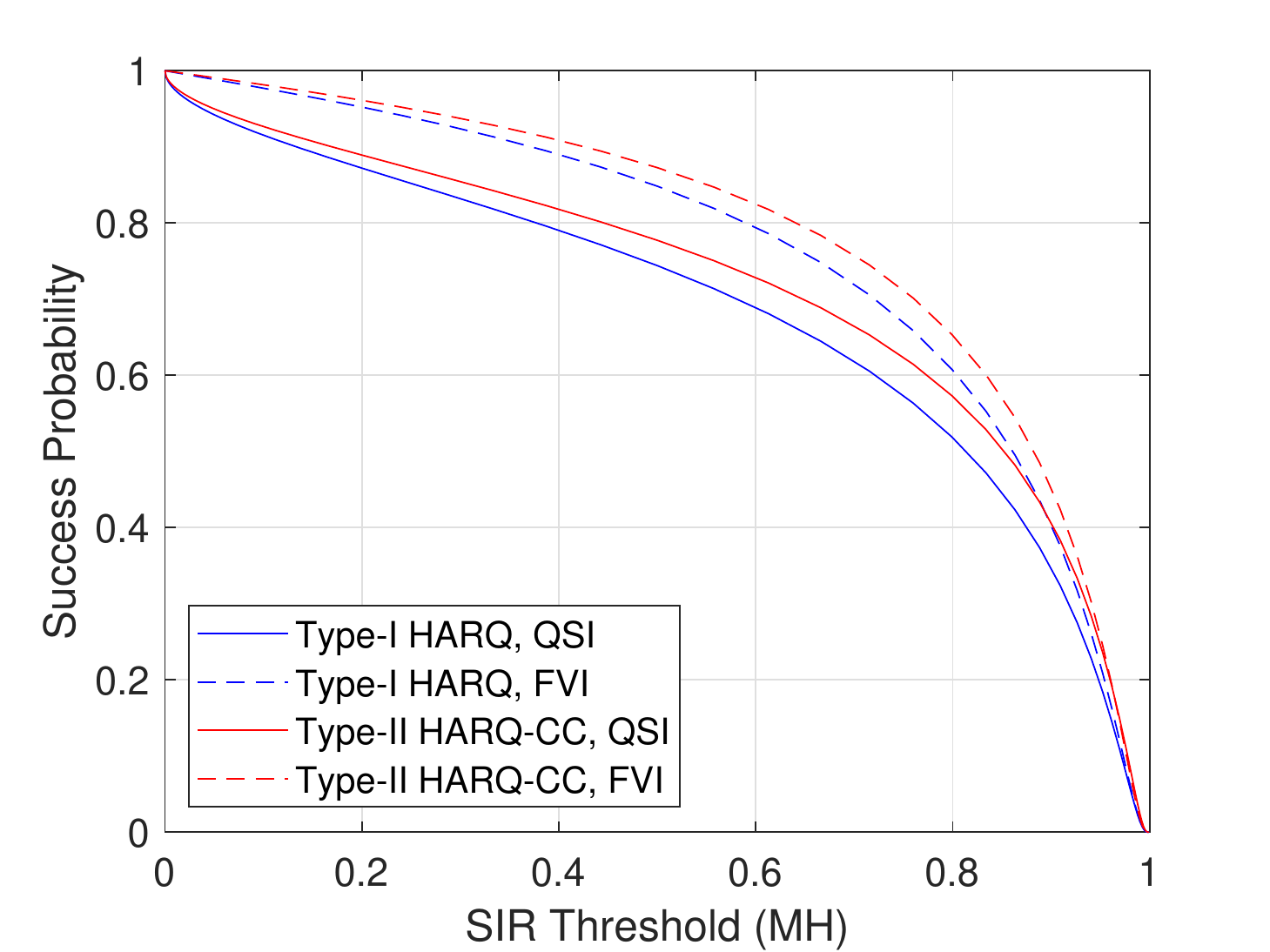} 
 \caption{Success probabilities of HARQ retransmission schemes in Poisson ad hoc networks ($K=2$).  } \label{fig:CP_HARQ} 
 \end{figure}

Fig.~\ref{fig:CP_HARQ} depicts the success probabilities of Type-I HARQ and Type-II HARQ-CC when $K=2$.
As shown, Type-II HARQ-CC provides a higher success probability at all SIR thresholds. 
Nevertheless, the success probabilities of the two schemes are comparable when the SIR threshold is small (e.g., $\theta<0.05$ MH). 
This indicates that, in the high-coverage regime, the SIR gain 
due to MRC has an imperceptible effect on the success probability.
The reason is that low SIR thresholds render a high success probability of the first transmission, which makes the impact of retransmission negligible.

\subsection{Summary and Discussion}

This section has discussed the influence of temporal interference correlation on the retransmission performance of a  transmission link in a Poisson field of interferers. In particular, we have derived the JSP of multiple transmissions, the correlation coefficient for two transmissions, and the success probability for a given number of transmission attempts. 
Furthermore, based on the above results, we have shown how to analyze the success probabilities under Type-I HARQ and Type-II HARQ retransmission schemes. 
The lessons learned are  as follows. 

\begin{itemize}

\item Temporally-correlated interference results in a higher JSP than temporally-independent interference. The performance gap between the two cases increases with the number of transmission attempts.

\item The CSP of $K$-th transmission given that the previous $K-1$ transmissions all succeed increases with $K$ and remains the same with temporally-correlated interference and temporally-independent interference, respectively.

\item Temporally-correlated interference  results in higher success probability with retransmission than temporally-independent interference.
The gap between the success probabilities with retransmission under the two types of interference increases with the number of transmission attempts. 

\item The performance gain of Type-II HARQ over Type-I HARQ is larger with temporally-correlated interference than with temporally-independent interference.  

\end{itemize}

\vspace{0.2cm}
\noindent
{\em Open Technical Issues}: This section has presented models to characterize the impact of temporal interference correlation in a retransmission-based (or multi-packet transmission) system. We have considered both Type-I HARQ and Type-II HARQ systems and evaluated both CSP and JSP. From a practical point of view, the impact of ``interacting queues" may need to be considered in a retransmission scenario. Also, the performance of HARQ systems in the presence of interference correlation in different wireless systems (e.g., non-orthogonal multiple access (NOMA) systems and MIMO systems) will be worth investigating.  

In some emerging scenarios, such as device-to-device communication \cite{N2014Tehrani} and mobile social networks~\cite{Laha2015A}, 
an information source may rely on mobile users to spread messages to multiple destinations. To model the process of information spreading in such systems, a possible direction is to incorporate user mobility models into the stochastic geometry analysis.

\section{Spatially and Temporally Correlated Interference in Mobile Systems}  \label{sec:mobility}
This section investigates the effect of mobility on the spatial-temporal SIR correlation. We analyze the JSP in mobile networks with spatially and temporally correlated interference. 

\subsection{Introduction}
We consider two mobile network scenarios. In the first scenario, similar to \cite{Krishnan2017}, we consider a downlink communication in a Poisson cellular network, and provide the steps to calculate the spatial-temporal JSP of a mobile user at two different time instants. In the second scenario, similar to \cite{Z.2014Gong}, we consider a Poisson bipolar model where the desired transmitter and receiver are static while other transmitters in the network move according to a random mobility model. For this scenario, we also explain the steps to calculate the JSP at the static receiver \cite{Tabassum2019}.

Let us denote the success event indicator at time $t_1$ at location $u_1$ by 
$A_1=\mathbbm1_{\{\text{SIR}_1>\theta \}}$ and at time $t_2$ at location $u_2$ by $A_2=\mathbbm1_ {\{\text{SIR}_2>\theta \} }$, where $\text{SIR}_1$ is the SIR at the receiver at location $u_1$ at time $t_1$, and $\text{SIR}_2$ is the SIR at the receiver at location $u_2$ at time $t_2$; $\theta$ is  the target SIR. In the following, we present the steps to calculate $\mathbb{E}[A_1A_2]$ (i.e., the spatial-temporal JSP) for the two mentioned network scenarios. Note that when the network is stationary over time, $\mathbb{E}[A_1A_2]$ does not depend on $t_1$ and $t_2$; it only depends on $\Delta=|t_2-t_1|$. Generally, the random variables $A_1$ and $A_2$ are spatially 
and temporally 
correlated. When $u_1=u_2$, $\mathbb{E}[A_1A_2]$ captures the temporal correlation, i.e., the correlation at one location in two different time instants. When $t_1=t_2$, $\mathbb{E}[A_1A_2]$ captures the spatial correlation, i.e., correlation at two different locations at the same time.

\subsection{Analysis of Poisson Downlink Networks}

\subsubsection{Model I}
\begin{figure} 
\centering
 \subfigure [No handoff occurs. ]
  {
 \label{fig:Cox_system1}
 \centering   
 \includegraphics[width=0.65 \textwidth]{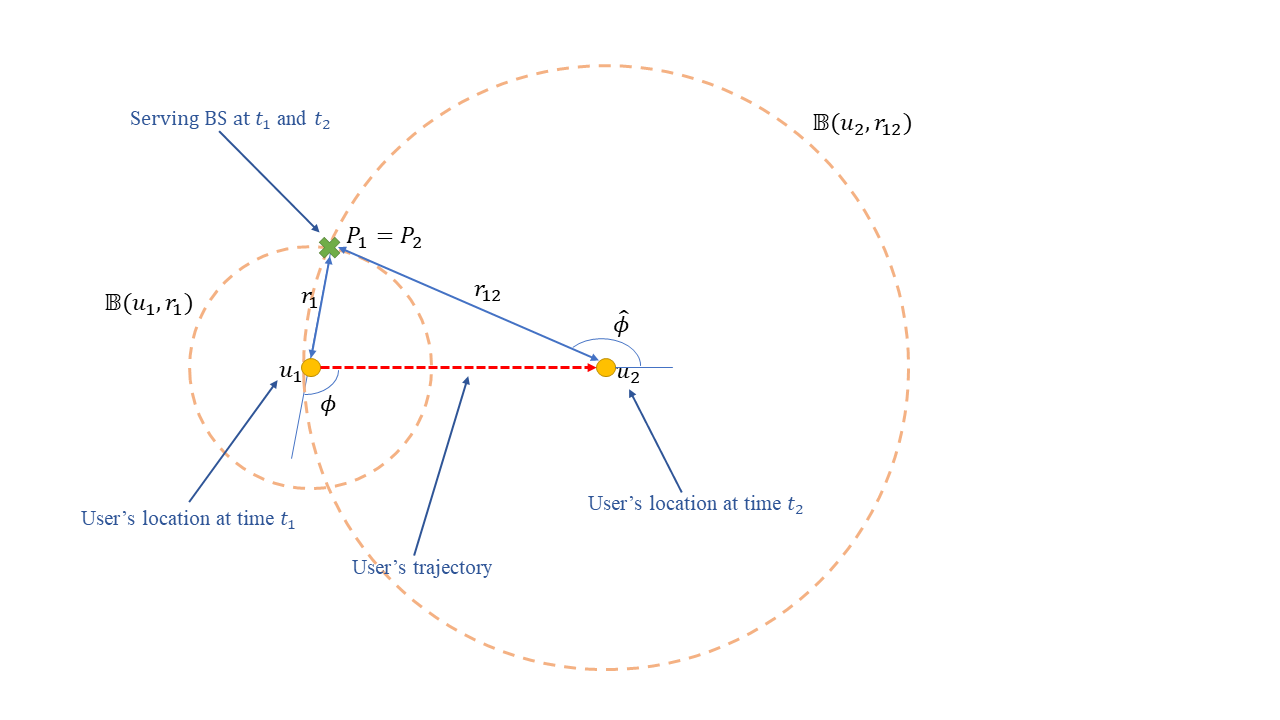}}
 \centering  
 \subfigure  [Handoff occurs.
 ] {
\label{fig:Cox_system2}
 \centering 
\includegraphics[width=0.65 \textwidth]{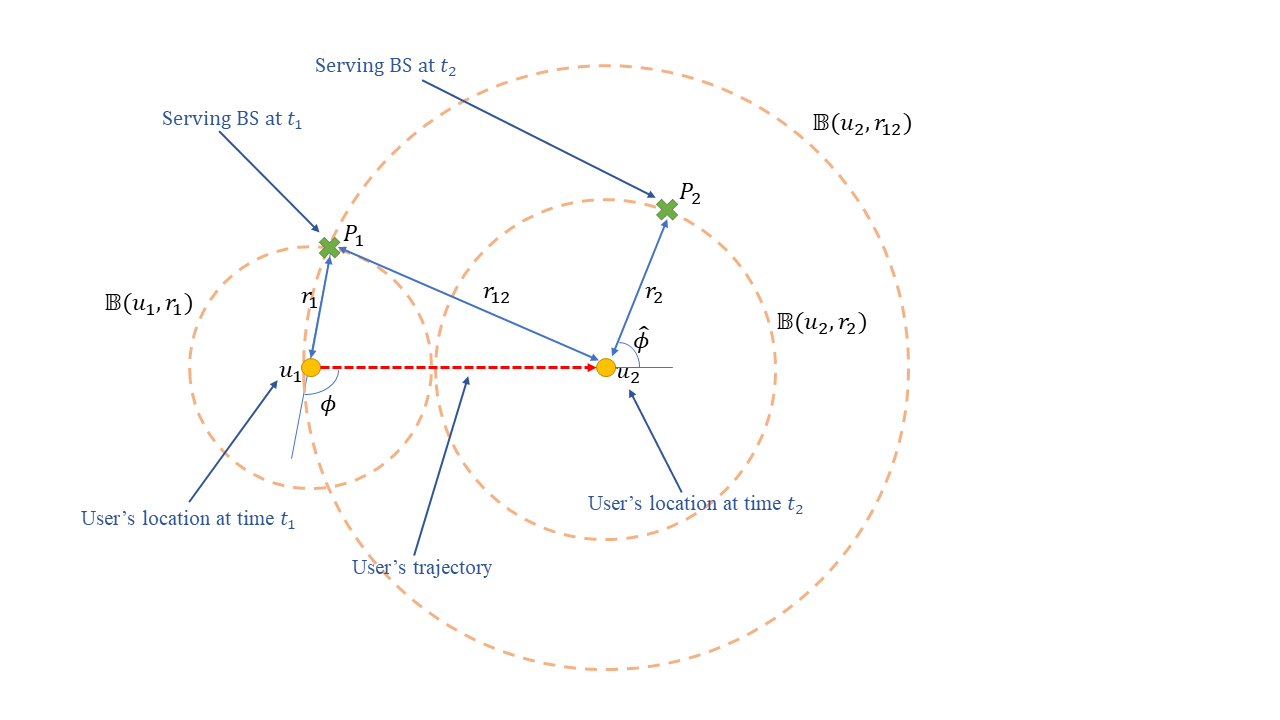}}
\caption{Scenarios with and without handoff \cite{Krishnan2017}. } 
\centering
\label{fig:Scenario1}
\end{figure}

Consider a Poisson downlink network $\Phi$ with intensity $\lambda$ where 
all the users are associated with their nearest BSs. The  small-scale fading is i.i.d. across time and space. Consider a mobile user that is located at distance $r_1$ from its associated BS at time $t_1$. 
The user moves away from  the associated BS with a constant speed $v$ at an angle $\phi$ at time $t_1$, where $\phi$ is uniformly distributed in $[0,2\pi]$ (Fig. \ref{fig:Scenario1}). Before further discussion, let us introduce some notations that help us in deriving the spatial-temporal JSP at time $t_1$ and $t_2$ with $t_{2} > t_{1}$ \cite{Krishnan2017}.
\begin{itemize}
    \item $P_1$, $P_2$: Location of the serving BS at time $t_1$ and $t_2$, respectively.
	\item $r_1$, $r_2$: Distance between the mobile user and its associated BS at time $t_1$ and time $t_2$, respectively.
	\item $r_{12}$: Distance between the mobile user at time $t_2$ and its associated BS at time $t_1$.
	\item $\mathbb{B}(x,r)$: Closed disk with center at $x$ and radius being $r$. For brevity, $\mathbb{B}(u_1,r_1)$, $\mathbb{B}(u_2,r_2)$, and $\mathbb{B}(u_2,r_{12})$ are denoted by $\mathcal{B}_1$, $\mathcal{B}_2$, and $\mathcal{B}_{12}$, respectively.
	\item $H$: Event that at least one handoff occurs in the interval $[t_1,t_2]$.
\end{itemize}

As the mobile user moves from $u_1$ to $u_2$, the mobile user may be served by the same BS (no handoff occurs) or it may be handed off to other BSs (handoff occurs). At time $t_2$, when there is no BS in $\mathcal{B}_{12}$, the mobile user is still connected to 
the same BS as time $t_1$, i.e., $P_1=P_2$. In this case, according to the definitions and from triangle equations, we have $r_2=r_{12}=\sqrt{r_1^2+v^2(t_2-t_1)^2+2r_1v(t_2-t_1)\cos\phi}$ and $\mathcal{B}_{12}=\mathcal{B}_2$. On the other hand, when there is a BS in $\mathcal{B}_{12}$, at time $t_2$, the mobile user is served by a new BS, i.e., $P_1 \neq P_2$. In this case $r_2<r_{12}$ and $\mathcal{B}_2 \subset \mathcal{B}_{12}$. In the following, we briefly provide the steps for deriving the JSP at time $t_1$ and $t_2$. Note that since $u_1 \neq u_2$ and $t_1 \neq t_2$, the JSP captures the spatial-temporal SIR correlation.

\subsubsection{Steps to Derive the JSP \cite{Krishnan2017}}   
\textbf{Step 0}: From the contact distribution function of the homogeneous PPP, the PDF of the contact distance at time $t_1$ follows as 
(\ref{eqn:PDF_CD_PPP}). Moreover, due to the symmetry, we can assume $\phi$ is uniformly distributed in $[0,\pi]$.

\textbf{Step (i)}: Calculating $\mathbb{P}\left(H \mid r_1,\phi\right)$ (handoff probability given $r_1$ and $\phi$):
\begin{align}
	\mathbb{P}\left(H \mid r_1,\phi\right) & \stackrel{\text{(a)}}{=} 1-\mathbb{P}\left( \Phi(\mathcal{B}_{12} \setminus \mathcal{B}_1)=0 \mid r_1,\phi \right) \nonumber \\
											 & \stackrel{\text{(b)}}{=} 1-\exp\big(  -\lambda |\mathcal{B}_{12} \setminus \mathcal{B}_{1}| \big), \nonumber
											 \end{align}
where $(a)$ follows the fact that handoff occurs when there is a BS in $\mathcal{B}_{12}$. In $(a)$, $\mathcal{B}_{1}$ is excluded from $\mathcal{B}_{12}$ since there is no BS inside $\mathcal{B}_{1}$. $(b)$ follows from the void probability of the PPP.  

\textbf{Step (ii)}: Deriving the distribution of $r_{2}$: Given $r_1$ and $\phi$, when there is no handoff,  $r_2=\sqrt{r_1^2+v^2(t_2-t_1)^2+2r_1v(t_2-t_1)\cos\phi}$. However, when handoff occurs, as discussed earlier, $r_2<r_{12}$. In this case, we can derive the conditional PDF of $r_2$ as $f_{r_2}(z \mid H, r_1, \phi) = \frac{{\rm d}}{{\rm d}z} F_{r_2}(z \mid H, r_1, \phi)$, where $F_{r_2}(z \mid H, r_1, \phi)$ is the conditional CDF of $r_{2}$ and can be obtained by
\begin{align}
	& F_{r_2}(z \mid H, r_1, \phi) \nonumber \\
	&= \frac{ \mathbb{P}\left(r_2 \le z, H \mid r_1,\phi\right) }{ \mathbb{P}\left(H \mid r_1,\phi\right) } \nonumber \\
	 								 &= \frac{ \mathbb{P}\left(\Phi(\mathbb{B}(u_2,z) \setminus \mathcal{B}_{1})>0, \Phi(\mathcal{B}_{12} \setminus \mathcal{B}_{1})>0 \mid r_1,\phi\right) }{ \mathbb{P}\left(H \mid r_1,\phi\right) } \nonumber \\
	 								 &\stackrel{\text{(a)}}{=} \frac{ \mathbb{P}\left(\Phi(\mathbb{B}(u_2,z) \setminus \mathcal{B}_{1})>0 \mid r_1,\phi\right) }{ \mathbb{P}\left(H \mid r_1,\phi\right) } \nonumber \\
	 								 &= \frac{1 \! - \! \exp\big( -\lambda |\mathbb{B}(u_2,z) \setminus \mathcal{B}_{1}| \big)}{ \mathbb{P}\left(H \mid r_1,\phi\right) }, \end{align}	 
for $z\in[\max(0,r_1-v),r_{12}]$. $(a)$ follows since $\mathbb{B}(u_2,z)\subset\mathcal{B}_{12}$ in the case of handoff. When $r_1>v$, the nearest BS is at least at distance $r_1-v$, therefore, $F_{r_2}(z \mid H, r_1, \phi)=0$ for $z<\max(0,r_1-v)$. Moreover, as explained earlier, $F_{r_2}(z \mid H, r_1, \phi)=1$ for $r_{12}<z$.

When handoff occurs, we also need the distribution of $\hat{\phi}$, angle between the vectors  $\overrightarrow{u_1u_2}$ and $\overrightarrow{u_2P_2}$. Given $r_1$, $r_2$, and $v$, one of the two following cases holds\footnote{If $v \le |r_1-r_2|$, no handoff occurs when the user moves from $u_1$ to $u_2$.}: 
1) $|r_1-r_2|< v \le r_1+r_2$ ($\mathcal{B}_{1}$ and $\mathcal{B}_{2}$ overlap): $\hat{\phi}$ is uniformly distributed in 
$\left[-\pi+\arccos\left(\frac{r_2^2+v^2-r_1^2}{2r_2v}\right),\pi-\arccos\left(\frac{r_2^2+v^2-r_1^2}{2r_2v}\right)\right]$, 
2) $r_1+r_2<v$ ($\mathcal{B}_{1}$ and $\mathcal{B}_{2}$ are disjoint): $\hat{\phi}$ is uniformly distributed in $[0,2\pi]$.

\textbf{Step (iii)}: Calculating the spatial-temporal JSP \mbox{$\mathbb{P}(\text{SIR}_1>\theta,\text{SIR}_2>\theta)$} as in (\ref{eq:scenario1_main}), where $\bar{H}$ is the complement of $H$ and denotes the event that handoff does not occur.
\begin{figure*}
	\begin{align}
	 \mathbb{P}(\text{SIR}_1>\theta,\text{SIR}_2>\theta) 
&	=
	\mathbb{E}\left[ \mathbb{P}(\text{SIR}_1>\theta,\text{SIR}_2> \theta \mid r_1, \phi) \right] \nonumber \\ 
	&=\underbrace{
		\mathbb{E}\left[ \mathbb{P}(\text{SIR}_1>\theta,\text{SIR}_2> \theta \mid \bar{H}, r_1, \phi) \mathbb{P}(\bar{H} \mid r_1, \phi) \right] }_{\textbf{Term I}}\nonumber \\ 
	&+\underbrace{
		\mathbb{E}\left[ \mathbb{P}(\text{SIR}_1>\theta,\text{SIR}_2> \theta \mid H, r_1, \phi) \mathbb{P}(H \mid r_1, \phi) \right] }_{\textbf{Term II}} \IEEEeqnarraynumspace \label{eq:scenario1_main}.
		\end{align}
		\hrulefill
		\end{figure*}

\textbf{Term I} in \eqref{eq:scenario1_main} can be obtained by (\ref{eqn:term1}), 
	\begin{figure*}
	\begin{align}
	& \mathbb{E}\left[ \mathbb{P}(\text{SIR}_1>\theta,\text{SIR}_2> \theta \mid \bar{H}, r_1, \phi) \mathbb{P}(\bar{H} \mid r_1, \phi) \right]  \nonumber \\ &= \mathbb{E}\left[ \mathbb{E}_{\Phi_{\rm I}}\left[ \mathbb{P}(\text{SIR}_1>\theta,\text{SIR}_2> \theta \mid \Phi_{\rm I}, \bar{H}, r_1, \phi) \right] \mathbb{P}(\bar{H} \mid r_1, \phi) \right] \nonumber \\
	&\stackrel{\text{(a)}}{=} \mathbb{E}\left[ \mathbb{E}_{\Phi_{\rm I}}\left[ \mathbb{P}(\text{SIR}_1>\theta \mid \Phi_{\rm I}, \bar{H}, r_1, \phi) \mathbb{P}(\text{SIR}_2>\theta \mid \Phi_{\rm I}, \bar{H}, r_1, \phi) \right] \mathbb{P}(\bar{H} \mid r_1, \phi) \right], \label{eqn:term1}
	\end{align}
	\hrulefill
	\end{figure*}
where $\Phi_{\rm I}$ denotes the locations of the interferers. In the case of no handoff ($\bar{H}$), $\Phi_{\rm I}$ is a PPP with intensity $\lambda$ in $\mathbb{R}^2\setminus \left(\mathcal{B}_{1}\cup\mathcal{B}_{12}\right)$. $(a)$ is obtained using the fact that, given $\Phi_{\rm I}$, $\bar{H}$, $r_1$, and $\phi$, $\text{SIR}_1$ and $\text{SIR}_2$ are independent since small-scale fading is i.i.d. at different time instants and spatial locations. The inner expectation w.r.t. $\Phi_{\rm I}$ can be calculated by using the PGFL of PPP, and the outer expectation by using PDF of $r_1$ and $\phi$ provided in \textbf{Step 0}. $\mathbb{P}(\bar{H} \mid r_1, \phi)$, probability that handoff does not occur, is also provided in \textbf{Step (i)}. 
\textbf{Term II} in \eqref{eq:scenario1_main} can be derived as \eqref{eqn:term2}.
\begin{figure*}
\begin{align}
	& 
	 \mathbb{E}\left[ \mathbb{P}(\text{SIR}_1>\theta,\text{SIR}_2>\theta \mid H, r_1, \phi) \mathbb{P}(H \mid r_1, \phi) \right]  \nonumber \\ 
	& = \mathbb{E}\left[ \mathbb{E}_{r_2}\left[ \mathbb{E}_{\hat{\phi}}\left[ \mathbb{P}(\text{SIR}_1>\theta,\text{SIR}_2>\theta \mid \hat{\phi}, r_2, H, r_1, \phi) \right] \right] \mathbb{P}(H \mid r_1, \phi) \right] \nonumber \\
	&\stackrel{\text{(a)}}{=} 
	\mathbb{E}\left[ \mathbb{E}_{r_2}\left[ \mathbb{E}_{\hat{\phi}}\left[ 
	\mathbb{E}_{\Phi_{\rm I}}\left[ \mathbb{P}(\text{SIR}_1>\theta \mid \Phi_{\rm I}, \hat{\phi}, r_2, H, r_1, \phi) \mathbb{P}(\text{SIR}_2>\theta \mid \Phi_{\rm I}, \hat{\phi}, r_2, H, r_1, \phi) \right]
	\right] \right] \mathbb{P}(H \mid r_1, \phi) \right] . 
	\label{eqn:term2}
	\end{align}
	\hrulefill
	\end{figure*}
In the case of handoff ($H$), $\Phi_{\rm I}$ is a PPP with intensity $\lambda$ in $\mathbb{R}^2\setminus (\mathcal{B}_{1} \cup \mathcal{B}_{2})$. Note that $\Phi_{\rm I} \cup \{P_2\}$ is the set of interferers at time $t_1$ and $\Phi_{\rm I} \cup \{P_1\}$ is the set of interferers at time $t_2$. $(a)$ is obtained using the fact that, given $\Phi_{\rm I}$, $\hat{\phi}$, $r_2$, $H$, $r_1$, and $\phi$, $\text{SIR}_1$ and $\text{SIR}_2$ are independent. The expectation w.r.t. $\Phi_{\rm I}$ can be obtained by using the PGFL of PPP. The expectation w.r.t. $\hat{\phi}$ and $r_2$ can be obtained from \textbf{Step (ii)}. $\mathbb{P}(H \mid r_1, \phi)$ is provided in \textbf{Step (i)}. Outer expectation w.r.t. $r_1$ and $\phi$ can also be obtained by using PDFs given in \textbf{Step 0}.    

\subsection{Analysis of Poisson Bipolar Networks}

\subsubsection{Model II}
In the second scenario, similar to \cite{Z.2014Gong}, we consider a Poisson bipolar network, where, at time $t_1$, transmitters form a PPP $\Phi$ with intensity $\lambda$. Each transmitter is assigned a unique receiver at distance $r_0$ at a random direction. Due to the stationarity of PPP, we investigate the performance of the typical receiver and consider its location as the origin. Since the homogeneous PPP is isotropic, without loss of generality, we can assume that the transmitter of this receiver is located at $[r_0,0]^T$.
From Slivnyak's theorem \cite{M.2013Haenggic}, at time $t_1$, the other transmitters (interferers), denoted by  $\Phi$, form a homogeneous PPP with intensity $\lambda$. Assume that the selected receiver and its corresponding transmitter are static while interferers are mobile according to a uniform mobility model. Hence, at any time $t$, the interferers form a homogeneous PPP with intensity $\lambda$.  

We study the JSP at two different time slots $t_1$ and $t_2$. Note that since the desired receiver is fixed at the origin ($u_1=u_2$), the JSP captures the temporal correlation. Let $w_{i,\Delta}$ denote the displacement vector for transmitter $i$\footnote{$\Delta$ reflects the fact that statistics of the displacement vector depends on the time difference $\Delta=|t_2-t_1|$ rather than $t_1$ and $t_2$ separately.}. We consider a random individual mobility model (such as random walk), where each interferer moves independently of other interferers following the same random mobility model. Thus, $w_{i,\Delta}$ is independently and identically distributed for each interferer. 

\subsubsection{Derivation of Spatial-temporal JSP}
We can calculate the JSP 
in the case of random individual mobility model where $u_{1}=u_{2}$ for spatially and temporally i.i.d. fading as follows:
\begin{align}
	& \mathbb{P}(\text{SIR}_1>\theta,\text{SIR}_2>\theta) \nonumber \\
	&= \mathbb{E}_{\Phi_{\rm I}} \left[ \mathbb{P}(\text{SIR}_1>\theta,\text{SIR}_2>\theta \mid \Phi_{\rm I}) \right] \nonumber \\
	&\stackrel{\text{(a)}}{=} \mathbb{E}_{\Phi_{\rm I}} \left[ \mathbb{P}(\text{SIR}_1>\theta \mid \Phi_{\rm I}) \mathbb{P}( \text{SIR}_2>\theta \mid \Phi_{\rm I}) \right] \nonumber \\
	&\stackrel{\text{(b)}}{=} \mathbb{E}_{\Phi_{\rm I}} \Big[ \mathbb{P}(\text{SIR}_1>\theta \mid \Phi_{\rm I}) \nonumber \\																						& \hspace{5mm} \times
	\mathbb{E}_{(w_{\Delta ,i})} \big[ \mathbb{P}( \text{SIR}_2>\theta \mid \{w_{\Delta,i}\}, \Phi_{\rm I}) \big] \Big], \label{eq:scenarioIImobility} \end{align} 
where $(a)$ is obtained, since given the initial location of the interferers $\Phi_{\rm I}$, $\text{SIR}_1$ and $\text{SIR}_2$ are independent. In $(b)$, the inner expectation is w.r.t. the displacement vectors of all interferers, where $\Delta=|t_2-t_1|$.

\subsection{Numerical Results}
To understand the effect of mobility, in this section, we consider the CSP
	\begin{align}
	   \mathbb{P}( \text{SIR}_{2}> \theta \mid \text{SIR}_{1} >\theta ) 
	  = 
	\frac{ \mathbb{P}( \text{SIR}_{1} >\theta, \text{SIR}_2>\theta ) }{ \mathbb{P}( \text{SIR}_1>\theta ) }. \nonumber
	\end{align}
Since the network performance depends on the nodes' displacements\footnote{The displacement of a node that moves at a speed $v$ for a time duration $|t_2-t_1|=\Delta$ is the same as when it moves with speed $v\Delta$ for duration $|t_2-t_1|=1$.}, we set $|t_2-t_1|=1$ and only consider the effect of the speed on the spatial-temporal JSP.

\begin{figure} 
\centering
 \subfigure [Model I.]
  {
 \centering   
 \includegraphics[width=0.5  \textwidth]{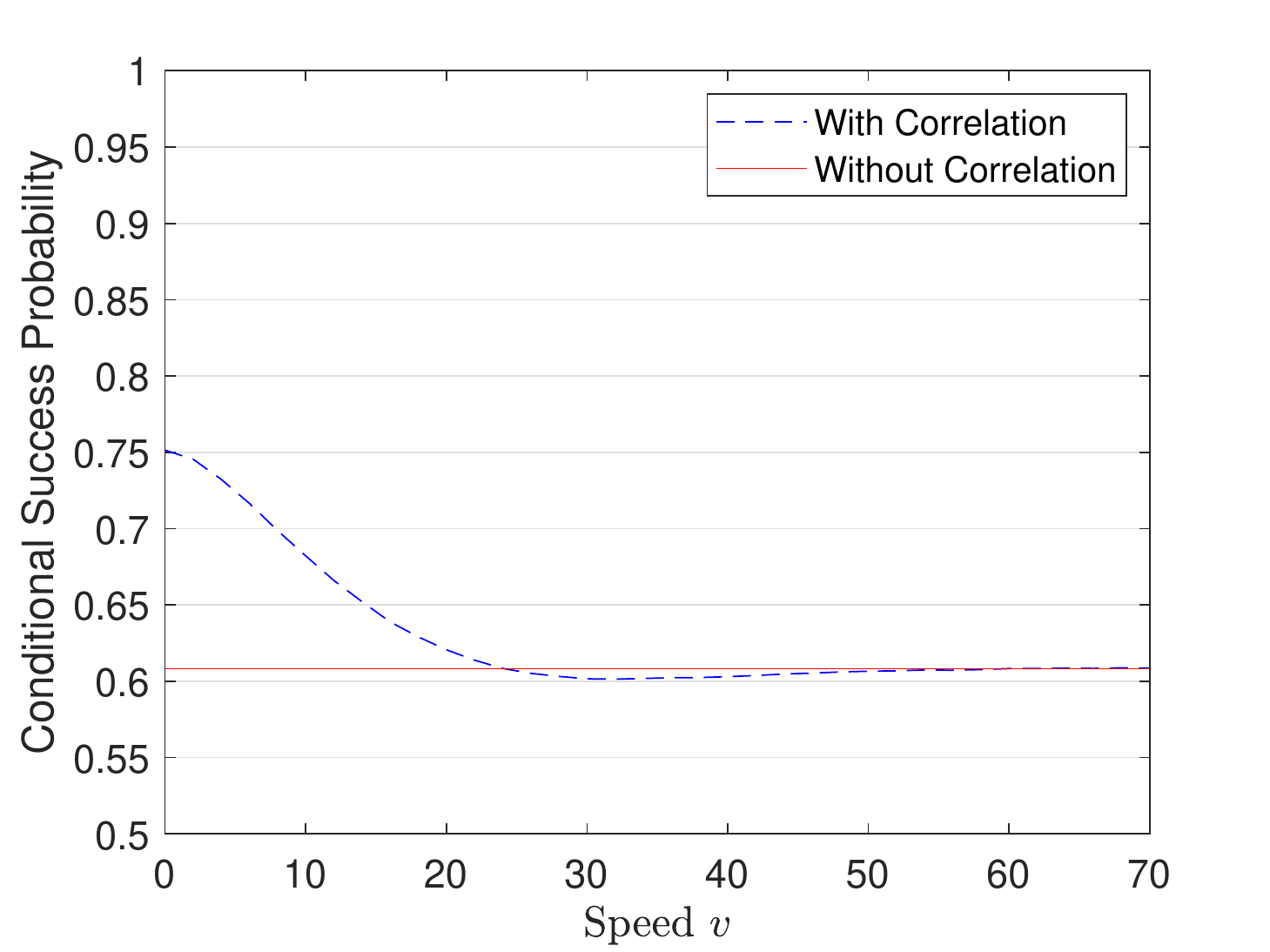}}
 \centering  
 \subfigure  [Model II.
 ] {
 \centering 
\includegraphics[width=0.5 \textwidth]{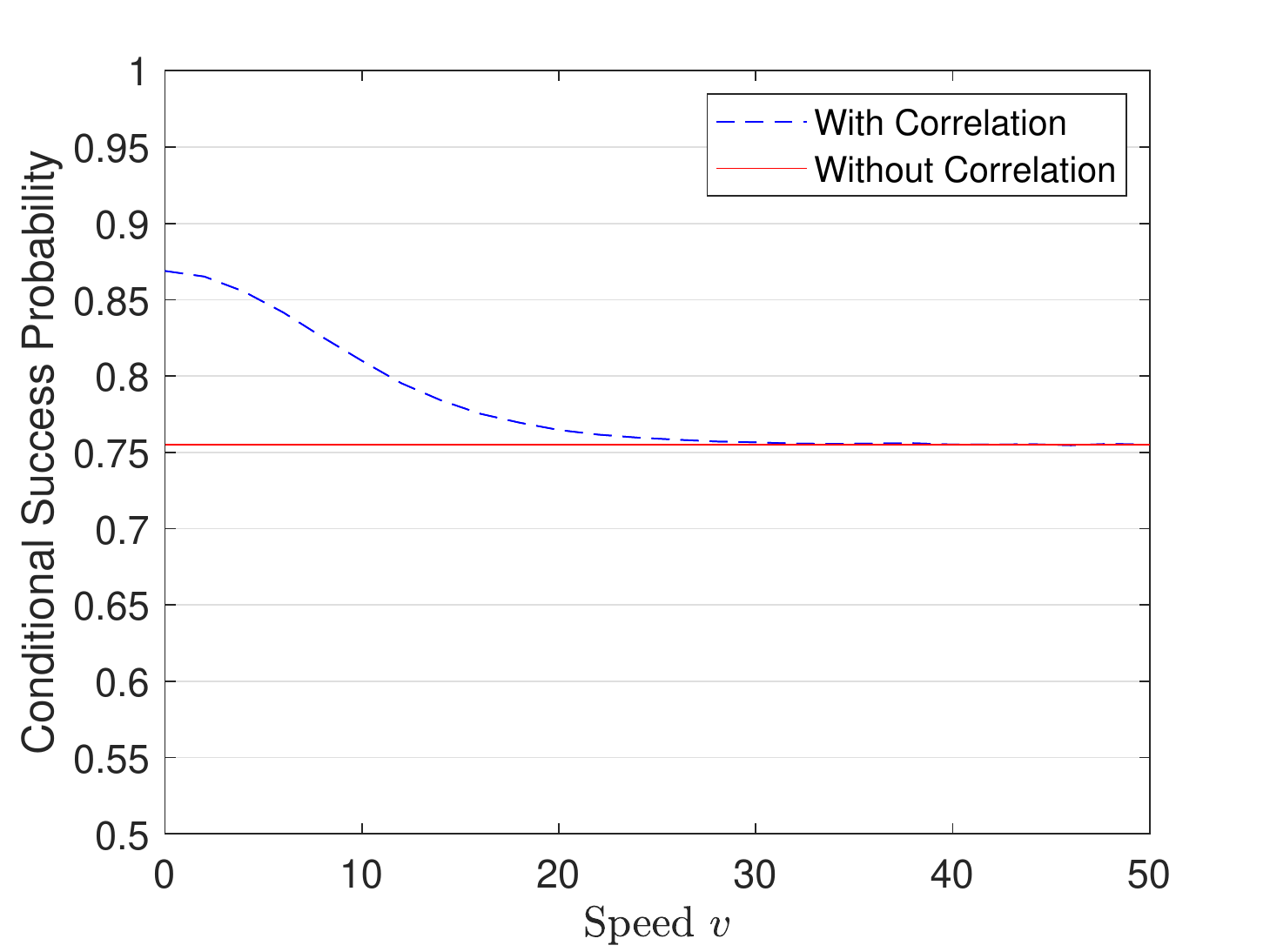}}
\caption{Simulation of the effect of speed on the spatial-temporal JSP for Model I in $(a)$ and Model II in $(b)$ for $\lambda=0.001$ and $\theta=-1{\rm dB}$.  In Model II, the desired link distance is set to $d=8$.  } 
\centering
 	\label{fig:Simulation_results}
\end{figure}

In Fig.~\ref{fig:Simulation_results}, we illustrate the effect of correlation on the CSP for both scenarios. For the random individual mobility model, we have assumed all the mobile users move with the same speed $v$, and each mobile user moves in a random independent direction $\varphi_i$,  uniformly distributed in the interval $[0,2\pi]$. To better understand the impact of correlation, in Fig.~\ref{fig:Simulation_results}, we also compare the results with the independent case, where the CSP is equal to $\mathbb{P}( \text{SIR}_2>\theta )$. According to the displacement theorem \cite{M.2013Haenggic}, at any time $t$, interferers form a PPP with intensity $\lambda$. Thus, in Model II, $\mathbb{P}( \text{SIR}_{u_2}(t)>\theta )$ does not depend on $t$ (and consequently $v$). Using the stationarity of the PPP, we can obtain the same result for Model I.

As shown in Fig. \ref{fig:Simulation_results}, at high speeds, we can ignore the effect of correlation, and assume that the two sets of interferers at time instants $t_1$ and $t_2$ are independent. As a result, most of the existing works in mobile networks only focus on the highly mobile scenario where they can ignore the effect of correlation. It is worthwhile to note that the gap between the results with correlation (dashed blue curves) and without correlation (solid red lines) is proportional to the correlation coefficient between the success event at $t_1$ and success event at $t_2$. As illustrated in Fig. \ref{fig:Simulation_results}, the correlation between success events at $t_1$ and $t_2$ decreases as the speed increases.

\subsection{Summary and Discussion}

This section has characterized the spatial-temporal JSP of two network scenarios 1) with mobile users and static interferers and 2) with mobile interferers and mobile users. The main 
lessons learned are as follows.

\begin{itemize}

\item The spatial-temporal correlation among different transmission events decreases with speed and can be ignored when the speed is high enough.

\item The gap between the spatial-temporal CSP is proportional to the correlation coefficient between  successful transmission events.  

\end{itemize}

In Poisson networks, 
mobility introduces spatial randomness which decreases the correlation across time \cite{Krishnan2017}; thus, spatial diversity provides time diversity when nodes are mobile. This is also discussed in the literature by studying the mean local delay \cite{Haenggi2010}.  																																	
In a static network, when there is a close interferer, retransmitting an unsuccessfully received signal is usually not helpful. As a result, in static networks, the mean local delay may be infinite, i.e., a significant fraction of users suffer from large delays. However, in mobile networks, interference correlation decreases across time and there is a higher chance to receive a retransmitted signal successfully. Therefore, mobility reduces the mean local delay by providing spatial diversity. Specifically, in the highly mobile scenario, the mean local delay is always finite.     

\vspace{0.2cm}
\noindent
{\em Open Technical Issues}: Analyzing the performance of mobile wireless networks with correlation effects is challenging for most scenarios and usually yields complicated results. Consequently, most of the related works consider simple system models such as downlink communication in a Poisson cellular network (as in Model I) or Poisson bipolar networks with mobile interferers (as in Model II), but even for the latter model, there are no analytical results.

In the downlink communication in Model I, we have studied the performance of a mobile node in a network of static transmitters (i.e., BSs). In contrast, in the uplink, we have mobile transmitters (the desired transmitter and interferers). Due to the random mobility of the nodes and irregular shapes of the cells, it may not be possible to calculate the load distribution while considering the correlation effect, which in turn, makes the analysis of the uplink scenario very challenging. In fact, an exact analysis is impossible, but  approximations, bounds, or asymptotic results can almost always be derived.

In Model II, as shown in \eqref{eq:scenarioIImobility}, we need the distribution of the displacement vector $w_{\Delta}$. However, it may not be possible to derive this distribution for most of the mobility models and for all values of $\Delta$. Therefore, most of 
the existing works resort to a simplified mobility model or consider only one movement step during which the mobile node moves along a straight line.

In addition to the spatial-temporal JSP, the SIR meta distribution in cellular networks with mobile BSs has been studied in~\cite{X.June2020Tang}. The joint SIR meta distribution and product SIR meta distribution in mobile networks are also worth investigating.

\section{Future Directions}

In this section, we discuss  future applications of stochastic geometry analysis and related technical challenges. 

\subsection{Novel Performance Metrics and Stochastic Geometry Analysis}

 Although the existing literature has extensively investigated different types of large-scale communication networks, 
 the majority of them focus on the spatial and temporal average performance, such as mean success probability, delay and throughput. Recently, 
 the SIR meta distribution~\cite{M.Apr.2016Haenggi} has been introduced to provide insights on the distribution of success probability over space and/or time. The meta distribution has been extended to evaluate the distributions of rate and energy~\cite{N.2019Deng}, mean delay~\cite{Y.Jun.2017Zhong},
 transmission rate~\cite{S.2019Kalamkar2},  secrecy rate~\cite{J.2018Tang}, and secrecy success probability~\cite{Yu2019Xinlei}.

 The concept of meta distribution can also be applied to evaluate the distribution of other performance metrics such as the covert probability for low probability of detection/intercept communications~\cite{A.2013Bash,X.Oct2020Lu}.  
Moreover, the
age-of-information (AoI) \cite{S.2012Kaul,H.2020Yang}  that quantifies the time elapsed from the generation to observation of the information explicitly measures the information freshness that the conventional metrics (e.g., link delay) do not. 
Some recent works have shown that network protocol designs based on AoI provide better consumer service quality in  real-time status updating, such as Google Scholar \cite{Bastopcu2020Melih}. To enhance information freshness, the existing network protocols need to be re-examined and stochastic geometry analysis can be exploited in understanding AoI in large-scale systems. Characterization of the above-mentioned performance metrics by considering the spatial and temporal correlations introduced in this tutorial opens a broad research area. 

\subsection{Stochastic Geometry Analysis for Cross-layer Study}
 
The majority of the existing literature on stochastic geometry analysis assumes that the transmission demand is uniformly distributed and regulated by some medium access control (MAC) protocols such as Aloha or CSMA. However, in practice, the MAC-level transmission scheduling can be  affected by application-level service scheduling~\cite{Lin2006Xiaojun}. For example, in a mobile edge computing system where the mobiles can offload computation tasks to the edge, different requests can be scheduled to different edge servers based on the quality-of-service requirements of the tasks as well as the functionality and the status of the servers. In such a scenario, to analyze the communications performances for the offloaded tasks in a large-scale system, stochastic geometry can be used. For such an analysis, characterizing the point process corresponding to the interferers would be challenging due to the dynamic resource scheduling introduced by the cross-layer design. 
In general, modeling and analysis of cross-layer effects is an important research direction.
 


\subsection{Machine Learning Approaches for Large-Scale Networks and Stochastic Geometry Analysis}

Machine learning (ML) techniques have been increasingly adopted for resource allocation and control in wireless networks due to their capabilities in learning network environment variations (e.g., due to traffic pattern and channel uncertainties), classification of the relevant quantities, predicting future outcomes~\cite{R.2017Li}. For instance, deep reinforcement learning approaches are capable of estimating channel state information (CSI) without explicit feedback/detection and thus can be exploited to design resource management schemes with low  communication  and computation overhead compared to conventional CSI-based schemes~\cite{Sun2019Yaohua,Boutaba2018R}. With ML-based designs, network resource allocations will be more intelligent and environment-aware. Stochastic geometry tools can be used to analyze the effects of correlation introduced by environment-aware resource allocation in large-scale systems. Also, stochastic geometry can be used to study the impact of network interference as well as interference correlation on the performance of ML-based designs (e.g., in terms of stability, convergence, and accuracy).
 
\subsection{Stochastic Geometry Analysis of Emerging Network Scenarios } 
 
\subsubsection{Programmability of Radio Environment with Reconfigurable Intelligent Surfaces} A 
{\em Reconfigurable Intelligent Surface} (RIS)~\cite{D.May2019Renzo,S.Arxiv2020Gong}, also known as {\em software-defined hypersurface} and {\em large intelligent surface/antennas}, is a digitally-tunable metasurface. 
A metasurface is a thin and planar structure composed of sub-wavelength passive scattering particles which alter the electromagnetic waves through a surface impedance boundary condition.
The electromagnetic response of an RIS can be reprogrammed without re-fabrication. Due to this distinguishing feature, an RIS can be utilized to dynamically alter the electromagnetic behavior (e.g., amplitude, phase, and frequency \cite{E.2019Basar}) of the incident signals, and thus reshaping the wireless propagation environments. 

RIS has been envisioned as a candidate technology for future generation networks \cite{D.Arxiv2020Renzo} and expected to be widely deployed on, e.g., surfaces of buildings, vehicles and billboards, to facilitate the ever-increasing traffic demands. In such an environment, the signal distributions  
are jointly manipulated by multiple RISs. The incident signals at nearby RISs come from common sources and exhibit spatial correlations. Characterization of the aggregate effect of large-scale RIS deployments by taking into account the spatial correlation is an open research problem.

\subsubsection{Joint Radar and Communication}

Joint radar and communication (JRC) \cite{F.Tcom2020Liu} is a novel paradigm  that facilitates radar detection and data communication over the same frequency bands to improve the spectrum efficiency.  In JRC systems, a dual-functional transmitter (e.g., autonomous cars or unmanned drones) simultaneously transmits data and detects radar targets (e.g., barriers, mobile vehicles and pedestrians) based on a shared and integrated hardware platform. Specifically, a portion of the transmitted signal received at the target receiver is used for communication, while the other portion of the signal reflected from the radar targets at the dual-function transmitter is utilized for detection~\cite{C.Luongarxiv}. Due to the common signal source and interferers, the communication and radar performance in JRC systems are correlated.  
Furthermore, to accommodate both radar and communication demands, JRC brings a series of resource allocation problems, 
such as power allocation, spectrum allocation and sharing, and beamforming design.
Analyzing  radar and communication performance in large-scale systems based on stochastic geometry analysis under different resource allocation schemes  is a promising future research direction.
 
\subsubsection{Internet of Space Things}
 
The Internet of Space Things (IoST)~\cite{Ian2020F} is an expansion of the ground Internet of Things (IoT) to the aerial and space domains enabled by drones and miniaturized satellites/CubeSats~\cite{H.2001Heidt}, respectively.
With the holistic integration of the ground with aerial and space communication systems, IoST is expected to realize ubiquitous connectivity virtually to support a broad range of applications including monitoring, tracking, in-space backhauling~\cite{H.2017Kaushal}, wireless power transfer~\cite{X.2016Lu2,X.Lu2016} and remote healthcare solutions.
Performance characterization of IoST requires efforts of developing three-dimensional modeling incorporating disparate propagation characteristics, mobility patterns, resource allocation schemes of communication systems as well as their intricate correlations at different domains.  

\subsubsection{Terahertz Communications}
The terahertz band (i.e., 0.1-10 THz) provides ultra-wide spectrum resources and enable new 
applications for beyond 5G~\cite{T.2011Kleine-Ostmann}. For example,  terabits per second (Tbps) data rate can be realized to support ``fiber-like"  
communication performance 
which offers a seamless transition between optical fibers and THz wireless links with no latency.
Moreover, the micro-scale wavelength makes THz band more suitable for {\em in-vivo} nano-network communications than any other frequency bands because terahertz waves have strong penetrating force and can be absorbed by in-vivo substances such as liquids and organs~\cite{S.2020Ghafoor}.
The availability of new spectrum bands will necessitate novel resource allocation designs. An intriguing research direction is to develop stochastic geometry models that take into account the  correlation effects from: i) physical blockages, ii) direction of arrival/departure of high-gain beam steering, and iii)  temporal broadening
effect resulted from frequency selectivity in the ultra-wide THz band. 
  
\subsubsection{Internet of Nanothings (IoNT)}

Internet of Nanothings (IoNT) is a connected molecular system empowered by nanomachines, which are microscopic devices capable of performing high-precision functions for many real-world applications such as biomedicine, food industry and military operations~\cite{Ian2010}.
As IoNT is expected to work in the environments of fluids, gases or particulates, IoNT will foreseeably function in a manner that is drastically different from the IoT due to differences in propagation environments, the scale of deployment, as well as physical constraints (e.g., in energy and computations) of miniaturized nano devices. 
Therefore, it is imperative to develop novel mathematical models to characterize the features of IoNT. 
Moreover, the protocol stack design and analysis for large-scale IoNT still remains an open field for exploration.
 
\section{Summary}

Stochastic geometry tools can be used to develop analytical frameworks for large-scale wireless systems considering the effects of spatial-temporal interference correlation, which is generally ignored in traditional performance analyses. In this tutorial, we have presented a comprehensive spatial-temporal analysis of large-scale communications systems.
In particular, we have formulated models to characterize correlations in interference (and hence SIR) due to different effects such as distributions of the interferers, distribution of contact distance, shadowing, transmission buffer status (or network queues), multihop transmissions, retransmissions, and user mobility. The performance of a target link in a large-scale system has been demonstrated both analytically and numerically, considering the effect of spatial-temporal interference correlation in different scenarios.
In particular, we have derived the joint distribution of spatially-correlated SIR with multihop relaying, temporally-correlated SIR with retransmissions and spatially and temporally-correlated SIR with mobility. For each of the above scenarios, we have shed light on the technical challenges in stochastic geometry analysis.   
Finally, we have presented future research directions in stochastic geometry analysis  of emerging wireless communication systems the performance of which will be affected by spatial-temporal interference correlations.

\section*{Appendix}

\subsection*{A. Proof of Theorem~\ref{thm:SP_MCP_field} (Moments of the CSP$_\Phi$ in Random Fields of Interferers)}  

\begin{proof} 
Given $\Phi$, the CSP of a link with transmission distance $r_{\rm t}$ can be derived as   
\begin{align}  
 \bar{F}_{\eta|\Phi} ( \theta )   
&   = \mathbb{P} \Bigg[  h_{\rm t}   > \theta r_{\rm t}^{\alpha} \sum_{ j \in \mathbb{N} } h_{j} \|x_{j}\| ^{-\alpha}  \Bigm| \Phi  \Bigg] \nonumber \\
&   =  \mathbb{E}_{(h_{j})}  \bigg[  \exp \Big( - \theta r^{\alpha}_{\rm t}  \sum_{j \in \mathbb{N}} h_{j}   \|x_{j}\|  ^{-\alpha} \Big) \Bigm| \Phi \bigg] \nonumber \\
& =   \prod_{ j \in \mathbb{N} } \frac{1}{1+ \theta r^{\alpha}_{\rm t}   r_{j} ^{-\alpha} }  .   \label{eqn:spatial_SP1} 
\end{align}

The moments of the CSP given Mat\'{e}rn and Poisson fields of interferers can be obtained by applying their PGFLs as in (\ref{eqn:Moments_MCP_PPP_fields}), shown on the top of the next page,  
\begin{figure*}
\begin{align}
 \mathcal{M}_{P_{\rm s}} (b) &  = \mathbb{E} \Bigg[ \prod_{ j \in \mathbb{N} }  \bigg(  \frac{1}{1+ \theta \|x_{\rm t}\|^{\alpha}   \|r_{j}\|^{-\alpha} }  \bigg)^{\!b} \Bigg] \nonumber \\
& \overset{(a)}{=}  \! \begin{dcases} \exp \! \Bigg( \!\! - \! \lambda_{\rm p} \!  \int_{\mathbb{R}^2} \! \bigg[ 1 - \exp \! \bigg( \!  - \bar{c}  \bigg( 1 \!- \! \!\int_{\mathbb{R}^2} \!\! \bigg( \frac{1}{ 1 \! 
+   \theta r^{\alpha}_{\rm t}  \|x-y\|^{-\alpha}   } \! \bigg)^{\!b}   f^{\rm M}(y) \mathrm{d} y   \bigg) \! \bigg) \! \bigg] \mathrm{d} x  \Bigg),     \, \,\, \textup{MCP}    \\
 \exp \! \Bigg( \! - \! \lambda \!\! \int_{\mathbb{R}^2} \! \bigg( 1 - \bigg( \frac{1}{  1+ \theta r^{\alpha}_{\rm t}  \|x\|^{-\alpha}   } \! \bigg)^{\!b} \bigg) \mathrm{d} x  \Bigg), \hspace{50mm}   \textup{PPP}  ,
\end{dcases} 
\label{eqn:Moments_MCP_PPP_fields}
\end{align}  
\hrulefill
\end{figure*}
where $(a)$ applies the PGFLs of MCP and PPP given in (\ref{eqn:PGFL_MCP}) and (\ref{eqn:PPP_PGFL}), respectively. 
By inserting $f^{\rm M}(y)$ given in (\ref{PDF_MCP_daughter})  in (\ref{eqn:Moments_MCP_PPP_fields}), the final result with a Mat\'{e}rn cluster field of interferers in (\ref{eqn:moments_field}) directly yields.

We further extend the expression for the PPP in (\ref{eqn:Moments_MCP_PPP_fields}) as 
\begin{align} 
& \mathcal{M}_{P_{\rm s}} (b) \nonumber \\
& \overset{(b)}{=} \exp \! \Bigg( \!  \! -  \pi \lambda \delta \! \int^{\infty}_{0} \! \! \bigg(  1 - \!  \Big( 1 - \! \frac{\theta r^{\alpha}_{\rm t} }{  z+ \theta r^{\alpha}_{\rm t}   } \Big)^{\!b}     \bigg)  z^{\delta  - 1} \mathrm{d} z  \Bigg) \nonumber \\
& \overset{(c)}{=} \exp \! \bigg( \!\! - \pi \lambda \delta  \sum^{\infty}_{k=1} \binom{b}{k} (-1)^{k+1} \theta^{k} r^{\alpha k }_{\rm t}  \nonumber \\
& \hspace{30mm} \times \int^{\infty}_{0} \frac{z^{\delta-1}}{(z+\theta r^{\alpha}_{\rm t} )^{k}}\mathrm{d} z \bigg),  \label{eqn:M_PPP}
\end{align}
where $(b)$ employs the change of variable $z = r^2_{\rm t} \theta^{-\delta}  x^{2}$ and the substitution $\delta= \frac{2}{\alpha}$, $(c)$ takes the binomial expansion of  
 $\big( 1 - \frac{\theta r^{\alpha}_{\rm t} }{  z+ \theta r^{\alpha}_{\rm t}   } \big)^{b} $ for $b \in \mathbb{C}$.

The integral in (\ref{eqn:M_PPP}) can be transformed as follows~\cite[Eq. 3.196.2]{S.Gradshteyn2007} 
\begin{align} \label{eqn:JCP_integral}
&\int^{\infty}_{0} \!\! \frac{z^{\delta-1}}{(z+\theta r^{\alpha}_{\rm t} )^{k}} \mathrm{d}z \nonumber \\
& = \theta^{\delta-k} r^{ 2 - \alpha k}_{\rm t} \frac{ \Gamma(k-\delta) \Gamma( \delta) }{ \Gamma(k) } \nonumber \\
& \overset{(d)}{=} \theta^{\delta-k} r^{2 - \alpha k}_{\rm t}  \frac{(-1)^{k+1} \pi \csc(\pi \delta) \Gamma(\delta) }{ \Gamma(\delta-k+1) \Gamma(k) } \nonumber  \\  %
& = \theta^{\delta-k} r^{2 - \alpha k}_{\rm t} (-1)^{k+1} \pi \csc(\pi \delta) \frac{(\delta-1)\ldots (\delta-k+1)}{(k-1)!} \nonumber \\
& =  \theta^{\delta-k} r^{ 2 - \alpha k}_{\rm t} (-1)^{k+1}  \pi \csc(\pi \delta)  \binom{\delta-1}{k-1}, 
\end{align}
where $(d)$ follows as 
\begin{align}
\Gamma(k-\delta) = \frac{\pi \csc(\pi(k-\delta))}{\Gamma(\delta-k+1)} = \frac{(-1)^{k+1} \pi \csc(\pi \delta)}{\Gamma(\delta-k+1)}. \nonumber
\end{align}

Inserting (\ref{eqn:JCP_integral}) into (\ref{eqn:M_PPP}) yields, 
\begin{align}
 \mathcal{M}_{P_{\rm s}} (b)& \!=\! \exp \! \Bigg( \!\! -  \lambda  \theta^{\delta } r^{ 2 }_{\rm t} \pi \delta \csc(\pi \delta) \sum^{\infty}_{k=1} \binom{b}{k}   \binom{\delta\!-\!1}{k\!-\!1}  \!  \Bigg)  \nonumber \\
& \!\overset{(e)}{=}\! \exp \! \bigg( \!\! -   \lambda  \theta^{\delta } r^{ 2 }_{\rm t}   \frac{\Gamma(1-\delta)  \Gamma(b+\delta)}{\Gamma(b) }   \bigg) ,   \label{eqn:M_PS_PPP}
\end{align}
where $(e)$ follows as $\pi \delta \csc(\pi \delta)=\Gamma(1-\delta)\Gamma(\delta)$ and $\sum^{\infty}_{k=1} \binom{b}{k}   \binom{\delta-1}{k-1}= \frac{\Gamma(b+\delta)}{\Gamma(b)\Gamma(1+\delta)}$.

With a $\beta$-Ginibre field of interferers, the success probability can be obtained from (\ref{eqn:spatial_SP1}) as
\begin{align} 
\mathcal{M}_{P_{\rm s}}(b) & \overset{(e)}{=}  \mathbb{E}_{Q_{j}} \Bigg[  \prod_{ j \in \mathbb{N}   }   \bigg( \frac{\beta}{1 + \theta r^{\alpha }_{\rm t}   Q_{j}^{-\alpha/2}   }  + 1 - \beta \bigg)^{\!b} \Bigg]  \nonumber \\
& \overset{(f)}{=}      \prod_{ j \in \mathbb{N}   } \int^{\infty}_{0} \frac{ (\pi \lambda /\beta)^{j}}{\Gamma(j) } q^{j-1} e^{-\pi \lambda q/\beta} \nonumber \\
&  \hspace{15mm} \times \bigg( \frac{\beta}{1 + \theta   r^{\alpha}_{\rm t}   q  ^{-\alpha/2} } + 1 - \beta \bigg)^{\!b}  \mathrm{d} q   ,   \nonumber  
\end{align}  
where $(e)$ adopts the substitution $Q_{j}=r^2_{j}$ and $(f)$ follows from $Q_{j} \sim \mathcal{G}(j,\beta/\pi \lambda)$.

Summarizing the above results, we have the final results in Theorem~\ref{thm:SP_MCP_field}.
\end{proof}

\subsection*{B.  Proof of Theorem~\ref{thm:2D_SP_CCU} (Moments of the CSP$_\Phi$ of the Typical Cell-center User in a Poisson Downlink Network)}


\begin{proof}
Let $\Phi_{\rm c} \triangleq \{\Phi \mid o\in \mathcal{R}_{\rm c}\}$.
Given a Poisson point process $\Phi_{\rm c}$, the CSP of the typical cell-center user is
\begin{align}
 & \bar{F}^{\rm c}_{\eta|\Phi_{\rm c}} (\theta) \nonumber \\
  & = \mathbb{P} \big[ h > \theta \|x_{1}\|^{\alpha} I_{o} \mid 
  \Phi_{\rm c} \big] \nonumber \\
& = \mathbb{E} \bigg[ \exp \bigg( \! -  \theta  \| x_{1}\| ^{\alpha}   \sum^{\infty}_{ j=2 }    h_{j} \|x_{j}\|^{-\alpha} \bigg) \Bigm| 
\Phi_{\rm c} \bigg] \nonumber \\
& \overset{(a)}{=} \mathbb{E} \bigg[ \exp \bigg( \! -  \theta   \sum^{\infty}_{ j=2 }    h_{j} \varrho_{j}^{\alpha} \bigg) \Bigm| 
\Phi^{\rm R}_{\rm c} \bigg] \nonumber \\
& \overset{(b)}{=}  \prod^{\infty}_{\substack{ j= 2 \\   \varrho_{2} \leq \rho    }  }   \frac{1}{ 1 + \theta \varrho^{\alpha}_{j}  }  
, \label{eqn:CSP_CCU1}
\end{align}
where $(a)$ changes the PPP to its RDP, i.e., $\rho_{j}=\|x_{1}\|/\|x_{j}\|$, and the condition $\varrho_{2} \leq \rho$ in $(b)$ implies $o\in \mathcal{R}_{\rm c}$. 

Then, the $b$-th moments of the CSP$_\Phi$ in a Poisson downlink network can be derived as 
\begin{align}
 & \mathcal{M}^{\rm c}_{P_{\rm s}}(b) \nonumber \\ 
&  = \mathbb{E}    \Bigg[  \prod^{\infty}_{ j = 2   }  \bigg( \frac{1}{ 1 +  \theta \varrho_{j}^{\alpha} }  \bigg)^{\!b}   
    \bigm |     \varrho_{2}  \leq \rho   \Bigg]  \nonumber \\ 
  & = \mathbb{E}    \Bigg[  \prod^{\infty}_{ j = 2   }  \bigg( \frac{1}{ 1 +  \theta \rho^{\alpha} \big(\frac{\varrho_{j}}{\rho}\big)^{\alpha} }  \bigg)^{\!b}   
        \bigm |     \varrho_{2}  \leq \rho   \Bigg]  \nonumber \\ 
  & \overset{(b)}{=} \mathbb{E}    \Bigg[  \prod^{\infty}_{ j = 2   }  \bigg( \frac{1}{ 1 +  \theta \rho^{\alpha}  \varrho_{j} ^{\alpha} }  \bigg)^{\!b}      
            \Bigg]  \nonumber \\       
 & \overset{(c)}{=} \bigg( 1 + 2 \int^{1}_{0}  \bigg( 1 -  \Big( 1 -  \frac{\theta \rho^{\alpha} \varrho^{\alpha}}{ 1 + \theta \rho^{\alpha} \varrho^{\alpha} } \Big)^{b}  \bigg)   \varrho^{-3} \mathrm{d} \varrho \bigg)^{-1}     \nonumber \\
 & \overset{(d)}{=} \! \bigg(1  - 2  \sum^{\infty}_{k=1} \! {b \choose k} (- \theta \rho^{\alpha}  )^{k}  \nonumber \\
 & \hspace{20mm} \times \int^{1}_{0} \Big(\frac{\varrho^{\alpha}}{1 + \theta \rho^{\alpha} \varrho^{\alpha}} \Big)^{k} \varrho^{-3} \mathrm{d} \varrho  \bigg)^{-1} \nonumber \\
 & \overset{(e)}{=} \bigg(1 - 2  \sum^{\infty}_{k=1} {b \choose k}  \frac{ (- \theta \rho^{\alpha}  )^{k}}{\alpha}  \nonumber \\
  & \hspace{20mm} \times \int^{1}_{0} \Big(\frac{ x }{1 \!+\! \theta \rho^{\alpha} x } \Big)^{k} x^{-\frac{2}{\alpha}-1} \mathrm{d} x  \bigg)^{-1} \nonumber \\
&  \overset{(f)}{=}  \bigg(  1 -  \sum^{\infty}_{k=1} {b \choose k}  (- \theta \rho^{\alpha}  )^{k}\frac{\delta}{k   - \delta} \nonumber \\
& \hspace{10mm} \times {_2} F_{1} \big(k,k-\delta; k+1-\delta; -\theta \rho^{\alpha}\big) \bigg)^{-1},  
 \label{eqn:moments_CCU1}
\end{align} 
where $(b)$ holds due to the fact that 
the probability law of $\varrho_{2}$ conditioned on the probability
law of $\varrho_{2}, \varrho_{3}, \dots$ conditioned on $\varrho_{2} < \rho$  is the same as the law of $\varrho_{2},\varrho_{3},\dots$ without conditioning 
and $(c)$ applies the PGFL of the RDP of a PPP, $(d)$ takes the binomial expansion of  
$\big( 1 - \frac{\theta \rho^{\alpha} \varrho^{\alpha} }{ 1 + \theta \rho^{\alpha} \varrho^{\alpha} } \big)^{b} $ for $b \in \mathbb{C}$, $(e)$ applies the change of variable $x=\varrho^{\alpha}$, and $(f)$ substitutes $\frac{2}{\alpha}$ with $\delta$.
   
After some mathematical manipulations of (\ref{eqn:moments_CCU1}), we have the final result in Theorem~\ref{thm:2D_SP_CCU}. 
\end{proof}

\subsection*{C. Proof of Theorem~\ref{thm:vertex} (Moments of the CSP$_\Phi$ for the Typical Vertex User in a Poisson Downlink Network)}

\begin{proof}
Based on the stationarity of a PPP, we have the success probability of the typical vertex user conditioned on $\Phi$ as
\begin{align} 
&\bar{F}^{\rm v}_{\eta|\Phi}(\theta) \nonumber\\
& = \mathbb{E} \big[ e^{- \theta r^{\alpha}_{1} I_{o}} \mid r_{1}=r_{2}=r_{3} , \Phi \big] \nonumber \\ 
&  =  \mathbb{E}  \Bigg[ \exp \Bigg( -  \theta r^{\alpha}_{1} \bigg( h_{1} r^{-\alpha}_{2} + h_{2}r^{-\alpha}_{3}  \nonumber \\ 
& \hspace{40mm} +  \sum^{\infty}_{ j = 4 }  h_{j} r^{-\alpha}_{j} \bigg) \Bigg)  \Bigm|  \Phi \Bigg]  \nonumber \\ 
& = \big( 1 \! + \!  \theta  \big)^{\!-2 }      \prod^{\infty}_{ j = 4 }     \frac{1}{ 1 +  \theta r^{\alpha}_{1} r^{-\alpha}_{j} }     . \label{eqn:CCP_vertex}
\end{align}

With $\bar{F}^{\rm v}_{\eta|\Phi}(\theta)$ in (\ref{eqn:CCP_vertex}), the moments of the CSP$_\Phi$ can be derived as 
\begin{align}
&\mathcal{M}^{\rm v}_{P_{\rm s}}(b)  \nonumber \\ 
& =   \Big( 1   +    \theta  \Big)^{\! -2 b }   \mathbb{E}  \Bigg[   \prod^{\infty}_{ j = 4 }     \bigg( \frac{1}{1  +    \theta \|x_{1}\|^{\alpha} \|x_{j}\|^{-\alpha} }   \bigg)^{\!b}   \Bigg] \nonumber  \\ 
& \overset{(a)}{=} \Big( 1 \!+\!  \theta   \Big)^{-2 b}   \nonumber\\ 
& \hspace{2mm}
\times \mathbb{E}_{r_{1} } \! \Bigg[  \exp \! \bigg( \!\! -  2 \pi \lambda  \! \int^{\infty}_{r_{1}} \!\! \bigg ( 1 \!-\!  \bigg( \frac{1}{1 \!+\!  \theta r^{\alpha}_{1} x^{-\alpha}} \! \bigg)^{\!b}  \bigg) x \mathrm{d} x \bigg) \! \Bigg] \nonumber \\ 
& \overset{(b)}{=}  \Big( 1 \! +\! \theta \Big)^{\!-2 b} 2 \pi^2 \lambda^2 \!\!\int^{\infty}_{  0}\!\!   \exp(   - \lambda \pi r^2) r^3   \nonumber \\
&  \hspace{2mm} \times \exp \! \bigg( \!\! -  2 \pi  \lambda \! \int^{\infty}_{r} \!\! \bigg ( 1 \!-\!  \bigg( \frac{1}{1 \!+\!  \theta r^{\alpha}  x^{-\alpha} } \! \bigg)^{\!b}  \bigg) x \mathrm{d} x  \! \bigg)  \mathrm{d} r,  \label{M_CP_worstcase}
\end{align}
where $(a)$ applies 
\cite[Lemma 1]{J.1995Meche} and the PGFL of a PPP, and $(b)$ plugs in the PDF $f^{\rm v}_{r_{1}}(r)$ given in (\ref{eqn:PDF_VU_2D}).

After some mathematical manipulations of (\ref{M_CP_worstcase}), 
 we have the final results in Theorem~\ref{thm:vertex}.
\end{proof}

\subsection*{D.  Proof of Theorem~\ref{thm:shadowing} (Laplace Transform of Interference with Independent and Correlated Shadowing)}

\begin{proof}
We first derive the Laplace transform of the interference with correlated shadowing as follows. 
\begin{align}  
 &  \mathcal{L}_{I^{\rm Cor}_{o}  } (s) \nonumber \\
 & = \mathbb{E}  \big[ \exp \big( - s I^{\rm Cor}_{o} \big)    \big]   \nonumber \\
& =   \mathbb{E}  \bigg[ \exp \Big( - s \sum_{ j \in \mathbb{N}}   h_{j}  \ell(\|x_{j}\|)    S_{j} \Big)    \bigg]  \nonumber \\& \overset{(a)}{=} \mathbb{E}  \bigg[  \sum_{ k \in \mathbb{N}} \sum_{ x_{j} \in \mathbb{S}_{k} }  \frac{1}{ 1 + s    \ell(\|x_{j}\|)  T_{k}   }   \bigg]  \nonumber \\
& \overset{(b)}{=} \mathbb{E}_{(T_{k})} \! \Bigg[ \exp \! \Bigg( \! - \lambda  \sum_{k \in \mathbb{N}}    \int_{\mathbb{S}_{k} } \! \!  \bigg( 1  -       \frac{1}{1\!+\!s   \ell(\|x_{j}\|)   T_{k}  }    \bigg)     \mathrm{d} x \Bigg) \Bigg],  \nonumber 
\end{align}
where $(a)$ follows as $h_{j} \sim \mathcal{E}(1)$ and the points within the same cell share the same shadowing coefficient, and $(b)$ applies the PGFL of a PPP and holds as the accumulative interference of each shadowing cell $\mathbb{S}_{k}$ is averaged over the same random variable $T_{k}$ for $x_{j} \in \mathbb{S}_{k}$.

Similarly, the Laplace transform of $I^{\rm Ind}_{o}$ is given by 
\begin{align} 
& \mathcal{L}_{I^{\rm Ind}_{o}  } (s) \nonumber \\ &   = \mathbb{E}  \big[ \exp \big( - s I^{\rm Ind}_{o} \big)    \big]   \nonumber \\
&  =     \mathbb{E}  \bigg[ \exp \Big( - s \sum_{ j \in \mathbb{N}}  h_{j} \ell(\|x_{j}\|)   S_{j}  \Big)   \bigg]  \nonumber \\
& \overset{(c)}{=} \!   \exp  \! \Bigg( \! - \! \lambda  \sum_{k \in \mathbb{N}}   \int_{\mathbb{S}_{k}} \! \! \bigg(  1  -   \mathbb{E}_{T_{k}} \bigg[   \frac{1}{1\!+\!s    \ell(\|x_{j}\|)    T_{k} }   \bigg] \! \bigg)      \mathrm{d} x \! \Bigg), \nonumber 
\end{align} 
where $(c)$ follows as $S_{j}$ is an i.i.d mark for all $x_{j} \in \mathbb{S}_{k}$ 
with independent shadowing, with CDF $F_{k}$, as $T_{k}$.
\end{proof}

\subsection*{E. Proof of Theorem~\ref{thm:DSP} (Success Probability in a Poisson Downlink Network with Unsaturated Queues)}

\begin{proof}
To bypass the difficulty of modeling coupled queue statuses, we assume the typical receiver observes temporally-independent interference across different time slots. 
If each interfering BS is active with independent probability $p_{\rm A} = \mathbb{E}[\iota_{j}]$, the success probability can be derived as
\begin{align} 
P_{\rm s} &  = \mathbb{P} [ \eta > \theta ] \nonumber \\
& \overset{(a)}{\approx} \mathbb{E} \bigg[ \frac{ h_{1} \|x_{1}\|^{-\alpha} }{ \sum^{\infty}_{ j = 2 } \iota_{j}    h_{j} \|x_{j}\| ^{-\alpha} } > \theta  \bigg] \nonumber \\
& \overset{(b)}{=} \mathbb{E}  \Bigg[ \prod^{\infty}_{ j = 2 }    \bigg (  \frac{p_{\rm A}}{1 + \theta \|x_{1}\|^{\alpha}  \|x_{j}\|^{-\alpha} }  + 1 - \!p_{\rm A}  \bigg)  \Bigg]  \nonumber  \\ %
& \overset{(c)}{=} \mathbb{E}  \Bigg[ \prod^{\infty}_{ j = 2 }    \bigg (  \frac{p_{\rm A}}{1 + \theta \varrho_{j}^{\alpha} }  + 1 - p_{\rm A}  \bigg)  \Bigg] \nonumber \\
&   \overset{(d)}{=} \frac{1}{ 1 + 2 \int^{1}_{0}     \frac{p_{\rm A} \theta \varrho^{\alpha}}{ 1 + \theta \varrho^{\alpha} }    \varrho^{-3} \mathrm{d} \varrho }  \nonumber \\
&  \overset{(e)}{=}  \frac{1}{ 1 + \delta p_{\rm A} \theta   \int^{1}_{0} \frac{x}{1+ \theta x} x^{-\delta-1} \mathrm{d} x   }  \nonumber \\
& = \frac{1}{ 1 + \frac{ \delta  p_{\rm A} \theta}{1-\delta} {_2}F_{1} (1,1- \delta; 2 - \delta ;-\theta)  } \nonumber \\
&   \overset{(f)}{=} \frac{1}{ 1 - p_{\rm A} + p_{\rm A}   {_2}F_{1} (1, - \delta; 1 - \delta ; -\theta)  },
\label{eqn:CP_DL_queue}
\end{align}
where $(a)$ applies the independent interference assumption, $(b)$ follows as $h_{1}, h_{j} \sim \mathcal{E}(1)$, $(c)$ changes the PPP to its RDP, $(d)$ applies the PGFL of an RDP given in~(\ref{eqn:PGFL_RDP}), $(e)$ changes the variable $x= \varrho^\alpha$  and substitutes $\frac{2}{\alpha}$ with $\delta$. 

As can be seen from (\ref{eqn:CP_DL_queue}), $p_{\rm A}$ is the only unknown parameter.  To obtain $p_{\rm A}$, we need to analyze the probability that each BS is active from the perspective of queue dynamics.
As a BS sends out a packet once its selected queue is not empty, the active probability of the BS is equivalent to the utilization of its queue.
Since the service rate of a BS is equivalent to the success probability, with random scheduling, 
we have the service rate of the typical queue at a BS conditioned on $\Phi_{\rm B}$ as
\begin{align} 
\mu = \frac{1}{N_{\rm u}}  \mathbb{P} [ \eta > \theta \mid \Phi_{\rm B}],
\end{align} 
where $N_{\rm u}$ denotes the number of users associated with the typical BS.

According to the law of total probability, the mean service rate of any user's queue at its serving BS is calculated as
\begin{align}
\bar{\mu} &  = \sum^{\infty}_{n=1} p_{ N_{\rm u} }(n)  \mu   \nonumber \\
& =   \sum^{\infty}_{n=1} p_{ N_{\rm u} }(n)    \frac{ \mathbb{P} [\eta > \theta] }{n}, 
\label{eqn:mean_service_rate}
\end{align}
where $ \mathbb{P} [\eta > \theta]$ has been obtained in (\ref{eqn:CP_DL_queue}), $p_{ N_{\rm u} }(n)$ denotes the PDF of $N_{\rm u}$. In a Poisson downlink network, $p_{ N_{\rm u} }(n)$ is approximated as \cite{L.2018Wang}
\begin{align} 
p_{ N_{\rm u} }(n) \approx \frac{\nu^{\nu} \Gamma(n+\nu) (\lambda_{\rm u}/\lambda_{\rm B})^n}{n!\Gamma(\nu)(\lambda_{\rm u}/\lambda_{\rm B} + \nu)^{n+\nu}}, \nonumber 
\end{align}
where $\nu=3.5$. 

From  (\ref{eqn:CP_DL_queue}) and (\ref{eqn:mean_service_rate}), the mean service rate $\bar{\mu}$ is a function of active probability $p_{\rm A}$.
Based on the random scheduling, the active probability of a BS equals the utilization factor of the queue for the typical general user. Thus, we can establish a fixed-point equation as follows:
\begin{align}
p_{\rm A} = \frac{\xi_{\rm u} }{\bar{\mu}} =  \frac{  \xi_{\rm u} }{\sum^{\infty}_{n=1} p_{N_{\rm u}} (n)  P_{\rm s} /n }.  \nonumber
\end{align}  
  
Subsequently, we have the success probability of a Poisson downlink network with unsaturated queues given in Theorem \ref{thm:DSP}.
\end{proof}

\subsection*{F. Proof of Theorem \ref{thm:SP_MCP_field} (Success Probability in a Poisson Bipolar Network) }

\begin{proof}
Using the assumption that each transmitter is active with independent probability, the success probability in Poisson bipolar networks is given by 
\begin{align}  
  P_{\rm s} 
& \overset{(a)}{\approx} \mathbb{E}  \bigg[ \frac{ h_{\rm t} \|x_{\rm t}\|^{-\alpha} }{ \sum_{j \in \mathbb{N} \backslash \{ \varsigma(x_{\rm t}) \} }  \iota_{j}   h_{j} \| x_{j} \|^{-\alpha} } > \theta  \bigg] \nonumber \\ 
&   \overset{(b)}{=}  \mathbb{E} \bigg[   \prod_{ j \in \mathbb{N} } \bigg( \frac{p_{\rm A}}{1+ \theta \|x_{\rm t}\|^{\alpha} \| x_{j} \|^{-\alpha} } + 1 - p_{\rm A} \bigg)   \bigg] \nonumber \\
&  \overset{(c)}{=}  \exp \bigg(  \! \! -  2 \pi \lambda p_{\rm A}    \int^{\infty}_{0} \! \!    \frac{ \theta r^{\alpha}_{\rm t} x^{-\alpha+1}  }{1 + \theta r^{\alpha}_{\rm t} x^{-\alpha}  }      \mathrm{d} x    \bigg)  \nonumber \\ 
&  \overset{(d)}{=}  \exp \bigg(  \! \! -    \pi \lambda p_{\rm A} \delta    \int^{\infty}_{0} \! \!    \frac{ \theta r^{\alpha}_{\rm t} u^{ \frac{2}{\alpha}-1}  }{ u + \theta r^{\alpha}_{\rm t}    }      \mathrm{d} u   \bigg)  \nonumber \\  
& \overset{(e)}{=} \exp \Big(  - p_{\rm A} \lambda \pi r^{2}_{\rm t} \theta^{\delta} \Gamma(1 +   \delta )  \Gamma(1 -  \delta) \Big), \label{eqn:CP_queue}
\end{align}
where $(a)$ applies the independent interference assumption, $(b)$ holds 
due to Slivnyak's theorem,   $(c)$ applies the PGFL of a PPP and employs the substitution $\delta=2/\alpha$, $(d)$ applies the change of varaible $u=x^{\alpha}$, and $(e)$ follows from~\cite[Eq. 3.196.2]{S.Gradshteyn2007} that $\int^{\infty}_{0}  \frac{u^{\delta-1}}{u+A} \mathrm{d} u = A^{\delta} \Gamma(1-\delta) \Gamma(\delta)$.    


Since the service rate of the typical transmitter is equivalent to the success probability of a time slot, $p_{\rm A}$ can be written as $\min\{ \frac{\xi}{ P_{\rm s} }, 1\}$ which is a well-known result from Geo/Geo/1 queue~\cite{I.2004Atencia}. Subsequently, we can establish the following equation.
\begin{align}
& \min \big\{ \xi /  P_{\rm s}, 1 \big\}\!  \nonumber \\ 
& = \min \!\bigg\{ \!\frac{\xi}{ \exp \big( - p_{\rm A} \lambda \pi r^2_{\rm t} \theta^{\delta}    \Gamma(1\!+\!\delta) \Gamma(1\!-\!\delta) \big) }, 1  \bigg\} \nonumber\\
&= p_{\rm A}. \label{eqn:equation_PA}
\end{align}

Since $P_{\rm s}$ is a function of $P_{A}$ as shown in (\ref{eqn:CP_queue}), solving the equation in (\ref{eqn:equation_PA}) yields 
\begin{align}
\!\!\! p_{\rm A}= \min \!\bigg\{ \! - \frac{ \mathcal{W} \big(\! - \xi \lambda \pi r^2_{\rm t}  \theta^{\delta}    \Gamma(1\!+\!\delta) \Gamma(1\!-\! \delta ) \big) }{ \lambda \pi r^2_{\rm t} \theta^{\delta}  \Gamma(1+\delta) \Gamma(1-\delta) }  , 1 \bigg\} \label{P_A},
\end{align}
where $\mathcal{W}$ denotes the Lambert-$\mathcal{W}$ function \cite{M.1997Corless}. 

 By inserting $p_{\rm A}$ in (\ref{P_A}) into the expression of $P_{\rm s}=\frac{\xi}{p_{\rm A}}$, we have the final results presented in Theorem~\ref{thm:buffer_bipolar}. 
\end{proof}

\subsection*{G. Proof of Theorem \ref{thm:M_ESP} (Moments of the End-to-End CSP$_\Phi$)}

\begin{proof}
From (\ref{def:SP_relaying}), the end-to-end CSP$_{\Phi}$ of $M$-hop relaying given a Poisson field of interferers can be expressed as 
\begin{align} 
&\bar{F}_{\eta|\Phi }(\theta)  \nonumber \\
& \! =  \!  \mathbb{P} [ h_{1} \! > \! \theta d^{\alpha}_{1}  I_{1},  h_{2} \! > \! \theta  d^{\alpha}_{2} I_{2}, \ldots \!, h_{M} \! > \! \theta  d^{\alpha}_{M} I_{M} \mid \Phi_{m}  ]    \nonumber\\
& \! = \!   \mathbb{E}_{(h_{j,m})}   \Bigg[ \! \exp \! \bigg( \! - \theta  \! \sum^{M}_{m=1}  \sum_{ j \in \mathbb{N} }  \! \frac{  h_{j,m} d^{\alpha}_{m} }{  \| x_{j,m} \! - \! y_{m} \|^{\alpha} } \! \bigg) \Bigm|  \Phi_{m}  \Bigg] ,   \nonumber \\ 
& \! = \!    \prod_{ j \in \mathbb{N} }  \prod^{M}_{m=1}  \frac{1}{ 1 + \theta d^{\alpha}_{m} \|   x_{j,m} - z_{m} \|^{-\alpha}  }    .  \label{eqn:ESP_PPP} 
\end{align} 

With the end-to-end success probability with a Poisson field of interferers in (\ref{eqn:ESP_PPP}), we further compute the $b$-th moment of it as in (\ref{eqn:M_SP_Poifield}), 
where $(a)$ follows as $h_{j,m}\sim \mathcal{E}(1)$ and $(b)$ applies the PGFL of a PPP.    
\begin{figure*} \vspace{-5mm}
\begin{align}
 \mathcal{M}^{\textup{Poi}}_{P_{\rm s}}(b,M)  
& \overset{(a)}{=}    \mathbb{E}_{  \Phi }  \Bigg[  \prod_{ j \in \mathbb{N} }   \prod^{M}_{m=1}  \bigg(  \frac{1}{  1 + \theta d^{\alpha}_{m} \|   x_{j,m}   - z_{m} \|^{-\alpha} } \bigg)^{\!b}   \Bigg ] , \hspace{10mm}  
\nonumber \\   
& \overset{(b)}{=}  \begin{dcases}  \exp \Bigg( \!\! - \lambda \int_{\mathbb{R}^2} \bigg(  1  -   \prod^{M}_{m=1} \bigg(  \frac{1}{ 1 + \theta d^{\alpha}_{m} \|   x   - z_{m} \|^{-\alpha}  } \bigg)^{\!b} \bigg)   \mathrm{d} {  x}  \Bigg) ,   \hspace{10mm} \, \mathrm{QSI}  \\
\exp \Bigg( \!\! - \lambda  \prod^{M}_{m=1} \int_{\mathbb{R}^2}  \bigg(  1  -   \bigg( \frac{1}{  1 + \theta d^{\alpha}_{m} \|  x   - z_{m} \|^{-\alpha}  } \bigg)^{\!b} \bigg)   \mathrm{d} {  x}  \Bigg), \hspace{10mm}  \mathrm{FVI}  
\end{dcases}
\label{eqn:M_SP_Poifield}
\end{align}
\hrulefill 
\end{figure*} 
\end{proof}

\subsection*{H. Proof of Theorem~\ref{thm:Retx_JSP} (JSP of Multiple Transmissions)}

\begin{proof}
The JSP for $K$ transmissions is given as 
\begin{align}  
 \mathcal{J}_{K}   & \! =  \mathbb{P} \bigg[ \bigcap^{K}_{k=1} \Big \{ h^{(k)} > \theta r^{\alpha}_{\rm t} I^{(k)}_{o} \Big \}   \bigg]
\nonumber\\ 
& \overset{\text{(a)}}{=} \mathbb{E}   \bigg[  \prod^{K}_{k=1} \exp \Big(  - \theta r^{ \alpha }_{\rm t}     I^{(k)}_{o}       \Big) \bigg]  \nonumber \\   
& =  \mathbb{E} 
\bigg[   \prod^{K}_{k=1}  \exp  \bigg( \! - \theta r^{ \alpha }_{\rm t}    \sum_{j \in \mathbb{N} } h^{(k)}_{j} \|x^{(k)}_{j}\|^{-\alpha}  \bigg) \bigg]  , 
\end{align} 
where $(a)$ follows as $h^{(k)}\sim \mathcal{E}(1)$ is independent across different time slots.
 
In the scenarios with QSI, 
the JSP can be expressed as 
 \begin{align}
  \mathcal{J}^{\textup{QSI}}_{K}   & \!\overset{\text{(b)}}{=}\!   \mathbb{E}  \Bigg[   \prod_{ j \in \mathbb{N} }   \bigg( \frac{1}{ \! 1 + \theta r^{ \alpha }_{\rm t}  \|x^{(k)}_{j}\|^{-\alpha}  }   \bigg)^{\!\!K}    \Bigg]  \nonumber \\
 & \!\overset{\text{(c)}}{=}  \!  \exp \Bigg(\! - 2 \pi \lambda \! \int^{\infty}_{0} \! \! \bigg( \! 1 \!- \! \bigg( \frac{1}{  1+ \theta r^{\alpha}_{\rm t} x^{-\alpha}     } \bigg)^{\!\!K}  \!   \bigg)   x \mathrm{d} x  \! \Bigg),   \label{eqn:JCP_1}
 \end{align} 
 where $(b)$ follows as 
 $K$ transmissions are subject to the same point process of interferers, and $(c)$ applies the PGFL of a PPP. As~(\ref{eqn:JCP_1}) is equivalent to the moments of the CSP$_\Phi$ in Poisson ad hoc networks in (\ref{eqn:Moments_MCP_PPP_fields}) with $b$ replaced by $K$, we have the final results of $\mathcal{J}^{\textup{QSI}}_{K}$ in (\ref{eqn:J_QSI}). 
 
 Moreover, in the scenarios with FVI, as different transmissions are affected by independent point processes, the JSP can be derived as  
  \begin{align} 
  \mathcal{J}^{\textup{FVI}}_{K} & =   \prod^{K}_{k=1} 
   \mathbb{E}_{  h^{(k)}_{j}}   \Bigg[  \exp  \bigg( \! - \theta r^{ \alpha    }_{\rm t}   \!   \sum_{ j \in \mathbb{N} }   h^{(k)}_{j}   \|x^{(k)}_{j}\|^{-\alpha}    \bigg) \Bigg]  \nonumber \\
  & \overset{(d)}{=}  \prod^{K}_{k=1}  \mathbb{E}    \Bigg[    \prod_{ j \in \mathbb{N} }     \frac{1}{ \! 1 + \theta r^{ \alpha  }_{\rm t}       \|x^{(k)}_{j}\|^{-\alpha}  }    \Bigg] \nonumber \\
  & \overset{(e)}{=}   \exp \Bigg( \!\!- 2 \pi \lambda \int^{\infty}_{0}  \bigg(  1 -   \frac{1}{  1+ \theta r^{\alpha}_{\rm t} x^{-\alpha}     }      \bigg)  x \mathrm{d} x  \Bigg)^{\!\!K}  \nonumber \\
  & \overset{(f)}{=}  \exp \Big( -  c \lambda     r^2_{\rm t} \theta^{\delta}  K  \Big),   \nonumber
  \end{align}  
   where $(d)$ holds as $h^{(k)}_{j} \sim \mathcal{E}(1)$, $(e)$ follows as the point process of interferers across different transmissions are i.i.d. and applies the PGFL of a PPP, and $(f)$ follows the same steps of the derivations of 
   (\ref{eqn:CP_queue}). 
\end{proof}

\subsection*{I. Proof of Corollary~\ref{cor:typeII} (Success Probability of Type-II HARQ)}

\begin{proof}
When $K=2$, the success probability with Type-II HARQ-CC given in (\ref{eqn:Type_II}) can be rewritten as
\begin{align} 
\mathcal{P}_{\mathrm{II}} & = \mathbb{P} \big[ \eta^{(1)} > \theta  \big] + \mathbb{P} \big[ \eta^{(1)} + \eta^{(2)} > \theta , \eta^{(1)} < \theta  \big] \nonumber \\ 
&  =  \mathbb{P} \big[ \eta^{(1)} > \theta  \big]  \nonumber \\
& \hspace{10mm} +   \mathbb{E} \Big[ \mathbb{P}  \big[ \eta^{(2)}> \theta - \eta^{(1)} \mid  \eta^{(1)} \big]    \mathbbm1 _{ \{ \eta^{(1)} < \theta  \} }  \Big]  \nonumber  \\
& = \mathcal{J}_{1} +     \int^{T}_{0} \mathbb{P}  \big[ \eta^{(2)} > \theta - u    \big]  f_{\eta^{(1)}  } ( u )   \mathrm{d} u \Big] , \label{eqn:CP_typeII}  
\end{align}
where $\mathcal{J}_{1}$ is given in Theorem~\ref{thm:Retx_JSP} and $f_{\eta^{(1)}  } ( t )$ represents the PDF of $\eta^{(1)}$.

Then, we derive the conditional CDF of $\eta^{(1)}$ given $\Phi^{(1)}$ as follows 
\begin{align}
\mathbb{P} \big[ \eta^{(1)} \leq  t   \big] 
& =   \mathbb{E} \bigg[ 1 - \exp \Big( \! - u r^{\alpha}_{\rm t}  \sum_{ j \in \mathbb{N} } h^{(1)}_{j} \|x_{j} \|^{-\alpha}   \Big) \bigg]  \nonumber \\
& =   1 - \prod_{ j \in \mathbb{N} }    \frac{1}{ 1+ u r^{\alpha}_{\rm t}    r_{j}  ^{-\alpha}} . \label{eqn:CDF_eta}
\end{align} 

Subsequently, the PDF of $\eta^{(1)}$ can be derived by taking the derivative of (\ref{eqn:CDF_eta}) with respective to $u$ as follows. 
\begin{align} 
 f_{\eta^{(1)}} (u) 
 & = \frac{\partial \Big(  1 - \prod_{j \in \mathbb{N}  }    \frac{1}{ 1+ u r^{\alpha}_{\rm t}  r_{j}^{-\alpha}}  \Big) }{ \partial u } \nonumber \\
& = \sum_{ j \in \mathbb{N}  } \frac{ r^{\alpha}_{\rm t} r_{j}^{-\alpha}   }{ ( 1+ u r^{\alpha}_{\rm t}   r_{j}^{-\alpha} )^{2} }  \prod^{i \neq j }_{ i \in \mathbb{N} }   \frac{1}{ 1+ u r^{\alpha}_{\rm t}  r_{i}^{-\alpha}}. \label{eqn:PDF_eta}
\end{align}

By plugging (\ref{eqn:PDF_eta}) into the second term of (\ref{eqn:CP_typeII}), we have (\ref{eqn:secondterm_CP_typeII}), 
\begin{figure*}
\vspace{-3mm}
\begin{align} 
&  \int^{T}_{0} \mathbb{P}  \big[ \eta^{(2)} > \theta - u    \big]  f_{\eta^{(1)}  } ( u )   \mathrm{d} u \Big]   \nonumber \\
& = \! \begin{dcases}   \int^{T}_{0} \! \mathbb{E} \Bigg[  \sum_{ j \in \mathbb{N}  } \frac{ r^{\alpha}_{\rm t} r_{j}^{-\alpha}    }{ ( 1+ u r^{\alpha}_{\rm t}   r_{j}^{-\alpha} )^{2} } 
\frac{1}{  1 +  ( \theta - u ) r^{\alpha}_{\rm t}  r_{j}^{-\alpha} }    \prod^{i \neq j}_{ i  \in \mathbb{N}  }  \frac{1}{ 1+ u r^{\alpha}_{\rm t}    r_{i}^{-\alpha}}   \frac{1}{  1 +  ( \theta - u ) r^{\alpha}_{\rm t}  r_{i}^{-\alpha} }   \Bigg]   \mathrm{d} u , \hspace{15mm} \textup{QSI}   \\
 \int^{T}_{0} \! \mathbb{E} \Bigg[  \prod_{ j^{\prime} \in  \mathbb{N} }      \frac{1}{  1 +  ( \theta - u ) r^{\alpha}_{\rm t}  r_{j^{\prime}} ^{-\alpha} } \sum_{ j \in \mathbb{N}  } \frac{ r^{\alpha}_{\rm t} r^{-\alpha}_{j}   }{ ( 1+ u r^{\alpha}_{\rm t}   r_{j}^{-\alpha} )^{2} }  \prod^{i \neq j}_{ i \in \mathbb{N}  }   \frac{1}{ 1+ u r^{\alpha}_{\rm t}  r_{i}^{-\alpha}}  \Bigg]   \mathrm{d} u  ,   \hspace{35mm} \textup{FVI}
 \end{dcases}  \nonumber \\
  & \overset{(a)}{=} \! \begin{dcases} \!2 \pi \lambda    \!  \int^{T}_{0} \! \! \int^{\infty}_{0} \!\!  \frac{ r^{\alpha}_{\rm t} r^{-\alpha}   }{ ( 1+ u r^{\alpha}_{\rm t}   r^{-\alpha} )^{2} }     \frac{1}{  1 +  ( \theta - u ) r^{\alpha}_{\rm t}  r^{-\alpha}  }  r \mathrm{d} r   \mathbb{E} \Bigg[   \prod^{i \neq j}_{ i  \in  \mathbb{N} }   \frac{1}{  1 +  u r^{\alpha}_{\rm t}  r_{i}^{-\alpha} }    \frac{1}{  1 +  ( \theta - u ) r^{\alpha}_{\rm t}  r_{i}^{-\alpha} } \Bigg]  \mathrm{d} u, \hspace{3mm} \textup{QSI}  \\
  \!2 \pi \lambda    \!  \int^{T}_{0} \! \! \mathbb{E} \Bigg[  \prod_{j^{\prime} \in  \mathbb{N}  }  \frac{1}{  1 +  ( \theta - u ) r^{\alpha}_{\rm t}  x_{j^{\prime}}^{-\alpha}  } \Bigg] \int^{\infty}_{0} \!\!  \frac{ r^{\alpha}_{\rm t} r^{-\alpha}   }{ ( 1+ u r^{\alpha}_{\rm t}   r^{-\alpha} )^{2} }   r \mathrm{d} r   \mathbb{E} \Bigg[ \prod_{ i \in  \mathbb{N}}    \frac{1}{  1 +  u r^{\alpha}_{\rm t}  r_{i}^{-\alpha} }   \Bigg]   \mathrm{d} u, \hspace{16mm} \textup{FVI} ,
  \end{dcases} \label{eqn:secondterm_CP_typeII}
 \end{align}
   \hrulefill
 \end{figure*}
 where $(a)$ applies the Campbell-Mecke formula given in~(\ref{eqn:Campbell-Mecke}). 
 
 Subsequently, by computing the expectations in (\ref{eqn:secondterm_CP_typeII}) based on the PGFL of the PPP, we have the final results in {\bf Corollary}~\ref{cor:typeII}.
\end{proof}

\bibliographystyle{IEEEtran}

\begin{IEEEbiography}   
[
{
\includegraphics[width=1in,height=1.25in,clip,keepaspectratio]{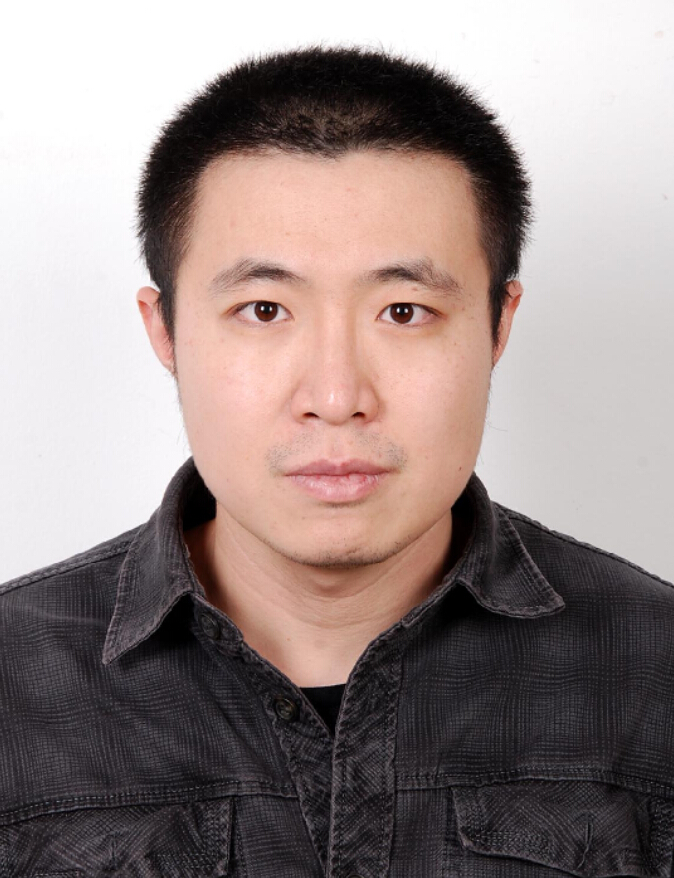}
}
]
{Xiao Lu} received the Ph.D degree in the University of Alberta, Canada, the M.Eng. degree in computer engineering from Nanyang Technological University, and the B.Eng. degree in communication engineering from Beijing University of
Posts and Telecommunications. His current research interests include design, analysis, and optimization of future generation cellular wireless networks.
\end{IEEEbiography} 

\begin{IEEEbiography}   
[
{
\includegraphics[width=1in,height=1.25in,clip,keepaspectratio]{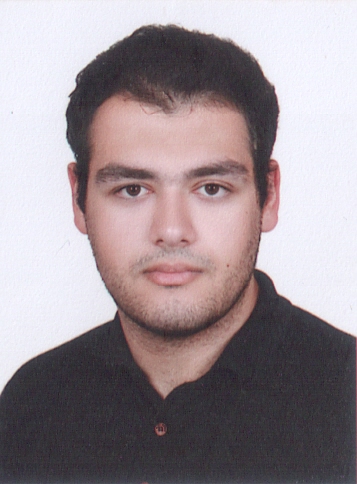}
}
]
{Mohammad Salehi}   received the B.Sc. degree in electrical engineering from K. N. Toosi University of Technology, Tehran, Iran, in 2014; the M.Sc. degree in electrical engineering from Amirkabir University of Technology, Tehran, Iran, in 2017; and the Ph.D. degree in electrical engineering from the University of Manitoba, Winnipeg, Canada, in 2021. His research interests include modeling and analyzing wireless networks.  
\end{IEEEbiography} 

\begin{IEEEbiography}   
[
{
\includegraphics[width=1in,height=1.25in,clip,keepaspectratio]{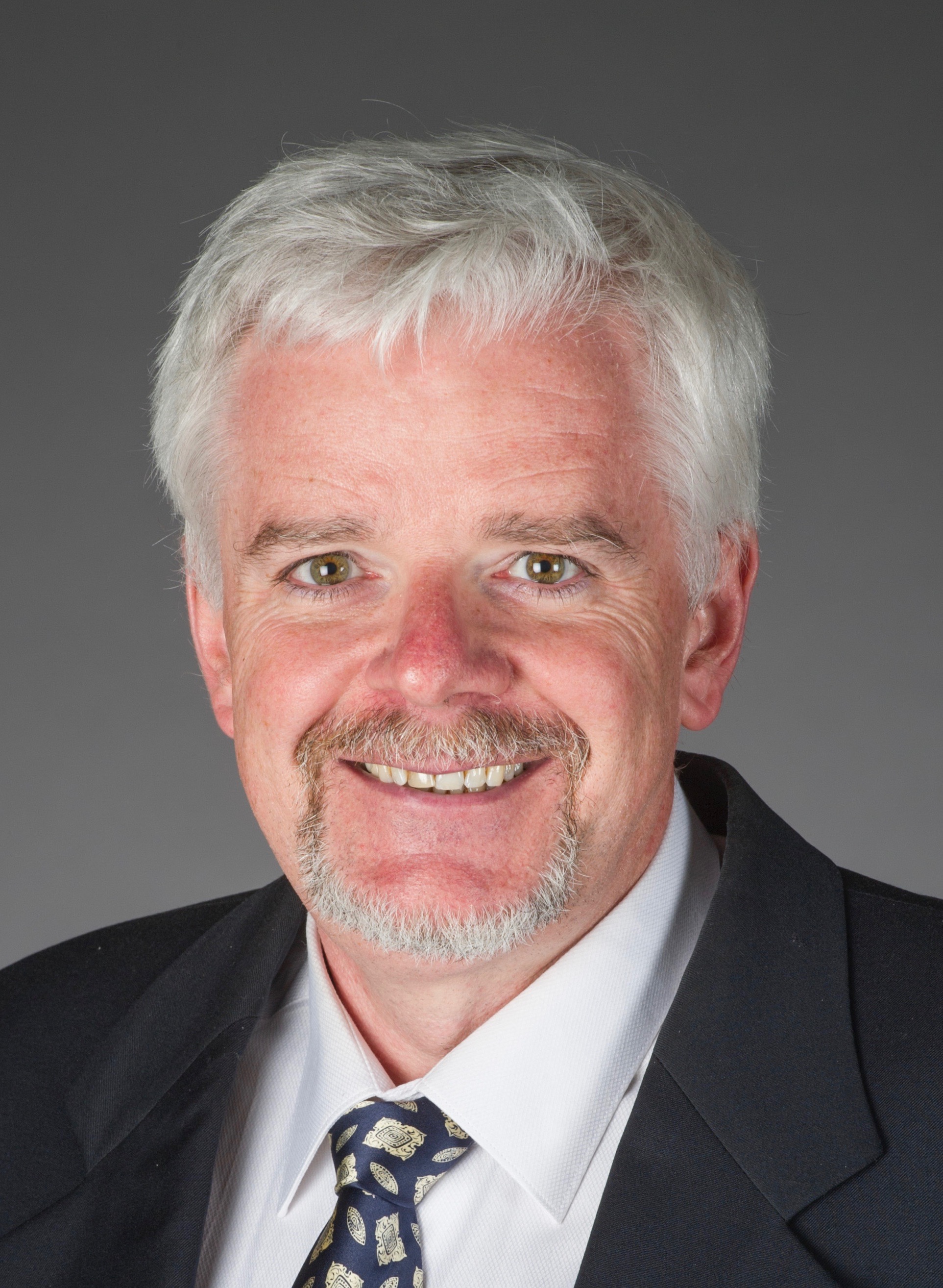}
}
]
{Martin Haenggi} (S’95-M’99-SM’04-F’14) received the Dipl.-Ing. (M.Sc.) and Dr.sc.techn. (Ph.D.) degrees in electrical engineering from the Swiss Federal Institute of Technology in Zurich (ETHZ) in 1995 and 1999, respectively. Currently he is the Freimann Professor of Electrical Engineering and a Concurrent Professor of Applied and Computational Mathematics and Statistics at the University of Notre Dame, Indiana, USA. In 2007-2008, he was a Visiting Professor at the University of California at San Diego, in 2014-2015 he was an Invited Professor at EPFL, Switzerland, and in 2021-2022 he is a Guest Professor at ETHZ.
He is a co-author of the monographs "Interference in Large Wireless Networks" (NOW Publishers, 2009) and “Stochastic Geometry Analysis of Cellular Networks” (Cambridge University Press, 2018) and the author of the textbook "Stochastic Geometry for Wireless Networks" (Cambridge, 2012) and the blog stogblog.net, and he published 18 single-author journal articles. His scientific interests lie in networking and wireless communications, with an emphasis on cellular, amorphous, ad hoc (including D2D and M2M), cognitive, vehicular, and wirelessly powered networks.
He served as an Associate Editor for the Elsevier Journal of Ad Hoc Networks, the IEEE Transactions on Mobile Computing (TMC), the ACM Transactions on Sensor Networks, as a Guest Editor for the IEEE Journal on Selected Areas in Communications, the IEEE Transactions on Vehicular Technology, and the EURASIP Journal on Wireless Communications and Networking, as a Steering Committee member of the TMC, and as the Chair of the Executive Editorial Committee of the IEEE Transactions on Wireless Communications (TWC). From 2017 to 2018, he was the Editor-in-Chief of the TWC. Currently he is an editor for MDPI Information. 
For both his M.Sc. and Ph.D. theses, he was awarded the ETH medal. He also received a CAREER award from the U.S. National Science Foundation in 2005 and three paper awards from the IEEE Communications Society, the 2010 Best Tutorial Paper award, the 2017 Stephen O. Rice Prize paper award, and the 2017 Best Survey paper award, and he is a Clarivate Analytics Highly Cited Researcher.
  
\end{IEEEbiography}

\begin{IEEEbiography} 
[
{
\includegraphics[width=1in,height=1.25in,clip,keepaspectratio]{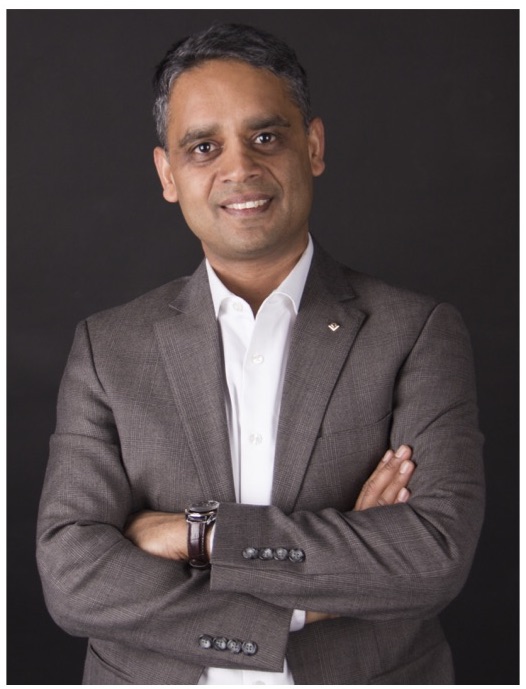}
}
]
{Ekram Hossain} (F’15) is a Professor and Associate Head (Graduate Studies) in the Department of Electrical and Computer Engineering at University of Manitoba, Canada (http://home.cc.umanitoba.ca/$\sim$hossaina). He is a Member (Class of 2016) of the College of the Royal Society of Canada, a Fellow of the Canadian Academy of Engineering, and a Fellow of the Engineering Institute of Canada. Dr. Hossain’s current research interests include design, analysis, and optimization of wireless networks with emphasis on beyond 5G cellular networks. He was elevated to an IEEE Fellow ``for contributions to spectrum management and resource allocation in cognitive and cellular radio networks". He received the 2017 IEEE ComSoc TCGCC (Technical Committee on Green Communications \& Computing) Distinguished Technical Achievement Recognition Award ``for outstanding technical leadership and achievement in green wireless communications and networking". He was listed as a Clarivate Analytics Highly Cited Researcher in Computer Science in 2017, 2018, 2019, and 2020. Currently he serves as the Editor-in-Chief of IEEE Press (2018-2021) and the Director of Magazines (2020-2021) for the IEEE Communications Society. Previously he served as the Editor-in-Chief for the IEEE Communications Surveys and Tutorials (2012–2016).
\end{IEEEbiography}

\begin{IEEEbiography}   
[
{
\includegraphics[width=1in,height=1.25in,clip,keepaspectratio]{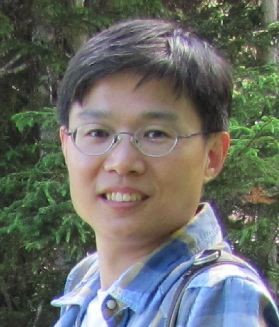}
}
]
  {Hai Jiang} (SM’15) received the B.Sc. and M.Sc. degrees in electronics engineering from Peking University, Beijing, China, in 1995 and 1998, respectively, and the Ph.D. degree in electrical engineering from the University of Waterloo, Waterloo, ON, Canada, in 2006. Since 2007, he has been a Faculty Member with the University of Alberta, Edmonton, AB, Canada, where he is currently a Professor with the Department of Electrical and Computer Engineering. His research interests include radio resource management, cognitive radio networking, mobile edge computing, and cooperative communications.
\end{IEEEbiography} 
\end{document}